\documentclass[letterpaper, 12pt]{article}
\usepackage{style_new_designs}

\usepackage{enumitem}
\usepackage[colorlinks=true, allcolors=blue]{hyperref}
\usepackage{setspace} 
\newcommand{\mr}{\mathbb{R}}
\newcommand{\wh}{\widehat}
\newcommand{\op}{o_p}
\newcommand{\Op}{O_p}
\newcommand{\one}{\mathds{1}}
\newcommand{\inv}{^{-1}}

\newcommand{\half}{^{1/2}}
\newcommand{\onehalf}{\frac{1}{2}}
\newcommand{\sub}{\subseteq}
\newcommand{\rootn}{\sqrt{n}}

\newcommand{\negrootn}{n^{-1/2}}

\newcommand{\convp}{\overset{p}{\to}}

\newcommand{\eqdist}{\overset{d}{=}}
\newcommand{\convw}{\Rightarrow}
\newcommand{\convwprocess}{\Rightarrow}

\newcommand{\normal}{\mathcal{N}}

\newcommand{\mc}{\mathcal}
\newcommand{\indep}{\mathrel{\perp \!\!\! \perp}}
\newcommand{\est}{\wh{\theta}}
\newcommand{\estg}{\wh{\theta}_{\group}}
\newcommand{\estgone}[1]{\wh{\theta}_{\group_1(#1)}}
\newcommand{\estgtwo}[1]{\wh{\theta}_{\group_2(#1)}}
\newcommand{\estreg}{\est_{\text{ols}}}

\newcommand{\catefn}{c}
\newcommand{\hk}{\sigma^2}
\newcommand{\hkte}{\sigma_{\tau}^2}
\newcommand{\hksd}{\sigma}
\newcommand{\hkest}{\wh{\sigma}^2}

\newcommand{\hkavg}{\sigma_{s}^2}
\newcommand{\hkavgest}{\wh{\sigma}_{s}^2}
\newcommand{\sdavg}{\sigma_s}
\newcommand{\sdavgest}{\wh{\sigma}_s}

\newcommand{\Dn}{D_{1:n}}
\newcommand{\Di}{D_i}
\newcommand{\Dj}{D_j}
\newcommand{\mn}{M_n}
\newcommand{\mjn}{M_{j,n}}

\newcommand{\propfn}{p}
\newcommand{\propfng}{p_{\group}}

\newcommand{\filtration}{\mathcal{F}}

\newcommand{\residuali}{\epsilon_i}
\newcommand{\residualj}{\epsilon_j}
\newcommand{\ykn}{Y_{k, n}}
\newcommand{\ikn}{I_{k, n}}
\newcommand{\sigkn}{\sigma^2_{k, n}}
\newcommand{\Sigkn}{\Sigma_{k, n}}
\newcommand{\Sign}{\Sigma_{n}}
\newcommand{\sn}{S_{n}}
\newcommand{\ceffn}{m}
\newcommand{\ceffnest}{\wh m}
\newcommand{\en}{E_n}

\DeclareMathOperator{\localdesigncond}{Loc}

\newcommand{\psisamp}{\psi_1}
\newcommand{\psiassign}{\psi_2}
\newcommand{\ylevel}{\bar Y}
\newcommand{\yleveli}{\bar Y_i}
\newcommand{\ylevelj}{\bar Y_j}
\newcommand{\hti}{H_i}
\newcommand{\te}{\tau}

\newcommand{\psii}{\psi_i}
\newcommand{\psiin}{\psi_{i,n}}
\newcommand{\psij}{\psi_j}
\newcommand{\psijn}{\psi_{j,n}}
\newcommand{\psin}{\psi_{1:n}}
\newcommand{\psiione}{\psi_{1, i}}
\newcommand{\psiitwo}{\psi_{2, i}}

\newcommand{\Ti}{T_i}
\newcommand{\Tj}{T_j}
\newcommand{\Tn}{T_{1:n}}

\newcommand{\propfnestn}{\wh \propfn_n}

\newcommand{\propfni}{\propfn_i}

\newcommand{\group}{g}

\newcommand{\groupls}{\group_{l,s}}

\newcommand{\groupset}{\mathcal{G}}
\newcommand{\groupsetn}{\mathcal{G}_n}
\newcommand{\groupsetnl}{\mathcal{G}_{nl}}

\newcommand{\remainderset}{\mathcal{R}}
\newcommand{\groupsetnu}{\mathcal{G}_n^{\nu}}
\newcommand{\nl}{n_l}
\newcommand{\filtrationcand}{\mathcal{F}_n}

\newcommand{\filtrationgeneric}{\mathcal{A}_n}
\newcommand{\filtrationalt}{\mathcal{A}}
\newcommand{\permn}{\pi_n}
\newcommand{\propbound}{\delta}

\newcommand{\hinbar}{\bar h_n(W_i)}

\newcommand{\Wn}{W_{1:n}}

\newcommand{\varest}{\wh{V}}

\newcommand{\levels}{L_n}
\newcommand{\dimpsi}{\dim(\psi)}

\newcommand{\nlevels}{L_n}

\newcommand{\denoml}{k_l}

\newcommand{\kboundn}{\overline{k}_n}
\newcommand{\snl}{S_{nl}}
\newcommand{\ml}{m_l}

\newcommand{\varlocal}{V}

\newcommand{\Siglimit}{\sigma^2}
\newcommand{\proprand}{\xi_n}

\newcommand{\opg}{o_{p, \filtrationgeneric}}
\newcommand{\Opg}{O_{p, \filtrationgeneric}}

\newcommand{\propselect}{q}
\newcommand{\propselectavg}{\bar q}
\newcommand{\propselectopt}{q^*}
\newcommand{\propselectest}{\wh q}
\newcommand{\propselectestn}{\wh q_n}
\newcommand{\propselecti}{q_i}

\newcommand{\propselectl}{q_l}

\newcommand{\nsampled}{n_T}
\newcommand{\nobs}{n_{\rm obs}}
\newcommand{\npilot}{n_{\rm pilot}}
\newcommand{\via}{v_{ia}}
\newcommand{\vib}{v_{ib}}
\newcommand{\vja}{v_{ja}}
\newcommand{\vjb}{v_{jb}}
\newcommand{\zin}{z_{in}}
\newcommand{\zjn}{z_{jn}}

\newcommand{\ugroup}{u_{\group}}
\newcommand{\groupmatching}{\nu}

\newcommand{\cost}{C}
\newcommand{\budget}{B}

\newcommand{\indfn}{\rho}
\newcommand{\remainderl}{r_l}
\newcommand{\opone}{\op(1)}
\newcommand{\al}{a_l}
\newcommand{\kl}{k_l}

\newcommand{\ui}{u_i}
\newcommand{\diff}{a}

\newcommand{\filtrationcandpsi}{\filtrationcand^{D}}
\newcommand{\filtrationcandt}{\filtrationcand^{T}}
\newcommand{\filtrationhd}{\mathcal{H}_n^{D}}
\newcommand{\filtrationht}{\mathcal{H}_n^{T}}

\newcommand{\ut}{u_t}
\newcommand{\ur}{u_r}

\newcommand{\Dt}{D_t}
\newcommand{\Dr}{D_r}

\DeclareMathOperator*{\argmin}{argmin}

\DeclareMathOperator{\ate}{ATE}
\DeclareMathOperator{\sate}{SATE}

\DeclareMathOperator{\bern}{Bernoulli}
\DeclareMathOperator{\unif}{Unif}

\DeclareMathOperator{\cov}{Cov}

\DeclareMathOperator{\var}{Var}

\DeclareMathOperator{\simiid}{\overset{iid}{\sim}}

\DeclareMathOperator{\identity}{Id}

\DeclareMathOperator{\crdist}{CR}

\date{\today}
\title{Fine Stratification of Survey Experiments\footnote{I wish to thank Alberto Abadie, Anna Mikusheva, and Victor Chernozhukov for their support and guidance during this project.
This paper also benefited from conversations with Isaiah Andrews, Mert Demirer, and numerous seminar participants.}}
\author{Max Cytrynbaum\footnote{Yale Department of Economics. Correspondence: max.cytrynbaum@yale.edu}}

\begin{document}
\maketitle
\begin{abstract}
This paper studies a two-stage model of experimentation, where the researcher first samples representative experimental participants from an eligible pool, then assigns each sampled unit to treatment or control, using matched $k$-tuples randomization at both stages.
To implement such designs, we develop a fast algorithm for matching units into $k$-tuples for any $k \ge 2$ and any dimension of covariates.
By surveying 200 recent experimental working papers, we estimate that our algorithm newly enables multivariate fine stratification with provable match quality guarantees in about 48\% of recent papers.
We show that finely stratified sampling and assignment both nonparametrically reduce the variance of treatment effect estimation, with the gains from stratified sampling increasing in the size of the eligible pool and how well covariates predict treatment effect heterogeneity. 
We develop new inference methods that fully exploit the efficiency gains from both design stages, allowing researchers to report smaller standard errors if they designed a representative experiment. 
An application to nine published experiments quantifies the efficiency gains.
\end{abstract}



\onehalfspacing
\newpage

\section{Introduction} \label{section:introduction}

A key objective in the design of experiments is to construct an experimental sample that is representative of a broader population of interest.
Given such a sample, experimenters also want to randomize so that the treatment and control groups are finely balanced on covariates.
One way to achieve both goals is through a survey experiment that uses stratified randomization for the sampling of participants and assignment of treatments.

To illustrate, consider the OpenResearch Unconditional Income Study (ORUS), a large-scale experiment which distributed over \$40 million in unconditional cash transfers in the United States in an effort to study the effects of guaranteed income on a range of social and economic outcomes \citep{broockman2024income}.
OpenResearch recruited $14{,}573$ individuals who were eligible and willing to participate, but its budget only allowed for the final enrollment of $3{,}000$ participants in the experiment.
To ensure representativeness, the authors drew a stratified sample of size $3{,}000$ from the eligible pool.
Among the participants, they also assigned treatments by stratified randomization, matching units into triples on income, race, and state, along with several dozen other covariates $\psi$, then assigning one unit in each triple to the \$1{,}000 per month transfer.

Beyond the ORUS study, the value of stratification for both sampling participants and assigning treatments is recognized by practitioners.
The J-PAL guide to randomized experiments recommends stratifying the sample ``based on observed covariates that are expected to moderate the treatment effect,'' and observes that ``stratifying at the sampling stage \ldots reduces the variance of the sample relative to the underlying population,'' just as ``stratifying at the treatment assignment stage creates treatment and control groups that are more similar to each other'' \citep{jpalSampling}.

Similarly, \citet{muralidharan2017scale} note that ``the lowest-hanging fruit may be to make samples more representative of the populations about which we wish to learn'' and suggest ``researchers have devoted more effort to persuading their institutional partners to randomize (for internal validity) than to be representative.''

Motivated by such considerations, we study a general family of stratified designs for survey experiments, in which participants are representatively sampled from an eligible pool, then assigned to a binary treatment.
We allow the sampling proportions $\propselect(\psi)$ to vary with covariates $\psi$, as in the ORUS study above, implementing these proportions using fine stratification.
Such covariate-dependent sampling rates are a classical feature of stratified survey sampling \citep{cochran1977} and may arise for a variety of reasons, including a desire to oversample subgroups of particular interest for heterogeneous treatment effects (as discussed in the J-PAL guide), heterogeneous costs of experimentation in different regions, or administrative and logistical constraints in constructing the sample.
Likewise, we allow treatment probabilities $\propfn \ne 1/2$, e.g.\ $\propfn=1/3$ as in the ORUS study.

To implement such designs, we develop a new matching algorithm, which produces high quality matched $k$-tuples for any $k \ge 2$ by matching on covariates $\psi$ in any dimension.
Our algorithm provably satisfies an asymptotic tight matching condition, which is needed to guarantee that fine stratification delivers nonparametric efficiency improvements at both the sampling and assignment stages. 

Prior work has provided algorithms with such guarantees for matched pairs $k = 2$ \citep{bai2022pairs}, their recursive extension to $k = 2^\ell$ \citep{bai2024factorial}, and univariate matched $k$-tuples produced by sorting \citep{bai2020pairs}.
Our algorithm is the first to possess such match quality guarantees for tuples of any size $k \ge 2$ and $\dim(\psi) \ge 1$. 

Leaving theory aside, even as a purely practical matter, the literature has few recommendations on how to form matched $k$-tuples for general $k \ge 2$ and multivariate $\psi$ beyond greedy heuristics \citep{moore2012blocking} and mixed-integer schemes that do not scale above several hundred units \citep{brixius2025}.
This is perhaps not surprising, since solving for the globally optimal matched $k$-tuples in Euclidean distance is NP-hard if $k \ge 3$  and $\dim(\psi) > 1$ \citep{pyatkin2017npHardness}.
We navigate around this hardness result with an algorithm that scales to at least 100,000 units and also satisfies the appropriate tight matching condition.
We also provide a more basic version of the procedure that easily scales to millions of units. 
Thus, our procedure also closes a practical gap.

To assess the impact of such designs on experiments in economics, we surveyed 200 recent NBER working papers involving randomized experiments from 2024 to 2026.
We find that 96 of these papers contain an experiment requiring $k\neq2^\ell$, for which existing methods do not provide multivariate match quality guarantees.
Thus, for about 48\% of recent working papers, provably high quality stratified treatment assignment with $\dim(\psi) > 1$ is newly enabled by our procedure.

This is before considering finely stratified sampling, which has propensity $\propselect \ne 1/2$ unless the number of eligible units is exactly twice the size of the experiment, generically requiring matched $k$-tuples for $k > 2$.
In our survey, we find that about $22\%$ of these working papers already perform such a sampling step, either by randomly sampling participants from a larger eligible pool ($11\%$) or using a third-party service such as Prolific or YouGov to recruit a representative sample ($11\%$).
Our results show that experimenters who construct a representative sample in this way can report smaller standard errors, providing an incentive for researchers to make experiments representative. 

In Section~\ref{section:method}, we formally introduce a family of finely stratified designs that randomizes sampling and assignment variables within a matched partition $\groupsetn$. 
For matched groups $g \in \groupsetn$ and centroids $\bar \psi_\group = |\group|\inv \sum_{i \in \group} \psi_i$, we require the tight matching condition:
\begin{equation} \label{equation:matching-intro}
\frac{1}{n} \sum_{\group \in \groupsetn} \sum_{i \in \group} \bigl| \psi_i - \bar \psi_\group \bigr|_2^2 = \op(1).
\end{equation}
This generalizes similar match quality conditions that have appeared in \citet{bai2022pairs} and \citet{bai2020pairs}.
In the remaining sections, we make the following contributions:
\begin{enumerate}
\item In Section~\ref{section:algorithms}, we develop a new matching algorithm, which uses a fast spatial sorting procedure as a warm start to initialize the balanced $k$-means procedure of \citet{malinen2014balanced}.
We show that the resulting partition $\groupsetn$ satisfies the tight matching condition~\eqref{equation:matching-intro} under weak assumptions.

\item In Section~\ref{section:asymptotics}, we provide asymptotic theory for finely stratified survey experiments with general sampling proportions $\propselect(\psi)$ and assignment propensity $\propfn \ne 1/2$.
We show finely stratified sampling and assignment both provide nonparametric variance reductions, attaining the relevant semiparametric variance bound. 
Stratified sampling attenuates the variance due to treatment effect heterogeneity, with larger reductions as the size of the eligible pool grows.
We characterize the optimal stratification variables, showing that for sampling one should stratify on a small set of covariates $\psi$ most predictive of treatment effect heterogeneity, formally justifying the J-PAL recommendation above.

\item In Section~\ref{section:optimal_designs}, we apply these designs to formulate and solve a budget-constrained optimal sampling problem with heterogeneous costs, deriving the optimal sampling proportions $\propselectopt(\psi) \propto \cost(\psi)^{-1/2} \hksd(\psi)$ for costs $\cost(\psi)$ and an appropriate conditional variance function $\hk(\psi)$.
Under a homoskedasticity assumption, this design can be implemented using the known costs.
We also propose versions that use available pilot or observational data to estimate the unknown $\hk(\psi)$ term, showing these asymptotically minimize variance over all budget-feasible proportions $\propselect(\psi)$.

\item In Section~\ref{section:inference}, we develop new inference methods for experiments with joint finely stratified sampling and assignment. 
We enable asymptotically non-conservative inference on the population average treatment effect, as well as asymptotically valid inference on the average treatment effect among the eligible units.  
\end{enumerate}

In Section \ref{section:empirical}, we present an empirical application to nine experiments published in the American Economic Review, which demonstrates the value of our methods.

\subsection{Related Literature}

For overviews of experimental design theory, see \cite{rosenberger2016book}, \cite{athey2017survey}, or \cite{bai2025primer}.
For reviews of the classical literature on survey sampling, see \cite{cochran1977} and \cite{lohr2021}. 

The idea that trial findings should be extended from enrolled participants to the broader population of trial-eligible individuals originates in the clinical literature \citep{rothwell2005}.
A subsequent literature in statistics develops estimators that correct for non-representative experimental participation and imbalanced assignment by ex-post adjustment, as in \cite{dahabreh2019generalizing} and \cite{li2023generalization}. 
See \citet{degtiar2023review} for a review of this area.
By contrast, our finely stratified sampling and assignment designs prevent such imbalances from arising at the randomization stage.

Contemporaneous with the first version of our paper, \citet{yang2021} propose a two-stage design using rerandomization for both sampling and assignment.
Under rerandomization, difference-of-means estimation is asymptotically slightly less efficient than ex-post linear covariate adjustment.
By contrast, we show that two-stage fine stratification is asymptotically equivalent to nonparametric covariate adjustment for the imbalances in both the sampling and assignment variables.

Recent work on stratified treatment assignment includes \cite{imai2009}, \cite{bugni2018inference}, \cite{fogarty2018}, \cite{wang2021}, \cite{bai2022pairs}, \cite{dechaisemartin21}, \cite{bai2020pairs}, \cite{tabord-meehan2020}, \cite{bai2024factorial}, and \cite{bai2026efficiency}.
For treatment assignment, our work is most closely related to \cite{bai2022pairs}, who study matched pairs designs, and \cite{bai2020pairs}, who studies finely stratified designs with constant propensity $\propfn = a/k$ and univariate stratification variables, $\dim(\psi)=1$.
Follow-up work in \citet{cytrynbaum2024adjustment} uses the local-randomization framework developed here to study optimal ex-post linear covariate adjustment under stratified randomization.
Aside from stratification, other recent proposals for balanced treatment assignment include \cite{kasy2016}, \cite{kallus2017balance}, \cite{li2018rerandomization}, \cite{krieger2019}, and \cite{harshaw2021}.

From a computational perspective, finely stratified designs require the construction of high-quality matched groups.
Multivariate fine stratification has previously relied on greedy algorithms such as that of \citet{moore2012blocking}, implemented in the \texttt{blockTools} package. For large experiments, the threshold blocking method of \citet{higgins2016blocking} forms blocks of uneven sizes bounded below by some chosen $k$. 
These methods do not deliver the tight matching guarantees that drive our efficiency results.

The most closely related previous work providing such guarantees is \citet{bai2022pairs}, who bound the average squared Euclidean distance of the optimal non-bipartite matching \citep{derigs1988} using a space-filling path argument. 
\citet{bai2020pairs} forms univariate $k$-tuples with a match quality guarantee by sorting units along their $\psii$ values for $\dim(\psi)=1$. 
We generalize the space-filling construction and adapt it into a practical algorithm, deploying this as the warm start for the balanced $k$-means procedure of \cite{malinen2014balanced}. 
This results in a fast, practical algorithm that provably satisfies the appropriate tight matching condition for any $k \ge 2$ and $\dim(\psi) \ge 1$.

Our budget-constrained optimal sampling design extends the classical theory of optimal allocation in survey sampling \citep{cochran1977} to finely stratified sampling into an experiment.
Our results on design using a pilot study are related to previous work in \cite{hahn2012}, \cite{bai2020pairs}, \cite{tabord-meehan2020}, and \cite{kasy2021adaptive}.
Our inference results build on the method of collapsed-strata in \cite{hansen1953} and its modern ``pairs of pairs'' variants studied in \cite{abadie2008}, \cite{bai2022pairs}, and \cite{bai2026variance}.

\section{Setting and Designs} \label{section:method}

Consider running an experiment to estimate the effect of a binary treatment.
There are $n$ eligible units with potential outcomes $Y_i(d)$ for $d \in \{0,1\}$ and observed baseline covariates $\psi_i \in \mr^v$ for $i = 1, \dots, n$, where $v = \dim(\psi)$.
We write $\psi_{1:n} = (\psi_i)_{i=1}^n$, and similarly for other variables.
Let $\Ti = 1$ if an eligible unit is sampled to participate in the experiment and $\Ti=0$ otherwise, so experiment size $\nsampled = \sum_i \Ti$.
In general, we may want to implement sampling proportions $\propselect(\psi) = P(T=1 | \psi) \in (0, 1]$.
Sampled units are assigned to treatment or control $\Di \in \{0,1\}$ with $\propfn(\psi) = P(D=1 | \psi) \in (0,1)$. 
In practice, $\propfn$ and $\propselect$ will often be constant.
Outcomes are $Y_i = \Ti[\Di Y_i(1) + (1-\Di) Y_i(0)]$. 

For example, \cite{shim2026fiscal} study the effect of fiscal news on household spending, sampling $\nsampled = 11{,}262$ respondents into their experiment from an eligible population of $n \approx 200{,}000$ Koreans in a survey panel, so that $\propselect = \nsampled / n \approx 1/18$.
Similarly, \cite{chioda2026making} evaluate an entrepreneurship training program, sampling $\nsampled = 4{,}402$ participants into the experiment from a pool of $n = 7{,}431$ Ugandan applicants so that $\propselect \approx 3/5$.

Denote $W_i = (\psi_i, Y_i(0), Y_i(1))$ and suppose $(W_i)_{i=1}^n \simiid P$, modeling the eligible units as an iid sample from a broader superpopulation.
One natural estimand in this case is the average treatment effect $\ate = E[Y(1)-Y(0)]$.
For example, in the ORUS basic income study in the introduction, the $\ate$ is the treatment effect among the broader population of low-income Americans from which the $n = 14{,}573$ eligible participants were recruited, using a combination of direct mailers and online advertisements.

Another possible estimand is the $\sate = n\inv \sum_{i=1}^n (Y_i(1)-Y_i(0))$ among the $n$ eligible units that the experimenter physically samples from.
In the ORUS study, the $\sate$ is the treatment effect among these $n = 14{,}573$ eligible units, from which they sampled the ultimate experimental population of size $\nsampled = 3{,}000$.
Note the $\sate$ is sometimes alternatively defined as the treatment effect among the $\nsampled$ experimental participants, $\nsampled\inv \sum_{i=1}^n \Ti(Y_i(1)-Y_i(0))$.
This definition has a narrower scope, describing effects only among the actual experimental participants rather than the broader eligible population.
The two definitions coincide in the special case $\propselect = 1$.
We mostly focus on the $\ate$ in what follows, though for completeness we extend our main asymptotic and inference results to the $\sate$ as well.

In either case, our goal is to sample a representative subset of the eligible units into the experiment, then assign them to treatment and control in a way that finely balances the baseline covariates $\psi_i$.
We implement both the sampling and treatment assignment steps using a form of matched $k$-tuples randomization.

\medskip

\textbf{Matched $k$-Tuples.}
We first describe matched $k$-tuples in the context of random sampling.  
Suppose we want to sample $\propselect = 1/5$ of the $n$ eligible units into the experiment. 
To do so, we first match them into homogeneous groups of size $k=5$ using the covariates $\psin$, then sample one out of every five units in each matched group into the experiment, uniformly at random.
Doing so balances covariates $\psii$ between the sampled and non-sampled units, $\Ti=1$ and $\Ti=0$. 
More generally, we introduce a family of matched $k$-tuples designs that allows for potentially varying sampling and assignment propensities $\propselect(\psi) = P(T = 1 | \psi)$ and $\propfn(\psi) = P(D = 1 | \psi)$:

\begin{defn}[Local Randomization] \label{defn:local_randomization}
Let $\propselect : \mr^v \to (0, 1]$ be a propensity with rational levels $\propselect(\psi) \in \{a_l/k_l : l = 1, \dots, L\}$.
We write $\Tn \sim \localdesigncond(\psi, \propselect(\psi))$ and say that the design $\Tn = (\Ti)_{i=1}^n$ is a \emph{locally randomized} implementation of $\propselect(\cdot)$ with respect to $\psi$ if there exists a partition of $\{1, \dots, n\}$ into matched groups $g \in \groupsetn$  such that:

\begin{enumerate}[label={\rm(\arabic*)}, itemindent=.5pt, itemsep=.4pt]
\item \emph{Tight matching.} 
The partition $\groupsetn$ is determined only by $\psin$ and independent randomness $\permn$ used to break ties in group formation.
For centroids $\bar \psi_\group \equiv |\group|\inv \sum_{i \in \group} \psi_i$, the partition has matching objective $F$ with 
\begin{equation} \label{equation:homogeneity}
F(\groupsetn) \equiv \frac{1}{n} \sum_{\group \in \groupset_n} \sum_{i \in \group} \bigl | \psi_i - \bar \psi_\group \bigr|_2^2 = \op(1).
\end{equation}
\item \emph{Complete randomization.}
For each $g \in \groupsetn$, units $i, j \in \group$ have equal sampling propensities $\propselect(\psii) = \propselect(\psij) = a_l/k_l$ for some $l=1, \dots, L$.
Conditional on $(\Wn, \permn)$, each $(\Ti)_{i \in \group}$ is drawn uniformly over all binary vectors in $\{0, 1\}^{k_l}$ with exactly $a_l$ out of $\kl$ units having $\Ti = 1$, independently between groups. 
\end{enumerate}
\end{defn}

Condition (2) is the usual definition of stratified randomization.
Note we may require one remainder group with $1 \le |g| < k_l$ per propensity level due to integer effects. 
For example, if $n=101$ and $\propselect = 1/5$, the remainder group would have size $|g| = 1$, and we can draw $\Ti \sim \bern(\propselect)$ for this unit.

\medskip

\textbf{Tight Matching Condition.}
The key requirement on the partition $\groupsetn$ is the tight matching condition $F(\groupsetn) = \op(1)$ in Equation \ref{equation:homogeneity}, which asks that the matched groups be tightly clustered in $\psi$-space.
We use it to guarantee strong control over the in-sample covariate imbalances that may arise during randomization, allowing us to show that finely stratified sampling and assignment both provide nonparametric variance reductions for treatment effect estimation.
This condition generalizes related match quality guarantees previously shown for optimal matched pairs in \citet{bai2022pairs} and univariate matched $k$-tuples in \citet{bai2020pairs}.

Constructing fast algorithms that provably satisfy this condition for any $k \ge 2$ and $v \ge 1$ is a core contribution, enabling the finely stratified sampling and assignment designs we study in this paper.
However, the value of this algorithm extends beyond the binary setting.
For example, \citet{bai2024factorial} study finely stratified assignment for experiments with multiple treatment arms, which requires an asymptotically equivalent matching condition, but provide an algorithm only in the case $k = 2^\ell$.
Our matching algorithm thus makes their asymptotic results applicable to designs outside this special case, for example experiments with three treatment arms, requiring matched triples with $k = 3$.
We present the algorithm and its formal matching guarantee in Section \ref{section:algorithms} below.

\medskip

\textbf{Experiment Design.}
We consider a two-stage finely stratified design:  
\begin{enumerate}[label={(\arabic*)}, itemindent=.5pt, itemsep=.4pt] 
\item Sample eligible units $\Tn \sim \localdesigncond(\psi, \propselect(\psi))$. 
\item Assign treatments $\Dn \sim \localdesigncond(\psi, \propfn(\psi))$ to the sampled units $\{i: \Ti=1\}$.
\end{enumerate}
The most common use case of this design is to let the propensities be constant, sampling $\Tn \sim \localdesigncond(\psi, \propselect)$ for $\propselect = a/k$ of the eligible units, then assigning treatments $\Dn \sim \localdesigncond(\psi, \propfn)$ to the sampled units.
In our survey of 200 experiments in recent NBER working papers, 65 are binary-treatment RCTs, of which at least 32\% deliberately use a treatment proportion $\propfn \ne 1/2$, requiring our matched $k$-tuples algorithm for multivariate fine stratification not only at the sampling stage, but also for treatment assignment.

More generally, the experimenter may desire non-constant propensities $\propselect(\psi)$ or $\propfn(\psi)$ for various budgetary or logistical reasons, or purely to improve statistical power. 
For example, in Section \ref{section:optimal_designs}, we study stratified implementation of the variance minimizing sampling proportions $\propselectopt(\psi)$ under a fixed budget constraint. 

\emph{Sampling-subordinate assignment.}
When $\propselect(\psi)$ takes multiple values $\al/\kl$, our inference methods in Section \ref{section:inference} additionally require the assignment groups in stage (2) to be matched within sampling propensity strata $\{i: \Ti=1,\ \propselect(\psii)=\al/\kl\}$.
Our matching and efficiency results in Sections \ref{section:algorithms} to \ref{section:optimal_designs} hold with or without this restriction, but we recommend imposing it to improve finite-sample efficiency.

\setlength{\tabcolsep}{1pt}
\begin{figure}[t]
\centering
\begin{tabular}{cc}
\includegraphics[width=0.5\textwidth]{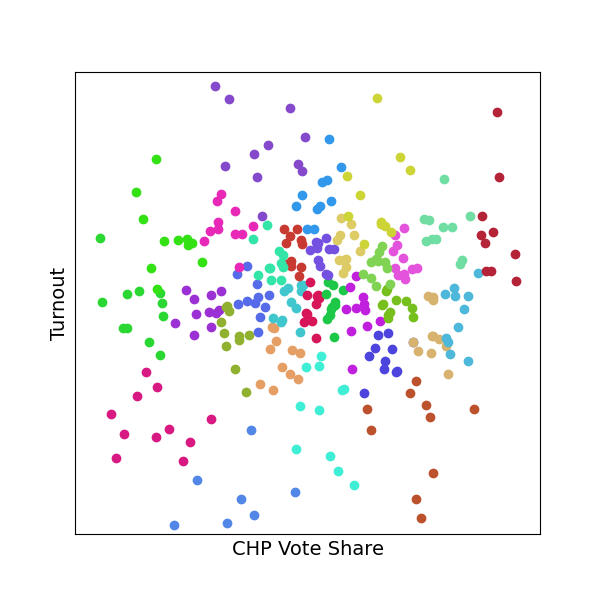} &
\includegraphics[width=0.5\textwidth]{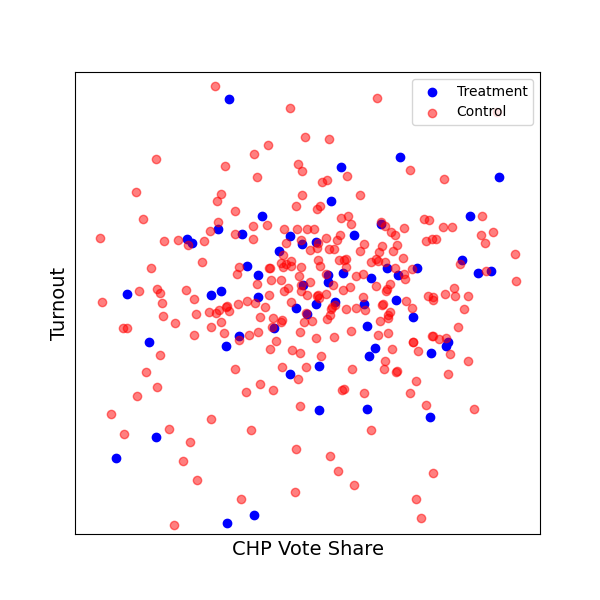} \\
\end{tabular}
\caption{Assignment groups and variables for $\Dn \sim \localdesigncond(\psi, \propfn)$ with $\propfn=2/11$.}
\label{fig:assignment}
\end{figure}
\setlength{\tabcolsep}{6pt}

\begin{ex}[Matched Tuples] \label{ex:intro:matched-tuples}
We illustrate the two-stage procedure using data from the electoral information campaign of \citet{baysan2022}.
First, suppose we sample $\propselect = 3/5$ of the units by forming groups of size $|\group| = 5$, then randomly setting $\Ti=1$ for exactly $3$ out of $5$ units in each group.
This ensures that units of each ``type'' in the space of covariates $\psi = $ (vote share, turnout) are represented within the smaller experiment.
\citet{baysan2022} assigned only $\propfn=2/11$ of the units to $D=1$ due to the high cost of the treatment.
Figure \ref{fig:assignment} illustrates this second stage.
Each color represents a different group of sampled units.
We match the sampled units $\{i: \Ti=1\}$ into groups of size $|\group| = 11$, then randomly assign $\Di = 1$ to exactly $2$ out of $11$ units from each group.
\end{ex}

\begin{remark}[Coarse Stratification] \label{rem:coarse-and-cr}
The literature has also considered coarsely stratified designs with fixed strata $X_i \in \{1, \dots, m\}$ and stratum sizes $n_x = |\{i : X_i = x\}| \to \infty$ \citep[e.g.][]{bugni2018inference}.
Such designs are nested in our framework by letting $\psi = X$ and matching units sharing the same discrete covariate value $X=x$ into groups at random.
In this case, $\groupsetn$ depends on both $\psin$ and an extra randomization device $\permn$, which we accommodate in our analysis.
Such groups trivially satisfy Equation \ref{equation:homogeneity}. 
Complete randomization is obtained by setting $\psi = 1$ and forming groups $|\group| = k$ at random. 
Thus, we provide a unified theory for finely stratified, coarsely stratified, and completely randomized designs at both the sampling and assignment stages. 
\end{remark}

\section{Algorithm and Match Quality Guarantee} \label{section:algorithms}

This section establishes the main theoretical properties of our matching algorithm.
See Section \ref{appendix:computational-comparison} in the appendix for a detailed evaluation of its computational properties and a comparison with alternative matching methods.

In what follows, denote $[z] = \{1, \dots, z\}$ for any positive integer $z$.
Let $a \wedge b = \min(a, b)$ and $|S|$ denote the cardinality of $S$.
We begin by describing the algorithm for a single propensity level $a/k$, where a set of units $S \sub [n]$ is to be matched into groups of a common size $k$.
For sampling, we apply the algorithm to the eligible units $S = [n]$.
For assignment, we apply it to $S = \{i: \Ti=1\}$, or separately to each $S = \{i: \Ti=1,\ \propselect(\psii)=\al/\kl\}$ under sampling-subordinate assignment.
Our algorithm constructs a matched partition $\groupsetn$ of the units in $S$ as follows: 

\begin{enumerate}[label={\rm(\arabic*)}, itemindent=.5pt, itemsep=.4pt]
\item \emph{Preprocessing.}  Form a remainder group $r$ of size $|r| < k$ units so that $S' = S \setminus r$ has cardinality $|S'|$ exactly divisible by the tuple size $k$.
For the remaining units $i \in S'$, standardize covariates into the unit cube $\psii \in [0, 1]^v$ by transforming $\psi_{i}^{(j)} \to (\psi_{i}^{(j)} - \min_{i' \in S'} \psi_{i'}^{(j)}) / (\max_{i' \in S'} \psi_{i'}^{(j)} - \min_{i' \in S'} \psi_{i'}^{(j)})$ for $j \in [v]$.
\item \emph{Spatial sorting.} Compute an index function $\indfn : [0,1]^v \to \mathbb N$ returning the position of $\psi \in [0, 1]^v$ along a space-filling curve.
Sort units by their index values $(\indfn(\psii))_{i \in S'}$, then match consecutive units into groups $g \in \groupsetn$ of size $k$, breaking ties at random.
\item \emph{Balanced $k$-means polishing.} Refine the resulting partition $\groupsetn$ by running the balanced $k$-means algorithm of \cite{malinen2014balanced} for $T \ge 0$ iterations.
In practice, we run until convergence, which happens in finitely many iterations.
\end{enumerate}

\begin{figure}[t]
\centering
\includegraphics[width=\textwidth]{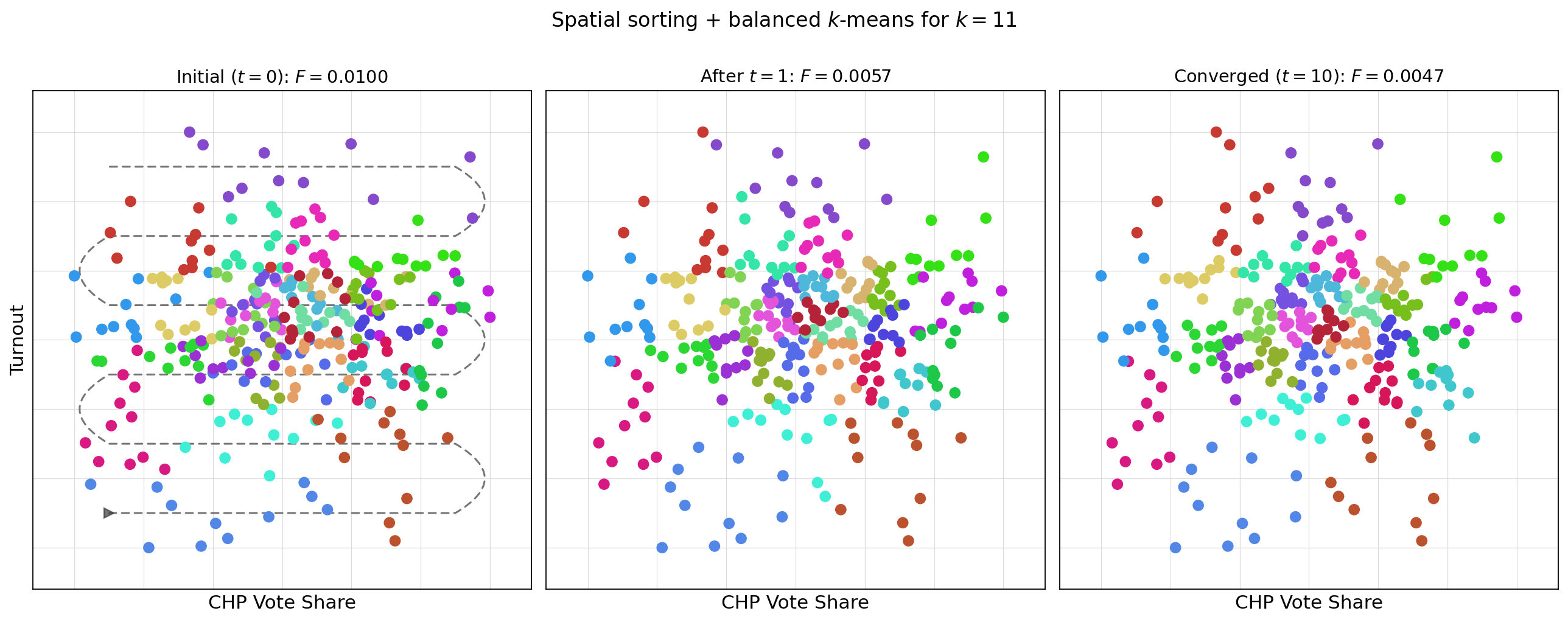}
\caption{Matching algorithm of Section \ref{section:algorithms} applied to the \cite{baysan2022} data with tuple size $k = 11$ and $m=6$.
Preprocessed data and spatial sorting algorithm (left), one iteration of balanced $k$-means (middle), convergence after $t=10$ iterations (right).}
\label{fig:block-path}
\end{figure}

Readers primarily interested in treatment effect estimation can see the match quality guarantee in Theorem \ref{thm:matching-text} and skip the technical construction below on a first reading.

For the preprocessing in step (1), note that if $\max_{i' \in S'} \psi_{i'}^{(j)} = \min_{i' \in S'} \psi_{i'}^{(j)}$ then covariate $j$ has no variation and can be excluded from the procedure without changing match quality.
Then without loss, $\max_{i' \in S'} \psi_{i'}^{(j)} - \min_{i' \in S'} \psi_{i'}^{(j)} > 0$ for each dimension $j \in [v]$.
Also, our statistical guarantees are invariant to the choice of remainder group. 
In practice, we use the remainder to trim outliers, which improves the performance of the spatial sorting algorithm. 

\medskip

\textbf{Spatial Sorting.}
We partition the unit cube $[0,1]^v$ into axis-aligned blocks of side length $1/m$ for grid size parameter $m \ge 1$.
This creates a grid with cells of the form $C(z) = \prod_{j=1}^{v} [z_j/m, (z_j+1)/m]$ for integers $z_j \in \{0, \dots, m-1\}$. See the gray lines in Figure \ref{fig:block-path} for an illustration with $m=6$.
The index function $\indfn(\cdot)$ in step (2) then enumerates these blocks along a space-filling path through $[0,1]^v$, such that any two consecutive cells share a $(v-1)$-dimensional face.
Such paths have also been used to analyze optimal matched pairs in \cite{bai2022pairs}, see Remark \ref{remark:spatial-sorting} below for comparison.

For each dimension $v \ge 1$, denote the position of cell $C(z)$ along this path by $\indfn_v(z) \in \mathbb N$.
In the univariate case, we can set $\indfn_1(z_1) = z_1$. 
In dimension $v \ge 2$, the position of cell $C(z)$ can be computed recursively as in \cite{er1984nary}:
\begin{equation} \label{equation:gray-code-text}
\indfn_v(z_1, \dots, z_v) \equiv z_v \cdot m^{v-1} + \begin{cases} \indfn_{v-1}(z_1, \dots, z_{v-1}) & z_v \text{ even}, \\ m^{v-1} - 1 - \indfn_{v-1}(z_1, \dots, z_{v-1}) & z_v \text{ odd}. \end{cases}
\end{equation}
To understand the recursion, note $\indfn_{v-1}(\cdot)$ defines an order for traversing the $(v-1)$-dimensional subspace.
If the coordinate $z_v$ in the new dimension is even, we traverse the $v-1$ dimensional layer in forward order, while if $z_v$ is odd we go in reverse order, so that consecutive layers are traversed without jumps.
See Figure \ref{fig:block-path} for an illustration.

We extend $\indfn_v$ as a map defined on vectors of integers to a map on the unit cube $\indfn: [0, 1]^v \to \mathbb{N}$ by setting $\indfn(\psi) \equiv \indfn_{v}(\lfloor m \psi \rfloor  \wedge (m-1))$, where the min and floor are applied componentwise.
For a given $m$, we sort the units in $S'$ by their index values $\indfn(\psi_i)$, matching consecutive units into groups of size $k$ and breaking ties at random with $\permn$.
This procedure yields an initial partition $\groupsetn^{(0)}(m)$ of $S'$ into groups $|g|=k$.
Valid choices of $m$ for our theoretical guarantees below are $m_0 = \lfloor (|S'|/k)^{1/2} \rfloor$ when $v=1$ and $m_0 = \lfloor (|S'|/(k(v-1)))^{1/(v+1)} \rfloor$ when $v \geq 2$.
We set $m_0=1$ if either floor is zero.

\emph{Adaptive Grid Choice.}
We also allow data-adaptive grid size selection.
For any partition $\groupset$ of $S'$, define the matching objective $F(\groupset) \equiv n\inv \sum_{g \in \groupset} \sum_{i \in g} |\psi_i - \bar{\psi}_g|_2^2$.
Then for any finite set $\mathcal M$ of positive integers containing $m_0$, choose $\widehat m \in \argmin_{m \in \mathcal M} F(\groupsetn^{(0)}(m))$.
Since $m_0 \in \mathcal M$, this weakly reduces the objective $F$ relative to the theoretically valid partition and thus retains its match quality guarantee.
We set the initial spatial sorting partition to be $\groupsetn^{(0)} = \groupsetn^{(0)}(\widehat m)$.

In our empirical applications, $\mathcal M$ consists of all positive integers within $\pm 10$ of $m_0$.
Computing the partition for any fixed $m$ requires only $O(n v + n \log n)$ operations, so even searching over a large candidate set remains highly tractable, though in practice it suffices to search near $m_0$.
Drawing from covariate distributions fitted using our nine empirical applications, spatial sorting matches one million units with a median runtime of 36 seconds for tuple size $k=20$ (Appendix~\ref{appendix:computational-comparison}).

\medskip

\textbf{Balanced $k$-means polishing.}
Given the initial partition $\groupsetn^{(0)}$, we use balanced $k$-means \citep{malinen2014balanced} to further decrease $F(\groupset)$ subject to the constraint  $|g| = k$.
At iteration $t$, label the groups in $\groupsetn^{(t)}$ as $(g_s^{(t)})_s$, with centroids $(\bar{\psi}_{g_s}^{(t)})_s$. Write a candidate improved partition as $\groupset=(g_s^{(t+1)})_s$. We update by
\begin{equation} \label{equation:balanced-kmeans-text}
\groupsetn^{(t+1)} \in \argmin_{\groupset : |g| = k} \sum_s \sum_{i \in g_s^{(t+1)}} |\psi_i - \bar{\psi}_{g_s}^{(t)}|_2^2.
\end{equation}
The right-hand side of \eqref{equation:balanced-kmeans-text} is a balanced linear assignment problem that can be solved in polynomial time using the Hungarian algorithm \citep{malinen2014balanced}.
This objective is non-increasing, with $0 \le F(\groupsetn^{(t+1)}) \le F(\groupsetn^{(t)})$ for all $t \geq 0$. 
With ties broken in favor of the current partition, the objective decreases strictly whenever the partition changes.
Since there are finitely many such partitions, convergence to a fixed point $\groupsetn^{(t+1)} = \groupsetn^{(t)}$ happens in finitely many iterations.
The output of the single-level algorithm is the partition $\groupsetn = \groupsetn^{(T)} \cup \{r\}$, where $T \ge 0$ is the total number of iterations.

The balanced $k$-means polishing step is memory-intensive for larger experiments, but we found it tractable at $n=25{,}000$ units on a laptop with 16 GiB of memory and at $n=100{,}000$ units on an institutional compute node with 75 GiB of peak memory use.

\medskip

\textbf{Multiple Propensity Strata.}
In our general definition of local randomization, the propensity $\propselect(\psi)$ may take several rational values $\propselect(\psi) \in \{a_l / k_l : l \in [L]\}$.
In this case, we group units sharing the same propensity into strata $S_l = \{i : \propselect(\psi_i) = a_l / k_l\}$ and write $S_l' = S_l \setminus r_l$ for the units remaining after removing the remainder.
We run the algorithm above separately within each $S_l'$, standardizing the covariates within the stratum and matching units into groups of size $k_l$.
The grid size $\widehat m_l$ is selected by the adaptive rule from a finite candidate set $\mathcal M_l$ containing the theoretically valid grid size $m_{0,l}$. This is obtained by replacing $|S'|$ and $k$ in the formulas above with $|S_l'|$ and $k_l$.
We then refine the partition using $T_l \geq 0$ balanced $k$-means iterations.
Writing $\groupsetnl$ for the partition of $S_l$ produced in this way, the final output is the union $\groupsetn = \cup_{l=1}^{L} \groupsetnl$.

Our main result shows $\groupsetn$ satisfies the tight matching condition in Equation \ref{equation:homogeneity}.

\begin{thm}[Algorithm Guarantee] \label{thm:matching-text}
Suppose $E[|\psi|_2^\alpha] < \infty$ for some $\alpha > v + 1$, with $v = \dim(\psi)$.
For each $l \in [L]$, let $\mathcal M_l$ be any finite candidate set of positive integers containing $m_{0,l}$, and choose $\widehat m_l$ by the adaptive rule above.
Then for any number of balanced $k$-means iterations $T_l \geq 0$, the partition $\groupsetn$ satisfies
\begin{equation} \label{equation:algorithm-guarantee-text}
n\inv \sum_{g \in \groupsetn} \sum_{i \in g} |\psi_i - \bar{\psi}_g|_2^2 = \op\!\left( n^{2/\alpha - 2/(v+1)} \right) = \op(1).
\end{equation}
\end{thm}
Note that the match quality guarantee already holds for the initial partition $\groupsetn^{(0)}$, corresponding to $T_l=0$.
By monotonicity of the balanced $k$-means objective, this match quality guarantee continues to hold after any number of polishing iterations.
Across the nine empirical datasets and $k=2,\ldots,20$, the final polishing step reduced $F$ by $56\%$ on average relative to adaptive spatial sorting.

\emph{Curse of Dimensionality.}
Theorem \ref{thm:matching-text} shows that with a large enough sample size, our algorithm guarantees the tight matches required for our efficiency results below.
However, the strength of such match quality guarantees degrades rapidly with increasing covariate dimension $v$. 
In particular, if $\max_{i \leq n}|\psi_i|_\infty \leq B$, the appendix proves the exact finite sample bound $F(\groupsetn) \leq C v \cdot n^{-2/(v+1)}$ for a constant $C$ independent of $(n, v)$.
If $v = v_n \geq c \log n$ for a fixed $c>0$ as $n \to \infty$, a calculation shows that the right-hand side diverges, making our bound vacuous.  By contrast, if $v_n = o(\log n / \log \log n)$ then the right hand side is $o(1)$ and tight matching holds.

The same logarithmic dimensionality barrier emerges from \citet{davezies2025cube}'s analysis of matched pairs randomization.
They prove a lower bound on covariate imbalance that can be used to show tight matching is impossible when $v_n/\log n \to \infty$ for covariates with uniformly bounded density.
In particular, in the appendix we use their result to show there exists a $c_0>0$ such that every sequence of matched-pair partitions $\groupsetn$ has $\liminf_{n \to \infty} P(F(\groupsetn) \geq c_0 v_n)>0$. 
See Remark \ref{remark:matching-lower-bound} for the argument.

This shows the importance of using fine stratification only for a small set of highly predictive covariates, such as baseline outcomes.
For settings with many covariates expected to be only weakly predictive of outcomes, it may be preferable to use alternative methods that attempt to globally balance linear, instead of fully nonparametric, functions of the covariates. 
Some examples include rerandomization \citep{li2018rerandomization}, the Gram-Schmidt walk \citep{harshaw2021}, and the cube method \citep{davezies2025cube}.

\begin{remark}[Sorting Construction] \label{remark:spatial-sorting}
\cite{bai2022pairs} were the first to use the existence of space-filling paths to study the efficiency of optimal matched pairs. 
In our notation, they construct a partition with $F(\groupsetn) = O(n^{-1/v})$ for bounded $\psi$ and $k=2$.
For bounded $\psi$, we obtain the sharper rate $F(\groupsetn) = O(n^{-2/(v+1)})$ for any fixed $k \ge 2$ by a more careful accounting of the pairwise objective contributions $|\psii - \psij|_2^2$ within and across grid-cell boundaries.
See Remark~\ref{remark:sorting-rate-comparison} for a detailed technical comparison.
\end{remark}

\textbf{Application to Survey Experiments.}
Recall in our two-stage design above we sample $\Tn \sim \localdesigncond(\psi, \propselect(\psi))$ and assign $\Dn \sim \localdesigncond(\psi, \propfn(\psi))$.
By definition, such designs require partitions $\groupset_n^T$ of the eligible units $[n]$ and $\groupset_n^D$ of the sampled units $\{i : \Ti = 1\}$ that both satisfy the tight matching condition.
The following corollary shows that our algorithm provides such partitions.

\begin{cor}[Survey Experiments] \label{cor:tight-matching-survey}
Suppose $E[|\psi|_2^\alpha] < \infty$ for some $\alpha > v + 1$.
Then the sampling and assignment partitions $\groupset_n^T$ and $\groupset_n^D$ produced by the algorithm have 
\[
F(\groupset_n^T) = \op(1), \qquad F(\groupset_n^D) = \op(1).
\]
\end{cor}

Theorem \ref{thm:algorithm-guarantee} presents our most general tight matching results, in asymptotics allowing maximum group size $\kboundn = \max_{l \in [\nlevels]} \denoml$ and number of matching strata $\nlevels$ to grow with $n$.
This extension supports our work in Section \ref{section:optimal_designs}, which consider a sampling design $\Tn \sim \localdesigncond(\psi, \propselectestn(\psi))$, where $\propselectestn(\psi)$ is a rational-valued estimate of the optimal sampling propensity $\propselectopt(\psi)$, discretized with increasing fineness, so that $\kboundn, \nlevels \to \infty$ as $n \to \infty$.

\section{Asymptotic Theory} \label{section:asymptotics}

This section contains our main asymptotic theory, showing nonparametric efficiency gains from both finely stratified sampling and assignment.
Below, we denote $\en[a_i] = n\inv \sum_{i=1}^n a_i$ for any array $(a_i)_{i=1}^n$.
We assume the following throughout:
\begin{assumption} \label{assumption:asymptotics-regularity}
Moments $E[Y(d)^2] < \infty$ for $d = 0, 1$.
\end{assumption}

\textbf{Estimation.} Recall in our definition of local randomization (Definition \ref{defn:local_randomization}), we allowed sampling and assignment propensities $\propselect(\psi)$ and $\propfn(\psi)$ to have finitely many rational levels, such as $\propselect(\psi) = \al / \kl $ for $l \in [L]$.
This may be required for logistical reasons, or to maximize statistical power subject to a budget constraint, as we study in Section~\ref{section:optimal_designs} below.
To accommodate the most general designs, we state our main asymptotic results using the inverse propensity weighting (IPW) estimator
\begin{equation} \label{equation:estimator:varying}
\est = \en\left[\frac{\Ti}{\propselect(\psii)} \left(\frac{\Di Y_i}{\propfn(\psii)} - \frac{(1-\Di) Y_i}{1-\propfn(\psii)}\right) \right].
\end{equation}

In the special case with constant $\propfn(\psi)=\propfn$ and $\propselect(\psi)=\propselect$, we could use a standard regression estimator, like the coefficient on $\Di$ in the regression $Y_i \sim 1 + \Di$, estimated among the sampled units. 
Denoting this estimator by $\estreg$, Lemma~\ref{lemma:ols-ipw-equivalence} shows that $\estreg = \est + \Op(n\inv)$ under stratified sampling and assignment, so switching to OLS does not change our efficiency or inference results.
Because of this, we restrict our attention to IPW in what follows. 
Note also that this equivalence between OLS and IPW is a consequence of stratification and is not true for iid designs. 

In what follows, let $\catefn(\psi) = E[Y(1) - Y(0) | \psi]$ denote the conditional average treatment effect (CATE) and $\hk_d(\psi) = \var(Y(d) | \psi)$ the conditional variance function.

\begin{thm}[CLT] \label{thm:clt_fixed}
Assume sampling and assignment are locally randomized with $\Tn \sim \localdesigncond(\psi, \propselect(\psi))$ and $\Dn \sim \localdesigncond(\psi, \propfn(\psi))$. 
Then $\rootn (\est - \ate) \convwprocess \normal(0, \varlocal)$
\begin{equation} \label{equation:asymptotic-theorem}
\varlocal = \var(\catefn(\psi)) + E\left[\frac{1}{\propselect(\psi)} \left (\frac{\hk_1(\psi)}{\propfn(\psi)} + \frac{\hk_0(\psi)}{1-\propfn(\psi)} \right) \right].
\end{equation}
\end{thm}

We begin with the most common case of constant sampling and assignment propensities $\propselect(\psi) = \propselect$ and $\propfn(\psi) = \propfn$.
Then the variance in Theorem \ref{thm:clt_fixed} simplifies to
\begin{equation} \label{equation:variance:representative-sampling}
\varlocal = \var(\catefn(\psi)) + \frac{1}{\propselect} E\left[\frac{\hk_1(\psi)}{\propfn} + \frac{\hk_0(\psi)}{1-\propfn} \right].
\end{equation}

\textbf{Stratified Assignment.} 
If $\propselect = 1$, then the experiment size $\nsampled = \sum_i \Ti$ satisfies $\nsampled = n$ and our framework recovers the conventional superpopulation sampling model.
In this case, the variance $V$ is precisely the \cite{hahn1998} semiparametric variance bound for estimating the $\ate$. 
Building on \citet{armstrong2022}, \citet{bai2026efficiency} show that this efficiency bound is unchanged under a broad class of covariate-adaptive treatment assignment mechanisms, including the finely stratified designs considered here.

For iid treatments $\Di \simiid \bern(\propfn)$, this variance can typically only be attained by nonparametrically estimating and weighting by the realized propensity $\wh \propfn(\psi)$  \citep{hirano2003}, nonparametric regression adjustment \citep{robins1994}, or a combination of the two \citep{chernozhukov2017dml}.
By contrast, here we achieve this efficient variance using a simple IPW estimator that plugs in the known propensity or OLS $Y_i \sim 1 + \Di$ if $\propfn(\psi) = \propfn$. 
This shows how fine stratification nonparametrically controls the covariate imbalances by design, without the need for ex-post adjustment.

The previous work of \cite{bai2022pairs} showed the variance $V$ is attained by $\estreg$ under matched pairs randomization if the pairs satisfy a tight matching condition, corresponding to the case $\propfn=1/2$ in our framework.
\cite{bai2020pairs} extended this to the case $\propfn = a/k$ with covariate dimension $v = 1$.

\medskip

\textbf{Stratified Sampling and Assignment.}
Next, consider finely stratified sampling into the experiment with $\propselect < 1$.
For a fixed eligible population of size $n$, lowering $\propselect$ reduces the experiment size $\nsampled \approx \propselect n$.
Equation~\ref{equation:variance:representative-sampling} shows that this inflates the residual variance by a factor $\propselect^{-1}$ but leaves the first term $\var(\catefn(\psi))$, which captures treatment effect heterogeneity predictable by $\psi$, unchanged even as $\propselect$ decreases.

To see why, imagine an oracle setting where we exactly observe the treatment effect level $\catefn(\psii)$ for each sampled unit $\Ti = 1$, then estimate the $\ate$ by the sampled average $(\propselect n)\inv \sum_i \Ti \catefn(\psii)$.
Finely stratified sampling within groups satisfying the tight matching condition (Equation \ref{equation:homogeneity}) makes the distribution of covariates among the sampled units tightly coupled to that of the broader eligible population.
Because $\catefn(\psii)$ is a function of these covariates, the sampled average behaves like the infeasible average over all $n$ eligible units:
\begin{equation} \label{equation:sampling-coupling}
\frac{1}{\propselect n} \sum_{i=1}^n \Ti \catefn(\psii) = \en[\catefn(\psii)] + \op(n^{-1/2}).
\end{equation}
Equation~\eqref{equation:sampling-coupling} is the sampling-stage analog of what subsequent work in \citet{bai2026efficiency} has termed the ``fast-balancing'' property of fine stratification, with the relevant balance here occurring between sampled $\Ti=1$ and non-sampled units $\Ti=0$.
The tight matching conditions delivered by our algorithm (Corollary \ref{cor:tight-matching-survey}) guarantee such fast-balancing properties for fine stratification at both the sampling and assignment stages.

The second term in Equation~\eqref{equation:variance:representative-sampling} accounts for the variance from treatment effect heterogeneity orthogonal to $\psi$, as well as random imbalances in the potential outcomes during treatment assignment. Such fluctuations cannot be affected by stratified sampling on $\psi$ with constant $\propselect$.
To reduce this residual term, we need to change the sampling proportions $\propselect(\psi)$ as in Equation~\eqref{equation:asymptotic-theorem} to take more samples of higher-variability units.

\medskip

\textbf{Adjustment for Sampling Imbalances.} One could also try to correct for sampling imbalances between the eligible and experimental populations using ex-post adjustment.
\citet{dahabreh2019generalizing} implement this in a related setting with iid sampling and assignment $\Ti \simiid \bern(\propselect(\psii))$ and $\Di \simiid \bern(\propfn(\psii))$.
Their estimator adjusts for covariate imbalances during both sampling and assignment by estimating regressions $\ceffn_d(\psi) = E[Y(d)|\psi]$ and realized propensities $\wh \propfn(\psi)$ and $\wh \propselect(\psi)$.
\citet{li2023generalization} showed that this estimator has the same asymptotic variance $V$ as in Theorem \ref{thm:clt_fixed}, also proving that $V$ is the semiparametric efficiency bound in this iid setting.
We note that \citet{armstrong2022} appears to provide a route to extend this efficiency bound beyond the iid case to more general sampling and assignment mechanisms, including finely stratified sampling with $\propselect(\psi) \ne 1$, though we leave a formal validation of such bounds to future work.

In contrast to the adjustment strategies above, Theorem~\ref{thm:clt_fixed} shows that under finely stratified sampling $\Tn \sim \localdesigncond(\psi, \propselect(\psi))$ and assignment $\Dn \sim \localdesigncond(\psi, \propfn(\psi))$ we can attain the variance $\varlocal$ using a simple estimator $\est$ that plugs in the known propensities, or just using the linear regression $Y \sim 1 + D$ in the constant propensity case.
This shows how fine stratification nonparametrically controls covariate imbalances during both the sampling and assignment stages, obviating the need for ex-post adjustment. 

\begin{remark}[Comparison with Other Designs] \label{remark:comparison-other-designs}
One alternative to fine stratification is rerandomization, where one redraws random treatment assignments or sampling variables until a covariate balance criterion is satisfied.
\cite{li2018rerandomization} show that under rerandomized treatment assignment, $\est$ is asymptotically slightly less efficient than an interacted regression that adjusts linearly for $\psi$.
Similarly, \cite{yang2021} study a two-stage design that rerandomizes at both the sampling and assignment stages, showing it is asymptotically slightly less efficient than ex-post linear adjustment for covariate imbalances in $\psi$ at both stages.
By contrast, fine stratification using the algorithm in Section~\ref{section:algorithms} nonparametrically controls covariate imbalances at both stages.
\end{remark}

\begin{remark}[Smoothness Conditions] \label{remark:smoothness}
Previous work on efficiency of fine stratification has required Lipschitz continuity of the functions $E[Y(d)|\psi]$ and $\var(Y(d) | \psi)$, e.g.\ see \cite{bai2022pairs}.
We develop a novel technical analysis allowing these conditions to be removed, showing the universality of the efficiency gains from fine stratification under the mild moment condition $E[Y(d)^2] < \infty$.
\end{remark}

\subsection{Extensions} \label{subsection:asymptotics:extensions}
 
In this section, we extend our main result to allow for stratified sampling and assignment with different covariates $\psisamp$ and $\psiassign$, characterizing the optimal covariates to stratify on at each stage. 
We also extend our asymptotic theory to target estimation of the $\sate = \en[Y_i(1) - Y_i(0)]$ in the eligible population instead of the $\ate$.

\medskip

It may be useful to stratify on different sets of variables during the sampling and assignment stages.
For example, $\psisamp$ could contain publicly available administrative data, while $\psiassign$ contains further survey covariates collected after sampling.
The next result extends Theorem~\ref{thm:clt_fixed} to this setting.
To show the effect of targeting the $\ate$ vs.\ the $\sate$, we write the variance in an alternative form that more clearly disambiguates the contributions from each design stage.
To that end, define the individual treatment effect $\te = Y(1) - Y(0)$ and \emph{outcome level} $\ylevel = (1-\propfn) Y(1) + \propfn Y(0)$.

\begin{thm}[CLT] \label{thm:clt_two_psi}
Let $\psisamp = f(\psiassign)$ for some function $f(\cdot)$.
If $\Tn \sim \localdesigncond(\psisamp, \propselect)$ and $\Dn \sim \localdesigncond(\psiassign, \propfn)$, then $\rootn (\est - \ate) \convwprocess \normal(0, \varlocal)$ with
\begin{equation} \label{equation:variance:two-psi}
\varlocal = \var(\te) + \frac{1-\propselect}{\propselect} E[\var( \te | \psisamp)] + \frac{E[\var(\ylevel | \psiassign)]}{\propselect\propfn (1-\propfn)}.
\end{equation}
\end{thm}

Writing $\varlocal = V_1 + V_2 + V_3$, the term $V_1 = \var(\te)$ is the irreducible variance from iid sampling of the eligible units.
The term $V_2 = (\propselect^{-1}-1) E[\var(\te | \psisamp)]$ is the additional variance from finely stratified sampling from the eligible pool, which is small if $\psisamp$ is highly predictive of treatment effect heterogeneity.
Finally, $V_3 \propto E[\var(\ylevel | \psiassign)]$ is the variance from random imbalances in the outcome levels $\ylevel$ between $D=1$ and $D=0$, which is small if $\psiassign$ is highly predictive of $\ylevel$.
See \citet{abadie2020} for a related decomposition separating sampling vs.\ assignment based uncertainty.
If $\psisamp = \psiassign = \psi$, Equation \eqref{equation:variance:two-psi} can be rearranged to yield the variance in our earlier theorem \eqref{equation:asymptotic-theorem}.

The technical condition $\psisamp = f(\psiassign)$ is satisfied, for instance, if $\psiassign$ contains all of the variables originally included in $\psisamp$.
It is not immediately clear whether it can be removed with a finer analysis.

\medskip

\textbf{Finite Population Estimand.}
In some settings, the eligible population of size $n$ may constitute the entire population of interest.
In this case it is natural to instead target the sample average treatment effect, $\sate = \en[Y_i(1) - Y_i(0)]$.

\begin{cor}[$\sate$] \label{cor:clt4sate}
In the setting of Theorem~\ref{thm:clt_two_psi}, $\rootn(\est - \sate) \convwprocess \normal(0, V_{\sate})$,
\begin{equation} \label{equation:variance:clt4sate}
V_{\sate} = \frac{1-\propselect}{\propselect} E[\var(\te | \psisamp)] + \frac{E[\var(\ylevel | \psiassign)]}{\propselect\propfn (1 - \propfn)}.
\end{equation}
\end{cor}
Targeting the $\sate$ simply removes the variance from the first stage sampling of $n$ eligible units from the superpopulation, resulting in $V_{\sate} = \varlocal - V_1$.

Note that this differs from some alternative definitions of the $\sate$ as the average treatment effect among the $\nsampled$ experimental participants, $\nsampled\inv \sum_{i=1}^n \Ti(Y_i(1)-Y_i(0))$.
The proof also shows that targeting this alternative quantity gives asymptotic variance $\varlocal - V_1 - V_2 = V_3$, removing all sampling variance and leaving only the variance due to random imbalances in outcome levels during treatment assignment.

See Appendix~\ref{appendix:sate-targeting} for more discussion about targeting the $\sate$ vs.\ the $\ate$ and the interpretation of the population measure $P$ in this context.  

For completeness, we also provide versions of the preceding results for general sampling and assignment propensities $\propselect(\psi)$ and $\propfn(\psi)$ in Theorem~\ref{thm:clt-two-psi-extended} and Corollary~\ref{cor:clt-sate-extended} in the appendix.
Lemma~\ref{lemma:variance-equivalence} also provides an algebraically equivalent representation of the ATE variance in the more standard efficiency-bound form. 
Such expressions fail to cleanly separate the variance contributions from different sources of randomness.  

\medskip

\textbf{Optimal Covariates.}
The results above reveal an important asymmetry between the optimal covariates to use for stratified sampling and assignment.
For sampling, we want covariates $\psisamp$ that are highly predictive of individual treatment effect heterogeneity $\te = Y(1)-Y(0)$.
By contrast, for random assignment we want covariates $\psiassign$ that are highly predictive of outcome levels $\ylevel$.
To see this even more clearly, observe that  
\begin{equation*}
E[\var(\te | \psisamp)] = \min_{f(\cdot)} E[(\te - f(\psisamp))^2] \quad \text{and} \quad  E[\var(\ylevel | \psiassign)] = \min_{f(\cdot)} E[(\ylevel - f(\psiassign))^2].
\end{equation*}
In either case, it is important to choose a small set of such highly predictive covariates, since the finite sample variance $n \var(\est)$ converges to the variance bound $V$ slowly in high dimensions $v = \dim(\psi)$, with $n \var(\est) = V + O(n^{-2/(v + 1)})$ under suitable regularity conditions.
See the paper's arXiv version for details.
It is an open question to determine exactly how many covariates to use in a specific study given experimenter beliefs. 
One heuristic is to only finely stratify on baseline versions of the important endline outcomes.    

\section{Optimal Sampling for Survey Experiments} \label{section:optimal_designs}

In this section we apply our main asymptotic results to study the optimal sampling problem for budget-constrained experimentation with heterogeneous costs.
After characterizing the oracle sampling design, we propose a feasible implementation that estimates unknown quantities using an available pilot or observational data, providing conditions under which this is asymptotically optimal.

\subsection{Budget-Constrained Optimal Sampling}

Our goal is to choose a sampling propensity $\propselect(\psi)$ that minimizes the asymptotic variance of ATE estimation in an experiment with heterogeneous costs and fixed budget constraint.
For simplicity, throughout this section we let assignment propensity $\propfn(\psi) = \propfn$ constant, e.g.\ with leading case $\propfn=1/2$.  
By Theorem \ref{thm:clt_fixed}, $\rootn (\est - \ate) \convwprocess \normal(0, V(\propselect))$, 
\begin{equation} \label{equation:variance-objective}
V(\propselect) = \var(\catefn(\psi)) + E\left[\frac{1}{\propselect(\psi)} \left (\frac{\hk_1(\psi)}{\propfn} + \frac{\hk_0(\psi)}{1-\propfn} \right) \right].
\end{equation}

Similarly, by Corollary~\ref{cor:clt-sate-extended} in the appendix $\rootn(\est - \sate) \convwprocess \normal(0, V_{\sate}(\propselect))$ for $\Tn \sim \localdesigncond(\psi, \propselect(\psi))$, where $V_{\sate}(\propselect) = V(\propselect) - \var(\te)$ for $\te = Y(1) - Y(0)$.
Thus the optimal sampling propensity is the same whether we target the $\sate$ or the $\ate$.
We therefore restrict attention to the $\ate$ in what follows.

Consider minimizing Equation \ref{equation:variance-objective} over all sampling propensities $\propselect(\psi)$.
Clearly the unconstrained solution is $\propselect^*(\psi) = 1$, making the experiment as large as possible.
More generally, we can formalize a problem of sampling with heterogeneous costs subject to a fixed budget constraint.

\medskip

\textbf{Costs.} Define $\cost(\psi)$ to be the known cost of including a unit of type $\psi$ in the experiment.
One natural specification is $\cost(\psi) = \cost_s(\psi) + \propfn \cost_1(\psi) + (1-\propfn)\cost_0(\psi)$, where $\cost_s(\psi)$ is the sampling cost and $\cost_1(\psi)$ and $\cost_0(\psi)$ are the marginal costs of assigning treatment and control, respectively.
For example, in a development context $\cost_s(\psi)$ could be the cost of paying volunteers to collect outcome data in a village $\psii = \psi$.
We can interpret $\cost(\psi)$ as the ex-ante cost of sampling a unit with $\psii = \psi$ into the experiment, prior to realization of its random treatment assignment $\Di \in \{0,1\}$.

Define the ex-ante conditional variance  
\begin{equation} \label{equation:ex-ante-variance}
\hkavg(\psi) = \frac{\hk_1(\psi)}{\propfn} + \frac{\hk_0(\psi)}{1-\propfn}.
\end{equation}
Also let $\sdavg(\psi) = \hkavg(\psi)^{1/2}$ denote the corresponding ex-ante conditional standard deviation.
We interpret $\hkavg(\psi)$ as the expected residual variance from sampling a unit with covariates $\psii = \psi$ into the experiment.
Then $V(\propselect) = \var(\catefn(\psi)) + E[\hkavg(\psi)/\propselect(\psi)]$, so for a given budget $\budget$ we would like to solve 
\begin{equation} \label{equation:optimal_sampling:problem} 
\min_{0 < \propselect(\cdot) \leq 1} E\left[\frac{\hkavg(\psi)}{\propselect(\psi)}\right] \quad \text{s.t.} \quad E[\cost(\psi) \propselect(\psi)] = \budget.
\end{equation}

The next theorem characterizes the interior solutions to this problem.
\begin{thm} \label{thm:optimal_sampling_population}
Assume $E[Y(d)^2] < \infty$ for $d = 0, 1$, $\inf_{\psi} \hk_d(\psi) \geq c > 0$, and $\cost(\psi) \in [C_l, C_u] \sub (0, \infty)$. Define the candidate solution
\begin{equation} \label{equation:optimal_sampling:solution}
\propselectopt(\psi) = \budget \cdot \frac{\sdavg(\psi) \cost(\psi)^{-1/2}}{E[\sdavg(\psi) \cost(\psi) \half]}.
\end{equation}
If $\sup_{\psi} \propselectopt(\psi) \leq 1$, then $\propselect^*$ is optimal in Equation \ref{equation:optimal_sampling:problem}.
\end{thm}

To build intuition for the form of the solution, consider the following special cases. 
\begin{enumerate}[label={(\alph*)}, itemindent=.5pt, itemsep=.4pt] 
    \item \emph{Homoskedasticity}. 
Suppose $\hk_d(\psi) = \hk_d$ constant for $d \in \{0,1\}$.
Then $\hkavg(\psi)$ is constant and $\propselectopt(\psi) = \budget \cdot \cost(\psi)^{-1/2} / E[\cost(\psi) \half]$. 
In particular, $\propselectopt(\psi) \propto \cost(\psi)^{-1/2}$, so that more costly units are sampled with lower probability.
    \item \emph{Homogeneous Costs}.
If $\cost(\psi) = \cost$, then the budget constraint is $E[\propselect(\psi)] = \budget / \cost$.
Denoting $\propselectavg = \budget / \cost$, the optimal solution has form $\propselect^*(\psi) = \propselectavg \sdavg(\psi) / E[\sdavg(\psi)]$.
This shows how we would like to oversample units with above-average residual standard deviation, with $\propselect^*(\psi) > \propselectavg$ when $\sdavg(\psi) > E[\sdavg(\psi)]$ and the reverse otherwise.
\end{enumerate}

Under the homoskedasticity assumption in (a), the optimal propensity $\propselectopt(\psi)$ only depends on the known costs $\cost(\psi)$ and distribution of $\psi$, so we do not require outcome data to estimate $\propselectopt(\psi)$.
However, heteroskedasticity is believed to be widely prevalent in economic data (e.g.\ \cite{romano2017}), so the propensity in (a) may leave efficiency gains on the table.
In the general case, we also require an estimate of the unknown residual variance $\hkavg(\psi)$.

Note also that Equation~\ref{equation:optimal_sampling:solution} can be viewed as an analog of classical optimal allocation theory in survey sampling \citep{cochran1977}, here applied to budget-constrained finely stratified sampling into an experiment.

\subsection{Feasible Optimal Sampling} \label{section:feasible-optimal-sampling}

If data from a pilot experiment $(\Ti, \Di, Y_i, \psii)_{i=1}^{\npilot}$ is available, one can use the strategy of \citet{fan1998} to estimate the individual conditional variance functions $\hkest_d(\psi)$ for $d \in \{0, 1\}$ and form the composite $\hkavgest(\psi) = \hkest_1(\psi)/\propfn + \hkest_0(\psi)/(1-\propfn)$.

To do so, one first forms the conditional-mean signal $S_i(1) = Y_i \Di \Ti / (\propfni \propselecti)$, where $\propfni$ and $\propselecti$ are the pilot propensities, noting $E[S_i(1) | \psii] = E[Y_i(1) | \psii]$. Projecting $S_i(1)$ on $\psii$ via a nonparametric regression yields $\ceffnest_1(\psii)$. A second projection of the IPW-weighted squared residuals $(Y_i - \ceffnest_1(\psii))^2 \Di \Ti / (\propfni \propselecti)$ on $\psii$ then yields $\hkest_1(\psii)$. The analogous construction with $1 - \Di$ in place of $\Di$ and $1 - \propfni$ in place of $\propfni$ gives $\hkest_0(\psi)$.

The choice of nonparametric regression is flexible. We tested multiple methods, including kernel ridge regression and random forests, and found random forests perform best across the settings in our empirical application. 
See Appendix~\ref{appendix:pilot-variance-estimation} for details.

However, experimenters do not always have the resources to run a pilot. 
Those that do often use a very small sample size, making estimation of $\hkavgest(\psi)$ challenging.
Even without a pilot, observational data $(\Ti, Y_i, \psii)_{i=1}^{\nobs}$ for units under the control condition $Y_i = Y_i(0)$ for $\Ti=1$ may still be available.
If we are willing to assume $\hk_1(\psi) \approx \hk_0(\psi)$ are similar, then $\hkavg(\psi) \approx \hk_0(\psi)/(\propfn(1-\propfn))$, which can be estimated consistently as above using only data on the control outcome $(\Ti, Y_i, \psii)_{i=1}^{\nobs}$.

To see why this equal heteroskedasticity assumption is plausible, consider the survey in \citet{cai2024} of 10 recent RCTs published in the AER. 
They estimated the within-arm variance ratio $\hk_1 /\hk_0$ for each measured outcome, finding $\hk_1/\hk_0 \approx 1$ for the vast majority of them.
This provides empirical support for such an assumption.

Given such an estimator $\hkavgest(\psi)$ of $\hkavg(\psi)$, Theorem \ref{thm:optimal_sampling_population} suggests using fine stratification to implement a suitable discretization of the estimated optimal propensity
\begin{equation} \label{equation:pilot:propensity-estimator}
\wh \propselect(\psi) = \budget \cdot \frac{\sdavgest(\psi) \cost(\psi)^{-1/2}}{\en[\sdavgest(\psii) \cost(\psii) \half]}.
\end{equation}

Note that under homoskedasticity, the residual variance factors out of the expression above, and we can just set $\propselectest(\psi) = \budget \cdot \cost(\psi)^{-1/2} / \en[\cost(\psii)\half]$.

\medskip

\textbf{Discretization.} 
To implement the propensity estimate $\propselectest(\psi)$ using fine stratification, we construct a discretization $\wh \propselect_n(\psi)$ of $\propselectest(\psi)$ into the rational approximating set $\{a_l/k_l\}_{l=1}^{\nlevels}$ such that $\en[(\wh \propselect_n(\psii) - \wh \propselect(\psii))^2] = \op(1)$.

For example, $\wh \propselect_n(\psi)$ can be obtained by rounding $\wh \propselect(\psi)$ to the nearest $a/k_n$ for some sequence $k_n \to \infty$.
Larger $k_n$ gives worse match quality, but better implementation of the estimated propensity $\wh \propselect(\psi)$.
Thus, in finite samples this parameter trades off between the variance due to random covariate imbalances in $\psi$ versus optimization of the residual variance using $\propselect(\psi)$.
Below we provide conditions on the maximum tuple size $\kboundn = \max_{l \in [\nlevels]} k_l$ and number of propensity levels $\nlevels$ for this match-quality degradation to be asymptotically lower order.
In practice, it is often possible to choose a reasonable level of discretization fineness by inspecting the estimated $\wh \propselect(\psii)$ to see how many different propensity levels are needed.

\begin{assumption}[Feasible Optimal Sampling] \label{assumption:pilot-clt}
Suppose the following: 
\begin{enumerate}[label={(\roman*)}, itemindent=.5pt, itemsep=.4pt]
\item $\Tn \sim \localdesigncond(\psi, \propselectestn(\psi))$ and $\Dn \sim \localdesigncond(\psi, \propfn)$.
Variance estimates $\hkest_d(\psi)$ for $d \in \{0,1\}$ are computed from external data $\proprand$, independent of the experimental units and design variables.
\item Moments $E[Y(d)^4] < \infty$ for $d \in \{0,1\}$ and $E[|\psi|_2^{\alpha}] < \infty$ for some $\alpha > v + 1$.
The conditional variances are bounded below, $\inf_\psi \hk_d(\psi) > 0$. The propensities satisfy $\propselectestn(\psi), \propselectopt(\psi) \in (\propbound, 1]$. The costs have $\cost(\psi) \in [C_l, C_u] \subset (0, \infty)$.
\item Conditional variance estimates satisfy $E_\psi[(\hkest_d(\psi) - \hk_d(\psi))^2] = \op(1)$ for $d \in \{0,1\}$.
\item The discretization has $\en[(\propselectestn(\psii) - \propselectest(\psii))^2] = \op(1)$, with $\kboundn \nlevels = o(n^{1 - (v+1)/\alpha})$ and $\kboundn = o(n^{1/4})$ for $\kboundn = \max_{l \in [\nlevels]} k_l$.
\end{enumerate}
\end{assumption}

The consistency requirement on the estimates $\hkest_d(\psi)$ in part (iii) can be satisfied, for example, by the estimator of \citet{fan1998}. 
See Appendix \ref{appendix:pilot-variance-estimation} for a more detailed discussion of our implementation.
Part (iv) requires a discretization $\propselectestn(\psi)$ of increasing fineness, subject to the upper bound $\kboundn \nlevels = o(n^{1 - (v+1)/\alpha})$ for $v=\dim(\psi)$.
For example, Lemma \ref{lemma:discretization-rounding} shows that rounding $\propselectest(\psi)$ to the nearest $a/k_n$, for any $k_n \to \infty$ with $k_n = o(n^{(1-(v+1)/\alpha)/2} \wedge n^{1/4})$, satisfies part (iv).
Finally, the reason for the rate $\kboundn \nlevels = o(n^{1 - (v+1)/\alpha})$ in part (iv) is to ensure the sequence of local randomizations in part (i) exist.
See Lemma \ref{lemma:tight-matching-exists} for the formal result showing $F(\groupset_n^T) = \op(1)$ and $F(\groupset_n^D) = \op(1)$ under these conditions, including sampling-subordinate assignment.

The conditions above ensure the variance due to random covariate imbalances is lower order, while also allowing asymptotic implementation of the optimal $\propselectopt(\psi)$.

\begin{thm}[Feasible Optimal Sampling] \label{thm:pilot-design}
Impose Assumption \ref{assumption:pilot-clt}.
Suppose that $\Tn \sim \localdesigncond(\psi, \wh \propselect_n(\psi))$ and $\Dn \sim \localdesigncond(\psi, \propfn)$.
Then $\rootn (\est - \ate) \convwprocess \normal(0, V^*)$
\[
V^* = \var(\catefn(\psi)) + \min_{\substack{0 < \propselect(\cdot) \leq 1 \\ E[\cost(\psi) \propselect(\psi)] = \budget}} E\left[\frac{\hkavg(\psi)}{\propselect(\psi)} \right].
\]
\end{thm}

In particular, this design minimizes the asymptotic variance $V(\propselect)$ of Equation~\ref{equation:variance-objective} over all budget-feasible sampling propensities $\propselect(\psi)$.

If the condition $\sup_\psi \propselectopt(\psi) \leq 1$ fails, we may obtain infeasible propensity estimates $\wh \propselect(\psij) > 1$.
To restore feasibility, we can iteratively set $\wh \propselect(\psij) = 1$ for such $j$ and recompute the optimal propensity for the remaining units.

To do so, define an index set $J = \emptyset$ and implement the following iterative rounding procedure.
(1) Find the largest $\wh \propselect(\psij) > 1$.
Set $\wh \propselect(\psij) = 1$ and add $j$ to $J$.
(2) Recompute the sampling propensity on the remaining units 
\[
\wh \propselect(\psii) = \frac{\budget - (1/n) \sum_{l \in J} \cost(\psi_l)}{1 - |J|/n} \cdot \frac{\sdavgest(\psii) \cost(\psii)^{-1/2}}{\en[\sdavgest(\psi_l) \cost(\psi_l) \half |l \not \in J]} \quad \quad \forall \, i \not \in J.
\]
If $\max_{i=1}^n \wh \propselect(\psii) \leq 1$, stop.
Otherwise, return to (1).
A quick calculation shows that this procedure satisfies the in-sample budget constraint $\en[\wh \propselect(\psii) \cost(\psii)] = \budget$ after each iteration and terminates with $\max_{i=1}^n \wh \propselect(\psii) \leq 1$.

\begin{remark}[Optimal Assignment Propensity] \label{remark:optimal-assignment-propensity}
In an earlier version of this paper, we noted that for any fixed sampling propensity $\propselect(\psi)$, the globally optimal assignment propensity is the conditional Neyman allocation $\propfn^*(\psi) = \hksd_1(\psi) / 
(\hksd_1(\psi) + \hksd_0(\psi))$, providing a feasible version $\Dn \sim \localdesigncond(\psi, \wh \propfn_n(\psi))$ of this 
optimal design.
Recent work by \citet{cai2024} has documented that the marginal variances $\hk_1$ 
and $\hk_0$ tend to be similar in experimental settings in economics, suggesting the baseline $\propfn = 1/2$ will often be approximately optimal under equal costs.
\end{remark}

\section{Inference Methods} \label{section:inference}
This section provides new methods for asymptotically exact inference on the $\ate$ and asymptotically conservative inference on the $\sate$ under finely stratified sampling and assignment.
We build on the pairs-of-pairs estimators originally introduced in \cite{abadie2008} and studied in \cite{bai2022pairs} and \cite{bai2026variance}, among others.

We focus on the case with constant assignment propensity $\propfn = a/k$, so that each non-remainder assignment group $\group \in \groupset_n$ has $|\group| = k$.
Throughout this section, we identify the partition $\groupset_n$ with the non-remainder groups of full size $k$.
For each assignment group $\group \in \groupset_n$, define the group-level difference in means
\begin{equation} \label{equation:group-estimator}
\estg = \frac{1}{a} \sum_{i \in \group} \Di Y_i - \frac{1}{k - a} \sum_{i \in \group} (1 - \Di) Y_i.
\end{equation}
First suppose $\propselect(\psi)=\propselect$ is constant.
Then the IPW estimator \eqref{equation:estimator:varying} equals the mean of the group-level estimates $\estg$ up to $\Op(n\inv)$ remainder terms.
Our variance estimators combine a between-group sample variance $\wh S^2 = \frac{1}{|\groupset_n|} \sum_{\group \in \groupset_n} (\estg - \est)^2$ with a residual variance estimator $\wh P^2$ defined below, which captures variation conditional on $\psi$.
The estimators for inference on the $\sate$ and $\ate$ are
\begin{equation} \label{equation:variance-estimator}
\varest_{\sate} = \frac{k}{\propselect} \wh P^2, \qquad \varest = \wh S^2 + \frac{k - \propselect}{\propselect} \wh P^2.
\end{equation}
We give two constructions of the residual estimator $\wh P^2$. Both have the same probability limit, but differ in their finite-sample bias.

\textbf{Within-group estimator.}
When $\min(a, k-a) \geq 2$, we can estimate the within-$\psi$ residual variation directly inside each group.
Let $\bar Y_{1, \group} = a\inv \sum_{i \in \group} \Di Y_i$ and $\bar Y_{0, \group} = (k-a)\inv \sum_{i \in \group} (1-\Di) Y_i$ be the treated and control means in group $\group$, so that $\estg = \bar Y_{1, \group} - \bar Y_{0, \group}$, and define within-arm sample variances $s^2_{1, \group} = \frac{1}{a-1} \sum_{i \in \group} \Di (Y_i - \bar Y_{1, \group})^2$ and $s^2_{0, \group} = \frac{1}{k-a-1} \sum_{i \in \group} (1-\Di) (Y_i - \bar Y_{0, \group})^2$.
The within-group estimator is
\begin{equation} \label{equation:P2-within}
\wh P^2_{N} = \frac{1}{|\groupset_n|} \sum_{\group \in \groupset_n} \biggl( \frac{s^2_{1, \group}}{a} + \frac{s^2_{0, \group}}{k-a} \biggr).
\end{equation}
The summand $s^2_{1, \group}/a + s^2_{0, \group}/(k-a)$ is the Neyman variance estimator for the difference in means within group $\group$, treating each matched group as a small completely randomized two-arm experiment \citep{neyman1990, imbens2015}.

\textbf{Pairs-of-pairs estimator.}
If $\min(a, k-a) = 1$, then at least one of the sample variances above is infeasible, necessitating the use of pairs-of-pairs style constructions \citep{abadie2008}. 
First, suppose that $|\groupset_n|$ is even.
Let $\groupmatching : \groupset_n \to \groupset_n$ be a bijective matching with $\groupmatching(\group) \neq \group$ and $\groupmatching^2 = \identity$ such that the centroids $\bar\psi_\group = |\group|\inv \sum_{i \in \group} \psii$ satisfy the pairwise tight matching condition $n\inv \sum_{\group \in \groupset_n} |\bar\psi_\group - \bar\psi_{\groupmatching(\group)}|_2^2 = \op(1)$.
In practice, $\groupmatching$ is obtained by matching the centroids into pairs using the algorithm of Section~\ref{section:algorithms}, which guarantees this condition under the moment assumption in Theorem~\ref{thm:matching-text}.
The pairs-of-pairs estimator is
\begin{equation} \label{equation:S2-P2}
\wh P^2_{\nu} = \frac{1}{2|\groupset_n|} \sum_{\group \in \groupset_n} \bigl(\estg - \est_{\groupmatching(\group)}\bigr)^2.
\end{equation}
It compares estimates between tightly matched groups and applies for any $\propfn = a/k$.
If $|\groupset_n|$ is odd, we match one of the groups twice, contributing an extra term to $\wh P^2_{\nu}$.

\emph{Bias Comparison.} Inference on the $\sate$ using $\varest_{\sate} = (k/\propselect) \wh P^2_{\nu}$ is conservatively biased in finite samples.
In a design-based framework without stratified sampling, \citet{bai2026variance} show that $E[\varest_{\sate}] = V_{\sate} + B_n$ for bias $B_n = \frac{k}{2|\groupset_n|} \sum_{\group \in \groupset_n} \bigl(\theta_\group - \theta_{\groupmatching(\group)}\bigr)^2 \geq 0$, where $\theta_\group = k\inv \sum_{i \in \group} \te_i$ is the average treatment effect in group $\group$.
The factor $k$ shows this bias can be severe for the large tuples required by imbalanced assignment propensities $\propfn$ far from $1/2$.
For example, this estimator performs very poorly for the \cite{baysan2022} example in our empirical application, with $\propfn=2/11$. 
By contrast, a calculation shows that the bias of $k \wh P^2_N$ is of constant order in $k$, see Lemma~\ref{lemma:within-group-unbiased} for details.

We therefore set $\wh P^2 = \wh P^2_N$ when $\min(a, k-a) \geq 2$ and $\wh P^2 = \wh P^2_{\nu}$ otherwise.
Both estimators are transparently non-negative for every realization: $\wh S^2$ and $\wh P^2_{\nu}$ are sums of squares, and $\wh P^2_N$ is a sum of sample variances.

\medskip

\textbf{Varying Sampling Propensity.}
When $\propselect(\psi)$ takes finitely many values $\propselectl = a_l / k_l$, $l = 1, \dots, L$, we apply the constructions above within each sampling-propensity stratum $\{i : \propselect(\psii) = \propselectl\}$.
Under sampling-subordinate assignment (Section \ref{section:method}), assignment groups lie within these strata, and we also pair groups only within the same sampling strata when constructing $\wh P^2_{\nu}$.
Let stratum size $\nl = |\{i : \propselect(\psii) = \propselectl\}|$ and let $\wh\theta_l$, $\wh S_l^2$, and $\wh P_l^2$ denote the within-stratum versions of $\est$, $\wh S^2$, and $\wh P^2$.
The variance estimators are
\begin{equation} \label{equation:variance-estimator-varying}
\varest_{\sate} = \sum_{l=1}^{L} \frac{\nl}{n} \cdot \frac{k}{\propselectl} \wh P_l^2, \qquad \varest = \sum_{l=1}^{L} \frac{\nl}{n} \biggl(\wh S_l^2 + \frac{k - \propselectl}{\propselectl} \wh P_l^2 + (\wh\theta_l - \est)^2 \biggr).
\end{equation}
The $\ate$ estimator adds a cross-stratum correction $(\wh\theta_l - \est)^2$ that recovers between-stratum variation in the conditional average treatment effect.

\begin{thm}[Inference] \label{thm:inference}
Assume the conditions of Theorem~\ref{thm:clt_fixed} with assignment groups matched within sampling strata.
Then
\[
\varest \convp \varlocal, \qquad \varest_{\sate} \convp V_{\sate} + E[\var(\te | \psi)].
\]
\end{thm}

In the constant-$\propselect$ case with $L = 1$, the cross-stratum correction vanishes and the estimators reduce to \eqref{equation:variance-estimator}.
By this theorem and the CLT in Section~\ref{section:asymptotics}, the interval $\wh C = [\est \pm \varest\half c_{1 - \alpha/2} / \rootn]$ with critical value $c_\alpha = \Phi\inv(\alpha)$ is asymptotically non-conservative for the $\ate$, with $P(\ate \in \wh C) = 1 - \alpha + o(1)$.

By contrast, inference on the $\sate$ is necessarily conservative, since the asymptotic variance $V_{\sate}$ in Corollary~\ref{cor:clt4sate} depends on $E[\var(\te | \psi)]$, which is not identified because individual treatment effects $\te_i$ are not observed.
Combined with Corollary~\ref{cor:clt4sate}, the result above shows that the interval $\wh C_{\sate} = [\est \pm \varest_{\sate}\half c_{1 - \alpha/2} / \rootn]$ is asymptotically valid for the $\sate$, with $P(\sate \in \wh C_{\sate}) \geq 1 - \alpha + o(1)$.

\begin{remark}[Prior Work] \label{remark:related-variance-estimators}
For finely stratified assignment without stratified sampling, the pairs-of-pairs construction of $\varest_{\sate}$ was originally introduced by \citet{abadie2008} for $\propfn=1/2$. This is a direct analog of the classical collapsed-strata estimator from the survey-sampling literature, e.g.\ \cite{hansen1953}.
\cite{bai2026variance} extend the analysis of this estimator to the case $\propfn = a/k$, providing finite-sample bias guarantees in a design-based framework.
For inference on the $\ate$, our estimator $\varest$ can be shown to coincide with that of \cite{bai2022pairs} for matched pairs, $\propfn = 1/2$.
\end{remark}

\section{Empirical Results} \label{section:empirical}

In this section, we quantify the finite sample performance of each of our designs on $N=9$ real DGPs from experimental papers covering a range of fields in applied economics.
\subsection{Designs and Empirical Papers} \label{section:empirical-setup}

Our theoretical results showed separate variance reductions from (a) finely stratified treatment assignment, (b) finely stratified sampling, and (c) stratified implementation of the estimated optimal sampling propensity $\propselectopt(\psi)$.
To quantify the efficiency gain from each of these tools, we simulate unadjusted $\ate$ estimation under the following designs:

\begin{enumerate}[label={}, itemindent=.5pt, itemsep=0pt]
    \item \textbf{CR}: Complete randomization $\Tn \sim \crdist(\propselectestn)$ and $\Dn \sim \crdist(\propfn)$, with $\propselectestn$ a discretization of the budget-exhausting sampling propensity $\propselectest = \budget / \en[\cost(\psii)]$ and $\propfn$ the fixed assignment propensity from the original paper. In particular, we let $\propselectestn = a/k$, using the minimal $k$ such that $\propselectestn \cdot \en[\cost(\psii)] \in [.95 \budget, 1.05 \budget]$. The costs $\cost(\psi)$ and discretization rule are discussed in more detail below.
    \item \textbf{CR, Loc}: As in CR but with stratified assignment $\Dn \sim \localdesigncond(\psi, \propfn)$.
    \item \textbf{Loc}: Stratified sampling and assignment $\Tn \sim \localdesigncond(\psi, \propselectestn)$ and $\Dn \sim \localdesigncond(\psi, \propfn)$.
    \item \textbf{Hom}: As in \textbf{Loc} but with sampling propensity $\wh\propselect_n(\psi)$ a discretization of $\wh\propselect(\psi) = \budget \cdot \cost(\psi)^{-1/2} / \en[\cost(\psii) \half]$, the optimal sampling propensity assuming homoskedasticity.
     This is feasible but possibly misspecified.
    \item \textbf{Pilot S/L}: As in \textbf{Loc}, but with $\Tn \sim \localdesigncond(\psi, \wh\propselect_n(\psi))$ where $\wh\propselect_n(\psi)$ is the discretization of a pilot-based estimate $\wh\propselect(\psi)$ formed from an experiment of size $\npilot = \nsampled/5$ (S) or $\npilot = 4\nsampled/5$ (L) with $\propfn=1/2$, $\propselect=1$, and matched-pair assignment.
We estimate the ex-ante variance $\hkavgest(\psi)$ from the pilot by random forest regression as in Appendix~\ref{appendix:pilot-variance-estimation}, then form $\wh\propselect(\psi)$ via the optimal sampling formula in Equation~\eqref{equation:pilot:propensity-estimator}.
    \item \textbf{Obs}: As in \textbf{Pilot}, but $\wh\propselect(\psi)$ is estimated from observational data $(Y_i(0), \psii)_{i=1}^{\nobs}$ of size $\nobs = 4\nsampled$ under the working assumption $\hk_1(\psi) = \hk_0(\psi)$, as discussed in Section~\ref{section:optimal_designs}.
We estimate $\hkest_0(\psi)$ from $(Y_i(0), \psii)_{i=1}^{\nobs}$ as in Appendix~\ref{appendix:pilot-variance-estimation} and use the plug-in $\hkavgest(\psi) = \hkest_0(\psi)/(\propfn(1-\propfn))$ for the ex-ante variance $\hkavg(\psi)$ in Equation~\eqref{equation:ex-ante-variance}, then form $\wh\propselect(\psi)$ via Equation~\eqref{equation:pilot:propensity-estimator}.
\end{enumerate}

We use data from nine experimental papers published in the AER between November 2020 and November 2022: \cite{abebe2021}, \cite{baysan2022}, \cite{bergquist2020}, \cite{casey2021}, \cite{dellavigna2022}, \cite{domurat2021}, \cite{fink2020}, \cite{hussam2022}, and \cite{lowe2021}.
We exclude papers without available unit-level data or predetermined covariates, or whose experiments do not fit our one-period treatment framework, e.g.\ because of multiple interventions on the same unit with a time series structure.

For each paper, we use the observed data to impute a DGP to sample from.
In particular, for $d = 0,1$ we estimate the conditional mean $\ceffnest_d(\psi) = \wh E[Y(d) | \psi]$ and variance $\hkest_d(\psi) = \wh \var(Y(d) | \psi)$, and take the outcome law to be $Y(d) | \psi \sim \normal\bigl(\ceffnest_d(\psi), \hkest_d(\psi)\bigr)$. For the two papers with binary outcomes, \cite{casey2021} and \cite{domurat2021}, we instead take $Y(d) | \psi \sim \bern\bigl(\ceffnest_d(\psi)\bigr)$, where $\ceffnest_d(\psi) = \wh P(Y(d) = 1 | \psi)$ is a fitted classifier.
The conditional mean regression model for imputation is selected separately for each paper by out-of-sample cross-validation between OLS, lasso with interactions, gradient-boosted trees, and random forests.
For continuous outcomes, we estimate the conditional variance separately by applying a regularized random forest to squared cross-fitted residuals.
See Appendix~\ref{appendix:imputation-details} for further details.
Let $N_0$ denote the size of the original experiment.

\medskip \textbf{Simulation Design.}
We do the following:
\begin{enumerate}[label={(\arabic*)}, itemindent=.5pt, itemsep=.4pt]
    \item We draw $4\nsampled$ eligible units by sampling with replacement from the $N_0$ original units $(\psii)_{i=1}^{N_0}$, then perturb selected coordinates within observed-value cells so that exact ties $\psi_i = \psi_j$ have probability zero (Appendix~\ref{appendix:imputation-details}). For the continuous-outcome papers, we form outcomes $\tilde Y_i(d) = \ceffnest_d(\psii) + \hkest_d(\psii)\half \residuali^d$, $d = 0, 1$, with the residual $\residuali^d \sim \normal(0,1)$ drawn independently across units and $d \in \{0, 1\}$. For the binary-outcome papers, we instead draw Bernoulli outcomes using the fitted probabilities described above. The budget-exhausting constant propensity $\propselect = 1/4$ results in a representative experiment of size $\nsampled$.
    \item Randomize $\Tn$ and $\Dn$ according to one of the designs above.
    \item Reveal outcomes $\tilde Y_i = \Ti \Di \tilde Y_i(1) + \Ti (1-\Di) \tilde Y_i(0)$ and form the estimator $\est$. 
    Form confidence intervals $\wh C = [\est \pm \varest\half c_{1-\alpha/2}/\rootn]$ for the $\ate$ and $\wh C_{\sate} = [\est \pm \varest_{\sate}\half c_{1-\alpha/2}/\rootn]$ for the $\sate$, with $\alpha = 0.05$ and $\varest$, $\varest_{\sate}$ as in Section~\ref{section:inference}.
\end{enumerate}

Since the ATE of the estimated population is known, we can compute the standard deviation (SD), coverage probabilities, and percent reduction in confidence interval length for each design on each DGP.
We focus on the $\ate$ in this section for brevity. 
Empirical results for the $\sate$ are presented in Appendix~\ref{appendix:empirical-sate}.

% Generated deterministically by the replication package.
% Design: jittered DGP; pool=4n; budget=0.75; constant q=1/4.
% Target estimand: $\ate$ (col_sd='est_sd', col_ci='ci_radius', col_cov='coverage').
% Top panel reports %\Delta SD (vs CR baseline) so the SD/CI rows are in the same units;
% the gap between |%\Delta SD| and |%\Delta CI| is a visual diagnostic for over-conservative inference.
\begin{table}[htbp]
\begin{adjustbox}{width=0.9\columnwidth,center}
  \centering
    \begin{tabular}{rrrrrrrrrrr}
\cmidrule{2-11}          & \multicolumn{1}{c}{Design, Paper} & \multicolumn{1}{c}{Abe.} & \multicolumn{1}{c}{Bay.} & \multicolumn{1}{c}{Ber.} & \multicolumn{1}{c}{Cas.} & \multicolumn{1}{c}{Del.} & \multicolumn{1}{c}{Dom.} & \multicolumn{1}{c}{Fin.} & \multicolumn{1}{c}{Hus.} & \multicolumn{1}{c}{Low.} \\
\cmidrule{2-11}    \multicolumn{1}{c}{\multirow{7}[2]{*}{$\%\Delta$SD}} & \multicolumn{1}{c|}{CR} & \multicolumn{1}{c}{0} & \multicolumn{1}{c}{0} & \multicolumn{1}{c}{0} & \multicolumn{1}{c}{0} & \multicolumn{1}{c}{0} & \multicolumn{1}{c}{0} & \multicolumn{1}{c}{0} & \multicolumn{1}{c}{0} & \multicolumn{1}{c}{0} \\
          & \multicolumn{1}{c|}{CR, Loc} & \multicolumn{1}{c}{-6} & \multicolumn{1}{c}{-45} & \multicolumn{1}{c}{-4} & \multicolumn{1}{c}{1} & \multicolumn{1}{c}{-39} & \multicolumn{1}{c}{0} & \multicolumn{1}{c}{-16} & \multicolumn{1}{c}{-6} & \multicolumn{1}{c}{-12} \\
          & \multicolumn{1}{c|}{Loc} & \multicolumn{1}{c}{-8} & \multicolumn{1}{c}{-45} & \multicolumn{1}{c}{-2} & \multicolumn{1}{c}{0} & \multicolumn{1}{c}{-39} & \multicolumn{1}{c}{1} & \multicolumn{1}{c}{-19} & \multicolumn{1}{c}{-6} & \multicolumn{1}{c}{-13} \\
          & \multicolumn{1}{c|}{Hom.} & \multicolumn{1}{c}{-13} & \multicolumn{1}{c}{-44} & \multicolumn{1}{c}{-14} & \multicolumn{1}{c}{-8} & \multicolumn{1}{c}{-41} & \multicolumn{1}{c}{-6} & \multicolumn{1}{c}{-23} & \multicolumn{1}{c}{-13} & \multicolumn{1}{c}{-10} \\
          & \multicolumn{1}{c|}{Pilot S} & \multicolumn{1}{c}{-12} & \multicolumn{1}{c}{-42} & \multicolumn{1}{c}{-20} & \multicolumn{1}{c}{-8} & \multicolumn{1}{c}{-39} & \multicolumn{1}{c}{1} & \multicolumn{1}{c}{-22} & \multicolumn{1}{c}{-11} & \multicolumn{1}{c}{-18} \\
          & \multicolumn{1}{c|}{Pilot L} & \multicolumn{1}{c}{-11} & \multicolumn{1}{c}{-44} & \multicolumn{1}{c}{-20} & \multicolumn{1}{c}{-7} & \multicolumn{1}{c}{-40} & \multicolumn{1}{c}{-4} & \multicolumn{1}{c}{-21} & \multicolumn{1}{c}{-10} & \multicolumn{1}{c}{-18} \\
          & \multicolumn{1}{c|}{Obs} & \multicolumn{1}{c}{-14} & \multicolumn{1}{c}{-44} & \multicolumn{1}{c}{-19} & \multicolumn{1}{c}{-9} & \multicolumn{1}{c}{-40} & \multicolumn{1}{c}{-3} & \multicolumn{1}{c}{-22} & \multicolumn{1}{c}{-12} & \multicolumn{1}{c}{-18} \\
\cmidrule{2-11}    \multicolumn{1}{c}{\multirow{7}[2]{*}{$\%\Delta$CI}} & \multicolumn{1}{c|}{CR} & \multicolumn{1}{c}{0} & \multicolumn{1}{c}{0} & \multicolumn{1}{c}{0} & \multicolumn{1}{c}{0} & \multicolumn{1}{c}{0} & \multicolumn{1}{c}{0} & \multicolumn{1}{c}{0} & \multicolumn{1}{c}{0} & \multicolumn{1}{c}{0} \\
          & \multicolumn{1}{c|}{CR, Loc} & \multicolumn{1}{c}{-6} & \multicolumn{1}{c}{-45} & \multicolumn{1}{c}{0} & \multicolumn{1}{c}{-1} & \multicolumn{1}{c}{-38} & \multicolumn{1}{c}{-1} & \multicolumn{1}{c}{-15} & \multicolumn{1}{c}{1} & \multicolumn{1}{c}{-12} \\
          & \multicolumn{1}{c|}{Loc} & \multicolumn{1}{c}{-6} & \multicolumn{1}{c}{-45} & \multicolumn{1}{c}{0} & \multicolumn{1}{c}{-1} & \multicolumn{1}{c}{-38} & \multicolumn{1}{c}{-1} & \multicolumn{1}{c}{-14} & \multicolumn{1}{c}{0} & \multicolumn{1}{c}{-11} \\
          & \multicolumn{1}{c|}{Hom.} & \multicolumn{1}{c}{-13} & \multicolumn{1}{c}{-45} & \multicolumn{1}{c}{-12} & \multicolumn{1}{c}{-7} & \multicolumn{1}{c}{-39} & \multicolumn{1}{c}{-5} & \multicolumn{1}{c}{-19} & \multicolumn{1}{c}{-2} & \multicolumn{1}{c}{-10} \\
          & \multicolumn{1}{c|}{Pilot S} & \multicolumn{1}{c}{-12} & \multicolumn{1}{c}{-41} & \multicolumn{1}{c}{-16} & \multicolumn{1}{c}{-7} & \multicolumn{1}{c}{-38} & \multicolumn{1}{c}{1} & \multicolumn{1}{c}{-19} & \multicolumn{1}{c}{-2} & \multicolumn{1}{c}{-16} \\
          & \multicolumn{1}{c|}{Pilot L} & \multicolumn{1}{c}{-12} & \multicolumn{1}{c}{-43} & \multicolumn{1}{c}{-18} & \multicolumn{1}{c}{-6} & \multicolumn{1}{c}{-39} & \multicolumn{1}{c}{-4} & \multicolumn{1}{c}{-18} & \multicolumn{1}{c}{-2} & \multicolumn{1}{c}{-17} \\
          & \multicolumn{1}{c|}{Obs} & \multicolumn{1}{c}{-13} & \multicolumn{1}{c}{-42} & \multicolumn{1}{c}{-16} & \multicolumn{1}{c}{-7} & \multicolumn{1}{c}{-39} & \multicolumn{1}{c}{-5} & \multicolumn{1}{c}{-20} & \multicolumn{1}{c}{-2} & \multicolumn{1}{c}{-17} \\
\cmidrule{2-11}    \multicolumn{1}{c}{\multirow{7}[2]{*}{Cover}} & \multicolumn{1}{c|}{CR} & \multicolumn{1}{c}{0.95} & \multicolumn{1}{c}{0.94} & \multicolumn{1}{c}{0.95} & \multicolumn{1}{c}{0.94} & \multicolumn{1}{c}{0.94} & \multicolumn{1}{c}{0.94} & \multicolumn{1}{c}{0.94} & \multicolumn{1}{c}{0.93} & \multicolumn{1}{c}{0.94} \\
          & \multicolumn{1}{c|}{CR, Loc} & \multicolumn{1}{c}{0.95} & \multicolumn{1}{c}{0.94} & \multicolumn{1}{c}{0.95} & \multicolumn{1}{c}{0.94} & \multicolumn{1}{c}{0.95} & \multicolumn{1}{c}{0.94} & \multicolumn{1}{c}{0.95} & \multicolumn{1}{c}{0.94} & \multicolumn{1}{c}{0.95} \\
          & \multicolumn{1}{c|}{Loc} & \multicolumn{1}{c}{0.95} & \multicolumn{1}{c}{0.94} & \multicolumn{1}{c}{0.95} & \multicolumn{1}{c}{0.94} & \multicolumn{1}{c}{0.95} & \multicolumn{1}{c}{0.95} & \multicolumn{1}{c}{0.95} & \multicolumn{1}{c}{0.95} & \multicolumn{1}{c}{0.94} \\
          & \multicolumn{1}{c|}{Hom.} & \multicolumn{1}{c}{0.95} & \multicolumn{1}{c}{0.94} & \multicolumn{1}{c}{0.95} & \multicolumn{1}{c}{0.94} & \multicolumn{1}{c}{0.95} & \multicolumn{1}{c}{0.94} & \multicolumn{1}{c}{0.95} & \multicolumn{1}{c}{0.95} & \multicolumn{1}{c}{0.94} \\
          & \multicolumn{1}{c|}{Pilot S} & \multicolumn{1}{c}{0.95} & \multicolumn{1}{c}{0.94} & \multicolumn{1}{c}{0.96} & \multicolumn{1}{c}{0.94} & \multicolumn{1}{c}{0.95} & \multicolumn{1}{c}{0.94} & \multicolumn{1}{c}{0.95} & \multicolumn{1}{c}{0.95} & \multicolumn{1}{c}{0.95} \\
          & \multicolumn{1}{c|}{Pilot L} & \multicolumn{1}{c}{0.95} & \multicolumn{1}{c}{0.95} & \multicolumn{1}{c}{0.95} & \multicolumn{1}{c}{0.94} & \multicolumn{1}{c}{0.95} & \multicolumn{1}{c}{0.94} & \multicolumn{1}{c}{0.95} & \multicolumn{1}{c}{0.95} & \multicolumn{1}{c}{0.95} \\
          & \multicolumn{1}{c|}{Obs} & \multicolumn{1}{c}{0.95} & \multicolumn{1}{c}{0.95} & \multicolumn{1}{c}{0.95} & \multicolumn{1}{c}{0.94} & \multicolumn{1}{c}{0.95} & \multicolumn{1}{c}{0.94} & \multicolumn{1}{c}{0.95} & \multicolumn{1}{c}{0.95} & \multicolumn{1}{c}{0.95} \\
\cmidrule{2-11}          & \multicolumn{1}{c|}{$\nsampled$} & \multicolumn{1}{c}{1451} & \multicolumn{1}{c}{550} & \multicolumn{1}{c}{1359} & \multicolumn{1}{c}{91} & \multicolumn{1}{c}{446} & \multicolumn{1}{c}{1000} & \multicolumn{1}{c}{175} & \multicolumn{1}{c}{116} & \multicolumn{1}{c}{770} \\
          & \multicolumn{1}{c|}{$\dim(\psi)$} & \multicolumn{1}{c}{8} & \multicolumn{1}{c}{3} & \multicolumn{1}{c}{9} & \multicolumn{1}{c}{3} & \multicolumn{1}{c}{3} & \multicolumn{1}{c}{4} & \multicolumn{1}{c}{9} & \multicolumn{1}{c}{4} & \multicolumn{1}{c}{3} \\
\cmidrule{2-11}          &       &       &       &       &       &       &       &       &       &  \\
    \end{tabular}%
    \end{adjustbox}
    \caption{Empirical Results (ATE).}
    \label{table:empirical}%
\end{table}%

Descriptions of each paper, including the treatment and outcome variables and our choice of stratification variables $\psi$ are provided in Appendix~\ref{appendix:paper-details}.
Experiment sizes $\nsampled$ are as in the original papers, ranging from $\nsampled=91$ for \cite{casey2021} to $\nsampled=1451$ for \cite{abebe2021}.
The one exception is \cite{domurat2021}, for which we fit the DGP on 87,394 comparison observations and set $\nsampled=1000$ for Monte Carlo tractability.
Baseline treatment propensities $\propfn$ are set to match or closely approximate those of the original published designs: $\propfn = 1/2$ for \cite{abebe2021}, \cite{bergquist2020}, \cite{casey2021}, and \cite{dellavigna2022}, $\propfn = 2/3$ for \cite{fink2020} and \cite{lowe2021}, $\propfn = 2/11$ for \cite{baysan2022}, $\propfn = 4/5$ for \cite{domurat2021}, and $\propfn = 3/4$ for \cite{hussam2022}.
Five of nine applications have $\propfn \ne 1/2$, requiring the use of our general matched $k$-tuples algorithm at the assignment stage.
We use sampling-subordinate assignment, as described in Section~\ref{section:method}.

\emph{Costs and Discretization}.
The marginal cost $\cost(\psi)$ of including different units in the experiment is not reported in the papers in our sample.
To quantify variance reductions from representative sampling with $\propselectopt(\psi)$ in the type of DGPs that occur in applied economics research, we specify $\cost(\psi) = \one(|\psi|_2 \leq \kappa) + 5 \one(|\psi|_2 > \kappa)$ with $\kappa$ the median of $|\psii|_2$ over the eligible pool and $\budget = 0.75$, yielding the budget-exhausting constant propensity $\propselect = 1/4$ for the \textbf{CR}, \textbf{CR, Loc}, and \textbf{Loc} designs.
This propensity discretizes to matched $4$-tuples ($\propselectestn = a/k$ with $k = 4$), so our general matched $k$-tuples algorithm is used at the sampling stage in all nine papers and at the assignment stage in the five papers with $\propfn \ne 1/2$.
The varying-$\propselect$ designs (\textbf{Hom}, \textbf{Pilot S/L}, \textbf{Obs}) realize $\propselectestn(\psi) \in [0.1, 1]$, with empirical mean sampling rate $\en[\propselectestn(\psii)] \approx 0.34$ averaged across these designs and papers.
By spending more of the budget on cheaper units, these designs raise the average sampling rate above the constant $\propselectestn$ of \textbf{CR}, \textbf{CR, Loc}, and \textbf{Loc}, which has $\propselectestn = 1/4$ in our cost specification.

For each varying-$\propselect$ design (\textbf{Hom}, \textbf{Pilot S/L}, \textbf{Obs}), the discretization $\wh\propselect_n(\psi)$ of $\wh\propselect(\psi)$ is formed by choosing values in $\{a/10 : a = 1, \dots, 10\}$ to minimize $\en[(\wh\propselect_n - \wh\propselect)^2(\psii)]$ subject to $\nlevels = |\text{Image}(\wh\propselect_n)| \leq 2$ for $\nsampled < 500$, $\nlevels \leq 3$ for $500 \leq \nsampled < 1000$, and $\nlevels \leq 4$ for $\nsampled \geq 1000$.
This is motivated by the $\kboundn\nlevels=o(n)$ condition for tight matching with bounded covariates in Theorem~\ref{thm:algorithm-guarantee}.
Across the nine applications above, this discretization gives $\kboundn\nlevels/n \leq 0.06$ in every case.
We use an explicit rule so the design can be specified automatically in our Monte Carlo exercise.
When designing a single experiment in practice, experimenters can manually inspect the estimated $\wh\propselect(\psii)$ values to see how many levels $\nlevels$ of discretization may be needed. They can also compare how finer discretization reduces the approximation error $\en[(\wh\propselect_n - \wh\propselect)^2(\psii)]$ relative to any finite sample deterioration in the matching objective $F(\groupset_n^T)$ due to larger tuple sizes.

\subsection{Results} \label{section:empirical-results}

Our main results are presented in Table~\ref{table:empirical}.
The largest single SD reduction comes from finely stratified treatment assignment (\textbf{CR} vs.~\textbf{CR, Loc}), with an average of $-14\%$ across our sample and reductions exceeding $30\%$ in \cite{baysan2022} and \cite{dellavigna2022}.
For these datasets, finely stratified sampling provides only a very small additional gain: the average change in SD from \textbf{CR, Loc} to \textbf{Loc} is $-0.3\%$.
This is expected from the theory, since stratified sampling reduces only the treatment effect heterogeneity component $\var(\catefn(\psi))$ of the variance, and here the nonparametric $R^2$ of $\psi$ for $Y(1) - Y(0)$ in our imputed DGPs averages only $4\%$ across the nine papers.
The gains from the varying-rate designs below instead arise from heteroskedasticity and cost heterogeneity, which $\psi$ does capture.

Averaging over the nine papers, the marginal reductions over \textbf{Loc} are $-4.9\%$ for \textbf{Hom}, $-4.6\%$ for \textbf{Pilot S}, $-5.5\%$ for \textbf{Pilot L}, and $-5.8\%$ for \textbf{Obs}.
The homoskedasticity optimal design \textbf{Hom} reduces standard deviation in seven of nine papers, while the feasible \textbf{Pilot L} and \textbf{Obs} designs deliver larger reductions on average, though pilot estimation noise can offset the gain on some papers (e.g.\ \textbf{Pilot S} on \cite{baysan2022}).
Large pilots ($\npilot = 0.8 \nsampled$) perform modestly better than small pilots ($\npilot = 0.2 \nsampled$) on average, while the observational-data design \textbf{Obs}, which uses only $(Y_i(0), \psii)$ observations under the assumption $\hk_1(\psi) = \hk_0(\psi)$, is empirically comparable to \textbf{Pilot L}.
This provides further empirical support for this equal-heteroskedasticity assumption.

For inference, coverage is close to nominal across all 63 (paper, design) cells, with mean coverage $0.946$ and range $[0.93, 0.96]$.
Across the 54 non-\textbf{CR} cells, the mean absolute difference between $\%\Delta$CI and $\%\Delta$SD is $2.3$ percentage points, so the confidence intervals are well calibrated and recover nearly all of the efficiency gains from finely stratified sampling and assignment.
Empirical results for the $\sate$, presented in Section~\ref{appendix:empirical-sate}, are similarly well calibrated, with confidence intervals nearly identical in length to their $\ate$ counterparts.

\section{Recommendations for Practice} \label{section:recommendations-for-practice}

First, we discuss the choice of matching algorithm.
For fine stratification with covariate dimension $v = 1$, one can simply sort the units by their $\psii$ values, as in \citet{bai2020pairs}.
For matched pairs randomization with $v \ge 2$, the \cite{derigs1988} algorithm implemented in the \texttt{nbpMatching} package \citep{beck2024nbpmatching} is tractable for low tens of thousands of units, with a match quality guarantee established by \citet{bai2022pairs}.

For all other cases, including larger matched pairs experiments with $v \ge 2$, we recommend the new matching algorithm developed in Section~\ref{section:algorithms}. This is the first algorithm to deliver provable match quality guarantees for matched $k$-tuples of any size $k \ge 2$ and $v \ge 1$.
The spatial sorting step alone scales to massive experiments with millions of units, while the refined procedure adding balanced $k$-means remains tractable at $n=100{,}000$ units.
These guarantees hold whether the stratification variable $\psi$ is fully continuous, fully discrete, or a mix of the two, requiring only the weak moment condition of Theorem~\ref{thm:matching-text}.

Next we turn to the design itself.
Finely stratified assignment delivers the single largest marginal reduction in standard deviation in our empirical application of Section~\ref{section:empirical}, averaging $-14\%$ across the nine papers.
Our theoretical results show that finely stratified sampling into the experiment can provide an additional nonparametric variance reduction, which is larger the more eligible units $n$ we have to choose from and the more predictive $\psi$ is of treatment effect heterogeneity.
This gives researchers an incentive to source a large pool of eligible units, then use fine stratification to sample representative participants.

In particular, for sampling, researchers should stratify on covariates expected to predict treatment effect heterogeneity, as has been advocated by the J-PAL guide on sampling for experiments \citep{jpalSampling}.
For treatment assignment, the optimal covariates are instead those most predictive of outcome levels, though baseline outcomes typically serve both roles well.
If these covariates are indeed predictive, our inference methods report smaller standard errors.

One should stratify on a small set of such covariates, since match quality deteriorates rapidly as covariate dimension $v$ grows.
Researchers should avoid matching on highly predictive covariates like baseline outcomes alongside many weakly predictive ones.

Finally, when sampling costs or residual variances vary across covariates, a non-constant sampling propensity can provide further variance reductions.
The simplest version of such a design assumes homoskedasticity and uses only cost information without the need for auxiliary outcome data.
Versions that estimate the optimal propensity using small pilots, large pilots, or observational data performed similarly, reducing standard deviation by roughly $5$-$6\%$ relative to constant-propensity fine stratification.
Since a large pilot is rarely available in practice, the observational-data design is an attractive default when control-arm data on $(Y(0), \psi)$ can be obtained at low cost. 

In discretizing the estimated optimal propensity, researchers can compare propensity approximation error vs.\ realized matching objective values, since finer discretizations may require larger tuple sizes and more propensity levels, which can worsen finite-sample match quality.
When costs and residual variance are roughly homogeneous, these considerations may favor a constant sampling propensity $\propselect$.

In all cases, our inference methods performed well, giving close-to-nominal coverage for the $\ate$ and valid coverage for the $\sate$, with confidence interval lengths reflecting the efficiency gains from both finely stratified sampling and assignment.

\typeout{}
\begingroup
\singlespacing
\fontsize{11.5pt}{13.8pt}\selectfont
\bibliography{design_references.bib}
\endgroup

\appendix
\renewcommand{\thesection}{A}

\section{Proof of Matching Algorithm Guarantee} \label{proofs:matching}

Write $v = \dim(\psi)$ and, when the dimension varies with $n$, write $v_n = \dim(\psi_{i,n})$.
In this section, we work with triangular arrays $(\psiin)_{i=1}^n \sub \mr^{v}$.
For any subset $S_n \sub [n]$, let $(\snl)_{l=1}^{\nlevels}$ partition $S_n$, with prescribed positive integer tuple sizes $\kl$ and $\kboundn = \max_l \kl$.
We apply steps (1) through (3) of Section \ref{section:algorithms} separately within each $\snl$, using tuple size $\kl$, the adaptive grid-selection rule, and $T_l \geq 0$ balanced $k$-means iterations.
The propensity strata of Section \ref{section:algorithms} give one instance of this construction.
Let $\remainderl$ be the remainder group removed in step (1), and write $\snl' = \snl \setminus \remainderl$ and $\nl = |\snl'|$.
Write $\remainderset_n = \{\remainderl : l \in [\nlevels], |\remainderl|>0\}$ for the collection of nonempty remainder groups.
Typical choices are $S_n = [n]$ at the sampling stage and $S_n = \{i : \Ti = 1\}$ at the assignment stage.
Define the squared maximal coordinate range $\mn^2 = \max_{1 \le j \leq v}(\max_{i \in [n]} \psi_{i,n}^{(j)} - \min_{i \in [n]} \psi_{i,n}^{(j)})^2$.

This section proves the guarantees claimed for the algorithm in Section \ref{section:algorithms}.
Denote objective function $F(\groupset) = n\inv \sum_{\group \in \groupset} \sum_{i \in \group} |\psiin - \bar\psi_\group|_2^2$ for any collection of disjoint groups $\groupset \sub 2^{[n]}$, where centroid $\bar\psi_\group = |\group|\inv \sum_{i \in \group} \psiin$.
Then $F(\groupset \cup \groupset') = F(\groupset) + F(\groupset')$ for disjoint $\groupset, \groupset'$.
The proofs often use the objective function's equivalent representation as a sum of within-group pairwise differences
\begin{equation} \label{equation:F-pair-form}
F(\groupset) = \frac{1}{2n} \sum_{\group \in \groupset} \frac{1}{|\group|} \sum_{i,j \in \group} |\psiin - \psijn|_2^2.
\end{equation}
This follows from the identity $\sum_{i,j \in \group} |\psiin - \psijn|_2^2 = 2|\group| \sum_{i \in \group} |\psiin - \bar\psi_\group|_2^2$.

\begin{assumption}[Matching] \label{ass:matching}
For every $l \in [\nlevels]$ with $\nl>0$ and $j \in [v]$, coordinate range $\max_{i \in \snl'} \psi_{i,n}^{(j)} - \min_{i \in \snl'} \psi_{i,n}^{(j)} > 0$.
\end{assumption}

As in Section \ref{section:algorithms}, this is without loss since a covariate with zero range in a nonempty stratum can be excluded from step (2) within that stratum without changing the matching objective.

\begin{thm}[Spatial Sorting] \label{thm:matching-per-stratum}
Require \ref{ass:matching}.
For every $n$ and $v \geq 1$, the spatial sorting algorithm described above returns $\groupset_n^{(0)} = \cup_{l=1}^{\nlevels} \groupsetnl^{(0)} \cup \remainderset_n$ with
\begin{equation} \label{equation:homogeneity-rate-per-stratum}
F(\groupset_n^{(0)}) \leq 6 \mn^2 v (n / \kboundn \nlevels)^{-2/(v+1)}.
\end{equation}
\end{thm}

\begin{proof}
We begin by reducing the problem to the case where all per-stratum points lie in $[0,1]^v$ and the remainder groups have been removed.
Additivity of $F$ and group disjointness implies $F(\groupset_n^{(0)}) = \sum_{l=1}^{\nlevels} F(\groupsetnl^{(0)}) + F(\remainderset_n)$.
By Lemma \ref{lem:remainder-negligible}, the remainder satisfies $F(\remainderset_n) \leq \mn^2 v (n/\kboundn \nlevels)^{-2/(v+1)}$.

Then consider $\sum_{l=1}^{\nlevels} F(\groupsetnl^{(0)})$.
For each $l \in [\nlevels]$ and $i \in \snl'$, define $\phi_{i,n} \in [0,1]^v$ to be the rescaled $\psiin$ defined in the preprocessing step in part (1) of Section \ref{section:algorithms}.
By Lemma \ref{lem:unit-cube-reduction} applied to $\mathcal{S} = \snl'$, we have $F(\groupsetnl^{(0)}, \psi) \leq \mn^2 \cdot F(\groupsetnl^{(0)}, \phi)$.
Summing over $l$ and using the remainder bound above, it suffices to show $\sum_{l=1}^{\nlevels}F(\groupsetnl^{(0)},\phi) \leq 5v(n/\kboundn\nlevels)^{-2/(v+1)}$.
Then we relabel $\phi$ as $\psi$ and assume $\{\psiin : i \in \snl'\} \sub [0,1]^v$ for each $l$ in what follows.

Fix $l \in [\levels]$ and grid size parameter $m_l \geq 1$.
Let $\indfn(\psi) = \indfn_v(\lfloor \ml \psi \rfloor \wedge (\ml-1))$ denote the sorting index and $B(0), \dots, B(m_l^v - 1)$ the block partition of $[0,1]^v$ from Lemma \ref{lem:snake-index}.
For each unit $i \in \snl'$ set $t(i) = \indfn(\psiin) \in \{0, \dots, m_l^v - 1\}$.
Then $\psiin \in B(t(i))$ holds since by definition $B(t) = \indfn^{-1}(t)$.
The partition $\groupsetnl^{(0)} = \{\groupls : s = 1, \dots, K_l\}$ is formed by sorting the non-remainder units on $t(i)$, breaking ties at random, so the groups satisfy the following \emph{monotonicity} property: for group indices $1 \leq s < s' \leq K_l$, we have 
\begin{equation} \label{equation:ordering_property_per_stratum}
i \in \groupls, \; j \in \group_{l, s'}, \; \; \implies \;\; t(i) \leq t(j).
\end{equation}

We will bound the terms in $F(\groupsetnl^{(0)})$ by a case analysis, depending on whether each pair $i, j \in g$ for $g \in \groupsetnl^{(0)}$ lies in the same block $t(i) = t(j)$, or different blocks $t(i) \ne t(j)$.  

\emph{Case 1.} Suppose $t(i) = t(j)$.
Then both points lie in a common block of side length $1/m_l$, so $|\psiin - \psijn|_2 \leq \sqrt{v}/m_l$ by item (2) of Lemma \ref{lem:snake-index}.
Then summing over $\groupsetnl^{(0)}$,
\begin{equation} \label{equation:same-block-bound}
\begin{aligned}
&\sum_{\group \in \groupsetnl^{(0)}} \frac{1}{|\group|} \sum_{i \ne j \in \group} \one\{t(i) = t(j)\} |\psiin - \psijn|_2^2 \le \frac{v}{m_l^2} \sum_{\group \in \groupsetnl^{(0)}} \frac{1}{|\group|} \sum_{i \ne j \in \group} \one\{t(i) = t(j)\}  \\
&\le \frac{v}{m_l^2} \sum_{\group \in \groupsetnl^{(0)}} \frac{|\group|(|\group|-1)}{|\group|}
= \frac{v}{m_l^2} \sum_{\group \in \groupsetnl^{(0)}} (|\group|-1) \leq \frac{v \nl}{m_l^2}.
\end{aligned}
\end{equation}
The final inequality holds since $\group \in \groupsetnl^{(0)}$ are disjoint subsets of $\snl'$ with $|\snl'| = \nl$.

\emph{Case 2.} Suppose $t(i) \neq t(j)$. 
Without loss, suppose $t(i) < t(j)$. 
By item (1) of Lemma \ref{lem:snake-index}, for each index $t$ with $t(i) < t < t(j)$ we can pick an arbitrary element $y_t \in B(t)$. 
Also set $y_{t(i)} = \psiin$ and $y_{t(j)} = \psijn$. 
Item (3) of Lemma \ref{lem:snake-index} gives the bound $|y_{t+1} - y_t|_2 \leq 2\sqrt{v} / m_l$ for any such $t$.
Then by telescoping, triangle inequality, and since $|\psiin - \psijn|_2 \leq \sqrt{v}$ on $[0,1]^v$, we have
\begin{equation} \label{equation:pair-snake-bound}
|\psiin - \psijn|_2^2 \leq \sqrt{v} \cdot |\psiin - \psijn|_2 \leq \sqrt{v} \sum_{t=t(i)}^{t(j)-1} |y_{t+1} - y_t|_2 \leq \frac{2 v}{m_l} \, |t(i) - t(j)|.
\end{equation}
Within each group $\groupls$, enumerate the units as $\groupls = \{u_{s,1}, \dots, u_{s,\denoml}\}$.
Re-indexing the inner sum by this enumeration and applying \eqref{equation:pair-snake-bound} termwise, we obtain 
\begin{align*}
& \sum_{i, j \in \groupls} \one\{t(i) \neq t(j)\} |\psiin - \psijn|_2^2 = \sum_{a, b = 1}^{\denoml} \one\{t(u_{s,a}) \neq t(u_{s,b})\} |\psi_{u_{s,a},n} - \psi_{u_{s,b},n}|_2^2 \\
&\leq \frac{2v}{m_l} \sum_{a, b = 1}^{\denoml} \one\{t(u_{s,a}) \neq t(u_{s,b})\} |t(u_{s,a}) - t(u_{s,b})| \leq \frac{2v}{m_l} \sum_{a, b = 1}^{\denoml} |t(u_{s,a}) - t(u_{s,b})|.
\end{align*}
For any $1 \leq s < K_l$ and $a,b \in \{1, \dots, \denoml\}$, monotonicity \eqref{equation:ordering_property_per_stratum} gives $t(u_{s,a}) \leq t(u_{s+1,b})$.
Then the $K_l$ intervals $I^s_{ab}$ of the form $[t(u_{s,a}), t(u_{s,b})]$ or $[t(u_{s,b}), t(u_{s,a})]$ are disjoint up to endpoints with $\cup_{s=1}^{K_l} I^s_{ab} \sub [0, m_l^v - 1]$.
With $\mc L(\cdot)$ Lebesgue measure, this implies
\begin{align*}
\sum_{s=1}^{K_l} |t(u_{s,a}) - t(u_{s,b})| = \sum_{s=1}^{K_l} \mc L(I^s_{ab}) = \mc L\left(\cup_{s=1}^{K_l} I^s_{ab}\right) \leq \mc L([0, m_l^v - 1]) \le m_l^v.
\end{align*}
Then aggregating using the bound above, we have 
\begin{align*} 
&\sum_{s=1}^{K_l} \frac{1}{\denoml} \sum_{i,j \in \groupls} \one\{t(i) \neq t(j)\} |\psiin - \psijn|_2^2 \leq \sum_{s=1}^{K_l} \frac{1}{\denoml} \frac{2v}{m_l} \sum_{a,b=1}^{\denoml} |t(u_{s,a}) - t(u_{s,b})| \\
&\leq \frac{1}{\denoml} \denoml(\denoml-1) \frac{2v}{m_l} m_l^v = 2 v (\denoml - 1) m_l^{v-1} \leq 2 v \denoml m_l^{v-1}.
\end{align*}
Combining the cases above and applying the pair-form representation \eqref{equation:F-pair-form}, we obtain $F(\groupsetnl^{(0)}) \leq (2n)\inv[v \nl / m_l^2 + 2 v \denoml m_l^{v-1}]$.
When $v=1$, write $x_l=(\nl/\denoml)^{1/2}$ and set $m_l=\max\{1,\lfloor x_l\rfloor\}$.
When $v \geq 2$, write $x_l=(\nl/[\denoml(v-1)])^{1/(v+1)}$ and again set $m_l=\max\{1,\lfloor x_l\rfloor\}$.
We claim that in either case,
\[
F(\groupsetnl^{(0)}) \leq \frac{5v}{n} \denoml^{2/(v+1)} \nl^{(v-1)/(v+1)}.
\]
If $\nl=0$, then $F(\groupsetnl^{(0)})=0$.
Hence suppose $\nl>0$, so divisibility by $\denoml$ gives $\nl/\denoml \geq 1$.
First suppose $v=1$.
Then $x_l \geq 1$, $m_l=\lfloor x_l\rfloor \geq x_l/2$, and the preceding bound gives $F(\groupsetnl^{(0)}) \leq (2n)\inv(4\nl/x_l^2+2\denoml)=3\denoml/n \leq 5\denoml/n$, proving the claim.
Now suppose $v \geq 2$.
If $x_l \geq 1$, then $x_l/2 \leq m_l \leq x_l$, and hence $F(\groupsetnl^{(0)}) \leq (2n)\inv(4v\nl/x_l^2+2v\denoml x_l^{v-1})$.
By the definition of $x_l$, we have $\nl=\denoml(v-1)x_l^{v+1}$.
Therefore,
\begin{align*}
F(\groupsetnl^{(0)}) &\leq \frac{v\denoml}{2n}\{4(v-1)+2\}x_l^{v-1} =\frac{v(2v-1)}{n}(v-1)^{-(v-1)/(v+1)}\denoml^{2/(v+1)}\nl^{(v-1)/(v+1)}.
\end{align*}
The constant $(2v-1)(v-1)^{-(v-1)/(v+1)}=2(v-1)^{2/(v+1)}+(v-1)^{-(v-1)/(v+1)} \leq 5$.
To see this, note $(v-1)^{-(v-1)/(v+1)} \leq 1$ for $v \ge 2$, while $v-1 \leq 2^{(v+1)/2}$ implies $2(v-1)^{2/(v+1)} \leq 4$.
Thus the claim holds when $x_l \geq 1$.
If $x_l<1$, then $m_l=1$, and the bound above is $F(\groupsetnl^{(0)}) \leq (2n)\inv[v\nl/m_l^2+2v\denoml m_l^{v-1}]=(2n)\inv(v\nl+2v\denoml)=v\denoml(\nl/\denoml+2)/(2n)$.
Then it suffices to establish $\nl/\denoml+2 \leq 4(\nl/\denoml)^{(v-1)/(v+1)}$.
Indeed, $x_l<1$ implies $\nl/\denoml<v-1$, so $(\nl/\denoml)^{2/(v+1)} \leq (v-1)^{2/(v+1)} \leq 2$ as shown above, so $\nl/\denoml \leq 2(\nl/\denoml)^{(v-1)/(v+1)}$.
Since the stratum is nonempty, $\nl/\denoml \geq 1$ also gives $2 \leq 2(\nl/\denoml)^{(v-1)/(v+1)}$.
Adding these inequalities proves the final claim.
Summing $l = 1, \dots, \nlevels$ and using $\denoml \leq \kboundn$ then gives
\begin{align*}
\sum_{l=1}^{\nlevels} F(\groupsetnl^{(0)}) &\leq \frac{5v}{n} \kboundn^{2/(v+1)} \sum_{l=1}^{\nlevels} \nl^{(v-1)/(v+1)} = \frac{5v}{n} \kboundn^{2/(v+1)} \nlevels \bigg( \frac{1}{\nlevels} \sum_{l=1}^{\nlevels} \nl^{(v-1)/(v+1)} \bigg).
\end{align*}
Applying Jensen to $x \mapsto x^{(v-1)/(v+1)}$ and using $\sum_l \nl \leq n$, this is bounded above by
\[
\frac{5v}{n} \kboundn^{2/(v+1)} \nlevels \bigg( \frac{1}{\nlevels} \sum_{l=1}^{\nlevels} \nl \bigg)^{(v-1)/(v+1)} \leq 5v (\kboundn \nlevels)^{2/(v+1)} n^{-2/(v+1)}.
\]
By the work above and the reductions at the start of the proof, this proves the theorem for the prescribed grid sizes $m_l$ given above.
For data-adaptive grid-size selection, let $\groupsetnl^{(0)}(m)$ denote the initial partition in stratum $l$ using grid size $m$. Let $\mathcal M_l \sub \{1,2,\ldots\}$ be any finite candidate set containing $m_l$, and let $\widehat m_l \in \argmin_{m \in \mathcal M_l}F(\groupsetnl^{(0)}(m))$.
Since $m_l \in \mathcal M_l$, we also have $\sum_{l=1}^{\nlevels}F(\groupsetnl^{(0)}(\widehat m_l)) \leq \sum_{l=1}^{\nlevels}F(\groupsetnl^{(0)}(m_l)) \leq 5v(\kboundn\nlevels)^{2/(v+1)}n^{-2/(v+1)}$.
\end{proof}

\begin{remark}[Rate Intuition] \label{remark:sorting-rate-comparison}
We briefly compare our proof with \citet{bai2022pairs}.
For $k=2$, they first bound $\sum_g \sum_{i \neq j \in g} |\psi_i-\psi_j|_2^2 \leq \sqrt{v} \sum_g \sum_{i \neq j \in g} |\psi_i-\psi_j|_2$.
They then control both within-cell and between-cell pair discrepancies using the unsquared distances on the right, giving terms of order $n/m$ and $m^{v-1}$.
Balancing terms with $m^v \asymp n$ yields rate $n^{-1/v}$.
By contrast, for a single propensity stratum and fixed group size $k$, we retain the squared cell-diameter bound $1/m^2$ for within-cell pairs and use the telescoping unsquared distance bound only for between cell pairs, giving terms of order $n/m^2$ and $m^{v-1}$.
Balancing terms with $m^{v+1} \asymp n$ yields the sharper rate $n^{-2/(v+1)}$.
\end{remark}

\begin{remark}[Lower Bound] \label{remark:matching-lower-bound}
Define covariate imbalance $B_{n,v}=2n^{-1}\sum_{i=1}^n\psi_i(2D_i-1)$.
For matched-pairs randomization with partition $\groupset_n$, $E[|B_{n,v}|_2^2 \mid \psi_{1:n}]=8F(\groupset_n)/n$.
Under Assumption~3 of \citet{davezies2025cube}, the covariates lie in $[0,1]^v$ and have a density bounded uniformly in $v$.
Equation~(A.8) in that paper shows that the matched-pair partition minimizing $F$ satisfies $E[|B_{n,v}|_2^2] \geq (1/3-o(1))v/n$ when $v=\omega(\log n)$.
The identity above then gives $E[F(\groupset_n)] \geq (1/24-o(1))v$ for this minimizing partition and hence for every matched-pair partition.
Moreover, $F(\groupset_n) \leq v$ under their normalization, so for any fixed $0<c_0<1/24$, $E[F(\groupset_n)] \leq c_0v+vP(F(\groupset_n) \geq c_0v)$.
It follows that $\liminf_{n \to \infty}P(F(\groupset_n) \geq c_0v) \geq 1/24-c_0>0$.
\end{remark}

\begin{thm}[Algorithm Guarantee] \label{thm:algorithm-guarantee}
Require Assumption \ref{ass:matching}.
For every $n$ and $v \geq 1$, the algorithm described above with any per-stratum balanced-$k$-means iteration counts $T_l \ge 0$ returns a partition $\groupset_n = \cup_{l=1}^{\nlevels} \groupsetnl^{(T_l)} \cup \remainderset_n$ with
\begin{equation} \label{equation:algorithm-guarantee-rate}
F(\groupset_n) \leq 6v \mn^2  (n/\kboundn \nlevels)^{-2/(v+1)}.
\end{equation}
\begin{enumerate}[label={(\arabic*)}, itemindent=.5pt, itemsep=.4pt, topsep=.4pt]
\item Let $v=v_n$. If $\sup_{n \geq 1}\mn \leq c < \infty$, then $F(\groupset_n) \leq 6c^2 \cdot v_n (n / \kboundn \nlevels)^{-2/(v_n+1)}$.
\item If $(\psi_{i,n})_{i=1}^n$ are iid within each row and $\{|\psi_{1,n}|_\infty^\alpha:n\geq 1\}$ is uniformly integrable for some $\alpha>0$, then $F(\groupset_n) = \op(v_n n^{2/\alpha} (n / \kboundn \nlevels)^{-2/(v_n+1)})$.
\item For a fixed covariate law, if $(\psi_i)_{i=1}^n$ are iid and $E[|\psi|_2^\alpha] < \infty$ for some $\alpha > 0$, then
\[
F(\groupset_n) = \op\bigl(n^{2/\alpha} (n / \kboundn \nlevels)^{-2/(v+1)}\bigr).
\]
\end{enumerate}
\end{thm}

\begin{proof}
For each $l \in [\nlevels]$ and $i \in \snl'$, let $\phi_{i,n} \in [0,1]^v$ denote the normalized covariates from preprocessing.
By additivity of $F$, Lemma \ref{lem:unit-cube-reduction}, and Lemma \ref{lem:remainder-negligible}, $F(\groupset_n,\psi) \leq \mn^2 \sum_{l=1}^{\nlevels}F(\groupsetnl^{(T_l)},\phi) + \mn^2 v(n/\kboundn\nlevels)^{-2/(v+1)}$.
Then it suffices to show $\sum_{l=1}^{\nlevels}F(\groupsetnl^{(T_l)},\phi) \leq 5v(\kboundn\nlevels)^{2/(v+1)}n^{-2/(v+1)}$.
We relabel $\phi$ as $\psi$ and work with these normalized covariates for the remainder of the proof.
Fix a stratum $l \in \{1, \dots, \nlevels\}$.
Recall there are $K_l = \nl / \denoml$ non-remainder groups in $\snl'$, each of size $\denoml$.
We show monotonicity $F(\groupsetnl^{(t+1)}) \leq F(\groupsetnl^{(t)})$ for every $t \geq 0$.
The balanced assignment step \eqref{equation:balanced-kmeans-text} chooses $\groupsetnl^{(t+1)}$ to minimize $\sum_s \sum_{i \in \group_s} |\psiin - \bar\psi_{\group_s}^{(t)}|_2^2$ over equal-size partitions. Since $\groupsetnl^{(t)}$ is feasible,
\begin{equation} \label{equation:kmeans-assignment-step}
\sum_{s=1}^{K_l} \sum_{i \in \group_s^{(t+1)}} |\psiin - \bar\psi_{\group_s}^{(t)}|_2^2 \leq \sum_{s=1}^{K_l} \sum_{i \in \group_s^{(t)}} |\psiin - \bar\psi_{\group_s}^{(t)}|_2^2 = n F(\groupsetnl^{(t)}).
\end{equation}
For each $s$, expanding the square and using $\sum_{i \in \group_s^{(t+1)}}(\psiin - \bar\psi_{\group_s}^{(t+1)}) = 0$ gives
\[
\sum_{i \in \group_s^{(t+1)}} |\psiin - \bar\psi_{\group_s}^{(t)}|_2^2 = \sum_{i \in \group_s^{(t+1)}} |\psiin - \bar\psi_{\group_s}^{(t+1)}|_2^2 + \denoml |\bar\psi_{\group_s}^{(t+1)} - \bar\psi_{\group_s}^{(t)}|_2^2 \geq \sum_{i \in \group_s^{(t+1)}} |\psiin - \bar\psi_{\group_s}^{(t+1)}|_2^2.
\]
Summing over $s$ and combining with \eqref{equation:kmeans-assignment-step} yields the inequality $n F(\groupsetnl^{(t+1)}) \leq n F(\groupsetnl^{(t)})$, hence $F(\groupsetnl^{(t+1)}) \leq F(\groupsetnl^{(t)})$.
Iterating, we have $F(\groupsetnl^{(T_l)}) \leq F(\groupsetnl^{(0)})$ for any $T_l \geq 0$.
Summing over $l \in [\nlevels]$ and using the normalized bound established in the proof of Theorem \ref{thm:matching-per-stratum}, we have $\sum_{l=1}^{\nlevels} F(\groupsetnl^{(T_l)}) \leq \sum_{l=1}^{\nlevels} F(\groupsetnl^{(0)}) \leq 5v(\kboundn\nlevels)^{2/(v+1)}n^{-2/(v+1)}$.
Together with the reduction above, this proves \eqref{equation:algorithm-guarantee-rate}.
Under the condition of item (1), $\mn^2 \leq c^2$, so substituting into \eqref{equation:algorithm-guarantee-rate} proves its claim.
The triangular-array conclusion of Lemma \ref{lem:coordinate-range-rate} gives $\mn^2 = \op(n^{2/\alpha})$, and its fixed-distribution conclusion gives the same rate when $v$ is fixed. Substituting these conclusions into \eqref{equation:algorithm-guarantee-rate} proves items (2) and (3), respectively.
\end{proof}

\begin{proof}[Proof of Corollary \ref{cor:tight-matching-survey}]
For sampling, apply Theorem \ref{thm:algorithm-guarantee} with $S_n=[n]$ and strata given by the levels of $\propselect(\psii)$.
For assignment, take $S_n=\{i:\Ti=1\}$ and use the levels of $\propfn(\psii)$ as strata, or their intersections with the levels of $\propselect(\psii)$ under sampling-subordinate assignment.
The theorem's finite-sample bound is pathwise in $S_n$, so it also applies to the random set of sampled units.
In each case, $\kboundn$ and $\nlevels$ are bounded, so item (3) and $\alpha>v+1$ give $F(\groupset_n)=\op(n^{2/\alpha-2/(v+1)})=\op(1)$ for both partitions.
\end{proof}

\begin{remark}[Convergence] \label{rem:algorithm-convergence}
Since each iteration strictly decreases $F$ whenever the partition changes and the set of equal-size partitions of $\snl'$ into $K$ groups is finite, the per-stratum sequence $(\groupsetnl^{(t)})_{t \geq 0}$ stabilizes in finitely many iterations at a partition $\groupsetnl^{(\infty)}$.
Since this holds for each $l = 1, \dots, \nlevels$, the full partition $\groupset_n^{(t)} = \cup_{l=1}^{\nlevels} \groupsetnl^{(t)} \cup \remainderset_n$ likewise converges in finitely many steps to $\groupset_n^{(\infty)} = \cup_{l=1}^{\nlevels} \groupsetnl^{(\infty)} \cup \remainderset_n$.
\end{remark}

\begin{lem}[Remainder] \label{lem:remainder-negligible}
For every $n$ and $v \geq 1$, the remainder groups in $\remainderset_n$, each of size $|\remainderl| < \denoml \leq \kboundn$, satisfy the deterministic bound
\[
F(\mc R_n) \leq \mn^2 v (n / \kboundn \nlevels)^{-2/(v+1)}.
\]
\end{lem}

\begin{proof}
For each $l$, $i \in \remainderl$, and $j \in [v]$, writing $\psi_{i,n}^{(j)} - \bar\psi_{\remainderl}^{(j)} = |\remainderl|^{-1}\sum_{i' \in \remainderl}(\psi_{i,n}^{(j)} - \psi_{i',n}^{(j)})$ and applying the triangle inequality gives $|\psi_{i,n}^{(j)} - \bar\psi_{\remainderl}^{(j)}| \leq \max_{i' \in \remainderl} \psi_{i',n}^{(j)} - \min_{i' \in \remainderl} \psi_{i',n}^{(j)} \leq \mn$.
Squaring and summing over $j$ gives $|\psiin - \bar\psi_{\remainderl}|_2^2 \leq \sum_{j=1}^{v} (\max_{i' \in \remainderl} \psi_{i',n}^{(j)} - \min_{i' \in \remainderl} \psi_{i',n}^{(j)})^2 \leq v \mn^2$, and summing over $i \in \remainderl$ gives $\sum_{i \in \remainderl} |\psiin - \bar\psi_{\remainderl}|_2^2 \leq v \mn^2 |\remainderl|$.
The remainder groups are disjoint, $\sum_l |\remainderl| \leq \kboundn\nlevels$, and $\sum_l |\remainderl| \leq n$. Hence $F(\mc R_n) \leq \mn^2v\min\{1,\kboundn\nlevels/n\}$.
Let $x=\kboundn\nlevels/n$. If $x \leq 1$, then $x \leq x^{2/(v+1)}$, while if $x>1$, then $1 \leq x^{2/(v+1)}$. Thus $\min\{1,x\} \leq x^{2/(v+1)}=(n/\kboundn\nlevels)^{-2/(v+1)}$, proving the claim.
\end{proof}

\begin{lem}[Unit Cube Reduction] \label{lem:unit-cube-reduction}
Let $(\psiin)_{i=1}^n \sub \mr^{v}$ be a triangular array, $\mn^2 = \max_{j \leq v}(\max_i \psi_{i,n}^{(j)} - \min_i \psi_{i,n}^{(j)})^2$, and fix any subset $\mathcal{S} \sub [n]$.
For $j = 1, \dots, v$ and $i \in \mathcal{S}$, define rescaling within $\mathcal{S}$ by $\phi_{i,n}^{(j)} = (\psi_{i,n}^{(j)} - \min_{i' \in \mathcal{S}} \psi_{i',n}^{(j)})/(\max_{i' \in \mathcal{S}} \psi_{i',n}^{(j)} - \min_{i' \in \mathcal{S}} \psi_{i',n}^{(j)})$ so that $\phi_{i,n} \in [0,1]^v$.
Then for any partition $\groupset$ of $\mathcal{S}$, $F(\groupset, \psi) \leq \mn^2 \cdot F(\groupset, \phi)$.
\end{lem}

\begin{proof}
Let $m_{n,j}(\mathcal{S}) = \max_{i \in \mathcal{S}} \psi_{i,n}^{(j)} - \min_{i \in \mathcal{S}} \psi_{i,n}^{(j)}$, so the affine relation $\psi_{i,n}^{(j)} = m_{n,j}(\mathcal{S}) \phi_{i,n}^{(j)} + \min_{i' \in \mathcal{S}} \psi_{i',n}^{(j)}$ holds for $i \in \mathcal{S}$. 
Averaging within any $\group \sub \mathcal{S}$ yields $\bar \psi_g^{(j)} = m_{n,j}(\mathcal{S}) \bar \phi_g^{(j)} + \min_{i' \in \mathcal{S}} \psi_{i',n}^{(j)}$, so $\psi_{i,n}^{(j)} - \bar \psi_g^{(j)} = m_{n,j}(\mathcal{S})(\phi_{i,n}^{(j)} - \bar \phi_g^{(j)})$.
Since $m_{n,j}(\mathcal{S}) \leq \mn$ for every $j$, $|\psiin - \bar \psi_g|_2^2 = \sum_{j=1}^{v} m_{n,j}(\mathcal{S})^2 (\phi_{i,n}^{(j)} - \bar \phi_g^{(j)})^2 \leq \mn^2 \cdot |\phi_{i,n} - \bar \phi_g|_2^2$.
Applying $n\inv \sum_{g \in \groupset} \sum_{i \in \group}$ to both sides gives $F(\groupset, \psi) \leq \mn^2 \cdot F(\groupset, \phi)$.
\end{proof}

\begin{lem}[Coordinate Range] \label{lem:coordinate-range-rate}
Let $v=v_n$ and suppose $(\psi_{i,n})_{i=1}^n$ are iid within each row. Define $\mjn = \max_i \psi_{i,n}^{(j)} - \min_i \psi_{i,n}^{(j)}$ and set $\mn^2 = \max_{j \leq v_n} \mjn^2$.
If $\{|\psi_{1,n}|_\infty^\alpha:n\geq1\}$ is uniformly integrable for some $\alpha>0$, then $\mn^2 = \op(n^{2/\alpha})$.
The same conclusion holds when $v$ is fixed, $\psi_{i,n}\sim P$ iid, and $E_P[|\psi_{i,n}|_2^\alpha]<\infty$.
\end{lem}

\begin{proof}
For each $j$, we have $\mjn \leq 2 \max_i |\psi_{i,n}^{(j)}|$, and hence $\max_j \mjn \leq 2 \max_i |\psiin|_\infty$.
Applying Lemma \ref{lemma:maximal-inequality} to $X_{i,n}=|\psi_{i,n}|_\infty$ proves the triangular-array conclusion.
For fixed $v$ and $P$, $|\psiin|_\infty\leq|\psiin|_2$, so the fixed-distribution conclusion of that lemma applied to $X_i=|\psi_i|_2$ proves the final claim.
\end{proof}

\subsection{Computational Details and Comparison} \label{appendix:computational-comparison}

\begin{figure}[h]
\centering
\includegraphics[width=\textwidth]{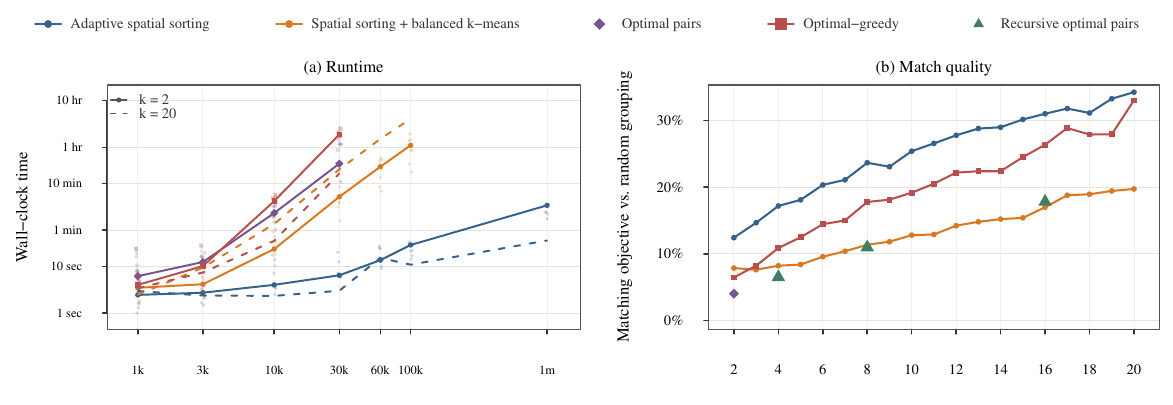}
\caption{Solid and dashed curves in panel (a) show medians for $k=2$ and $k=20$, with faint points showing application-specific runtimes for $k=2$. Panel (b) shows the matching objective relative to random grouping, averaged across the nine applications.}
\label{fig:computational-comparison}
\end{figure}

We benchmark adaptive spatial sorting and its balanced $k$-means refinement relative to optimal pairing \citep{bai2022pairs}, optimal-greedy matching \citep{moore2012blocking}, and recursive pairing \citep{bai2024factorial}. Runtime benchmarks draw covariates from the nine fitted empirical distributions. We run each method's standard serial implementation on one core. Exact pairing is a global optimization problem and cannot be divided into independent exact subproblems. For each $\groupsetn$, we report $F(\groupsetn)$, computed on the observed covariates, divided by its expectation under random balanced grouping.

Across applications, adaptive spatial sorting handles $n=1{,}000{,}000$ in a median of 205 seconds at $k=2$ and 36 seconds at $k=20$, using less than 1 GiB of RAM. At $n=100{,}000$, polishing takes 66 minutes at $k=2$ and 4.1 hours at $k=20$, using at most 75 GiB of RAM. At $n=30{,}000$, exact pairing takes from 11 minutes to 1.4 hours and uses up to 54 GiB of RAM. The \texttt{nbpMatching} package \citep{beck2024nbpmatching} used by \citet{bai2024factorial} is infeasible for $n>46{,}340$ due to its use of 32-bit matching-array indices. Across the 171 application-by-$k$ comparisons, polishing reduces $F$ by 56\% on average relative to adaptive spatial sorting, improves upon optimal-greedy matching for every $k\geq3$, and is within 1.8 percentage points of recursive pairing where available.

In our NBER audit, 96 of 200 recent papers (48\%) contain an experiment with $k \ne 2^\ell$, for which recursive pairing is infeasible.
Including the five $k=2^\ell$ experiments with $n\geq60{,}000$ that are infeasible for the \texttt{nbpMatching} package raises the share of experiments for which guaranteed tight matching is newly enabled to 50.5\%.

\clearpage

\renewcommand{\thesection}{B}
\pagenumbering{arabic}\renewcommand{\thepage}{\arabic{page}}

\begin{center}
{\Large Online Appendix to ``Fine Stratification of Survey Experiments''}
\end{center}

\section{Additional Results and Implementation Details}

\subsection{Conditional Variance Estimation} \label{appendix:pilot-variance-estimation}

The theory in Section \ref{section:optimal_designs} requires conditional variance function estimates $\hkest_d(\psi)$ for $d \in \{0, 1\}$, as well as the ex-ante variance estimate $\hkavgest(\psi)$ from Equation~\eqref{equation:ex-ante-variance} used in the optimal sampling formula in Equation~\eqref{equation:pilot:propensity-estimator}.
Our simulations and empirical application use a modification of the method in \cite{fan1998}.
In a regression model $Y = \ceffn(\psi) + \hksd(\psi) \epsilon$, they propose to (1) use local linear regression to estimate $\ceffn(\psi)$ and (2) use local linear regression to project estimated residuals $(Y - \ceffnest(\psi))^2$ on $\psi$.

In our setting, we can form signal $S_i(1) = Y_i \Di \Ti / (\propfni \propselecti)$, noting that $E[S_i(1) | \psii] = E[Y_i(1) | \psii] = \ceffn_1(\psii)$.
We then (1) project $S_i(1)$ on $\psii$ to estimate $\ceffn_1(\psii)$.
Next, (2) we project the IPW-weighted squared residuals $(Y_i - \ceffnest_1(\psii))^2 \Di \Ti / (\propfni \propselecti)$ on $\psii$ to estimate $\hkest_1(\psii)$, and similarly for $d=0$.

For the ex-ante variance $\hkavgest(\psi)$ in Equation~\eqref{equation:ex-ante-variance} used in the optimal sampling formula in Equation~\eqref{equation:pilot:propensity-estimator}, two routes are available: (i) the plug-in $\hkavgest(\psi) = \hkest_1(\psi)/\propfn + \hkest_0(\psi)/(1-\propfn)$ from the per-arm estimates above, or (ii) a single regression on $\psii$ of the combined IPW-weighted residual signal
\[
S_i = (Y_i - \ceffnest_1(\psii))^2 \Di \Ti / (\propfni^2 \propselecti) + (Y_i - \ceffnest_0(\psii))^2 (1-\Di) \Ti / ((1-\propfni)^2 \propselecti)
\]
This has the same conditional expectation as $\hkavg(\psii)$ in Equation~\eqref{equation:ex-ante-variance} and avoids combining per-arm fits with different smoothing. 
This is analogous to the DR-learner of \citet{kennedy2023drLearner} in CATE estimation, where a single influence-function signal is regressed on covariates rather than separately combining per-arm fits.
Our empirical implementation uses route (ii), which we found to be more numerically stable.

We tested linear ridge regression, RBF-kernel ridge regression, and random forests for each regression step.
Kernel ridge estimated $\hk_d(\psi)$ the most precisely in dimensions $\dim(\psi)=1,2$, while forests were superior in higher dimensions.
Our simulation and empirical results use random forest regression with 200 trees, maximum depth 3, and minimum leaf size 10.

\subsection{Imputation of Potential Outcomes and Simulation DGP} \label{appendix:imputation-details}

For each paper and arm $d \in \{0, 1\}$, we fit $\ceffnest_d(\psi) = \wh E[Y | \psi, D = d]$ on the observed-arm units and estimate $\hkest_d(\psi)$ from cross-fitted residuals: split arm-$d$ units into $K = 5$ folds, fit $\ceffnest_d^{(-k)}$ on $K - 1$ folds, form the held-out residuals $r_i = Y_i - \ceffnest_d^{(-k)}(\psii)$ for $i$ in the held-out fold, then regress $r_i^2$ on $\psii$ to obtain $\hkest_d(\psi)$.
The fitted $\ceffnest_d$ and $\hkest_d$ define the population we resample from. In each Monte Carlo replication, for each sampled unit $i$ we draw $\tilde Y_i(d) = \ceffnest_d(\psii) + \hkest_d(\psii)\half \residuali^d$ with $\residuali^d \sim \normal(0,1)$ independently across replications, units, and arms. 
This gives $\var(\tilde Y(d) | \psi) = \hkest_d(\psi)$ in the simulated population.
For the two binary-outcome papers we instead draw $\tilde Y_i(d) \sim \bern(\ceffnest_d(\psii))$. 

To sample from the imputed DGP, we draw a unit with replacement from $(\psii)_{i=1}^{N_0}$ and perturb selected coordinates having at least 15 distinct observed values.
For each such coordinate, let $z$ be the sampled unit's value and let $z^-$ and $z^+$ be the nearest distinct observed values below and above it, setting any missing neighbor equal to $z$.
We draw $\tilde z \sim \unif\bigl((z^-+z)/2,(z+z^+)/2\bigr)$.
This avoids artificial perfect matches while preserving the observed marginal covariate ranges.
We leave the top-coded baseline friend count in \cite{lowe2021} unperturbed and impose that the four included ethnicity shares in \cite{bergquist2020} sum to at most $1$.

The regression family $\mc R$ used for $\ceffnest_d$ is selected separately for each paper using $K$-fold cross-validation. We minimize mean squared error for continuous outcomes and Brier loss for binary outcomes, averaging the criterion across treatment arms, over OLS, lasso with degree-2 interactions, gradient-boosted trees, and random forest.
For continuous outcomes, we estimate $\hkest_d$ separately by applying a regularized random forest to the squared cross-fitted residuals above.
Hyperparameters are chosen by inner cross-validation where applicable.
The covariates are the same $\psi$ used for stratification in the design step.
Table~\ref{table:empirical} reports $\dim(\psi)$, Appendix~\ref{appendix:paper-details} lists $\psi$ paper by paper.

Some previous work, such as \citet{bai2020pairs} and earlier versions of this paper, instead uses a matching-based imputation $\wh Y_i(d) = Y_{j(i)}(d)$ with $j(i) = \argmin_{j: \Dj = d} |\psii - \psij|_2$. This can artificially inflate efficiency since nearest-neighbor pairs of units end up with identical imputed potential outcomes, making the within-stratum variance exactly zero with positive probability over the sampling distribution. We thank an anonymous referee for pointing this out.

\renewcommand{\thesection}{C}

\section{Proofs}

\subsection{Notation} \label{proofs:notation}

We denote $[n] = \{1, \dots, n\}$ and use $a \wedge b$ for the minimum and $\one(\cdot)$ for the indicator.
We write $\en[a_i] = n\inv \sum_{i=1}^n a_i$.
For $x \in \mr^{\dimpsi}$, $|x|_2$ denotes the Euclidean norm and $|x|_\infty$ the sup norm.
For $f : \mr^{\dimpsi} \to \mr$, $|f|_{\text{lip}}$ is the Lipschitz seminorm and $\|f\|_{\psi,2}^2 = E[f(\psi)^2]$ the $L_2(\psi)$ norm. We write $X_n \convp X$ for convergence in probability, $X_n \convwprocess X$ for weak convergence, and use the usual stochastic orders $\Op(\cdot)$ and $\op(\cdot)$. Notation $a_n \asymp b_n$ means $a_n / b_n$ and $b_n / a_n$ are both bounded and $a_n \lesssim b_n$ means $a_n \leq C b_n$ for some constant $C$. Independence is denoted $\indep$ and $E[\cdot | \mc F]$ is conditional expectation given a $\sigma$-algebra $\mc F$.

\subsection{Asymptotics} \label{proofs:asymptotics}

This subsection develops a self-contained proof of the extension results in Section~\ref{subsection:asymptotics:extensions}, from which the main theorem (Theorem~\ref{thm:clt_fixed}) of Section~\ref{section:asymptotics} is derived.
We allow different sampling and assignment stratification variables $\psisamp$ and $\psiassign$ and sampling and assignment propensities $\propselect(\psisamp)$ and $\propfn(\psiassign)$ taking finitely many rational levels.
This subsumes the results presented in Theorem~\ref{thm:clt_two_psi} and Corollary~\ref{cor:clt4sate}, which assume both constant.

First, we collect some notation.

\begin{enumerate}[label={\rm(\arabic*)}, noitemsep, topsep=0pt]
\item \emph{Filtrations.} At the assignment stage, let $\filtrationcandpsi = \sigma(\psi_{2,1:n}, \permn, \Tn)$ denote the information used to construct the assignment partition and $\filtrationhd = \sigma(\Wn, \permn, \Tn)$ the corresponding filtration carrying the full data $W_i = (Y_i(0), Y_i(1), \psi_{2,i})$ (recall $\psi_{1,i} = f(\psi_{2,i})$ by Assumption~\ref{assumption:feasible-clts}(i)). The analogs for the sampling design are $\filtrationcandt = \sigma(\psi_{1,1:n}, \permn)$ and $\filtrationht = \sigma(\Wn, \permn)$. Note the key relations $\filtrationcandt \sub \filtrationht$ and $\filtrationcandpsi \sub \filtrationhd$. In the technical lemma section, we sometimes use a generic $\sigma$-algebra $\filtrationgeneric$, which we set to be either $\filtrationgeneric = \filtrationcandpsi$ or $\filtrationgeneric = \filtrationhd$ depending on the required context.
\item \emph{Randomness.} Recall $\permn$ is independent randomness used to break ties during matching. In what follows, $\proprand$ denotes an independent auxiliary data set, such as a pilot or previous observational data. We use $\eta$ for the complete randomizations in the design itself, see the next definition for details. When necessary, we denote $\eta^T$ for the sampling complete randomizations and $\eta^D$ for the assignment randomizations, obeying the key relations $\eta^D \indep \filtrationhd$ and $\eta^T \indep \filtrationht$.
\end{enumerate}

We formalize the construction of the local randomization design from Definition~\ref{defn:local_randomization}.
We state results for the assignment design $\Dn \sim \localdesigncond(\psiassign, \propfn(\psiassign))$.
Our formalization also applies to the sampling design after the substitution $\Dn \to \Tn$, $\Tn \to 1$, $\propfn \to \propselect$, $\psiassign \to \psisamp$, under which the unit index set $I = \{i: \Ti=1\}$ becomes $I = [n]$.

\begin{defn}[Design Formalization] \label{defn:design-construction}
Let $\Dn \sim \localdesigncond(\psiassign, \propfn(\psiassign))$ with associated $\filtrationcandpsi$-measurable partition $\groupset_n$ of index set $I = \{i: \Ti=1\} \sub [n]$.
Recall $\propfn(\psiassign) \in \{\al/\kl : l \in [L]\}$, with $\kl$ the interior group size at level $l$.
Define a jointly independent array representing the design randomizations $(\eta_{l,j})_{l \in [L], j \in [n+1]}$ with $\eta \indep (\Wn, \permn, \Tn)$.
Let $\eta_{l,j} \sim \crdist(\al/\kl)$ for each $l \in [L]$ and $j \in [n]$ and $\eta_{l, n+1}$ be a $\kl$ vector with iid Bernoulli$(\al/\kl)$ components.
In particular, $n+1$ is a loose, fixed upper bound on the number of groups required in stratum $l \in [L]$. 
We define a function $s(\cdot)$ that specifies which randomization $\eta_{s(\group)}$ is assigned to group $\group$.
In particular, let $s: \groupsetn \to [L] \times [n+1]$ be an injection sending each level-$l$ interior group to some $(l, j)$ with $j \in [n]$ and the level-$l$ remainder group to $(l, n+1)$.
Note this function will not be surjective.

We can take $s(\cdot)$ to be $\filtrationcandpsi$-measurable since only the information in $\filtrationcandpsi$ is used to construct the interior (full size $\kl$) and remainder groups $g \in \groupsetn$ and evaluate the group propensities $\propfn(\psiitwo)$ for $i \in \group$.
For interior groups, we simply set $(\Di)_{i \in \group} = \eta_{s(\group)}$.
More generally, let $(\Di)_{i \in \group}$ be the first $|g|$ components of $\eta_{s(g)}$, accommodating remainder groups.
Formally, if $f(g, x) = x_{1:|g|}$ is truncation, we can write for $\group \in \groupsetn$ 
\begin{equation} \label{equation:design-construction}
(\Di)_{i \in \group} = f(\group, \eta_{s(\group)}).
\end{equation}
\end{defn}

\textbf{Strategy.}
We decompose $\est - \ate$ into three terms by projection onto an increasing chain of conditioning $\sigma$-algebras: an assignment term, a sampling term, and a superpopulation term.
We establish a CLT for each, then combine them using a characteristic function argument.
Define $\te_i = Y_i(1) - Y_i(0)$, $\ylevel_i = (1-\propfn(\psiitwo)) Y_i(1) + \propfn(\psiitwo) Y_i(0)$, and $\hti = (\Di - \propfn(\psiitwo))/[\propfn(\psiitwo)(1-\propfn(\psiitwo))]$.
Note the algebraic identity $\hti Y_i = \te_i + \hti \yleveli$, which holds pointwise.
Then the estimator can be expanded as
\begin{align*}
\est - \ate &= \en\bigl[\Ti \hti Y_i / \propselect(\psiione)\bigr] - \ate = \en\bigl[\Ti (\te_i + \hti \yleveli) / \propselect(\psiione)\bigr] - \ate \\
     &= \en\bigl[\te_i - \ate \bigr] + \en\bigl[\te_i (\Ti - \propselect(\psiione))/\propselect(\psiione)\bigr] + \en\bigl[\Ti \hti \yleveli / \propselect(\psiione)\bigr] \\
     &\equiv A_n + B_n + C_n.
\end{align*}
Our approach is to prove a marginal CLT for $\rootn A_n$ and conditional CLTs for $\rootn B_n | \filtrationht$ and $\rootn C_n | \filtrationhd$ given $\filtrationht = \sigma(\Wn, \permn)$ and $\filtrationhd = \sigma(\Wn, \permn, \Tn)$. 

\begin{assumption} \label{assumption:feasible-clts}
The following hold:
\begin{enumerate}[label={(\roman*)}, itemindent=.5pt, itemsep=.4pt]
\item $\Tn \sim \localdesigncond(\psisamp, \propselectestn(\psisamp))$ and $\Dn \sim \localdesigncond(\psiassign, \propfnestn(\psiassign))$, where $\propselectestn$ is $\sigma(\proprand, \psi_{1,1:n})$-measurable and $\propfnestn$ is $\sigma(\proprand, \psi_{2,1:n})$-measurable, with $\proprand \indep (\Wn, \permn, \eta^D, \eta^T)$ for the assignment and sampling randomness $\eta^D, \eta^T$ of Definition~\ref{defn:design-construction}.
Also, $\psisamp = f(\psiassign)$ for some measurable function $f(\cdot)$.
\item Propensities $\propselectestn(\psisamp) \in (\propbound, 1]$ and $\propfnestn(\psiassign) \in (\propbound, 1-\propbound)$, and each takes at most $\nlevels$ rational levels with denominators at most $\kboundn$, for deterministic sequences $\kboundn, \nlevels$ with $\kboundn \nlevels = o(n)$.
\item Convergence $\en[(\propselectestn(\psiione) - \propselect(\psiione))^2] \convp 0$ and $\en[(\propfnestn(\psiitwo) - \propfn(\psiitwo))^2] \convp 0$ for fixed $\propselect(\psisamp) \in (\propbound, 1]$ and $\propfn(\psiassign) \in (\propbound, 1-\propbound)$.
\item Moments $E[Y(d)^2] < \infty$ for $d \in \{0,1\}$.
\end{enumerate}
\end{assumption}

Assumption \ref{assumption:feasible-clts}(i) requires $\Tn$ and $\Dn$ to be locally randomized designs.
By Definition \ref{defn:local_randomization}, this requires the existence of group partitions $\groupset_n^T$ and $\groupset_n^D$ that satisfy the tight matching conditions $F(\groupset_n^T) = \op(1)$ and $F(\groupset_n^D) = \op(1)$.

We can constructively show the existence of such partitions using the algorithms developed in Section \ref{section:algorithms}.
In particular, at the sampling stage we apply this algorithm to the set of units $S_n = [n]$ with covariates $(\psi_{1,i})_{i=1}^n \sub \mr^{v_1}$ and propensity $\propselectestn(\cdot)$, producing $\groupset_n^T$ using its adaptive per-stratum grid-selection rule.

At the assignment stage, we apply the algorithm to $S_n = \{i : T_i = 1\}$ with covariates $(\psi_{2,i})_{i \in S_n} \sub \mr^{v_2}$ and propensity $\propfnestn(\cdot)$, producing $\groupset_n^D$ using the same rule.
When $\propfnestn \equiv \propfn=a/k$ is fixed, we may instead run the algorithm separately within sampling-propensity strata to obtain sampling-subordinate assignment.
Here $v_1 = \dim(\psisamp)$ and $v_2 = \dim(\psiassign)$.
The next lemma gives conditions sufficient for tight matching under our algorithm.

\begin{lem}[Tight Matching] \label{lemma:tight-matching-exists}
Suppose $E[|\psij|_2^{\alpha_j}] < \infty$ for some $\alpha_j > v_j + 1$, $j = 1, 2$.
Suppose also $\kboundn \nlevels = o(n^{1 - (v_j+1)/\alpha_j})$ for $j = 1, 2$.
Then
\begin{equation} \label{equation:tight-matching-algorithmic}
F(\groupset_n^T) = \op(1), \quad \; F(\groupset_n^D) = \op(1).
\end{equation}
\end{lem}

\begin{proof}
First consider the sampling partition. Apply Theorem \ref{thm:algorithm-guarantee} item (3) to $S_n = [n]$ with $\psi = \psisamp$, $v = v_1$, $\alpha = \alpha_1$.
Assumption \ref{ass:matching} holds WLOG after excluding zero-range covariates within each nonempty stratum.
Item (3) of the theorem gives $F(\groupset_n^T) = \op(v_1n^{2/\alpha_1}(n/\kboundn \nlevels)^{-2/(v_1+1)})$.
This is $\opone$ by the rate hypothesis, hence $F(\groupset_n^T) = \op(1)$.

The same argument applies to the assignment partition with $S_n = \{i : T_i = 1\}$, $\psi = \psiassign$, $v = v_2$, $\alpha = \alpha_2$, giving $F(\groupset_n^D) = \op(v_2n^{2/\alpha_2}(n/\kboundn \nlevels)^{-2/(v_2+1)}) = \op(1)$.
Under sampling-subordinate assignment with fixed $\propfn=a/k$, take the index set in Theorem \ref{thm:algorithm-guarantee} to be $S_n=\{i:\Ti=1\}$.
Let $q_{n,l}$ index the at most $\nlevels$ distinct values of $\propselectestn(\psi_{1,i})$ among eligible units $i\in[n]$.
Take the theorem's sub-strata to be $\snl=\{i\in S_n:\propselectestn(\psi_{1,i})=q_{n,l}\}$ for these levels and $\snl=\varnothing$ for unused indices $l\in[\nlevels]$.
Apply item (3) of the theorem with $\psi=\psiassign$, $v=v_2$, $\alpha=\alpha_2$, and $\kl=k$ for every $l$ to obtain
\[
F(\groupset_n^D) = \op\bigl(n^{2/\alpha_2}(k\nlevels/n)^{2/(v_2+1)}\bigr) = \op(1).
\]
The last equality follows from $k\nlevels\leq\kboundn\nlevels=o(n^{1-(v_2+1)/\alpha_2})$.
\end{proof}

The next lemma shows how local randomization annihilates the imbalances in any square integrable function $b(\psi)$ of the covariates used for stratification.

\begin{lem}[Nonparametric Balance] \label{lemma:balance}
Suppose Assumption~\ref{assumption:feasible-clts} holds.
For any function $E[b(\psi_2)^2] < \infty$, then $\en[\Ti (\Di - \propfnestn(\psiitwo)) b(\psiitwo)] = \op(\negrootn)$.
If $E[b(\psi_1)^2] < \infty$, then $\en[(\Ti - \propselectestn(\psiione)) b(\psiione)] = \op(\negrootn)$.
\end{lem}
\begin{proof}
We begin with the first claim.
By Definition~\ref{defn:local_randomization}, the assignment partition satisfies $F(\groupsetn) = n\inv \sum_{\group} \sum_{i \in \group} |\psi_{2,i} - \bar\psi_{2,\group}|_2^2 = \op(1)$, where $\bar\psi_{2,\group} = |\group|\inv \sum_{i \in \group} \psi_{2,i}$.
By Lemma~\ref{lemma:op-rate} there is a deterministic $\mu_n \to \infty$ with $\mu_n F(\groupsetn) = \op(1)$.
Since $\kboundn \nlevels = o(n)$ by Assumption~\ref{assumption:feasible-clts}, the deterministic sequence $n\inv \kboundn \nlevels \to 0$, so Lemma~\ref{lemma:op-rate} also provides a deterministic $\mu_n' \to \infty$ with $\mu_n' n\inv \kboundn \nlevels \to 0$.
Set $c_n = (\mu_n \wedge \mu_n')^{1/2}$, so $c_n \to \infty$, $c_n^2 F(\groupsetn) = \op(1)$, and $c_n^2 n\inv \kboundn \nlevels = o(1)$.

By hypothesis $b(\psi_2) \in L_2(\psi_2)$, and we approximate it by Lipschitz functions.
Define $\mc L_n = \{b'(\psi_2) \in L_2(\psi_2): |b'|_{lip} \vee |b'|_{\infty} \leq c_n\}$ and let $b_n \in \mc L_n$ satisfy $|b_n - b|_{2, \psi_2}^2 \leq 2\inf_{b' \in \mc L_n} |b' - b|_{2, \psi_2}^2$.
We claim $|b_n - b|_{2, \psi_2}^2 \to 0$.
Let $\epsilon > 0$.
By Lemma~\ref{lemma:lipschitz-approximation} there is a function $h$ with $|h|_{lip} \vee |h|_{\infty} < \infty$ and $|h - b|_{2, \psi_2}^2 < \epsilon$.
Since $c_n \to \infty$, $h \in \mc L_n$ for all $n$ large enough, so $|b_n - b|_{2, \psi_2}^2 \leq 2|h - b|_{2, \psi_2}^2 < 2\epsilon$ for such $n$.
Since $\epsilon$ was arbitrary, this shows the claim.
Now expand
\begin{align*}
\en[\Ti (\Di - \propfnestn(\psiitwo)) b(\psiitwo)] &= \en[\Ti (\Di - \propfnestn(\psiitwo)) (b - b_n)(\psiitwo)] \\
&+ \en[\Ti (\Di - \propfnestn(\psiitwo)) b_n(\psiitwo)] \equiv A_n + B_n.
\end{align*}
We invoke Lemmas~\ref{lemma:design_properties}, \ref{lemma:stratified_wlln}, and \ref{lemma:variance_identity} with the generic $\filtrationgeneric$ taken to be $\filtrationcandpsi$.
The lemma hypotheses $\filtrationcandpsi \sub \filtrationgeneric$ and $\filtrationgeneric \indep \eta$ hold trivially and by Definition~\ref{defn:design-construction}.
By Assumption~\ref{assumption:feasible-clts}(i), $\propfnestn(\cdot)$ is $\sigma(\proprand, \psi_{2,1:n})$-measurable with $\proprand \indep (\Wn, \permn, \eta^T, \eta^D)$.
We subsume $\proprand$ in $\permn$, so that $\propfnestn$ is $\sigma(\psi_{2,1:n}, \permn) \subseteq \filtrationcandpsi$-measurable.
Hence $\propfnestn(\psiitwo)_{1:n}$, $\groupsetn$, $b(\psiitwo)_{1:n}$, and $b_n(\psiitwo)_{1:n}$ are $\filtrationcandpsi$-measurable.

Then by Lemma~\ref{lemma:design_properties}, $E[A_n | \filtrationcandpsi] = 0$.
By the variance bound in the proof of Lemma~\ref{lemma:stratified_wlln}, also $\var(\rootn A_n | \filtrationcandpsi) \leq 2 \en[(b - b_n)^2(\psiitwo)]$.
Since $E[\en[(b - b_n)^2(\psiitwo)]] = |b - b_n|_{2, \psi_2}^2 = o(1)$ by work above, conditional Markov (Lemma~\ref{lemma:conditional_markov}) gives $\var(\rootn A_n | \filtrationcandpsi) = \op(1)$, and then $A_n = \op(\negrootn)$, again by conditional Markov.

Next consider $B_n$.
By Lemma~\ref{lemma:design_properties}, $E[B_n | \filtrationcandpsi] = 0$.
By Lemma~\ref{lemma:variance_identity},
\begin{align*}
&\var(\rootn B_n | \filtrationcandpsi) \leq n\inv \sum_{\group} |\group|\inv \sum_{i,j \in \group} (b_n(\psi_{2,i}) - b_n(\psi_{2,j}))^2 + n\inv \kboundn \nlevels \max_{i=1}^n b_n(\psi_{2,i})^2 \\
&\leq c_n^2 \, n\inv \sum_{\group} |\group|\inv \sum_{i,j \in \group} |\psi_{2,i} - \psi_{2,j}|_2^2 + n\inv \kboundn \nlevels c_n^2 = 2 c_n^2 F(\groupsetn) + n\inv \kboundn \nlevels c_n^2.
\end{align*}
The second inequality uses $|b_n|_{lip} \vee |b_n|_{\infty} \leq c_n$, and the equality uses $\sum_{i,j \in \group} |\psi_{2,i} - \psi_{2,j}|_2^2 = 2|\group| \sum_{i \in \group} |\psi_{2,i} - \bar\psi_{2,\group}|_2^2$.
The first term is $\op(1)$ since $c_n^2 F(\groupsetn) = \op(1)$, and the second is $o(1)$ since $c_n^2 n\inv \kboundn \nlevels = o(1)$.
Then $\var(\rootn B_n | \filtrationcandpsi) = \op(1)$, so $B_n = \op(\negrootn)$, again by conditional Markov.
The second claim follows by the substitutions $\Dn \to \Tn$, $\Tn \to 1$, $\propfnestn \to \propselectestn$, $\psi_2 \to \psi_1$.
\end{proof}

We now develop the core CLT machinery of the paper.
The design propensities are allowed to be estimated from auxiliary data and the matching parameters $\kboundn, \nlevels$ are allowed to grow.
This form is general enough to accommodate the main result in Theorem~\ref{thm:clt_fixed}, the extension results in Section~\ref{subsection:asymptotics:extensions}, and the optimal sampling CLT in Section~\ref{section:feasible-optimal-sampling}.

\begin{thm}[CLT] \label{thm:assignment-clt-feasible}
Let Assumption~\ref{assumption:feasible-clts} hold, and suppose one of the following: (i) $\kboundn = O(1)$ and $E[\diff(W)^2] < \infty$, or (ii) $E[|\diff(W)|^{2+\gamma}] < \infty$ and $\kboundn = o(n^{\gamma/(2(\gamma+2))})$ for some $\gamma > 0$.
Let $\filtrationhd = \sigma(\Wn, \Tn, \permn)$ and $\filtrationht = \sigma(\Wn, \permn)$.
Then
\begin{align*}
\rootn \en[\Ti(\Di - \propfnestn(\psi_{2,i})) \diff(W_i)] | \filtrationhd &\convwprocess \normal(0, V_a), \; \, V_a = E[\propselect(\psisamp) \propfn(\psiassign)(1-\propfn(\psiassign)) \var(\diff | \psiassign)], \\
\rootn \en[(\Ti - \propselectestn(\psi_{1,i})) \diff(W_i)] | \filtrationht &\convwprocess \normal(0, V_s), \; \, V_s = E[\propselect(\psisamp)(1-\propselect(\psisamp)) \var(\diff | \psisamp)].
\end{align*}
\end{thm}

\begin{proof}
Since the auxiliary randomness $\proprand$ is independent of the data $\Wn$ and of all design randomness $(\permn, \eta^D, \eta^T)$ in Definition \ref{defn:design-construction}, adjoining $\proprand$ to $\permn$ preserves every independence relation among $(\Wn, \permn, \eta^D, \eta^T, \Tn)$. 
In particular, we have $\Wn \indep \permn$, $\eta^D \indep (\Wn, \permn, \Tn)$, and $\eta^T \indep (\Wn, \permn)$.
We adopt this convention throughout, so $\proprand$ is a component of $\permn$, the propensities $\propselectestn, \propfnestn$ are $\sigma(\psi_{2,1:n}, \permn) \subseteq \filtrationcandpsi$-measurable (using $\psi_{1,1:n} = f(\psi_{2,1:n})$ from Assumption~\ref{assumption:feasible-clts}(i)), and the $\sigma$-algebras $\filtrationhd = \sigma(\Wn, \Tn, \permn)$, $\filtrationcandpsi = \sigma(\psi_{2,1:n}, \permn, \Tn)$, and $\filtrationht = \sigma(\Wn, \permn)$ carry $\proprand$.

Note $E[\diff(W)^2] < \infty$ under either hypothesis, in case (ii) by Jensen's inequality.
Let $\ui = \diff(W_i) - E[\diff(W) | \psiitwo]$, so $E[\ui | \psiitwo] = 0$ and $E[\ui^2 | \psiitwo] = \var(\diff | \psiitwo)$. 
Under hypothesis (ii) also $E[|\ui|^{2+\gamma}] < \infty$: by the $c_p$ inequality $|\ui|^{2+\gamma} \leq 2^{1+\gamma}(|\diff(W_i)|^{2+\gamma} + |E[\diff(W) | \psiitwo]|^{2+\gamma})$, and $E[|E[\diff(W) | \psiitwo]|^{2+\gamma}] \leq E[|\diff(W)|^{2+\gamma}]$ by conditional Jensen and the tower law, giving $E[|\ui|^{2+\gamma}] \leq 2^{2+\gamma} E[|\diff(W)|^{2+\gamma}] < \infty$.
By Lemma~\ref{lemma:balance} applied to the design $\Dn \sim \localdesigncond(\psiassign, \propfnestn(\psiassign))$, $\rootn \en[\Ti(\Di - \propfnestn(\psi_{2,i})) E[\diff | \psiitwo]] = \op(1)$, since $E[E[\diff | \psiitwo]^2] \leq E[\diff^2] < \infty$ by conditional Jensen and tower law.
Hence we have $\rootn \en[\Ti(\Di - \propfnestn(\psi_{2,i})) \diff(W_i)] = X_n + \op(1)$ for $X_n \equiv \rootn \en[\Ti(\Di - \propfnestn(\psi_{2,i})) \ui]$.

Order the assignment groups $\group(j) \in \groupset_n$, $j = 1, \dots, M_n$.
Since $\propfnestn(\psi_2)$ takes at most $\nlevels$ rational levels $\wh\propfn_{n,l} = a_{n,l}/k_{n,l}$ for $l \in [\nlevels]$ with denominators $k_{n,l} \leq \kboundn$, interior groups at level $l$ have size $k_{n,l}$, with at most one remainder group per level of size $\leq k_{n,l} - 1$.
Define the basic terms $z_{j,n} = \negrootn \sum_{i \in \group(j)} (\Di - \propfnestn(\psi_{2,i})) \ui$, noting $\Ti = 1$ for $i \in \group(j)$, so $X_n = \sum_{j=1}^{M_n} z_{j,n}$.
Our plan is to show $X_n | \filtrationhd \convwprocess \normal(0, V_a)$ by justifying the conditions of Proposition~\ref{prop:conditional-clt}, in three steps: (1) the basic terms $(z_{j,n})_{j=1}^{M_n}$ are jointly independent conditional on $\filtrationhd$, (2) the variance process $\Sigma_n \equiv \sum_{j=1}^{M_n} E[z_{j,n}^2 | \filtrationhd]$ has convergence $\Sigma_n \convp V_a$, and (3) the conditional Lindeberg condition $\sum_{j=1}^{M_n} E[z_{j,n}^2 \one(|z_{j,n}| > \epsilon) | \filtrationhd] = \op(1)$ holds for each $\epsilon > 0$.
We begin with (1).

\medskip

\noindent (1) Note $\filtrationcandpsi = \sigma(\psi_{2, 1:n}, \permn, \Tn) \sub \filtrationhd$ and $\filtrationhd \indep \eta^D$ for assignment randomness $\eta^D$ of Definition~\ref{defn:design-construction}.
Observe $\propfnestn$ is $\filtrationcandpsi$-measurable since $\proprand$ is a component of $\permn$.
Each $z_{j,n}$ has the form $\phi((\Di)_{i \in \group(j)}, (\propfnestn(\psi_{2,i}), \ui)_{1:n}, \group(j))$ for the deterministic function $\phi(d, (p, u), g) = \negrootn \sum_{i \in g}(d_i - p_i) u_i$, with $(\propfnestn(\psi_{2,i}), \ui)_{1:n}$ being $\filtrationhd$-measurable.
Then by Lemma~\ref{lemma:group_aggregate_independence}, the $(z_{j,n})_{j=1}^{M_n}$ are jointly independent conditional on $\filtrationhd$.

\medskip

\noindent (2) By Lemma~\ref{lemma:design_properties}, $E[\Di | \filtrationhd] = \propfnestn(\psi_{2,i})$ for $i \in \{\Ti=1\}$, so $E[z_{j,n} | \filtrationhd] = 0$.
By $\filtrationhd$-measurability of $\ur, \ut$ and $E[\Di | \filtrationhd] = \propfnestn(\psi_{2,i})$, we have
\begin{equation*}
E[z_{j,n}^2 | \filtrationhd] = n\inv \sum_{r,t \in \group(j)} \ur \ut \cov(\Dr, \Dt | \filtrationhd).
\end{equation*}
By Lemma~\ref{lemma:design_properties}, $\var(\Di | \filtrationhd) = \propfnestn(\psi_{2,i})(1-\propfnestn(\psi_{2,i}))$ for $i \in \{\Ti=1\}$.
For $r \neq t$ in an interior $\group(j)$ at level $l$, the covariance $\cov(\Dr, \Dt | \filtrationhd) = -\wh\propfn_{n,l}(1-\wh\propfn_{n,l})/(k_{n,l}-1)$ with $\wh\propfn_{n,l} = \propfnestn(\psi_{2,i})$ constant within the group.
For the remainder group at level $l$, $(\Di)_{i \in \group}$ is conditionally iid Bernoulli$(\wh\propfn_{n,l})$ by Definition~\ref{defn:design-construction}, so $\cov(\Dr, \Dt | \filtrationhd) = 0$.
Let $l(j)$ denote the assignment propensity level of group $\group(j)$.
Summing over $j$, we obtain
\begin{equation*}
\Sigma_n = \en[\Ti \propfnestn(\psi_{2,i})(1-\propfnestn(\psi_{2,i})) \ui^2] - n\inv \sum_{j: \text{interior}} \frac{\wh\propfn_{n,l(j)}(1-\wh\propfn_{n,l(j)})}{k_{n,l(j)}-1} \sum_{r \neq t \in \group(j)} \ur \ut \equiv T_{n2} + T_{n1}.
\end{equation*}
The diagonal term $T_{n2}$ aggregates over all groups using $\sum_j \wh\propfn_{n,l(j)}(1-\wh\propfn_{n,l(j)}) \sum_{i \in \group(j)} \ui^2 = \sum_i \Ti \propfnestn(\psi_{2,i})(1-\propfnestn(\psi_{2,i})) \ui^2$, since $\groupset_n$ partitions $\{i : \Ti = 1\}$ and $\propfnestn(\psiitwo)$ is constant within each group.
The off-diagonal $T_{n1}$ vanishes on the remainder groups by the iid Bernoulli structure noted above.
We will show $T_{n2} \convp V_a$.
To do so, write $\rho_n(\psiitwo) = \propfnestn(\psiitwo)(1-\propfnestn(\psiitwo))$ and $\rho(\psiassign) = \propfn(\psiassign)(1-\propfn(\psiassign))$, and split $T_{n2} = \en[\Ti (\rho_n - \rho)(\psiitwo) \ui^2] + \en[\Ti \rho(\psiitwo) \ui^2] \equiv T_{n2}^{(1)} + T_{n2}^{(2)}$.

To analyze these terms, we repeatedly use the following fact. 
If $(b_{i,n})_{i=1}^n$ are uniformly bounded, $\sup_{i,n} |b_{i,n}| \leq M$, with $\en[b_{i,n}^2] = \op(1)$, and $(c_i)_{i=1}^n$ are iid with $E[|c_1|] < \infty$, then $\en[|b_{i,n} c_i|] = \op(1)$.
Indeed, for any $K > 0$, $\en[|b_{i,n} c_i|] \leq K \en[|b_{i,n}|] + M \en[|c_i| \one(|c_i| > K)]$.
The first term is $\op(1)$ since $\en[|b_{i,n}|] \leq \en[b_{i,n}^2]^{1/2}$ by Cauchy--Schwarz, the second converges in probability to $M E[|c_1| \one(|c_1| > K)]$ by the iid weak law, and this limit is arbitrarily small for $K$ large.

For $T_{n2}^{(1)}$, the identity $|x(1-x) - y(1-y)| = |x-y| |1-x-y| \leq |x-y|$ for $x, y \in [0, 1]$ gives $|\rho_n - \rho| \leq |\propfnestn - \propfn|$ pointwise, hence $\en[(\rho_n - \rho)^2(\psiitwo)] \leq \en[(\propfnestn - \propfn)^2(\psiitwo)] = \op(1)$ by Assumption~\ref{assumption:feasible-clts}.
So $|T_{n2}^{(1)}| = |\en[\Ti (\rho_n - \rho)(\psiitwo) \ui^2]| \leq \en[|\rho_n - \rho|(\psiitwo) \ui^2]$, using $\Ti \leq 1$, which is $\op(1)$ by the fact above with $b_{i,n} = (\rho_n - \rho)(\psiitwo)$, bounded by $1/4$, and $c_i = \ui^2$, iid with $E[\ui^2] \leq \var(\diff) < \infty$.

For $T_{n2}^{(2)}$, let $h(W) = \rho(\psiassign) u^2$, which has $E[|\rho(\psiassign) u^2|] \leq E[u^2] < \infty$.
We claim that $\en[\Ti h(W_i)] = E[\propselect(\psisamp) h(W)] + \op(1)$.
To see this, we can decompose $\en[\Ti h_i] = \en[(\Ti - \propselectestn(\psiione)) h_i] + \en[(\propselectestn - \propselect)(\psiione) h_i] + \en[\propselect(\psiione) h_i]$.
The first term is $\op(1)$ by Lemma~\ref{lemma:stratified_wlln} part (2) applied to the sampling design $\Tn \sim \localdesigncond(\psisamp, \propselectestn(\psisamp))$, with the constant sequence $h_n = h$.
The second term is bounded in absolute value by $\en[|\propselectestn - \propselect|(\psiione) |h_i|]$, which is $\op(1)$ by the fact above with $b_{i,n} = (\propselectestn - \propselect)(\psiione)$, bounded by $1$, and $c_i = h_i$, iid with $E[|h|] < \infty$.
The third term equals $E[\propselect(\psisamp) h(W)] + \op(1)$ by the iid WLLN, using that $|\propselect h| \leq |h|$.
Then since $h = \rho(\psiassign) u^2$, we have $T_{n2}^{(2)} = \en[\Ti \rho(\psiitwo) \ui^2] = E[\propselect(\psisamp) \propfn(\psiassign)(1-\propfn(\psiassign)) u^2] + \op(1)$.
This implies $T_{n2} \convp V_a$, since by tower law and $\psiassign$-measurability of $\propselect(\psisamp)$ and since $E[u^2 | \psiassign] = \var(a | \psiassign)$, 
\begin{equation*}
E[\propselect(\psisamp) \propfn(\psiassign)(1-\propfn(\psiassign)) u^2] = E[\propselect(\psisamp) \propfn(\psiassign)(1-\propfn(\psiassign)) \var(\diff | \psiassign)] = V_a.
\end{equation*}

Next, we apply Lemma~\ref{lemma:lln} to show $T_{n1} = \op(1)$, with conditioning $\sigma$-algebra $\filtrationcandpsi$. 
Note that $\groupsetn$ is $\filtrationcandpsi$ measurable, recalling that $\filtrationcandpsi = \sigma(\psi_{2,1:n}, \permn, \Tn)$ for the assignment design (Definition \ref{defn:design-construction}).
Define $v_j = \frac{\wh\propfn_{n,l(j)}(1-\wh\propfn_{n,l(j)})}{k_{n,l(j)} - 1} \sum_{r \neq t \in \group(j)} \ur \ut$ for interior $\group(j)$ and $v_j = 0$ for the remainder, so $T_{n1} = -n\inv \sum_{j=1}^{M_n} v_j$.

Lemma~\ref{lemma:lln} requires joint conditional independence of $(v_j)_{j=1}^{M_n}$ given $\filtrationcandpsi$.
To prove this, we apply Lemma~\ref{lemma:partitions} with $h_i = \psi_{2,i}$ and $\kappa = (\permn, \Tn)$, so that in the notation of the lemma $\filtrationcandpsi = \sigma(\psi_{2,1:n}, \permn, \Tn) = \sigma(h_{1:n}, \kappa)$.
The conditional independence required is $\kappa \indep \Wn | h_{1:n}$, equivalent to $(\permn, \Tn) \indep \Wn | \psi_{2,1:n}$.
To see that this holds, chain two reductions.
First, $\Wn \indep \permn$ unconditionally gives $\Wn \indep \permn | \psi_{2,1:n}$ via $(A, B) \indep C \Rightarrow A \indep C | B$ (with $A = \Wn$, $B = \psi_{2,1:n} \in \sigma(\Wn)$, $C = \permn$).
Second, $\Tn \indep \Wn | (\psi_{2,1:n}, \permn)$ was established in the argument above.
Composing, $\Wn | (\psi_{2,1:n}, \permn, \Tn) \eqdist \Wn | (\psi_{2,1:n}, \permn) \eqdist \Wn | \psi_{2,1:n}$, which is the required $(\permn, \Tn) \indep \Wn | \psi_{2,1:n}$.
Hence Lemma~\ref{lemma:partitions} gives $(v_j)_{j=1}^{M_n}$ jointly conditionally independent given $\filtrationcandpsi$.
Also, note the partition $\groupsetn$ is $\filtrationcandpsi$-measurable by construction.

Lemma~\ref{lemma:lln} also requires the conditional mean-zero property $E[v_j | \filtrationcandpsi] = 0$ for each $j = 1, \dots, M_n$.
The prefactor $\wh\propfn_{n,l(j)}(1-\wh\propfn_{n,l(j)})/(k_{n,l(j)}-1)$ is $\filtrationcandpsi$-measurable and bounded by $1/4$ since $\wh\propfn_{n,l}(1-\wh\propfn_{n,l}) \leq 1/4$ and $k_{n,l} - 1 \geq 1$.
We claim that $E[\ur \ut | \filtrationcandpsi] = 0$ for $r \neq t$, reducing the conditioning as follows:
\begin{align*}
E[\ur \ut | \filtrationcandpsi] &= E[\ur \ut | \psi_{2,1:n}, \permn, \Tn] = E[\ur \ut | \psi_{2,1:n}, \permn] = E[\ur \ut | \psi_{2,1:n}] \\
&= E[\ur \ut | \psi_{2,r}, \psi_{2,t}].
\end{align*}
The first equality is by definition of $\filtrationcandpsi$.
For the second equality, by Definition~\ref{defn:design-construction} $\Tn$ is a function of $(\psi_{1,1:n}, \permn, \eta^T)$ for sampling randomness $\eta^T \indep (\Wn, \permn)$.
Since $\psi_{1,1:n} = f(\psi_{2,1:n})$ is $\sigma(\psi_{2,1:n})$-measurable, $\Tn \indep \Wn | (\psi_{2,1:n}, \permn)$.
In particular, $(\ur, \ut) \indep \Tn | (\psi_{2,1:n}, \permn)$.
Recall the basic fact that $(A, B) \indep C \implies A \indep C | B$.
The third equality follows from this fact noting that $\Wn \indep \permn$, hence $(\ur, \ut) \indep \permn | \psi_{2,1:n}$.
The fourth equality uses iid sampling: $(\ur, \ut, \psi_{2,r}, \psi_{2,t}) \indep \psi_{2,-(r,t)}$, hence $(\ur, \ut) \indep \psi_{2,-(r,t)} | (\psi_{2,r}, \psi_{2,t})$, again using the basic fact.
Then by tower law, we can further reduce
\[
E[\ur \ut | \psi_{2,r}, \psi_{2,t}] = E[\ur E[\ut | \ur, \psi_{2,r}, \psi_{2,t}] | \psi_{2,r}, \psi_{2,t}] = E[\ur E[\ut | \psi_{2,t}] | \psi_{2,r}, \psi_{2,t}] = 0.
\]
The second equality uses $\ut \indep (\ur, \psi_{2,r}) | \psi_{2,t}$ from iid sampling and the basic fact above.
The third uses $E[\ut | \psi_{2,t}] = 0$ by definition.
Then we have shown $E[v_j | \filtrationcandpsi] = 0$.

Finally, Lemma~\ref{lemma:lln} requires $n\inv \sum_{j=1}^{M_n} E[|v_j| \one(|v_j| > c_n) | \filtrationcandpsi] = \op(1)$ for some sequence $c_n$ with $c_n = \omega(1)$ and $c_n = o(n^{1/2})$.
Under assumption (i), set $c_n = n^{1/4}$.
Under assumption (ii), set $c_n = a_n \kboundn^{1 + 2/\gamma}$, where $a_n \to \infty$ is a deterministic sequence chosen such that $c_n = o(n^{1/2})$. 
Such an $a_n$ exists by Lemma~\ref{lemma:op-rate} applied to the deterministic null sequence $\kboundn^{1 + 2/\gamma} n^{-1/2}$, which is $o(1)$ since the case~(ii) hypothesis $\kboundn = o(n^{\gamma/(2(\gamma+2))})$ is equivalent to $\kboundn^{1 + 2/\gamma} = o(n^{1/2})$.
Then under either assumption, we have a sequence $c_n = \omega(1)$ and $c_n = o(n^{1/2})$.

On interior group $\group(j)$ at level $l$, $\sum_{r \neq t \in \group(j)} \ur \ut = \bigl(\sum_{r \in \group(j)} \ur\bigr)^2 - \sum_{r \in \group(j)} \ur^2$ and Cauchy-Schwarz give $|v_j| \leq \wh\propfn_{n,l}(1-\wh\propfn_{n,l}) \sum_{r \in \group(j)} \ur^2 \leq (1/4) \sum_{r \in \group(j)} \ur^2$, using the bound $\max_p p(1-p) \leq 1/4$.
Recall interior groups have size $|\group(j)| \leq \kboundn$.
Note also the fact that for any $(a_k)_{k=1}^m$ with $a_k \geq 0$ that $\sum_k a_k \one(\sum_k a_k > c) \leq m \sum_k a_k \one(a_k > c/m)$.
Then 
\begin{align*}
E\bigl[|v_j| \one(|v_j| > c_n) \bigm| \filtrationcandpsi\bigr]
  &\leq (1/4) E\bigl[\sum_{r \in \group(j)} \ur^2 \one\bigl(\sum_{r \in \group(j)} \ur^2 > 4 c_n\bigr) \bigm| \filtrationcandpsi\bigr] \\
  &\leq (\kboundn/4) \sum_{r \in \group(j)} E\bigl[\ur^2 \one\bigl(\ur^2 > 4 c_n/\kboundn\bigr) \bigm| \filtrationcandpsi\bigr].
\end{align*}
Summing over $j$, since interior groups partition a subset of $\{i : \Ti = 1\}$ this is
\begin{equation*}
\frac{1}{n}\sum_{j} E\bigl[|v_j| \one(|v_j| > c_n) \bigm| \filtrationcandpsi\bigr] \leq \frac{\kboundn}{4 n} \sum_{i=1}^n \Ti E\bigl[\ui^2 \one(\ui^2 > 4 c_n/\kboundn) \bigm| \filtrationcandpsi\bigr].
\end{equation*}
Taking expectation and using $\Ti \leq 1$ together with iid sampling and tower law, the right side has expectation bounded by $(\kboundn/4) E\bigl[\ui^2 \one(\ui^2 > 4 c_n/\kboundn)\bigr]$.
In case (i), $\kboundn \leq C$ for a constant $C$, so this is at most $(C/4) E\bigl[\ui^2 \one(\ui^2 > 4 c_n/C)\bigr] \to 0$ by dominated convergence, since $4 c_n/C \to \infty$ and $E[\ui^2] \leq E[\diff^2] < \infty$.
In case (ii), since $u^2 \one(u^2 > t) \leq |u|^{2+\gamma} t^{-\gamma/2}$ for $t > 0$, the bound is at most $\tfrac{1}{4^{1+\gamma/2}} E[|\ui|^{2+\gamma}] \kboundn^{1+\gamma/2} c_n^{-\gamma/2} = \tfrac{1}{4^{1+\gamma/2}} E[|\ui|^{2+\gamma}] a_n^{-\gamma/2} \to 0$, using $c_n = a_n \kboundn^{1+2/\gamma}$.
Then by Markov inequality, $n\inv \sum_j E[|v_j| \one(|v_j| > c_n) | \filtrationcandpsi] = \op(1)$.
Lemma~\ref{lemma:lln} with this $c_n$, which satisfies $c_n = \omega(1)$ and $c_n = o(n^{1/2})$, gives $T_{n1} = \op(1)$.
Then we have shown the claim $\Sigma_n = T_{n1} + T_{n2} \convp V_a$.

\noindent (3) 
Finally, we verify the conditional Lindeberg condition $\sum_{j=1}^{M_n} E[z_{j,n}^2 \one(|z_{j,n}| > \epsilon) | \filtrationhd] = \op(1)$ holds for each $\epsilon > 0$.
Note that the bound $|\Di - \propfnestn(\psi_{2,i})| \leq 1$ and Cauchy-Schwarz imply $z_{j,n}^2 \leq \kboundn n\inv \sum_{r \in \group(j)} \ur^2$, using $|\group(j)| \leq \kboundn$.
Hence
\begin{align*}
z_{j,n}^2 \one(z_{j,n}^2 > \epsilon^2)
  &\leq \kboundn n\inv \sum_{r \in \group(j)} \ur^2 \one\bigl(\sum_{r \in \group(j)} \ur^2 > n\epsilon^2/\kboundn\bigr) \leq \kboundn^2 n\inv \sum_{r \in \group(j)} \ur^2 \one(\ur^2 > n\epsilon^2/\kboundn^2).
\end{align*}
The second inequality applies the indicator function fact above to $\sum_{r \in \group(j)} \ur^2$ as a sum of $m \leq \kboundn$ non-negative terms.
Summing over $j$, since $\groupsetn$ partitions $\{i : \Ti = 1\}$ we have
\begin{equation*}
\sum_{j=1}^{M_n} E\bigl[z_{j,n}^2 \one(z_{j,n}^2 > \epsilon^2) \bigm| \filtrationhd\bigr] \leq \frac{\kboundn^2}{n} \sum_{i=1}^n \Ti E\bigl[\ui^2 \one(\ui^2 > n\epsilon^2/\kboundn^2) \bigm| \filtrationhd\bigr].
\end{equation*}
Taking expectation of both sides above and using $\Ti \leq 1$ together with iid sampling and tower law, the right side has expectation bounded by $\kboundn^2 E\bigl[\ui^2 \one(\ui^2 > n\epsilon^2/\kboundn^2)\bigr]$.

Under assumption (i), $\kboundn \leq C$ for a constant $C$, so this is at most $C^2 E\bigl[\ui^2 \one(\ui^2 > n\epsilon^2/C^2)\bigr] \to 0$ by dominated convergence, since $n\epsilon^2/C^2 \to \infty$ and $E[\ui^2] \leq E[\diff^2] < \infty$.
Under assumption (ii), since $u^2 \one(u^2 > t) \leq |u|^{2+\gamma} t^{-\gamma/2}$ for $t > 0$, this is at most $\epsilon^{-\gamma} E[|\ui|^{2+\gamma}] \kboundn^{2+\gamma} n^{-\gamma/2} \to 0$, since $\kboundn^{2+\gamma} = o(n^{\gamma/2})$ by raising the assumed Assumption~(ii) rate $\kboundn = o(n^{\gamma/(2(\gamma+2))})$ to the power $2+\gamma$.
Then by Markov inequality, we have $\sum_{j=1}^{M_n} E[z_{j,n}^2 \one(|z_{j,n}| > \epsilon) | \filtrationhd] = \op(1)$, showing the conditional Lindeberg condition.

Steps (1)-(3) verify, with conditioning $\sigma$-algebra $\filtrationhd$, the conditional independence, variance, and Lindeberg hypotheses of Proposition~\ref{prop:conditional-clt} for the array $(z_{j,n})_{j=1}^{M_n}$, extended by $z_{j,n} = 0$ for $M_n < j \leq n$.
The proposition gives $E[e^{i t X_n} | \filtrationhd] = e^{-t^2 V_a/2} + \op(1)$, with $V_a$ a constant, hence trivially $\filtrationhd$-measurable; since $\proprand$ is a component of $\permn \sub \filtrationhd$, this convergence is conditional on the auxiliary randomness $\proprand$.
Combined with $\rootn \en[\Ti(\Di - \propfnestn(\psi_{2,i})) \diff(W_i)] = X_n + \op(1)$, this gives the assignment CLT.

For the sampling CLT, apply the substitution rule of Section~\ref{proofs:asymptotics}: $\Dn \to \Tn$, $\Tn \to 1$, $\propfnestn \to \propselectestn$, $\propfn \to \propselect$, $\psiassign \to \psisamp$.
Under this substitution $\filtrationhd = \sigma(\Wn, \Tn, \permn)$ becomes $\filtrationht = \sigma(\Wn, \permn)$, and every step of the argument above carries through with sampling-design analogs; in particular, under $\Tn \to 1$ the term $T_{n2}$ becomes the average $\en[\propselectestn(\psi_{1,i})(1-\propselectestn(\psi_{1,i})) \ui^2]$, treated by the same split into $\propselectestn - \propselect$, handled by the truncation fact above, and the iid weak law.
This yields $\rootn \en[(\Ti - \propselectestn(\psi_{1,i})) \diff(W_i)] | \filtrationht \convwprocess \normal(0, V_s)$ with $V_s$ obtained from $V_a$ under the same substitution.
\end{proof}

We now use Theorem~\ref{thm:assignment-clt-feasible} together with the asymptotic independence Lemma~\ref{lemma:charfn} to prove generalizations of Theorem~\ref{thm:clt_two_psi} and Corollary~\ref{cor:clt4sate} that allow both $\propselect(\psisamp)$ and $\propfn(\psiassign)$ to take finitely many rational levels.

\begin{thm}[CLT, ATE] \label{thm:clt-two-psi-extended}
Suppose $\Tn \sim \localdesigncond(\psisamp, \propselect(\psisamp))$ and $\Dn \sim \localdesigncond(\psiassign, \propfn(\psiassign))$, with $\psisamp = f(\psiassign)$ for some measurable $f$.
Then $\rootn(\est - \ate) \convwprocess \normal(0, V_{\ate})$,
\begin{equation} \label{equation:variance:two-psi-extended}
V_{\ate} = \var(\te) + E\bigl[\tfrac{1-\propselect(\psisamp)}{\propselect(\psisamp)} \var(\te | \psisamp)\bigr] + E\bigl[\tfrac{\var(\ylevel | \psiassign)}{\propselect(\psisamp) \propfn(\psiassign)(1-\propfn(\psiassign))}\bigr].
\end{equation}
\end{thm}

\begin{proof}
This design is a special case of Assumption~\ref{assumption:feasible-clts} obtained by taking $\propselectestn \equiv \propselect$ and $\propfnestn \equiv \propfn$, so its condition (iii) holds exactly and $\kboundn, \nlevels = O(1)$.
Condition (ii) holds with $\propbound > 0$ chosen so the finitely many rational levels of $\propselect$ lie in $(\propbound,1]$ and those of $\propfn$ lie in $(\propbound,1-\propbound)$.
Case (i) of Theorem~\ref{thm:assignment-clt-feasible} then applies, which requires only $E[\diff(W)^2] < \infty$ of each integrand below.

Above, we noted the decomposition $\est - \ate = A_n + B_n + C_n$ for
\begin{align*}
A_n = \en[\te_i - \ate], \quad B_n = \en[\te_i (\Ti - \propselect(\psiione))/\propselect(\psiione)], \quad C_n = \en[\Ti \hti \yleveli/\propselect(\psiione)].
\end{align*}
Where $\hti = (\Di - \propfn(\psiitwo))/[\propfn(\psiitwo)(1-\propfn(\psiitwo))]$ and $\yleveli = (1-\propfn(\psiitwo))Y_i(1) + \propfn(\psiitwo) Y_i(0)$.
Lemma~\ref{lemma:charfn} requires an increasing chain of $\sigma$-algebras $\filtration_{n,0} \sub \filtration_{n,1} \sub \filtration_{n,2} \sub \filtration_{n,3}$.
In this context, we set $\filtration_{n,0} = \sigma(\emptyset)$, $\filtration_{n,1} = \filtrationht$, $\filtration_{n,2} = \filtrationhd$, and $\filtration_{n,3}=\sigma(\Wn,\permn,\Tn,\Dn)$.
We show conditional CLTs for $\rootn A_n, \rootn B_n, \rootn C_n$ at the three levels of the chain, then combine via the lemma.

First consider $\rootn A_n$.
By Young's inequality, $\var(\te) \leq 2 E[Y(1)^2 + Y(0)^2] < \infty$, so the iid CLT gives $\rootn A_n \convwprocess \normal(0, V_A)$ with $V_A = \var(\te)$.
Since $\filtration_{n,0}$ is trivial, this is equivalently $\rootn A_n | \filtration_{n,0} \convwprocess \normal(0, V_A)$.

Next consider $\rootn B_n$.
Apply the sampling CLT of Theorem~\ref{thm:assignment-clt-feasible} with $\diff(W) = \te/\propselect(\psisamp)$, noting $E[\diff(W)^2] \leq \propbound^{-2} E[\te^2] \leq 2 \propbound^{-2} E[Y(1)^2 + Y(0)^2] < \infty$ by the propensity bound and Young's inequality.
Then by that theorem $\rootn B_n | \filtration_{n,1} \convwprocess \normal(0, V_B)$ with
\begin{align*}
V_B = E\bigl[\propselect(\psisamp)(1-\propselect(\psisamp)) \var(\te/\propselect(\psisamp) | \psisamp)\bigr] = E\bigl[\tfrac{1-\propselect(\psisamp)}{\propselect(\psisamp)} \var(\te | \psisamp)\bigr].
\end{align*}

Finally consider $\rootn C_n$.
Apply the assignment part of Theorem~\ref{thm:assignment-clt-feasible} with $\diff(W) = \yleveli/[\propselect(\psiione) \propfn(\psiitwo)(1-\propfn(\psiitwo))]$.
Note $\propselect(\psiione)$ and $\propfn(\psiitwo)$ are both $\psiitwo$-measurable, the latter directly and the former since $\psisamp = f(\psiassign)$.
By the propensity bounds, $|\diff(W)| \leq [\propbound^3]\inv |\yleveli|$, so $E[\diff(W)^2] \leq \propbound^{-6} E[\yleveli^2] \lesssim E[Y(1)^2 + Y(0)^2] < \infty$.
Then the result guarantees $\rootn C_n | \filtration_{n,2} \convwprocess \normal(0, V_C)$ with
\begin{align*}
V_C = E\bigl[\propselect(\psisamp) \propfn(\psiassign)(1-\propfn(\psiassign)) \var(\diff | \psiassign)\bigr] = E\bigl[\tfrac{\var(\ylevel | \psiassign)}{\propselect(\psisamp) \propfn(\psiassign)(1-\propfn(\psiassign))}\bigr].
\end{align*}

By construction $\rootn A_n \in \sigma(\Wn) \sub \filtration_{n,1}$ and $\rootn B_n \in \sigma(\Wn, \Tn) \sub \filtration_{n,2}$, and $\rootn C_n$ is measurable in $\sigma(\Wn, \permn, \Tn, \Dn) \supset \filtration_{n,2}$.
Lemma~\ref{lemma:charfn} applied with $X_{n,k} = \rootn (A_n, B_n, C_n)$ at $k=1,2,3$ gives that for each $t \in \mr$
\begin{equation*}
E\bigl[e^{i t \rootn(A_n + B_n + C_n)}\bigr] = e^{-t^2(V_A + V_B + V_C)/2} + \op(1).
\end{equation*}

Therefore $\rootn(\est - \ate) \convwprocess \normal(0, V_A+V_B+V_C)=\normal(0, V_{\ate})$.
\end{proof}

\begin{cor}[CLT, SATE] \label{cor:clt-sate-extended}
In the design of Theorem~\ref{thm:clt-two-psi-extended}, $\rootn(\est - \sate) \, | \, \filtrationht \convwprocess \normal(0, V_{\sate})$, with
\begin{equation} \label{equation:variance:sate-extended}
V_{\sate} = E\bigl[\tfrac{1-\propselect(\psisamp)}{\propselect(\psisamp)} \var(\te | \psisamp)\bigr] + E\bigl[\tfrac{\var(\ylevel | \psiassign)}{\propselect(\psisamp) \propfn(\psiassign)(1-\propfn(\psiassign))}\bigr].
\end{equation}
In particular, $\rootn(\est - \sate) \convwprocess \normal(0, V_{\sate})$ marginally.
\end{cor}

\begin{proof}
Since $\sate = \en[\te_i] = A_n + \ate$, $\est - \sate = (\est - \ate) - A_n = B_n + C_n$.
Using the chain $\filtrationht \sub \filtrationhd \sub \filtration_{n,3}$ from the proof of Theorem~\ref{thm:clt-two-psi-extended}, that proof establishes the conditional CLTs $\rootn B_n | \filtrationht \convwprocess \normal(0, V_B)$ and $\rootn C_n | \filtrationhd \convwprocess \normal(0, V_C)$.
Lemma~\ref{lemma:charfn} applied to this reindexed chain with $X_{n,1}=\rootn B_n$ and $X_{n,2}=\rootn C_n$ gives $\rootn(B_n + C_n) \mid \filtrationht \convwprocess \normal(0, V_B + V_C)$.
By tower law and bounded convergence, $E[e^{it\rootn(\est-\sate)}] = E\bigl[E[e^{it\rootn(\est-\sate)} | \filtrationht]\bigr] \to e^{-t^2 V_{\sate}/2}$.

For the alternative target $\nsampled\inv\sum_{i=1}^n \Ti(Y_i(1)-Y_i(0))$ discussed in the main text, specialize to $\propselect(\psisamp)=\propselect$ and $\propfn(\psiassign)=\propfn$ constant.
The decomposition in the proof of Theorem~\ref{thm:clt-two-psi-extended} gives $A_n+B_n=(n\propselect)\inv\sum_i\Ti\te_i-\ate$.
The count bound in the proof of Lemma~\ref{lemma:ols-ipw-equivalence} gives $\nsampled=n\propselect+O(1)$, while $\Ti\leq1$ and the weak law of large numbers give $n\inv\sum_i\Ti|\te_i|=\Op(1)$.
Consequently, $\nsampled\inv\sum_i\Ti\te_i=(n\propselect)\inv\sum_i\Ti\te_i+\Op(n\inv)=\ate+A_n+B_n+\Op(n\inv)$, so $\rootn(\est-\nsampled\inv\sum_i\Ti\te_i)=\rootn C_n+\op(1)$.
The conditional CLT for $C_n$ established above and the tower law imply $\rootn C_n \convw \normal(0,V_3)$.
The preceding equality and Slutsky's theorem therefore give $\rootn(\est-\nsampled\inv\sum_i\Ti\te_i) \convw \normal(0,V_3)$.
\end{proof}

The next lemma gives an equivalent representation of the variance $V_{\ate}$ in \eqref{equation:variance:two-psi-extended} and shows that it reduces to the variance $\varlocal$ of Theorem~\ref{thm:clt_fixed} when $\psisamp = \psiassign = \psi$.

\begin{lem}[Variance equivalence] \label{lemma:variance-equivalence}
In the setting of Theorem~\ref{thm:clt-two-psi-extended}, let $c_j = E[\te | \psi_j]$ for $j=1,2$, $q_1 = \propselect(\psisamp)$, and $p_2 = \propfn(\psiassign)$. Then
\begin{equation*}
V_{\ate} = \var(c_1) + E\biggl[\frac{\var(c_2 | \psisamp)}{q_1}\biggr] + E\biggl[\frac{1}{q_1}\biggl(\frac{\hk_1(\psiassign)}{p_2} + \frac{\hk_0(\psiassign)}{1-p_2}\biggr)\biggr].
\end{equation*}
If $\psisamp = \psiassign = \psi$, then $V_{\ate} = \var(\catefn(\psi)) + E[\propselect(\psi)^{-1}(\hk_1(\psi)/\propfn(\psi) + \hk_0(\psi)/(1-\propfn(\psi)))]$.
\end{lem}

\begin{proof}
Write $\hk_d = \hk_d(\psiassign)$ for brevity.
We first show the pointwise identity
\begin{equation} \label{equation:ylevel-variance-identity}
\frac{\var(\ylevel | \psiassign)}{p_2(1-p_2)} = \frac{\hk_1}{p_2} + \frac{\hk_0}{1-p_2} - \var(\te | \psiassign).
\end{equation}
Expanding $\ylevel = (1-p_2) Y(1) + p_2 Y(0)$ and $\te = Y(1) - Y(0)$ gives
\begin{align*}
\var(\ylevel | \psiassign) &= (1-p_2)^2 \hk_1 + p_2^2 \hk_0 + 2p_2(1-p_2) \cov(Y(1), Y(0) | \psiassign), \\
\var(\te | \psiassign) &= \hk_1 + \hk_0 - 2 \cov(Y(1), Y(0) | \psiassign).
\end{align*}
Multiplying the second equality by $p_2(1-p_2)$, adding to the first, and rearranging gives \eqref{equation:ylevel-variance-identity}.
Substituting this identity into \eqref{equation:variance:two-psi-extended} gives
\begin{equation*}
V_{\ate} = \var(\te) + E\biggl[\frac{1-q_1}{q_1}\var(\te | \psisamp)\biggr] - E\biggl[\frac{\var(\te | \psiassign)}{q_1}\biggr] + E\biggl[\frac{1}{q_1}\biggl(\frac{\hk_1}{p_2} + \frac{\hk_0}{1-p_2}\biggr)\biggr].
\end{equation*}
The law of total variance implies that the first two terms in the preceding display equal $\var(c_1)+E[\var(\te | \psisamp)/q_1]$.
Because $\psisamp=f(\psiassign)$, $\var(\te | \psisamp) = E[\var(\te | \psiassign) | \psisamp] + \var(c_2 | \psisamp)$, and measurability of $q_1$ with respect to $\psisamp$ therefore gives $E[\var(\te | \psisamp)/q_1] - E[\var(\te | \psiassign)/q_1] = E[\var(c_2 | \psisamp)/q_1]$ by iterated expectations.
Combining these identities with the final term in the preceding display gives the first claim.
The second follows by setting $\psisamp = \psiassign = \psi$, so $c_1=c_2=\catefn(\psi)$ and $\var(c_2 | \psisamp)=0$.
\end{proof}

\begin{proof}[Proof of Theorem \ref{thm:clt_fixed}]
We invoke Theorem \ref{thm:clt-two-psi-extended} with $\psisamp = \psiassign = \psi$, so the constraint $\psisamp = f(\psiassign)$ holds with $f$ the identity.
Theorem \ref{thm:clt-two-psi-extended} then gives $\rootn(\est - \ate) \convwprocess \normal(0, V_{\ate})$ with $V_{\ate}$ as in \eqref{equation:variance:two-psi-extended}.
By Lemma \ref{lemma:variance-equivalence}, $V_{\ate} = \varlocal$ of \eqref{equation:asymptotic-theorem}.
\end{proof}

\begin{lem}[OLS and IPW Equivalence] \label{lemma:ols-ipw-equivalence}
Suppose $(W_i)_{i=1}^n \simiid P$, Assumption~\ref{assumption:asymptotics-regularity} holds, and the sampling and assignment propensities are fixed constants $\propselect=a_T/k_T \in (0,1]$ and $\propfn=a_D/k_D \in (0,1)$.
If $\Tn \sim \localdesigncond(\psisamp,\propselect)$ and $\Dn \sim \localdesigncond(\psiassign,\propfn)$, then $\estreg=\est+\Op(n\inv)$.
\end{lem}
\begin{proof}
Let $\widehat n_1=\sum_i\Ti\Di$, $\widehat n_0=\sum_i\Ti(1-\Di)$, $S_1=\sum_i\Ti\Di Y_i(1)$, and $S_0=\sum_i\Ti(1-\Di)Y_i(0)$.
Then $\estreg=S_1/\widehat n_1-S_0/\widehat n_0$ and $\est=S_1/(n\propselect\propfn)-S_0/(n\propselect(1-\propfn))$.
Let $r_T$ and $r_D$ denote the sampling and assignment remainder groups, taking either to be empty if absent.
Complete randomization within non-remainder groups gives $\nsampled-n\propselect=\sum_{i\in r_T}(\Ti-\propselect)$ and $\widehat n_1-\propfn\nsampled=\sum_{i\in r_D}(\Di-\propfn)$, so $|\nsampled-n\propselect|<k_T$ and $|\widehat n_1-\propfn\nsampled|<k_D$.
We also have $\widehat n_1-n\propselect\propfn=(\widehat n_1-\propfn\nsampled)+\propfn(\nsampled-n\propselect)$, so the triangle inequality gives $|\widehat n_1-n\propselect\propfn|<\propfn k_T+k_D$.
Since $\widehat n_0=\nsampled-\widehat n_1$, similarly $|\widehat n_0-n\propselect(1-\propfn)|<(1-\propfn)k_T+k_D$.
The count bounds imply $\widehat n_1=n\propselect\propfn+\Op(1)$ and $\widehat n_0=n\propselect(1-\propfn)+\Op(1)$.
Thus $|\widehat n_1^{-1}-(n\propselect\propfn)^{-1}|=|\widehat n_1-n\propselect\propfn|/(\widehat n_1n\propselect\propfn)=\Op(n^{-2})$, and similarly $|\widehat n_0^{-1}-(n\propselect(1-\propfn))^{-1}|=\Op(n^{-2})$.
By Assumption~\ref{assumption:asymptotics-regularity} and WLLN, $n\inv(|S_1|+|S_0|)\leq\en[|Y_i(1)|+|Y_i(0)|]=\Op(1)$, so $|S_1|+|S_0|=\Op(n)$.
Therefore, $|\estreg-\est|\leq |S_1||\widehat n_1^{-1}-(n\propselect\propfn)^{-1}|+|S_0||\widehat n_0^{-1}-(n\propselect(1-\propfn))^{-1}|=\Op(n\inv)$.
\end{proof}

\subsection{Optimal Sampling Design} \label{proofs:optimal-design}

\begin{proof}[Proof of Theorem \ref{thm:optimal_sampling_population}]
Define $V(\propselect) = E[\hkavg(\psi)/\propselect(\psi)]$.     
For measurable functions $\propselect : \mr^{\dimpsi} \to \mr$, define the sets
\[
\mc Q' = \{\propselect : |\propselect|_{\infty} < \infty, \propselect > 0, E[\cost(\psi)\propselect(\psi)] = \budget,  V(\propselect) < \infty\} \quad \mc Q = \mc Q' \cap \{0 < \propselect \leq 1\}.
\]
Also recall the candidate optimal solution $\propselectopt(\psi) = \budget \cdot \sdavg(\psi) \cost(\psi)^{-1/2} / E[\sdavg(\psi) \cost(\psi) \half]$.
Let $t \in [0, 1]$ and $q_1, q_2 \in \mc Q'$. 
By convexity of $y \to 1/y$ on $(0, \infty)$, for each $\psi \in \mr^{\dimpsi}$ we have
\[
\frac{\hkavg(\psi)}{t q_1(\psi) + (1-t)q_2(\psi)} \leq t \frac{\hkavg(\psi)}{q_1(\psi)} + (1-t) \frac{\hkavg(\psi)}{q_2(\psi)}
\]
Taking expectations of both sides gives $V(t q_1 + (1-t)q_2) \leq t V(q_1) + (1-t)V(q_2)$, so $V$ is convex on $\{\propselect > 0\}$ and $\mc Q'$ is convex.
We claim that $\propselectopt \in \mc Q'$. 
First, $\propselectopt \in (0, 1]$ by assumption.
Since $\sup_{\psi} \propselectopt \leq 1$ we have $E[\sdavg(\psi) \cost(\psi) \half] > 0$.
By Holder $E[\sdavg(\psi) \cost(\psi) \half] \leq E[\hkavg(\psi)] \half E[\cost(\psi)] \half < \infty$.
Then we have $V(\propselectopt) = \budget \inv E[\sdavg(\psi) \cost(\psi) \half]^2 < \infty$.
Also clearly $E[\cost(\psi)\propselectopt(\psi)] = \budget$.
This shows the claim.
Next, suppose that $\propselectopt + t \Delta \in \mc Q'$ for some $t \in [0,1]$ and $|\Delta|_{\infty} < M$.  
Since $\propselectopt \in \mc Q'$, by convexity of $\mc Q'$, $\propselectopt + t' \Delta \in \mc Q'$ for all $0 \leq t' \leq t$.
Then $dV_{\propselectopt}[\Delta] \equiv \limsup_{t \to 0^+} t \inv (V(\propselectopt + t\Delta) - V(\propselectopt))$ is well-defined.
We claim that this limit exists.
Under our assumptions $\propselectopt(\psi) > \propbound > 0$ pointwise, so that for any $\psi$ and all $t \leq \propbound^2 / 2M < \infty$
\begin{align*}
\left | t \inv \left (\frac{\hkavg(\psi)}{\propselectopt(\psi) + t \Delta(\psi)} - \frac{\hkavg(\psi)}{\propselectopt(\psi)} \right) \right| = \frac{|\hkavg(\psi)\Delta(\psi)|}{((\propselectopt)^2 + \propselectopt t \Delta)(\psi)} \leq \frac{M |\hkavg(\psi)|}{\propbound^2 - tM} \leq \frac{2M |\hkavg(\psi)|}{\propbound^2}.
\end{align*}
Since $\|\hkavg \|_1 < \infty$, dominated convergence implies 
\begin{align*}
dV_{\propselectopt}[\Delta] &= \lim_{t \to 0^+} t \inv (V(\propselectopt + t\Delta) - V(\propselectopt)) = E\left[\lim_{t \to 0^+}\frac{-\hkavg(\psi)\Delta(\psi)}{((\propselectopt)^2 + \propselectopt t \Delta)(\psi)} \right] \\
&= -E\left[\frac{\hkavg(\psi)\Delta(\psi)}{(\propselectopt)^2(\psi)} \right] = -(1/\budget)^2 E\left[\cost(\psi) \Delta(\psi) \right] E\left[\sdavg(\psi) \cost(\psi) \half \right]^2 = 0
\end{align*}

The last line since $\propselectopt + t \Delta \in \mc Q'$ implies $\budget = E[\propselectopt(\psi)\cost(\psi) + t \Delta(\psi) \cost(\psi)] = \budget + t E[\Delta(\psi) \cost(\psi)]$, so that $E[\Delta(\psi) \cost(\psi)] = 0$.
Let $q \in \mc Q'$, so that $\Delta = q-\propselectopt$ has $|\Delta|_{\infty} < \infty$.
Then by convexity $V(q) - V(\propselectopt) \geq dV_{\propselectopt}[q - \propselectopt] = 0$, showing that $\propselectopt = \argmin_{q \in \mc Q'} V(q)$.
Since $\propselectopt \in \mc Q$ by assumption, it is optimal over $\mc Q \subseteq \mc Q'$ as well. 
\end{proof}

\begin{proof}[Proof of Theorem~\ref{thm:pilot-design}]
The estimator is $\est = \en[\Ti \hti Y_i/\propselectestn(\psii)]$, with Horvitz-Thompson weight $\hti = (\Di - \propfn(\psii))/[\propfn(\psii)(1-\propfn(\psii))]$, outcome level $\yleveli = (1-\propfn(\psii)) Y_i(1) + \propfn(\psii) Y_i(0)$, and individual effect $\te_i = Y_i(1) - Y_i(0)$.
The implemented sampling propensity is the discretized estimate $\propselectestn(\cdot)$, and the assignment propensity $\propfn$ is fixed.

Our plan is to apply the CLT of Theorem~\ref{thm:assignment-clt-feasible} with $\gamma = 2$.
To do so, we first verify the key Assumption~\ref{assumption:feasible-clts} with $\psisamp = \psiassign = \psi$, sampling propensity $\propselectestn$, and assignment propensity $\propfnestn \equiv \propfn$.

Part (i) of the assumption holds with $f$ the identity, $\propfnestn \equiv \propfn$ fixed, and $\proprand$ the pilot randomness of Assumption~\ref{assumption:pilot-clt}(i).
To see this, note by Assumption~\ref{assumption:pilot-clt}(i), the variance estimates $\hkest_d$ are functions of external pilot data $\proprand$ independent of the experimental data and design variables, including the auxiliary design randomness $(\eta^D, \eta^T)$ of Definition~\ref{defn:design-construction}, so $\proprand \indep (\Wn, \permn, \eta^D, \eta^T)$.
From the construction in Equation~\eqref{equation:pilot:propensity-estimator}, $\propselectestn(\cdot)$ is a function of $\proprand$ and the study covariates $\psi_{1:n}$, hence $\sigma(\proprand, \psi_{1:n})$-measurable as required by the assumption.
We subsume $\proprand$ in $\permn$, so that $\propselectestn$ is $\sigma(\psi_{1:n}, \permn)$-measurable.

Part (ii) of the assumption holds since $\propselectestn$ takes at most $\nlevels$ rational levels with denominators at most $\kboundn$ and $\propfn$ is a fixed rational level, the bounds $\propselectestn(\psi) \in (\propbound, 1]$ and $\propfn \in (\propbound, 1-\propbound)$ are Assumption~\ref{assumption:pilot-clt}(ii), and $\kboundn \nlevels = o(n)$ follows from Assumption~\ref{assumption:pilot-clt}(iv).
Part (iii) holds since $\en[(\propselectestn(\psii) - \propselectopt(\psii))^2] \convp 0$ by Lemma~\ref{lemma:propensity-convergence}, using the discretization rate of Assumption~\ref{assumption:pilot-clt}(iv) and the consistency of the conditional variance estimates in Assumption~\ref{assumption:pilot-clt}(iii), and $\en[(\propfnestn(\psii) - \propfn(\psii))^2] = 0$ identically, so the fixed limit propensities are $\propselect = \propselectopt$ and $\propfn$.
Part (iv) holds since $E[Y(d)^4] < \infty$ gives $E[Y(d)^2] < \infty$, and $E[|\psi|^{\alpha}] < \infty$ with $\alpha > \dim(\psi)+1$ is Assumption~\ref{assumption:pilot-clt}(ii).
The moment bound $E[|\diff(W)|^4] < \infty$ is verified for each integrand below, and the rate $\kboundn = o(n^{\gamma/(2(\gamma+2))}) = o(n^{1/4})$ is Assumption~\ref{assumption:pilot-clt}(iv).

Using the pointwise identity $\hti Y_i = \te_i + \hti \yleveli$, expand the estimator as in the proof of Theorem~\ref{thm:clt-two-psi-extended}, with the implemented propensity $\propselectestn$ in place of $\propselect$:
\begin{align*}
\est - \ate &= \en[\te_i - \ate] + \en[\te_i (\Ti - \propselectestn(\psii))/\propselectestn(\psii)] + \en[\Ti \hti \yleveli/\propselectestn(\psii)] \\
&\equiv A_n + B_n + C_n.
\end{align*}
We apply Lemma~\ref{lemma:charfn} to deduce convergence in distribution of $\rootn(A_n + B_n + C_n)$ from conditional CLTs for $\rootn A_n$, $\rootn B_n$, and $\rootn C_n$ over the lemma's required increasing chain.
In this context, we set $\filtration_{n,0} = \sigma(\emptyset)$, $\filtration_{n,1} = \filtrationht$, $\filtration_{n,2} = \filtrationhd$, and $\filtration_{n,3} = \sigma(\Wn, \permn, \Tn, \Dn)$.

First consider $\rootn A_n$.
By Young's inequality $\var(\te) \leq 2 E[Y(1)^2 + Y(0)^2] < \infty$, so the iid CLT gives $\rootn A_n \convwprocess \normal(0, V_A)$ with $V_A = \var(\te)$.
Since $\filtration_{n,0}$ is trivial, this is equivalently $\rootn A_n | \filtration_{n,0} \convwprocess \normal(0, V_A)$.

Next consider $\rootn B_n$.
Write $\delta_i = 1/\propselectestn(\psii) - 1/\propselectopt(\psii)$ and split $B_n = B_n' + B_n''$, where $B_n' = \en[(\Ti - \propselectestn(\psii)) \te_i/\propselectopt(\psii)]$ and $B_n'' = \en[(\Ti - \propselectestn(\psii)) \te_i \delta_i]$.
For $B_n'$, apply Theorem~\ref{thm:assignment-clt-feasible} with the fixed integrand $\diff(W) = \te/\propselectopt(\psi)$, which satisfies $E[|\diff(W)|^4] \leq \propbound^{-4} E[\te^4] \lesssim E[Y(1)^4 + Y(0)^4] < \infty$.
The theorem gives $\rootn B_n' | \filtration_{n,1} \convwprocess \normal(0, V_B)$ with
\begin{align*}
V_B = E[\propselectopt(\psi)(1-\propselectopt(\psi)) \var(\te/\propselectopt(\psi) | \psi)] = E\bigl[\tfrac{1-\propselectopt(\psi)}{\propselectopt(\psi)} \var(\te | \psi)\bigr].
\end{align*}
For $B_n''$, since $\propselectestn(\psii), \propselectopt(\psii) \geq \propbound$ by Assumption~\ref{assumption:pilot-clt}, we have $|\delta_i| \leq \propbound^{-2} |\propselectestn(\psii) - \propselectopt(\psii)|$, and $|\propselectestn - \propselectopt| \leq 1$.
Then using Lemma~\ref{lemma:propensity-convergence}, we have 
\begin{equation} \label{equation:pilot-delta-bound}
\en[\delta_i^4] \leq \propbound^{-8} \en[(\propselectestn(\psii) - \propselectopt(\psii))^4] \leq \propbound^{-8} \en[(\propselectestn(\psii) - \propselectopt(\psii))^2] = \op(1).
\end{equation}
The term $\te_i \delta_i$ in $B_n'' = \en[(\Ti - \propselectestn(\psii)) \te_i \delta_i]$ is $\filtration_{n,1}$-measurable, and $E[\Ti | \filtration_{n,1}] = \propselectestn(\psii)$ by Lemma~\ref{lemma:design_properties} applied to the sampling design, so $E[\rootn B_n'' | \filtration_{n,1}] = 0$.
By the variance bound in the proof of Lemma~\ref{lemma:stratified_wlln} applied to the sampling design,
\begin{align*}
\var(\rootn B_n'' | \filtration_{n,1}) \leq 2 \en[\te_i^2 \delta_i^2] \leq 2 \en[\te_i^4]^{1/2} \en[\delta_i^4]^{1/2} = \Op(1) \op(1) = \op(1).
\end{align*}
The second inequality is Cauchy-Schwarz, then $\en[\te_i^4] \convp E[\te^4] < \infty$ from the WLLN and \eqref{equation:pilot-delta-bound},
Chebyshev (Lemma~\ref{lemma:conditional_markov}) gives $\rootn B_n'' = \op(1)$, so $\rootn B_n = \rootn B_n' + \op(1)$.

Finally consider $\rootn C_n$.
Split $C_n = C_n' + C_n''$, where
\begin{align*}
C_n' &= \en[\Ti (\Di - \propfn(\psii)) \yleveli/(\propselectopt(\psii) \propfn(\psii)(1-\propfn(\psii)))], \\
C_n'' &= \en[\Ti (\Di - \propfn(\psii)) \yleveli \delta_i/(\propfn(\psii)(1-\propfn(\psii)))].
\end{align*}
For $C_n'$, apply the assignment part of Theorem~\ref{thm:assignment-clt-feasible} with the fixed integrand $\diff(W) = \ylevel/(\propselectopt(\psi) \propfn(\psi)(1-\propfn(\psi)))$, which satisfies $|\diff(W)| \leq \propbound^{-1}[\propbound(1-\propbound)]^{-1} |\ylevel|$, hence $E[|\diff(W)|^4] \lesssim E[\ylevel^4] \lesssim E[Y(1)^4 + Y(0)^4] < \infty$.
The theorem gives $\rootn C_n' | \filtration_{n,2} \convwprocess \normal(0, V_C)$ with
\begin{align*}
V_C = E[\propselectopt(\psi) \propfn(\psi)(1-\propfn(\psi)) \var(\diff | \psi)] = E\Bigl[\tfrac{\var(\ylevel | \psi)}{\propselectopt(\psi) \propfn(\psi)(1-\propfn(\psi))}\Bigr].
\end{align*}
For $C_n''$, the integrand $\yleveli \delta_i/(\propfn(\psii)(1-\propfn(\psii)))$ is $\filtration_{n,2}$-measurable, and $E[\Di | \filtration_{n,2}] = \propfn(\psii)$ by Lemma~\ref{lemma:design_properties}, so $E[\rootn C_n'' | \filtration_{n,2}] = 0$.
By the variance bound in the proof of Lemma~\ref{lemma:stratified_wlln}, then $\Ti \leq 1$ with $\propfn(\psii)(1-\propfn(\psii)) \geq \propbound(1-\propbound)$, then Cauchy--Schwarz with $\en[\yleveli^4] = \Op(1)$ from the iid weak law together with \eqref{equation:pilot-delta-bound},
\begin{align*}
\var(\rootn C_n'' | \filtration_{n,2}) &\leq 2 \en\bigl[\Ti \yleveli^2 \delta_i^2/(\propfn(\psii)(1-\propfn(\psii)))^2\bigr] \\
&\leq 2 [\propbound(1-\propbound)]^{-2} \en[\yleveli^4]^{1/2} \en[\delta_i^4]^{1/2} = \op(1).
\end{align*}
Conditional Chebyshev gives $\rootn C_n'' = \op(1)$, so $\rootn C_n = \rootn C_n' + \op(1)$.

By construction $\rootn A_n \in \sigma(\Wn) \sub \filtration_{n,1}$, $\rootn B_n' \in \sigma(\Wn, \permn, \Tn) = \filtration_{n,2}$, and $\rootn C_n' \in \sigma(\Wn, \permn, \Tn, \Dn) = \filtration_{n,3}$.
Lemma~\ref{lemma:charfn} applied with $X_{n,k} = \rootn(A_n, B_n', C_n')$ at $k = 1, 2, 3$ gives, for each $t \in \mr$,
\begin{align*}
E\bigl[e^{i t \rootn(A_n + B_n' + C_n')}\bigr] = e^{-t^2(V_A + V_B + V_C)/2} + \op(1).
\end{align*}
Hence $\rootn(A_n + B_n' + C_n') \convwprocess \normal(0, V_A + V_B + V_C)$, and since $\rootn(\est - \ate) = \rootn(A_n + B_n' + C_n') + \op(1)$ by the work above, Slutsky gives $\rootn(\est - \ate) \convwprocess \normal(0, V_A + V_B + V_C)$.

The variances $V_A, V_B, V_C$ are those of Theorem~\ref{thm:clt-two-psi-extended} evaluated at $\psisamp = \psiassign = \psi$ and sampling propensity $\propselectopt$.
The proof of Lemma~\ref{lemma:variance-equivalence} establishes a purely algebraic identity for these three terms, giving $V_A + V_B + V_C = \var(\catefn(\psi)) + E[\hkavg(\psi)/\propselectopt(\psi)]$ with $\hkavg(\psi) = \hk_1(\psi)/\propfn(\psi) + \hk_0(\psi)/(1-\propfn(\psi))$.
By Theorem~\ref{thm:optimal_sampling_population} and the interior solution condition of Assumption~\ref{assumption:pilot-clt}(ii), $\propselectopt$ minimizes $E[\hkavg(\psi)/\propselect(\psi)]$ over $0 < \propselect \leq 1$ with $E[\cost(\psi)\propselect(\psi)] = \budget$.
Therefore $\rootn(\est - \ate) \convwprocess \normal(0, V^*)$ with
\begin{align*}
V^* = \var(\catefn(\psi)) + \min_{\substack{0 < \propselect \leq 1 \\ E[\cost(\psi)\propselect(\psi)] = \budget}} E[\hkavg(\psi)/\propselect(\psi)].
\end{align*}
\end{proof}

\subsection{Inference} \label{proofs:inference}

Lemmas \ref{lemma:new-inference-consistency-conditioning} and \ref{lemma:sate-inference-consistency} below establish the $\ate$ and $\sate$ parts of Theorem \ref{thm:inference} for the pairs-of-pairs estimator. The constant-$\propselect$ case is the special case $L = 1$ of \eqref{equation:variance-estimator-varying}.
We focus on proving consistency for the harder case with $\wh P^2 = \wh P^2_{\nu}$. 
Consistency for $\wh P^2 = \wh P^2_{N}$ can be shown similarly. 

Below, it is convenient to work with partition $\groupsetnu = \{\group \cup \groupmatching(\group) : \group \in \groupset_n\}$ into unions of paired groups. Write $u = \group_1(u) \cup \group_2(u)$ for the two groups forming each union.
We write $\sum_{\group \in \groupset_n}(\estg - \est_{\groupmatching(\group)})^2 = 2\sum_{u \in \groupsetnu}(\estgone{u} - \estgtwo{u})^2$, where $\estgone{u}$ and $\estgtwo{u}$ are the corresponding group-level estimates.
The variance estimator of \eqref{equation:S2-P2} is then $\wh P^2 = |\groupset_n|\inv \sum_{u \in \groupsetnu}(\estgone{u} - \estgtwo{u})^2$, and similarly $\wh P_l^2 = |\groupset_{nl}|\inv \sum_{u \in \mathcal{G}_{nl}^\nu}(\estgone{u} - \estgtwo{u})^2$ within each stratum.
Under sampling-subordinate assignment every $\group \in \groupset_n$ and every $u \in \groupsetnu$ lies entirely in a single propensity stratum $\snl = \{i : \propselect(\psii) = \propselectl\}$, so we set $\propselect_\group = \propselectl$ for $\group \in \groupset_{nl} := \{\group \in \groupset_n : \group \sub \snl\}$ and $\propselect_u = \propselectl$ for $u \in \mathcal{G}_{nl}^{\nu} := \{u \in \groupsetnu : u \sub \snl\}$.

For each $l \in [L]$, let $\nl = |\snl|$, and define the within-stratum IPW estimator and the group-level diff-in-means
\[
\wh\theta_l = \frac{1}{\nl} \sum_{i \in \snl} \frac{\Ti(\Di - \propfn) Y_i}{\propselectl \propfn(1-\propfn)}, \qquad
\estg = \frac{1}{a} \sum_{i \in \group} \Di Y_i - \frac{1}{k-a} \sum_{i \in \group} (1-\Di) Y_i,
\]
so that $\est = \sum_{l=1}^{L}(\nl/n)\wh\theta_l$ is the global IPW estimator from \eqref{equation:estimator:varying}.
Throughout the proofs below, we use the within-group quadratic and across-group cross-product aggregates
\begin{equation} \label{equation:v1-v2}
\wh v_1 = \frac{1}{n}\sum_{\group \in \groupset_n} \frac{k^{2}}{\propselect_{\group}^{2}} \estg^{2}, \qquad
\wh v_2 = \frac{2}{n}\sum_{u \in \groupsetnu} \frac{k(k - \propselect_u)}{\propselect_u^{2}} \estgone{u}\estgtwo{u}.
\end{equation}

Two conventions simplify the bookkeeping.
First, we identify $\groupset_n$ with the algorithm's interior groups of full size $k$, excluding the potential remainder group of size $< k$ in each of the $L$ propensity strata.
Within each sampling-propensity stratum, the identity $\wh\theta_l = |\groupset_{nl}|\inv \sum_{\group \in \groupset_{nl}} \estg$ holds up to an $\op(1)$ remainder discrepancy that does not affect the consistency arguments below.
Second, we assume $|\groupset_{nl}|$ is even in the proofs.
As noted in the text, in practice the case of an odd number of groups is handled by matching one group twice in $\mathcal{G}_{nl}^{\nu}$, contributing a single extra cross-product term to $\wh P_l^2$ of order $1/|\groupset_{nl}|$ that is similarly negligible but tedious to keep track of.

\begin{lem}[Aggregated form] \label{lemma:aggregated-form-2}
The variance estimator $\varest$ of \eqref{equation:variance-estimator-varying} satisfies
\begin{equation} \label{equation:new-variance-aggregated}
\varest = \wh v_1 - \wh v_2 - \est^{2}.
\end{equation}
\end{lem}

\begin{proof}[Proof of Lemma \ref{lemma:aggregated-form-2}]
Write $\wh V_l = \wh S_l^2 + ((k - \propselectl)/\propselectl) \wh P_l^2$ for the within-stratum terms, so $\varest = \sum_l (\nl/n) [\wh V_l + (\wh\theta_l - \est)^2]$.
By the within-stratum definitions of $\wh S_l^2$ and $\wh P_l^2$ as analogues of $\wh S^2$ and $\wh P^2$ from \eqref{equation:variance-estimator},
\begin{equation} \label{equation:Vl-stratum}
\wh V_l = \frac{1}{|\groupset_{nl}|} \sum_{\group \in \groupset_{nl}} \bigl(\estg - \wh\theta_l\bigr)^{2} + \frac{k - \propselectl}{\propselectl |\groupset_{nl}|} \sum_{u \in \mathcal{G}_{nl}^{\nu}} \bigl(\estgone{u} - \estgtwo{u}\bigr)^{2}.
\end{equation}
First we record the identity $\wh\theta_l = |\groupset_{nl}|\inv \sum_{\group \in \groupset_{nl}} \estg$.
To see this, note
\begin{align*}
\frac{1}{|\groupset_{nl}|} \sum_{\group \in \groupset_{nl}} \estg
&= \frac{1}{|\groupset_{nl}|} \sum_{\group \in \groupset_{nl}} \biggl(\frac{1}{a}\sum_{i \in \group}\Di Y_i - \frac{1}{k-a}\sum_{i \in \group}(1-\Di) Y_i\biggr) \\
&= \frac{1}{|\groupset_{nl}|} \biggl(\frac{1}{a}\sum_{i \in \snl}\Ti \Di Y_i - \frac{1}{k-a}\sum_{i \in \snl}\Ti (1-\Di) Y_i\biggr) \\
&= \frac{1}{k|\groupset_{nl}|\propfn(1-\propfn)} \sum_{i \in \snl}\Ti (\Di - \propfn) Y_i = \frac{1}{\nl \propselectl \propfn(1-\propfn)} \sum_{i \in \snl}\Ti (\Di - \propfn) Y_i = \wh\theta_l.
\end{align*}
The second equality uses $\sum_{\group \in \groupset_{nl}}\sum_{i \in \group} = \sum_{i \in \snl}\Ti$, since $\groupset_{nl}$ partitions $\{i \in \snl : \Ti = 1\}$.
The third combines fractions using $a = \propfn k$ and the identity $(1-\propfn)\Di - \propfn(1-\Di) = \Di - \propfn$.
The fourth uses $|\groupset_{nl}| = \nl \propselectl / k$, and the final equality is by definition of $\wh\theta_l$.

Next, we expand squares in \eqref{equation:Vl-stratum}.
Using the identity above, the first sum satisfies $\sum_{\group \in \groupset_{nl}}(\estg - \wh\theta_l)^2 = \sum_{\group \in \groupset_{nl}}\estg^2 - |\groupset_{nl}|\wh\theta_l^2$.
For the second sum, each $\group \in \groupset_{nl}$ appears exactly once across $\mathcal{G}_{nl}^\nu$ as either $\group_1(u)$ or $\group_2(u)$, so $\sum_{u \in \mathcal{G}_{nl}^\nu}(\estgone{u} - \estgtwo{u})^2 = \sum_{\group \in \groupset_{nl}}\estg^2 - 2\sum_{u \in \mathcal{G}_{nl}^\nu}\estgone{u}\estgtwo{u}$.
Substituting into \eqref{equation:Vl-stratum}, we have
\begin{align*}
\wh V_l &= \frac{1}{|\groupset_{nl}|}\bigl(\sum_{\group \in \groupset_{nl}}\estg^2 - |\groupset_{nl}|\wh\theta_l^2\bigr) + \frac{k - \propselectl}{\propselectl |\groupset_{nl}|}\bigl(\sum_{\group \in \groupset_{nl}}\estg^2 - 2\sum_{u \in \mathcal{G}_{nl}^\nu}\estgone{u}\estgtwo{u}\bigr) \\
&= \bigl(\frac{1}{|\groupset_{nl}|} + \frac{k - \propselectl}{\propselectl |\groupset_{nl}|}\bigr)\sum_{\group \in \groupset_{nl}}\estg^2 - \frac{2(k - \propselectl)}{\propselectl |\groupset_{nl}|}\sum_{u \in \mathcal{G}_{nl}^\nu}\estgone{u}\estgtwo{u} - \wh\theta_l^2 \\
&= \frac{k^2}{\nl \propselectl^2}\sum_{\group \in \groupset_{nl}}\estg^2 - \frac{2k(k - \propselectl)}{\nl \propselectl^2}\sum_{u \in \mathcal{G}_{nl}^\nu}\estgone{u}\estgtwo{u} - \wh\theta_l^2.
\end{align*}
Next we aggregate across strata. 
Using the work above, $\sum_{l=1}^{L} \frac{\nl}{n}\wh V_l$ is
\begin{align*}
&= \frac{1}{n}\sum_{l=1}^{L}\sum_{\group \in \groupset_{nl}}\frac{k^2}{\propselectl^2}\estg^2 - \frac{2}{n}\sum_{l=1}^{L}\sum_{u \in \mathcal{G}_{nl}^\nu}\frac{k(k - \propselectl)}{\propselectl^2}\estgone{u}\estgtwo{u} - \sum_{l=1}^{L}\frac{\nl}{n}\wh\theta_l^2 \\
&= \frac{1}{n}\sum_{\group \in \groupset_n}\frac{k^2}{\propselect_\group^2}\estg^2 - \frac{2}{n}\sum_{u \in \groupsetnu}\frac{k(k - \propselect_u)}{\propselect_u^2}\estgone{u}\estgtwo{u} - \sum_{l=1}^{L}\frac{\nl}{n}\wh\theta_l^2 = \wh v_1 - \wh v_2 - \sum_{l=1}^{L}\frac{\nl}{n}\wh\theta_l^2.
\end{align*}
Finally, we handle the cross-stratum correction.
This expands as $\sum_l (\nl/n)(\wh\theta_l - \est)^2 = \sum_l (\nl/n)\wh\theta_l^2 - \est^2$ using $\est = \sum_l (\nl/n)\wh\theta_l$ and $\sum_l \nl = n$.
Adding this to the identity above, the $\sum_l(\nl/n)\wh\theta_l^2$ terms cancel, leaving $\sum_l (\nl/n)[\wh V_l + (\wh\theta_l - \est)^2] = \wh v_1 - \wh v_2 - \est^2$.
This gives \eqref{equation:new-variance-aggregated}.
\end{proof}

\begin{lem}[Variance Estimator Consistency] \label{lemma:new-inference-consistency-conditioning}
In the setting above, with $\psisamp = \psiassign = \psi$, design $\Tn \sim \localdesigncond(\psi, \propselect(\psi))$ and $\Dn \sim \localdesigncond(\psi, \propfn)$ with $\propfn=a/k$ constant, and sampling-subordinate assignment,
\begin{equation} \label{equation:varest-limit-conditioning}
\varest \convp \varlocal.
\end{equation}
\end{lem}

\begin{proof}[Proof of Lemma \ref{lemma:new-inference-consistency-conditioning}]

Let $\filtrationhd = \sigma(\Wn, \Tn, \permn)$.
Since $\psin \in \sigma(\Wn)$, the design filtration $\filtrationcandpsi = \sigma(\psin, \permn, \Tn)$ of Section~\ref{proofs:asymptotics} satisfies $\filtrationcandpsi \sub \filtrationhd$. 
Moreover, $\filtrationhd \indep \eta$ for the assignment randomness $\eta$ of Definition~\ref{defn:design-construction}, since $\eta \indep (\Wn, \permn, \Tn)$ by construction.
By Definition~\ref{defn:local_randomization} and the construction of $\groupsetnu$ above, the partitions $\groupset_n$ and $\groupsetnu$ are determined by $\psin$, $\permn$, and the sampled set $\{i : \Ti = 1\}$, hence are $\filtrationcandpsi$-measurable and thus $\filtrationhd$-measurable.
By Lemma~\ref{lemma:aggregated-form-2}, $\varest = \wh v_1 - \wh v_2 - \est^{2}$, where
\begin{equation*}
\wh v_1 = \frac{1}{n}\sum_{\group \in \groupset_n} \frac{k^{2}}{\propselect_{\group}^{2}} \estg^{2}, \qquad
\wh v_2 = \frac{2}{n}\sum_{u \in \groupsetnu} \frac{k(k - \propselect_u)}{\propselect_u^{2}} \estgone{u}\estgtwo{u}.
\end{equation*}

First consider $\wh v_1$.
Define $\hti = (\Di - \propfn)/(\propfn(1-\propfn))$.
Note the identity $\hti Y_i = \tau_i + \hti \yleveli$, where $\yleveli = (1-\propfn)Y_i(1) + \propfn Y_i(0)$.
Then $\estg = k^{-1}\sum_{i \in \group} \hti Y_i = \theta_\group + k^{-1}\sum_{i \in \group}\hti\yleveli$, where $\theta_\group = k^{-1}\sum_{i \in \group}\tau_i$.
Since $\groupset_n$, $W_{1:n}$ are $\filtrationhd$-measurable and $E[\hti | \filtrationhd] = 0$, it follows that $E[\estg | \filtrationhd] = \theta_\group$.
Then $E[\estg^2 | \filtrationhd] = \theta_\group^2 + \var(\estg | \filtrationhd)$.
By Lemma~\ref{lemma:design_properties}, $\var(\hti | \filtrationhd) = [\propfn(1-\propfn)]\inv$ and $\cov(\hti, H_j | \filtrationhd) = -[(k-1)\propfn(1-\propfn)]\inv$ for $i \neq j \in \group$, so that  
\begin{align*}
\var(\estg | \filtrationhd)
&= k^{-2}\var\bigl(\sum_{i \in \group} \hti \yleveli \bigm| \filtrationhd\bigr) = k^{-2}\sum_{i, j \in \group}\bar Y_i \bar Y_j \cov(\hti, H_j | \filtrationhd) \\
&= [k^2 \propfn(1-\propfn)]\inv \sum_{i \in \group}\bar Y_i^{2} - [k^2(k-1)\propfn(1-\propfn)]\inv \sum_{i \neq j \in \group}\bar Y_i \bar Y_j.
\end{align*}
Substituting $E[\estg^2 | \filtrationhd] = \theta_\group^2 + \var(\estg | \filtrationhd)$ into $E[\wh v_1 | \filtrationhd]$ gives
\begin{align*}
E[\wh v_1 | \filtrationhd] &= n\inv \sum_{\group \in \groupset_n}\frac{k^2}{\propselect_\group^2}\theta_\group^2 + n\inv \sum_{\group \in \groupset_n}[\propselect_\group^2 \propfn(1-\propfn)]\inv \sum_{i \in \group}\bar Y_i^2 \\
&- n\inv \sum_{\group \in \groupset_n}[\propselect_\group^2 (k-1) \propfn(1-\propfn)]\inv \sum_{i \neq j \in \group}\bar Y_i \bar Y_j \equiv A_n + B_n + C_n.
\end{align*}
First consider $B_n$.
We have $\propselect_\group = \propselect(\psii)$ for $i \in \group$ by the assumption of sampling-subordinate assignment.
Since $\cup_{\group \in \groupset_n} \{i \in \group\} = \{i : \Ti = 1\}$,
\[
B_n = [\propfn(1-\propfn)]\inv n\inv \sum_{\group \in \groupset_n}\sum_{i \in \group}\frac{\bar Y_i^2}{\propselect(\psii)^2} = [\propfn(1-\propfn)]\inv \en\bigl[\Ti \bar Y_i^2/\propselect(\psii)^2\bigr].
\]
By Lemma~\ref{lemma:stratified_wlln} part (3) applied to the sampling design with $h = \bar Y^2/\propselect(\psi)^2$, $\en[\Ti \bar Y_i^2/\propselect(\psii)^2] \convp E[\bar Y^2/\propselect(\psi)]$, so $B_n \convp [\propfn(1-\propfn)]\inv E[\bar Y^2/\propselect(\psi)]$.

Next consider $C_n$. Using $\propselect_\group = \propselect(\psii) = \propselect(\psij)$ for $i, j \in \group$, we have
\[
C_n = -[(k-1)\propfn(1-\propfn)]\inv n\inv \sum_{\group \in \groupset_n}\sum_{\substack{i, j \in \group \\ i \neq j}}\frac{\bar Y_i \bar Y_j}{\propselect(\psii)\propselect(\psij)}.
\]
By Lemma~\ref{lemma:bilinear-form} at $K = k$ on $\groupset_n$ with $A_i = B_i = \bar Y_i/\propselect(\psii)$, the inner double sum converges in probability to $(k-1)E[E[\bar Y | \psi]^2/\propselect(\psi)]$, so $C_n \convp -[\propfn(1-\propfn)]\inv E[E[\bar Y | \psi]^2/\propselect(\psi)]$.

Finally, consider $A_n$. Expanding $\theta_\group^2 = k^{-2}\sum_{i, j \in \group}\tau_i\tau_j$ in $A_n$, the factor of $k^2$ cancels, and splitting diagonal from off-diagonal,
\[
A_n = n\inv \sum_{\group \in \groupset_n}\frac{1}{\propselect_\group^2}\sum_{i \in \group}\tau_i^2 + n\inv \sum_{\group \in \groupset_n}\frac{1}{\propselect_\group^2}\sum_{\substack{i, j \in \group \\ i \neq j}}\tau_i\tau_j \equiv A_n^D + A_n^O.
\]
For the diagonal piece, using $\propselect_\group = \propselect(\psii)$ for $i \in \group$, $\cup_{\group \in \groupset_n}\group = \{i : \Ti = 1\}$, and Lemma~\ref{lemma:stratified_wlln} part (3) applied to the sampling design,
\[
A_n^D = n\inv \sum_{\group \in \groupset_n}\sum_{i \in \group}\frac{\tau_i^2}{\propselect(\psii)^2} = \en[\Ti \tau_i^2/\propselect(\psii)^2] \convp E[\tau^2/\propselect(\psi)].
\]
For the off-diagonal piece, using $\propselect_\group = \propselect(\psii) = \propselect(\psij)$ for $i, j \in \group$,
\[
A_n^O = n\inv \sum_{\group \in \groupset_n}\sum_{\substack{i, j \in \group \\ i \neq j}}\frac{\tau_i\tau_j}{\propselect(\psii)\propselect(\psij)}.
\]
By Lemma~\ref{lemma:bilinear-form} at $K = k$ on $\groupset_n$ with $A_i = B_i = \tau_i/\propselect(\psii)$, $A_n^O \convp (k-1)E[\catefn(\psi)^2/\propselect(\psi)]$, where $\catefn(\psi) = E[\tau | \psi]$.
Using $E[\tau^2 | \psi] = \hkte(\psi) + \catefn(\psi)^2$ for $\hkte(\psi) = \var(\te | \psi)$, we have the limit $A_n \convp E[\hkte(\psi)/\propselect(\psi)] + k E[\catefn(\psi)^2/\propselect(\psi)]$.
Note also $E[\bar Y^2/\propselect(\psi)] = E[E[\bar Y^2 | \psi]/\propselect(\psi)]$ and $E[\bar Y^2 | \psi] = E[\bar Y | \psi]^2 + \var(\bar Y | \psi)$, so the other terms
\[
B_n + C_n \convp [\propfn(1-\propfn)]\inv \bigl(E[\bar Y^2/\propselect(\psi)] - E[E[\bar Y | \psi]^2/\propselect(\psi)]\bigr) = [\propfn(1-\propfn)]\inv E[\var(\bar Y | \psi)/\propselect(\psi)].
\]
Combining with the limit of $A_n$, we have shown that
\[
E[\wh v_1 | \filtrationhd] \convp E[\hkte(\psi)/\propselect(\psi)] + k E[\catefn(\psi)^2/\propselect(\psi)] + \frac{E[\var(\bar Y | \psi)/\propselect(\psi)]}{\propfn(1-\propfn)}.
\]

Next consider $E[\wh v_2 | \filtrationhd]$.
By Lemma~\ref{lemma:group_aggregate_independence} applied with $u_\group = \estg$, which is a function of $(\Di)_{i \in \group}$ and the $\filtrationhd$-measurable potential outcomes $(Y_i(0), Y_i(1))_{i \in \group}$, the group-level estimates $(\estg)_{\group \in \groupset_n}$ are jointly independent conditional on $\filtrationhd$; in particular $\estgone{u} \indep \estgtwo{u} | \filtrationhd$ since $\group_1(u) \neq \group_2(u)$.
Combined with $E[\estg | \filtrationhd] = \theta_\group$, this gives $E[\estgone{u}\estgtwo{u} | \filtrationhd] = \theta_{\group_1(u)}\theta_{\group_2(u)}$.
Since $\groupsetnu$ is $\filtrationhd$-measurable, taking conditional expectation through the sum yields
\begin{align*}
E[\wh v_2 | \filtrationhd] &= \frac{2}{n}\sum_{u \in \groupsetnu}\frac{k(k - \propselect_u)}{\propselect_u^2}\theta_{\group_1(u)}\theta_{\group_2(u)} = \frac{2}{n}\sum_{u \in \groupsetnu}\frac{k - \propselect_u}{k \propselect_u^2}\sum_{\substack{i \in \group_1(u) \\ j \in \group_2(u)}}\tau_i\tau_j.
\end{align*}
Define $A_i = \sqrt{(k - \propselect(\psii))/k} \cdot \tau_i/\propselect(\psii)$, which is $\sigma(W_i)$ measurable and has $E[A_i^2] \leq E[\tau_i^2]/\propbound^2 < \infty$ using $0 \leq (k - \propselect)/k \leq 1$ and propensity bound $\propselect(\psi) \geq \propbound > 0$.
By sampling-subordinate assignment, $\propselect_u = \propselect(\psii) = \propselect(\psij)$ for $i, j \in u$, so $A_i A_j = (k - \propselect_u)/(k \propselect_u^2) \tau_i \tau_j$ for all $i, j \in u$.
Using $u = \group_1(u) \cup \group_2(u)$, we have the decomposition
\[
\sum_{i \neq j \in u} A_i A_j = \sum_{i \neq j \in \group_1(u)} A_i A_j + \sum_{i \neq j \in \group_2(u)} A_i A_j + 2\sum_{\substack{i \in \group_1(u) \\ j \in \group_2(u)}} A_i A_j.
\]
Solving for the cross-group piece and summing over the disjoint unions $u \in \groupsetnu$,
\begin{align*}
E[\wh v_2 | \filtrationhd] &= \frac{2}{n}\sum_{u \in \groupsetnu}\sum_{\substack{i \in \group_1(u) \\ j \in \group_2(u)}} A_i A_j = \frac{1}{n}\sum_{u \in \groupsetnu}\biggl[\sum_{i \neq j \in u} A_i A_j - \sum_{i \neq j \in \group_1(u)} A_i A_j - \sum_{i \neq j \in \group_2(u)} A_i A_j\biggr].
\end{align*}
Note that by construction of $\groupsetnu$, each group $\group \in \groupset_n$ appears exactly once as $\group_1(u)$ or $\group_2(u)$ for a unique union $u \in \groupsetnu$, so summing over $u$ recovers within-group sums over $g \in \groupsetn$, $\sum_{u \in \groupsetnu} \big [\sum_{i \neq j \in \group_1(u)} A_i A_j + \sum_{i \neq j \in \group_2(u)} A_i A_j \big ] = \sum_{\group \in \groupset_n}\sum_{i \neq j \in \group} A_i A_j$.
Substituting,
\begin{equation} \label{equation:v2-bilinear}
E[\wh v_2 | \filtrationhd] = \frac{1}{n}\sum_{u \in \groupsetnu}\sum_{i \neq j \in u} A_i A_j - \frac{1}{n}\sum_{\group \in \groupset_n}\sum_{i \neq j \in \group} A_i A_j.
\end{equation}
By the argument in the proof of Lemma~\ref{lemma:matched-union-tight-matching}, tight matching of $\groupset_n$ and the centroid-matching condition imposed in Section~\ref{section:inference} imply $F(\groupsetnu)=\op(1)$.
By sampling-subordinate assignment, $\propselect(\psii)$ is constant over $i \in u$ and $i \in \group$.
Applying Lemma~\ref{lemma:bilinear-form} with $B_i = A_i$ at $K = 2k$ on $\groupsetnu$ and at $K = k$ on $\groupset_n$, we obtain
\begin{align*}
\frac{1}{n}\sum_{u \in \groupsetnu}\sum_{i \neq j \in u} A_i A_j &\convp (2k-1) E\bigl[\propselect(\psii) E[A_i | \psii]^2\bigr], \\
\frac{1}{n}\sum_{\group \in \groupset_n}\sum_{i \neq j \in \group} A_i A_j &\convp (k-1) E\bigl[\propselect(\psii) E[A_i | \psii]^2\bigr].
\end{align*}
Then $E[\wh v_2 | \filtrationhd] \convp k E[\propselect(\psii) E[A_i | \psii]^2] = E[(k - \propselect(\psii)) \catefn(\psii)^2/\propselect(\psii)] = k E[\catefn(\psii)^2/\propselect(\psii)] - E[\catefn(\psii)^2]$.
Finally, by Theorem~\ref{thm:clt_fixed} $\est \convp \ate$, so $\est^2 \convp \ate^2$ by continuous mapping.

Putting things together, recall that $\varest = \wh v_1 - \wh v_2 - \est^2$.
By Lemma~\ref{lemma:coupling-v1-v2}, we have $\wh v_1 - E[\wh v_1 | \filtrationhd] = \op(1)$ and $\wh v_2 - E[\wh v_2 | \filtrationhd] = \op(1)$.
Moreover, from the limits above
\[
E[\wh v_1 - \wh v_2 | \filtrationhd] - \est^2 \convp E[\hkte(\psi)/\propselect(\psi)] + \frac{E[\var(\bar Y | \psi)/\propselect(\psi)]}{\propfn(1-\propfn)} + E[\catefn(\psi)^2] - \ate^2.
\]
We have $E[\catefn(\psi)^2] - \ate^2 = \var(\catefn(\psi))$ and $\hkte(\psi) = \hk_1(\psi) + \hk_0(\psi) - 2 \cov(Y(1), Y(0) | \psi)$ and $\var(\bar Y | \psi) = (1-\propfn)^2 \hk_1(\psi) + \propfn^2 \hk_0(\psi) + 2\propfn(1-\propfn) \cov(Y(1), Y(0) | \psi)$.
Then we have
\[
\hkte(\psi) + \frac{\var(\bar Y | \psi)}{\propfn(1-\propfn)} = \frac{\hk_1(\psi)}{\propfn} + \frac{\hk_0(\psi)}{1-\propfn}.
\]
Dividing through by $\propselect(\psi)$, taking expectations, and combining with the $\var(\catefn(\psi))$ term gives
\begin{equation} \label{equation:v1-v2-est2-limit}
\wh v_1 - \wh v_2 - \est^2 \convp \var(\catefn(\psi)) + E\biggl[\frac{\hk_1(\psi)}{\propfn \propselect(\psi)} + \frac{\hk_0(\psi)}{(1-\propfn) \propselect(\psi)}\biggr] = \varlocal,
\end{equation}
with the second equality by \eqref{equation:asymptotic-theorem}.
This finishes the proof.
\end{proof}

The next lemma is stated for a generic partition $\mathcal{P}_n$, which we will set to either $\groupset_n$ or $\groupsetnu$ depending on the context in which the lemma is applied.

\begin{lem}[Bilinear Form] \label{lemma:bilinear-form}
Let $\mathcal{P}_n$ be a $\filtrationcandpsi$-measurable partition of $\{i : \Ti = 1\}$ into groups of size $K$ with $n\inv \sum_{\group \in \mathcal{P}_n}\sum_{i \in \group}|\psii - \bar\psi_\group|_2^{2} = \op(1)$, for the standing filtration $\filtrationcandpsi = \sigma(\psin, \permn, \Tn)$ of Section~\ref{proofs:asymptotics}.
Let $A_i, B_i \in \sigma(W_i)$ with $E[A^{2} + B^{2}] < \infty$.
Then
\begin{equation} \label{equation:data-bilinear}
\frac{1}{n}\sum_{\group \in \mathcal{P}_n}\sum_{\substack{i, j \in \group \\ i \ne j}} A_i B_j \convp (K-1) E\bigl[\propselect(\psii) \cdot E[A_i | \psii] E[B_i | \psii]\bigr].
\end{equation}
\end{lem}

\begin{proof}[Proof of Lemma \ref{lemma:bilinear-form}]
Let $A_i = v_a(\psii) + \residuali^a$ with $E[\residuali^a | \psii] = 0$, similarly $B_i = v_b(\psii) + \residuali^b$.
\begin{align*}
\frac{1}{n}\sum_{\group \in \mathcal{P}_n}\sum_{\substack{i, j \in \group \\ i \ne j}} A_i B_j = \frac{1}{n}\sum_{\group \in \mathcal{P}_n}\sum_{\substack{i, j \in \group \\ i \ne j}}(v_a(\psii) + \residuali^a)(v_b(\psij) + \residualj^b) \equiv A_n + B_n + C_n + R_n,
\end{align*}
where $A_n$ collects the $v_a(\psii) v_b(\psij)$ terms, $R_n$ collects the $\residuali^a\residualj^b$ terms, and $B_n, C_n$ collect the cross terms $v_a(\psii) \residualj^b$ and $\residuali^a v_b(\psij)$ respectively.
First consider $A_n$.
Denote $\via = v_a(\psii)$, $\vja = v_a(\psij)$, and $\vib, \vjb$ for $b$.
Re-indexing $\sum_{i \neq j \in \group}\via\vjb = \sum_{i < j \in \group}(\via\vjb + \vja\vib)$. 
By the identity $a_i b_j + b_i a_j = -(a_i - a_j)(b_i - b_j) + a_i b_i + a_j b_j$, this is
\begin{align*}
A_n &= \frac{1}{n}\sum_{\group \in \mathcal{P}_n}\sum_{i < j \in \group}(\via\vjb + \vja\vib) \\
&= \frac{1}{n}\sum_{\group \in \mathcal{P}_n}\sum_{i < j \in \group}(\via\vib + \vja\vjb) - \frac{1}{n}\sum_{\group \in \mathcal{P}_n}\sum_{i < j \in \group}(\via - \vja)(\vib - \vjb) \equiv A_n^1 + A_n^2.
\end{align*}

First consider $A_n^2$.
Note the identity $\sum_{i, j \in \group}|\psii - \psij|_2^2 = 2K\sum_{i \in \group}|\psii - \bar\psi_\group|_2^2$, so the tight matching hypothesis implies $X_n \equiv n\inv \sum_{\group \in \mathcal{P}_n}\sum_{i, j \in \group}|\psii - \psij|_2^{2} = \op(1)$.
By Lemma \ref{lemma:op-rate}, we can choose $\mu_n \to \infty$ with $\mu_n X_n \convp 0$. 
Then for $\lambda_n = \mu_n^{1/2}$, we have $\lambda_n \to \infty$ and $\lambda_n^2 X_n \convp 0$.
By Lemma \ref{lemma:lipschitz-approximation:rate}, there exist sequences $(z_n^a)_{n \ge 1}$ and $(z_n^b)_{n \ge 1}$ with $|z_n^a|_{lip} \vee |z_n^b|_{lip} \leq \lambda_n$, $|z_n^a - v_a|_{2, \psi} = o(1)$, and $|z_n^b - v_b|_{2, \psi} = o(1)$.
Denote $\zin^a = z_n^a(\psii)$, $\zin^b = z_n^b(\psii)$.
We can expand relative to the Lipschitz approximations using $\via - \vja = (\via - \zin^a) + (\zin^a - \zjn^a) + (\zjn^a - \vja)$ and similarly for $b$.
Then $|xy| \leq x^2 + y^2$ and $(p_1 + p_2 + p_3)^2 \leq 3(p_1^2 + p_2^2 + p_3^2)$ give
\[
|(\via - \vja)(\vib - \vjb)| \lesssim \sum_{c \in \{a, b\}}\bigl(|v_{ic} - \zin^c|^2 + |\zin^c - \zjn^c|^2 + |\zjn^c - v_{jc}|^2\bigr).
\]
Using this to bound $|A_n^2|$ termwise, we have
\begin{align*}
|A_n^2| &\leq n\inv \sum_{\group} \sum_{i < j \in \group} \sum_{c \in \{a, b\}} \bigl(|v_{ic} - \zin^c|^2 + |\zin^c - \zjn^c|^2 + |\zjn^c - v_{jc}|^2\bigr) \\
&\lesssim \sum_{c \in \{a, b\}} \bigl[ K \en[\Ti |v_{ic} - \zin^c|^2] + \lambda_n^2 X_n \bigr] = \op(1).
\end{align*}
The first inequality is the pointwise bound above.
The second equality follows by counting $\sum_{i < j \in \group} |v_{ic} - \zin^c|^2 \leq (K-1) \sum_{i \in \group} |v_{ic} - \zin^c|^2$ and similarly for the third term.
Also, we use the Lipschitz bound $|\zin^c - \zjn^c|^2 \leq \lambda_n^2 |\psii - \psij|_2^2$, and $\sum_\group \sum_{i \in \group} f_i = \sum_i \Ti f_i$ since $\mathcal{P}_n$ partitions $\{\Ti = 1\}$.
For the final equality, $\en[\Ti |v_{ic} - \zin^c|^2] \leq \en[|v_{ic} - \zin^c|^2] = \op(1)$ by Markov, since $E[|v_{ic} - \zin^c|^2] = |z_n^c - v_c|_{2,\psi}^2 = o(1)$.
For the second term, $\lambda_n^2 X_n \convp 0$ by our choice of $\lambda_n$ above.
Then we have shown that $A_n^2 = \op(1)$.

Next consider $A_n^1$. 
Since $\sum_{i < j \in \group}(\via\vib + \vja\vjb) = (K - 1)\sum_{i \in \group}\via\vib$ for $|\group| = K$,
\begin{align*}
A_n^1 \;=\; \frac{K - 1}{n}\sum_{\group \in \mathcal{P}_n}\sum_{i \in \group}\via\vib \;=\; (K - 1)\,\en[\Ti\,\via\,\vib] \;\convp\; (K - 1)\,E[\propselecti\,\via\,\vib].
\end{align*}
The middle equality since $\mathcal{P}_n$ partitions $\{i : \Ti = 1\} \subseteq [n]$.
The convergence follows from Lemma~\ref{lemma:stratified_wlln} part (3) applied to the sampling design with $h = v_a v_b$, since $E[|\propselecti\via\vib|] \leq E[\via^2]^{1/2}E[\vib^2]^{1/2} \lesssim E[A_i^2]^{1/2}E[B_i^2]^{1/2} < \infty$.
Combining, $A_n \convp (K - 1)\,E[\propselect(\psii)\cdot v_a(\psii)\,v_b(\psii)] = (K - 1)\,E[\propselect(\psii)\cdot E[A_i | \psii]\,E[B_i | \psii]]$.

Next claim that $R_n = \op(1)$.
We verify the conditions of Lemma \ref{lemma:lln} with the filtration $\filtrationcandpsi$ for the terms
\[
u_\group \equiv \sum_{\substack{i, j \in \group \\ i \neq j}}\residuali^a\,\residualj^b.
\]
We apply Lemma~\ref{lemma:partitions} with $h_i = \psii$ and $\kappa = (\permn, \Tn)$, so that $\sigma(h_{1:n}, \kappa) = \sigma(\psin, \permn, \Tn) = \filtrationcandpsi$.
The required hypothesis $\kappa \indep \Wn | h_{1:n}$ is the statement $(\permn, \Tn) \indep \Wn | \psin$, which we verify by chaining two reductions.
First, $\Wn \indep \permn$ unconditionally gives $\Wn \indep \permn | \psin$ via $(A, B) \indep C \Rightarrow A \indep C | B$ with $A = \Wn$, $B = \psin \in \sigma(\Wn)$, $C = \permn$.
Second, by Definition~\ref{defn:design-construction} applied to the sampling design, $\Tn$ is a function of $(\psin, \permn, \eta^T)$ for sampling randomness $\eta^T \indep (\Wn, \permn)$, so $\Tn \indep \Wn | (\psin, \permn)$.
Composing, $\Wn | (\psin, \permn, \Tn) \eqdist \Wn | (\psin, \permn) \eqdist \Wn | \psin$, which is the required $(\permn, \Tn) \indep \Wn | \psin$.
Since $\mathcal{P}_n$ is $\filtrationcandpsi$-measurable by hypothesis, Lemma~\ref{lemma:partitions} gives $(u_\group)_{\group \in \mathcal{P}_n}$ jointly conditionally independent given $\filtrationcandpsi$.
Note the key fact that $(A, B) \indep C \Rightarrow A \indep C | B$.
We first reduce the conditioning:
\begin{equation} \label{equation:bilinear-conditioning-reduction}
E[\residuali^a \residualj^b | \psin, \permn, \Tn] = E[\residuali^a \residualj^b | \psin, \permn] = E[\residuali^a \residualj^b | \psin] = E[\residuali^a \residualj^b | \psii, \psij].
\end{equation}
For the first equality, $\Tn \indep \Wn |(\psin, \permn)$ as established above, hence $(\residuali^a, \residualj^b) \indep \Tn |(\psin, \permn)$ by the conditional independence fact above, since $\residuali^a, \residualj^b \in \sigma(\Wn)$.
For the second equality, by design $\Wn \indep \permn$, hence $(\residuali^a, \residualj^b) \indep \permn |\psin$ by the conditional independence fact above.
The third equality drops $\psi_{-(i,j)}$: by iid sampling, $(\residuali^a, \residualj^b, \psii, \psij) \indep \psi_{-(i,j)}$, hence $(\residuali^a, \residualj^b) \indep \psi_{-(i,j)} |(\psii, \psij)$ by the fact above.
Then note
\begin{align*}
E[\residuali^a \residualj^b | \psii, \psij] = E\bigl[\residualj^b E[\residuali^a | \residualj^b, \psii, \psij] \bigm| \psii, \psij\bigr] = E\bigl[\residualj^b E[\residuali^a | \psii] \bigm| \psii, \psij\bigr] = 0.
\end{align*}
The first equality is by tower law. 
The second uses $\residuali^a \indep (\residualj^b, \psij) |\psii$, which holds by the fact above and because iid sampling gives $(\residuali^a, \psii) \indep (\residualj^b, \psij)$.
The third is $E[\residuali^a | \psii] = 0$ by definition of the residual.
Hence $E[u_\group | \filtrationcandpsi] = 0$ for each $\group \in \mathcal{P}_n$.

Next, to invoke Lemma \ref{lemma:lln}, we need to show $\frac{1}{n}\sum_{\group} E[|u_\group| \one(|u_\group| > c_n) | \filtrationcandpsi] = \op(1)$ for suitable $c_n \to \infty$.
Note for any $a, b \geq 0$,  $ab \one(ab > c) \leq a^2 \one(a^2 > c) + b^2 \one(b^2 > c)$ and for any $a_k \geq 0$, $k = 1, \dots, m$, we have
\begin{equation} \label{equation:indicator-facts}
\textstyle\sum_k a_k \one(\sum_k a_k > c) \leq m\sum_k a_k \one(a_k > c/m).
\end{equation}
By triangle and Young's inequalities, $|u_\group| \leq \sum_{i \neq j \in \group} |\residuali^a| |\residualj^b| \leq \frac{1}{2}\sum_{i \neq j \in \group} ((\residuali^a)^2 + (\residualj^b)^2) = \frac{K-1}{2} \sum_{i \in \group, c \in \{a, b\}} (\residuali^c)^2$.
Then for any $c_n \to \infty$,
\begin{align*}
E[|u_\group| \one(|u_\group| > c_n) | \filtrationcandpsi]
  &\leq \tfrac{K-1}{2} E\bigl[\sum_{i \in \group, c} (\residuali^c)^2 \one\bigl(\tfrac{K-1}{2}\sum_{i \in \group, c} (\residuali^c)^2 > c_n\bigr) \bigm| \filtrationcandpsi\bigr] \\
  &\leq K(K-1) \sum_{i \in \group, c \in \{a, b\}} E\bigl[(\residuali^c)^2 \one\bigl((\residuali^c)^2 > c_n [K(K-1)]\inv\bigr) \bigm| \filtrationcandpsi\bigr] \\
  &= K(K-1) \sum_{i \in \group, c \in \{a, b\}} E\bigl[(\residuali^c)^2 \one\bigl((\residuali^c)^2 > c_n [K(K-1)]\inv\bigr) \bigm| \psii\bigr].
\end{align*}
The second inequality applies \eqref{equation:indicator-facts} to $\sum_{i \in \group, c}(\residuali^c)^2$ as a sum of $m=2K$ non-negative terms and threshold $c = 2c_n/(K-1)$, giving per-term threshold $c/m = c_n [K(K-1)]\inv$.
The third reduces conditioning to $\psii$ by the same argument as in the conditional-mean calculation.
Then we can bound $\frac{1}{n}\sum_\group E[|u_\group| \one(|u_\group| > c_n) | \filtrationcandpsi]$ above by
\begin{align*}
&\frac{K(K-1)}{n} \sum_{c \in \{a, b\}} \sum_\group \sum_{i \in \group} E[(\residuali^c)^2 \one((\residuali^c)^2 > c_n [K(K-1)]\inv) | \psii] \\
  &= \frac{K(K-1)}{n} \sum_{c \in \{a, b\}} \sum_{i=1}^n \Ti E[(\residuali^c)^2 \one((\residuali^c)^2 > c_n [K(K-1)]\inv) | \psii].
\end{align*}
The first expression follows from the bound on $E[|u_\group| \one(|u_\group| > c_n) | \filtrationcandpsi]$ developed above.
The equality uses that $\mathcal{P}_n$ partitions $\{\Ti = 1\}$, so $\sum_\group \sum_{i \in \group} f_i = \sum_{i=1}^n \Ti f_i$.
Since $\Ti \leq 1$, this is at most $K(K-1) \sum_{c \in \{a, b\}} \en\bigl[E[(\residuali^c)^2 \one((\residuali^c)^2 > c_n [K(K-1)]\inv) | \psii]\bigr]$.
Taking expectations, by iid sampling and tower law we obtain the bound 
\[
E\frac{1}{n}\sum_\group E[|u_\group| \one(|u_\group| > c_n) | \filtrationcandpsi] \le K(K-1) \sum_{c \in \{a, b\}} E[(\residuali^c)^2 \one((\residuali^c)^2 > c_n [K(K-1)]\inv)].
\]
The RHS converges to $0$ as $c_n \to \infty$ by dominated convergence since $E[(\residuali^c)^2] \leq \var(C_i) < \infty$ for $C \in \{A, B\}$.
Then by Markov $\frac{1}{n}\sum_\group E[|u_\group| \one(|u_\group| > c_n) | \filtrationcandpsi] = \op(1)$.
Choosing $c_n = n^{1/4}$ so that $c_n = \omega(1)$ and $c_n = o(n^{1/2})$, Lemma \ref{lemma:lln} gives $R_n = \op(1)$.

The terms $B_n, C_n = \op(1)$ by an analogous argument.
For example, for $B_n$ let $u_\group^B = \sum_{i \neq j \in \group}\via\,\residualj^b$.
The conditional mean $E[\via\,\residualj^b | \filtrationcandpsi] = 0$ since $\via \in \filtrationcandpsi$ and $E[\residualj^b | \filtrationcandpsi] = E[\residualj^b | \psij] = 0$ by a similar conditioning reduction argument as above.
The Lindeberg condition is verified by the same indicator-function argument with $(\residuali^a)^2$ replaced by $\via^2$, using $E[\via^2] = E[E[A_i | \psii]^2] \leq E[A_i^2] < \infty$ by Jensen's inequality.

Combining our work above, we have $A_n + B_n + C_n + R_n \convp (K - 1)\,E[\propselect(\psii)\cdot E[A_i | \psii]\,E[B_i | \psii]]$, completing the proof.
\end{proof}

\begin{proof}[Proof of Theorem \ref{thm:inference}]
The $\ate$ part $\varest \convp \varlocal$ is Lemma \ref{lemma:new-inference-consistency-conditioning}.
The $\sate$ part $\varest_{\sate} \convp V_{\sate} + E[\var(\te | \psi)]$ is Lemma \ref{lemma:sate-inference-consistency}, with $\hkte(\psi) = \var(\te | \psi)$.
In the constant-$\propselect$ case $L = 1$, the cross-stratum correction in \eqref{equation:variance-estimator-varying} vanishes ($\wh\theta_l = \est$) and the resulting expressions equal \eqref{equation:variance-estimator}.
\end{proof}

\begin{lem}[Block Sequence] \label{lem:snake-index}
Fix integers $v, m \geq 1$. Let $\indfn_v : \{0, \dots, m-1\}^v \to \{0, \dots, m^v - 1\}$ and $\indfn : [0, 1]^v \to \{0, \dots, m^v - 1\}$ be the index functions defined in Section \ref{section:algorithms} in dimension $v \ge 1$.
Define blocks $B(t) = \indfn^{-1}(t)$ for $t = 0, \dots, m^v - 1$.
Then:
\begin{enumerate}[label={(\arabic*)}, itemindent=.5pt, itemsep=.4pt, topsep=.4pt]
\item $\indfn_v$ is a bijection and $\{B(t) : t = 0, \dots, m^v - 1\}$ is a partition of $[0,1]^v$.
\item For $t = 0, \dots, m^v - 1$, if $x, y \in B(t)$ then $|x - y|_2 \leq \sqrt{v}/m$.
\item For $t = 0, \dots, m^v - 2$, if $x \in B(t)$, $y \in B(t+1)$ then $|x - y|_2 \leq 2\sqrt{v}/m$.
\end{enumerate}
\end{lem}

\begin{proof}
First consider (1). 
We show bijectivity of $\indfn_v$ by induction.
For $v = 1$, $\indfn_1(z_1) = z_1$ is the identity, so this holds trivially.
For $v \geq 2$, note that the second summand in \eqref{equation:gray-code-text} lies in $\{0, \dots, m^{v-1} - 1\}$.
Then we can recover $z_v = \lfloor \indfn_v(z) / m^{v-1} \rfloor$.
Moreover, $(z_1, \dots, z_{v-1})$ can be recovered from the parity of $z_v$ and using inductive bijectivity of $\indfn_{v-1}$.
Then we have constructed an inverse for $\indfn_v$, so it is a bijection.
For any $\psi \in [0,1]^v$, $\lfloor m \psi \rfloor \in \{0, \dots, m\}^v$, so $\lfloor m \psi \rfloor \wedge (m-1) \in \{0, \dots, m-1\}^v$, hence $\indfn(\psi) = \indfn_v(\lfloor m \psi \rfloor \wedge (m-1))$ is a well-defined function from $[0, 1]^v$ to $\{0, \dots, m^v - 1\}$.
Then the level sets $B(t) = \indfn\inv(t)$ of any function are pairwise disjoint with union equal to its domain.
Moreover, each $B(t)$ is non-empty since $\indfn(m\inv \indfn_v\inv(t)) = \indfn_v(\indfn_v\inv(t)) = t$ so $m\inv \indfn_v\inv(t) \in B(t)$.
Then $\{B(t) : t = 0, \dots, m^v - 1\}$ partitions $[0, 1]^v$.

For item (2), define $z(t) = \indfn_v\inv(t)$.
If $\psi \in B(t) = \indfn\inv(t)$, then $\indfn_v(\lfloor m \psi \rfloor \wedge (m-1)) = t$, so $z(t) = \indfn_v\inv(t) = \lfloor m \psi \rfloor \wedge (m-1)$.
Then componentwise, by definition of the floor function and the truncation at $m-1$, $z(t) \leq m \psi \leq z(t) + 1$, with equality in the upper bound when $\psi_j=1$.
Then for any $x, y \in B(t)$, we have $|x_j - y_j| \leq 1/m$ for each $j$, so $|x - y|_2^2 \leq v/m^2$, establishing (2).

For item (3), we claim $|z(t+1) - z(t)|_1 = 1$ for $t = 0, \dots, m^v - 2$.
Assuming the claim, let $x \in B(t)$, $y \in B(t+1)$. 
By the key inequality in the proof of (2), $z(t) \leq m x \leq z(t) + 1$ and $z(t+1) \leq m y \leq z(t+1) + 1$ componentwise. 
Subtracting inequalities, for each coordinate $j$ we obtain $z(t)_j - z(t+1)_j - 1 \leq m(x_j - y_j) \leq z(t)_j - z(t+1)_j + 1$, so $m(x_j - y_j)$ lies in an interval of length $2$ centered at $z(t)_j - z(t+1)_j$.
The maximum absolute value attained on such an interval is $|z(t)_j - z(t+1)_j| + 1$, so $|m(x_j - y_j)| \leq |z(t)_j - z(t+1)_j| + 1$.
By the claim, exactly one coordinate of $z(t+1) - z(t)$ has absolute value $1$ and the rest are $0$, so squaring and summing over $j$ gives $m^2 |x - y|_2^2 \leq \sum_j (|z(t)_j - z(t+1)_j| + 1)^2 = (1 + 1)^2 + (v-1) \cdot 1 = v + 3 \leq 4 v$.
Hence $|x - y|_2 \leq 2\sqrt{v}/m$, establishing (3).

To finish the proof, we show the key claim $|z(t+1) - z(t)|_1 = 1$ for $t = 0, \dots, m^v - 2$.
Recall we defined $z(t) = \indfn_v\inv(t)$.
For $v = 1$, $\indfn_1$ is the identity and the adjacency property $|z(t+1) - z(t)|_1 = 1$ is immediate.
For $v \geq 2$, assume the adjacency property holds for $\indfn_{v-1}$.
Suppose $\indfn_v(z') = \indfn_v(z) + 1$.
Denote $\bar z = (z_1, \dots, z_{v-1})$ and similarly for $\bar z'$.
Write $\indfn_v(z) = z_v m^{v-1} + r(\bar z, z_v)$, defining $r(\bar z, z_v) \in \{0, \dots, m^{v-1} - 1\}$ to be the second summand in \eqref{equation:gray-code-text}.
We consider two cases.

\emph{Case 1:} Suppose $z_v' = z_v$.
Then $r(\bar z', z_v) - r(\bar z, z_v) = 1$.
If $z_v$ is even, inspecting the formula this implies $\indfn_{v-1}(\bar z') - \indfn_{v-1}(\bar z) = +1$.
If $z_v$ is odd, then $\indfn_{v-1}(\bar z') - \indfn_{v-1}(\bar z) = -1$.
The inductive hypothesis applied to $(\bar z, \bar z')$ or $(\bar z', \bar z)$ yields that $\bar z, \bar z'$ differ in exactly one coordinate by exactly $1$.
Together with $z_v = z_v'$ this gives the claim.

\emph{Case 2:} Suppose $z_v' \neq z_v$.
Since $\indfn_v(z') - \indfn_v(z) = (z_v' - z_v) m^{v-1} + (r(\bar z', z_v') - r(\bar z, z_v)) = 1$ and $|r(\bar z', z_v') - r(\bar z, z_v)| \leq m^{v-1} - 1$, the only possibility is $z_v' = z_v + 1$, $r(\bar z, z_v) = m^{v-1} - 1$, $r(\bar z', z_v') = 0$.
If $z_v$ is even then $\indfn_{v-1}(\bar z) = m^{v-1} - 1$ and $z_v' = z_v + 1$ is odd so $\indfn_{v-1}(\bar z') = m^{v-1} - 1$.
Since inductively $\indfn_{v-1}$ is a bijection, we must have $\bar z = \bar z'$.
Similarly, if $z_v$ is odd then $\indfn_{v-1}(\bar z) = 0$ and $z_v'$ is even so $\indfn_{v-1}(\bar z') = 0$ so that $\bar z = \bar z'$ by inductive bijectivity.
Thus $z, z'$ differ only in the $v$-th coordinate, by exactly $1$.
Then by induction, $\indfn_v$ satisfies the adjacency property $|z(t+1) - z(t)|_1 = 1$.
\end{proof}

\clearpage

\renewcommand{\thesection}{D}
\renewcommand{\theHsection}{D}
\pagenumbering{arabic}\renewcommand{\thepage}{\arabic{page}}

\begin{center}
{\Large Secondary Online Appendix to ``Fine Stratification of Survey Experiments''}
\vskip 24pt
{\large Max Cytrynbaum}
\end{center}

This secondary online appendix is not intended for publication. 

\setcounter{subsection}{0}
\setcounter{thm}{0}
\setcounter{equation}{0}

\subsection{Lemmas} \label{proofs:lemmas}

\begin{lem}[Conditional Convergence] \label{lemma:conditional_markov}
Let $(\filtrationgeneric)_{n \geq 1}$ and $(A_n)_{n\geq1}$ a sequence of $\sigma$-algebras and RV's.
Define conditional convergence
\begin{align*}
&A_n = \opg(1) \iff P(|A_n| > \epsilon | \filtrationgeneric) = \op(1) \quad \forall \epsilon > 0 \\
&A_n = \Opg(1) \iff P(|A_n| > s_n | \filtrationgeneric) = \op(1) \quad \forall s_n \to \infty
\end{align*}
Then the following results hold 
\begin{enumerate}[label={(\roman*)}, itemindent=.5pt, itemsep=.4pt] 
\item $A_n = \op(1) \iff A_n = \opg(1)$ and $A_n = \Op(1) \iff A_n = \Opg(1)$
\item $E[|A_n| | \filtrationgeneric] = \op(1) / \Op(1) \implies A_n = \op(1)/\Op(1)$ 
\item $\var(A_n| \filtrationgeneric) = \op(c_n^2) / \Op(c_n^2) \implies A_n - E[A_n | \filtrationgeneric] = \op(c_n)/\Op(c_n)$. 
\item If $(A_n)_{n\geq1}$ has $A_n \leq \Bar{A} < \infty$ $\filtrationgeneric$-a.s. $\forall n$ and $A_n = \op(1)$ $\implies$ $E[|A_n| | \filtrationgeneric] = \op(1)$ 
\end{enumerate}
\end{lem}
\begin{proof}
(i) Consider that for any $\epsilon > 0$
\begin{align*}
P(|A_n| > \epsilon) = E[\one(|A_n| > \epsilon)] = E[E[\one(|A_n| > \epsilon) | \filtrationgeneric]] = E[P(|A_n| > \epsilon | \filtrationgeneric)]  
\end{align*}
If $A_n = \op(1)$, then $E[P(|A_n| > \epsilon | \filtrationgeneric)] = o(1)$, so $P(|A_n| > \epsilon | \filtrationgeneric) = \op(1)$ by Markov inequality.
Conversely, if $P(|A_n| > \epsilon | \filtrationgeneric) = \op(1)$, then $E[P(|A_n| > \epsilon | \filtrationgeneric)] = o(1)$ since $(P(|A_n| > \epsilon | \filtrationgeneric))_{n \geq 1}$ is uniformly bounded, hence UI.  
Then $P(|A_n| > \epsilon) = o(1)$. 
The second equivalence follows directly from the first.
(ii) follows from (i) and conditional Markov inequality.
(iii) is an application of (ii).
For (iv), note that for any $\epsilon > 0$
\begin{align*}
E[|A_n| | \filtrationgeneric] \leq \epsilon + E[|A_n| \one(|A_n| > \epsilon) | \filtrationgeneric] \leq \epsilon + \Bar{A} P(|A_n| > \epsilon | \filtrationgeneric) = \epsilon + \op(1)
\end{align*}
The equality is by (i) and our assumption. 
Since $\epsilon > 0$ was arbitrary $E[|A_n| | \filtrationgeneric] = \op(1)$.
\end{proof}

\begin{defn}[Conditional Weak Convergence] \label{def:conditional_weak_convergence}
For random variables $A_n, A \in \mr^d$ and $\sigma$-algebras $(\filtration_n)_n$, $\filtration$, let $\mu_n$ and $\mu$ denote the laws of $A_n$ given $\filtration_n$ and of $A$ given $\filtration$, which exist as regular conditional distributions (e.g.\ \cite{klenke2020}, Theorem~8.29).
Let $d_{BL}(\mu, \nu) = \sup_f \bigl| \int f\, d\mu - \int f\, d\nu \bigr|$ denote the bounded Lipschitz metric on $\mathcal{P}(\mr^d)$, where the supremum is over functions $f : \mr^d \to \mr$ with $\|f\|_\infty + |f|_{lip} \leq 1$.
Define 
\begin{align*}
A_n | \filtration_n \convwprocess A | \filtration \iff d_{BL}(\mu_n, \mu) = \op(1).
\end{align*}
\end{defn}

\begin{thm}[Conditional L\'evy Continuity] \label{thm:conditional_levy_continuity}
For $A_n, A \in \mr^d$, convergence holds $A_n | \filtration_n \convwprocess A | \filtration$ if and only if $E[e^{it'A_n} | \filtration_n] = E[e^{it'A} | \filtration] + \op(1)$ for each $t \in \mr^d$.
\end{thm}

\begin{proof}
We show the result for $d=1$. 
The general case is similar.
The laws $\mu_n$ and $\mu$ exist as regular conditional distributions (\cite{klenke2020}, Theorem~8.29).
In particular, $\omega \mapsto \mu_n^\omega(B)$ is $\filtration_n$-measurable for each Borel $B \sub \mr$, and likewise for $\mu^\omega$.
Equivalently, $\omega \mapsto \mu_n^\omega$ and $\omega \mapsto \mu^\omega$ are measurable into $\mathcal{P}(\mr)$ equipped with the $\sigma$-algebra generated by the evaluation maps $\mu \mapsto \mu(B)$, $B \in \mathcal{B}(\mr)$.
This evaluation $\sigma$-algebra coincides with the Borel $\sigma$-algebra of the weak topology on $\mathcal{P}(\mr)$ (\cite{klenke2020}, Exercise~24.1.2), which is metrized by $d_{BL}$.
Since $d_{BL}$ is continuous on $\mathcal{P}(\mr) \times \mathcal{P}(\mr)$, composing the joint map $\omega \mapsto (\mu_n^\omega, \mu^\omega)$ with $d_{BL}$ yields a real-valued random variable $\omega \mapsto d_{BL}(\mu_n^\omega, \mu^\omega)$.

Suppose first that $d_{BL}(\mu_n, \mu) \convp 0$.
Let $\varphi_n(t) = E[e^{itA_n} | \filtration_n]$ and $\varphi(t) = E[e^{itA} | \filtration]$ be the characteristic functions of $\mu_n$ and $\mu$.
Fix $t \in \mr$.
The functions $x \mapsto \cos(tx)$ and $x \mapsto \sin(tx)$ are bounded by $1$ and Lipschitz with constant $|t|$, so each lies in the $d_{BL}$ unit ball after scaling by $(1+|t|)\inv$.
By definition of $d_{BL}$ as a supremum over the bounded-Lipschitz unit ball, $|\varphi_n(t) - \varphi(t)| \leq 2(1 + |t|) \, d_{BL}(\mu_n, \mu) = \op(1)$.
Since $t$ was arbitrary, $\varphi_n(t) - \varphi(t) = \op(1)$ for every $t \in \mr$, which is Definition~\ref{def:conditional_weak_convergence}.

Suppose now that $\varphi_n(t) - \varphi(t) = \op(1)$ for every $t \in \mr$.
By Urysohn's subsequence principle (e.g.\ \cite{klenke2020}, Corollary~6.13), random variables $Y_n \convp Y$ if and only if every subsequence $(Y_{n_k})$ contains a further subsequence $(Y_{n_{k_j}})$ with $Y_{n_{k_j}} \to Y$ almost surely.
Applying this principle to $Y_n = d_{BL}(\mu_n, \mu)$ and $Y = 0$, it suffices to show that any subsequence $(n_k)$ admits a further subsequence $(n_{k_j})$ along which $d_{BL}(\mu_{n_{k_j}}, \mu) \to 0$ almost surely.

For each fixed $t \in \mr$, $|\varphi_n(t) - \varphi(t)| \leq 2$ and $\varphi_n(t) - \varphi(t) = \op(1)$, so $E|\varphi_n(t) - \varphi(t)| = o(1)$ by bounded convergence.
For any $K > 0$, Tonelli and dominated convergence give
\begin{equation*}
E \int_0^1 |\varphi_n(t/K) - \varphi(t/K)|\, dt = \int_0^1 E|\varphi_n(t/K) - \varphi(t/K)|\, dt = o(1).
\end{equation*}
Markov's inequality then implies for each rational $K > 0$,
\begin{equation} \label{equation:cwc-integral-convergence}
\int_0^1 |\varphi_n(t/K) - \varphi(t/K)|\, dt = \op(1).
\end{equation}

Fix the subsequence $(n_k)$.
We have countably many $\op(1)$ statements along this subsequence: $\varphi_n(t) - \varphi(t) \convp 0$ for each $t \in \mathbb{Q}$, and \eqref{equation:cwc-integral-convergence} for each rational $K > 0$.
Applying Urysohn's principle again to convert each in-probability convergence to almost-sure convergence on a sub-subsequence, then using a standard diagonal argument across this countable family, we obtain a further subsequence $(n_{k_j})$ and a probability-one event $\Omega_0$ on which simultaneously $\varphi_{n_{k_j}}(t, \omega) \to \varphi(t, \omega)$ for all $t \in \mathbb{Q}$ and
\begin{align*}
\int_0^1 |\varphi_{n_{k_j}}(t/K, \omega) - \varphi(t/K, \omega)|\, dt &\to 0 \quad \text{for all rational } K > 0.
\end{align*}

Fix $\omega \in \Omega_0$.
The tightness estimate inside the proof of L\'evy's continuity theorem (\cite{klenke2020}, proof of Theorem~15.23) gives, with $\alpha = 1 - \sin(1) > 0$, for each $K > 0$ and $\mathrm{Re}(z)$ denoting the real component of a complex number $z$,
\begin{equation*}
\mu_{n_{k_j}}^\omega\bigl([-K, K]^c\bigr) \leq \alpha\inv \int_0^1 \bigl(1 - \mathrm{Re}\,\varphi_{n_{k_j}}(t/K, \omega)\bigr)\, dt.
\end{equation*}
Using the pointwise bound $1 - \mathrm{Re}\,\varphi_{n_{k_j}}(t/K, \omega) \leq 1 - \mathrm{Re}\,\varphi(t/K, \omega) + |\varphi_{n_{k_j}}(t/K, \omega) - \varphi(t/K, \omega)|$ in the integral, the second integral term vanishes as $j \to \infty$ for each rational $K > 0$ on $\Omega_0$ by the second simultaneous fact above, leaving
\begin{equation*}
\limsup_j \mu_{n_{k_j}}^\omega\bigl([-K, K]^c\bigr) \leq \alpha\inv \int_0^1 \bigl(1 - \mathrm{Re}\,\varphi(t/K, \omega)\bigr)\, dt \equiv T(K, \omega).
\end{equation*}
The function $\varphi(\cdot, \omega)$ is the characteristic function of the probability measure $\mu^\omega$, hence continuous at $0$ with $\varphi(0, \omega) = 1$, so $T(K, \omega) \to 0$ as $K \to \infty$ by dominated convergence.
For any $\epsilon > 0$ we can therefore choose rational $K > 0$ with $T(K, \omega) < \epsilon$, giving $\limsup_j \mu_{n_{k_j}}^\omega([-K, K]^c) \leq \epsilon$.
Hence the sequence $(\mu_{n_{k_j}}^\omega)_j$ is tight for each $\omega \in \Omega_0$.

For each $t \in \mathbb{Q}$, $\int e^{itx}\, d\mu_{n_{k_j}}^\omega(x) = \varphi_{n_{k_j}}(t, \omega) \to \varphi(t, \omega) = \int e^{itx}\, d\mu^\omega(x)$ on $\Omega_0$ by the first simultaneous fact above.
The complex exponentials $\{x \mapsto e^{itx} : t \in \mathbb{Q}\}$ form a separating family for finite measures on $\mr$, since two such measures whose integrals against $e^{itx}$ agree for every $t \in \mathbb{Q}$ have characteristic functions agreeing on $\mathbb{Q}$, hence everywhere on $\mr$ by continuity of characteristic functions, hence are equal by uniqueness.
By tightness of $(\mu_{n_{k_j}}^\omega)_j$ and \cite{klenke2020} Theorem~13.34, $\mu_{n_{k_j}}^\omega \convwprocess \mu^\omega$ as $j \to \infty$ for each $\omega \in \Omega_0$.
Since $d_{BL}$ metrizes weak convergence on $\mathcal{P}(\mr)$, this gives $d_{BL}(\mu_{n_{k_j}}^\omega, \mu^\omega) \to 0$ for each $\omega \in \Omega_0$, i.e.\ almost surely.
This finishes the proof.
\end{proof}

We require a slight modification of the martingale difference CLT in \cite{billingsley1995}, allowing the weak limit to be a mixture of normals.

\begin{prop}[MDS-CLT] \label{prop:mds_clt}
Consider probability spaces $(\Omega_n, \mathcal{A}_n, P_n)$ each equipped with filtration $(\filtration_{k,n})_{k \geq 0}$.
Suppose $(\ykn)_{k=1}^{n}$ is adapted to $(\filtration_{k,n})_{k \geq 0}$ and has $E[\ykn | \filtration_{k-1,n}] = 0$ for all $k \geq 1$ with $n \to \infty$. 
Make the following definitions 
\begin{align*}
S_{k, n} &= \sum_{j=1}^k \ykn \quad \quad \sigkn = E[Y_{k, n}^2 | \filtration_{k-1,n}] \quad \quad \Sigkn = \sum_{j=1}^k \sigkn 
\end{align*}
Denote $\sn \equiv S_{n, n}$ and $\Sign \equiv \Sigma_{n, n}$.
Suppose that $\sigkn \in \filtration_{0,n}$ for all $k, n$ and $\Sign = \Siglimit + \op(1)$ with $\Siglimit \in \filtration_{0,n}$.
Also, suppose for each $\epsilon > 0$
\begin{equation} \label{equation:conditional_lindberg}
L_n^{\epsilon} = \sum_{k=1}^{n} E[\ykn^2 \one(|\ykn| \geq \epsilon) | \filtration_{0, n}] = \op(1)
\end{equation}
Then $E[e^{it \sn} | \filtration_{0,n}] = e^{-\frac{1}{2}t^2 \Siglimit} + \op(1)$.
\end{prop}
\begin{proof}
We modify the argument in Theorem 35.12 of \cite{billingsley1995}.
\begin{align*}
E \left[e^{it \sn} - e^{-\onehalf t^2 \Siglimit} | \filtration_{0, n} \right] &= E[e^{it \sn}(1-e^{\onehalf t^2 \Sign}e^{- \onehalf t^2 \Siglimit})| \filtration_{0, n}] \\
&+ E[e^{-\onehalf t^2 \Siglimit}(e^{\onehalf t^2 \Sign}e^{it \sn} - 1)| \filtration_{0, n}] 
\end{align*}
For the first term, by conditional Jensen inequality 
\begin{align*}
|E[e^{it \sn}(1-e^{\onehalf t^2 \Sign}e^{- \onehalf t^2 \Siglimit})| \filtration_{0, n}]| &\leq E[|(1-e^{\onehalf t^2 \Sign}e^{- \onehalf t^2 \Siglimit})|| \filtration_{0, n}] \\
&= |(1-e^{\onehalf t^2 \Sign}e^{- \onehalf t^2 \Siglimit})| = \op(1) 
\end{align*}
The first equality since $\Sign, \sigma^2 \in \filtration_{0, n}$.
Since $\Sign = \Siglimit + \op(1)$, the second equality follows by continuous mapping. 
The second term has 
\begin{align*}
&|E[e^{-\onehalf t^2 \Siglimit}(e^{\onehalf t^2 \Sign}e^{it \sn} - 1)| \filtration_{0, n}]| = e^{-\onehalf t^2 \Siglimit} |E[(e^{\onehalf t^2 \Sign}e^{it \sn} - 1)| \filtration_{0, n}]| \\
&= e^{-\onehalf t^2 \Siglimit}\left | \sum_{k=1}^{n} E[e^{it S_{k-1, n}} e^{\onehalf t^2 \Sigkn}(e^{it\ykn} - e^{-\onehalf t^2 \sigkn}) | \filtration_{0, n}] \right | \\
\leq & \, e^{-\onehalf t^2 \Siglimit} e^{\onehalf t^2 \Sign} \sum_{k=1}^{n}  E[|e^{it S_{k-1, n}}E[e^{it\ykn} - e^{-\onehalf t^2 \sigkn}|\filtration_{k-1, n}]| | \filtration_{0, n}] \\
= & \, e^{-\onehalf t^2 \Siglimit} e^{\onehalf t^2 \Sign} \sum_{k=1}^{n}  E[|E[e^{it\ykn} - e^{-\onehalf t^2 \sigkn}|\filtration_{k-1, n}]| | \filtration_{0, n}] \equiv \op(1) Z_n  
\end{align*}
The first equality since $\sigma^2 \in \filtration_{0,n}$.
The second equality by telescoping.
The first inequality by triangle inequality and since $\Sigkn \in \filtration_{0, n}$, $\Sigkn \leq \Sign$, and $S_{k-1,n} \in \filtration_{k-1, n}$ for $1 \leq k \leq n$.
The final equality by continuous mapping since $\Sign \convp \sigma^2$. 
We want to show that $Z_n = \Op(1)$.
Fix $\epsilon > 0$ and let $\ikn = \one(|\ykn| > \epsilon)$. 
Note the facts $|e^{ix} - (1 + ix - (1/2)x^2)| \leq (1/6)|x|^3 \wedge |x|^2$ and $|e^z -  (1+z)| \leq |z|^2 e^{|z|}$ for real $x$, complex $z$.
By the MDS property and $E[\ykn^2 | \filtration_{k-1,n}] = \sigkn$, combined with these facts 
\begin{align*}
&|E[e^{it\ykn} - e^{-\onehalf t^2 \sigkn}|\filtration_{k-1, n}]| \leq E[|t \ykn|^3 \wedge |t \ykn|^2 + (1/4)t^4 \sigma_{k,n}^4 e^{\onehalf t^2 \sigkn} |\filtration_{k-1, n}] \\
&\leq (t^2 + t^4 + |t|^3 + e^{\onehalf t^2 \Sign})E[(\epsilon |\ykn|^2 + |\ykn|^2 \ikn + \sigma_{k,n}^4) |\filtration_{k-1, n}]\\
&\equiv A_{n,t} (\epsilon \sigkn + E[|\ykn|^2\ikn |\filtration_{k-1, n}] + E[\sigma_{k,n}^4 |\filtration_{k-1, n}])
\end{align*}
Then we have 
\begin{align*}
Z_n &\leq A_{n,t} \sum_{k=1}^n E[\epsilon \sigkn + E[|\ykn|^2\ikn |\filtration_{k-1, n}] + E[\sigma_{k,n}^4 |\filtration_{k-1, n}] | \filtration_{0,n}] \\
& = A_{n,t}(\epsilon \Sign + L_n^{\epsilon}) + A_{n,t} \sum_{k=1}^n E[\sigma_{k,n}^4 |\filtration_{0, n}] \leq A_{n,t}(\epsilon \Sign + L_n^{\epsilon} + \Sign (\epsilon^2 + L_n^{\epsilon}))  
\end{align*}

To see the final inequality, note that $\sigma_{k,n}^4 \leq \sigma_{k,n}^2 \max_{k=1}^n \sigkn$ and 
We have $\sigkn = E[\ykn^2 | \filtration_{k-1,n}] \leq \epsilon^2 + E[\ykn^2 \ikn | \filtration_{k-1,n}] \leq \epsilon^2 + \sum_{j=1}^n E[Y_{j,n}^2 I_{j,n} | \filtration_{j-1,n}]$.
Taking $\max_{k=1}^n$ on both sides gives $\max_{k=1}^n \sigkn \leq \epsilon^2 + \sum_{j=1}^n E[Y_{j,n}^2 I_{j,n} | \filtration_{j-1,n}]$.
Then $\sum_{k=1}^n E[\sigma_{k,n}^4 |\filtration_{0, n}] \leq \sum_{k=1}^n E[\sigma_{k,n}^2 (\epsilon^2 + \sum_{j=1}^n E[Y_{j,n}^2 I_{j,n} | \filtration_{j-1,n}]) |\filtration_{0, n}] = \Sign (\epsilon^2 + L_n^{\epsilon})$.
Note that since $\Sign \convp \sigma^2$, we have $A_{n,t}, \Sign = \Op(1)$ and $L_n^{\epsilon} = \op(1)$ by assumption.
Since $\epsilon$ was arbitrary, this shows $Z_n = \op(1)$.
\end{proof}

The next proposition specializes Proposition~\ref{prop:mds_clt} to triangular arrays that are jointly independent conditional on a single $\sigma$-algebra, which may itself vary with $n$.
This is the form we use in Section~\ref{proofs:asymptotics}.

\begin{prop}[Conditional CLT] \label{prop:conditional-clt}
For each $n$, let $\filtration_{0,n}$ be a $\sigma$-algebra and $(\ykn)_{k=1}^n$ random variables that are jointly independent conditional on $\filtration_{0,n}$, with $E[\ykn | \filtration_{0,n}] = 0$ and $E[\ykn^2 | \filtration_{0,n}] < \infty$.
Let $\sn = \sum_{k=1}^n \ykn$ and $\Sign = \sum_{k=1}^n E[\ykn^2 | \filtration_{0,n}]$, and suppose $\Sign = \Siglimit + \op(1)$ for some $\Siglimit \in \filtration_{0,n}$.
Suppose the conditional Lindeberg condition holds: for each $\epsilon > 0$, $\sum_{k=1}^n E[\ykn^2 \one(|\ykn| \geq \epsilon) | \filtration_{0,n}] = \op(1)$.
Then $E[e^{it \sn} | \filtration_{0,n}] = e^{-\frac{1}{2} t^2 \Siglimit} + \op(1)$ for each $t \in \mr$.
\end{prop}
\begin{proof}
Define the increasing filtration $\filtration_{k,n} = \sigma(\filtration_{0,n}, Y_{1,n}, \dots, \ykn)$ for $k \geq 1$, to which $(\ykn)_{k=1}^n$ is adapted.
By joint conditional independence, $\ykn \indep (Y_{1,n}, \dots, Y_{k-1,n}) | \filtration_{0,n}$.
Hence $E[\ykn | \filtration_{k-1,n}] = E[\ykn | \filtration_{0,n}] = 0$ and $\sigkn = E[\ykn^2 | \filtration_{k-1,n}] = E[\ykn^2 | \filtration_{0,n}] \in \filtration_{0,n}$.
So $(\ykn, \filtration_{k,n})$ is a martingale difference sequence with $\filtration_{0,n}$-measurable conditional variances, and $\Sign = \sum_{k=1}^n \sigkn$.
The variance and Lindeberg hypotheses of Proposition~\ref{prop:mds_clt} hold by assumption.
The proposition then gives $E[e^{it \sn} | \filtration_{0,n}] = e^{-\frac{1}{2} t^2 \Siglimit} + \op(1)$.
\end{proof}

\begin{lem}[Lipschitz Approximation] \label{lemma:lipschitz-approximation}
Let $Z \in \mr^d$ be a random variable.
Define $\mathcal L = \{g(Z) \in L_2(Z): |g|_{lip} \vee |g|_{\infty} < \infty\}$.
\begin{enumerate}[label={(\roman*)}, ref={\ref*{lemma:lipschitz-approximation}.(\roman*)}, itemindent=.5pt, itemsep=.4pt]
\item $\mathcal L$ is dense in $L_2(Z)$. \label{lemma:lipschitz-approximation:density}
\item For any $f \in L_2(Z)$ and any positive sequence $\lambda_n \to \infty$, there exists a sequence $(z_n)_{n \geq 1}$ in $\mathcal L$ with $|z_n|_{lip} \leq \lambda_n$ and $|z_n - f|_{2, Z} \to 0$. \label{lemma:lipschitz-approximation:rate}
\end{enumerate}
\end{lem}

\begin{proof}
For (i), let $\mu$ denote the law of $Z$ and fix a Borel set $A \subseteq \mr^d$.
The claim is immediate if $\mu(A)=0$, so suppose $\mu(A)>0$ and fix $0<\epsilon<\mu(A)$.
Since every finite Borel measure on $\mr^d$ is inner regular \cite[Theorem~7.8]{folland}, there is a compact set $K \subseteq A$ such that $\mu(A\setminus K)<\epsilon$.
Hence $|\one(\cdot\in A)-\one(\cdot\in K)|_{2,Z}^2=\mu(A\setminus K)<\epsilon$.
Define $f_m(z)=(1+m d(z,K))\inv$ for $m\geq1$, where $d(z,K)=\inf_{y\in K}|z-y|_2$.
The function $z\mapsto d(z,K)$ is $1$-Lipschitz, so $|f_m(z)-f_m(y)|\leq m|z-y|_2$ and $|f_m|_\infty\leq1$, and hence $f_m(Z)\in\mathcal L$.
If $z\in K$, then $d(z,K)=0$ and $f_m(z)=1$ for every $m$.
If $z\notin K$, the openness of $K^c$ gives $r>0$ such that $B(z,r)\subseteq K^c$, so $d(z,K)\geq r$ and $f_m(z)\leq(1+mr)\inv\to0$.
Thus $f_m(z)\to\one(z\in K)$ pointwise, and dominated convergence gives $|f_m-\one(\cdot\in K)|_{2,Z}\to0$.
The triangle inequality and arbitrariness of $\epsilon$ show that bounded Lipschitz functions are dense in the indicator functions of Borel sets.
Finite linear combinations of bounded Lipschitz functions are bounded Lipschitz, so they are dense in the bounded simple functions.
The bounded simple functions are dense in $L_2(Z)$ (e.g.\ \cite{folland}), proving (i).

For (ii), by (i) there exists a sequence $(\tilde z_k)_{k \geq 1}$ in $\mathcal L$ with $M_k := |\tilde z_k|_{lip} < \infty$ and $|\tilde z_k - f|_{2, Z} \to 0$.
Set $\tilde z_0 \equiv 0$ and define $k(n) = \max\{1 \leq k \leq n : M_k \leq \lambda_n\}$, taking $k(n) = 0$ if no such $k$ exists.
For any fixed $K \geq 1$, eventually $n \geq K$ and $\lambda_n \geq M_K$, hence $k(n) \geq K$ eventually and so $k(n) \to \infty$.
Setting $z_n := \tilde z_{k(n)}$ gives $|z_n|_{lip} \leq \lambda_n$ by construction and $|z_n - f|_{2, Z} \to 0$ since $k(n) \to \infty$.
\end{proof}

\begin{lem}[Asymptotic Independence] \label{lemma:charfn}
Consider a probability space $(\Omega, \filtrationalt, P)$ with $\sigma$-algebras $\filtration_{n,k} \sub \filtration_{n,k+1} \sub \filtrationalt$ for all $n \geq 1$ and $0 \leq k \leq m-1$.
Let $X_{n,k}$ be $\filtration_{n,k}$-measurable random variables for all $n$ and $1 \leq k \leq m$. 
Suppose that $X_{n,k} | \filtration_{n,k-1} \convwprocess L_k | \filtration_{n,0}$, as in Definition \ref{def:conditional_weak_convergence}.
Then $(X_{n,1}, \dots, X_{n,m}) | \filtration_{n,0} \convwprocess (L_1, \dots, L_m) | \filtration_{n,0}$, with jointly independent limit.
\end{lem}
\begin{proof}
By Levy continuity, it suffices to show $E[e^{it'(X_{1,n}, \dots X_{m,n})} | \filtration_{n,0}] \to \prod_{k=1}^{m} E[e^{it_k L_k} | \filtration_{n,0}]$ for all  $t \in \mr^m$.
We work by induction on $k$.
By assumption, $E[e^{itX_{n,1}}|\filtration_{n,0}] \convp E[e^{itL_1} | \filtration_{n,0}]$ for all $t \in \mr$.
Assume by induction that the conclusion holds for $1 \leq k \leq k' \leq m$.
Then
\begin{align*}
&E[e^{i \sum_{k=1}^{k'+1} t_k X_{n,k}} | \filtration_{n,0}] = E[e^{i \sum_{k=1}^{k'} t_k X_{n,k}} E[e^{it_{k'+1} X_{n, k'+1}} | \filtration_{n,k'}]| \filtration_{n,0}] \\
&= E[e^{i \sum_{k=1}^{k'} t_k X_{n,k}} (E[e^{it_{k'+1} X_{n, k'+1}} | \filtration_{n,k'}] - E[e^{i t_{k'+1} L_{k'+1}} | \filtration_{n,0}]) | \filtration_{n,0}] \\
&+ E[e^{i t_{k'+1} L_{k'+1}}| \filtration_{n,0}] E[e^{i \sum_{k=1}^{k'} t_k X_{n,k}}| \filtration_{n,0}] = \prod_{k=1}^{k' + 1} E[e^{i t_k L_k}| \filtration_{n,0}] + \op(1)
\end{align*}
The first equality is by tower law and our measurability and increasing $\sigma$-algebra assumption.
For the final equality, note 
$|e^{i \sum_{k=1}^{k'} t_k X_{n,k}} (E[e^{it_{k'+1} X_{n, k'+1}} | \filtration_{n,k'}] - E[e^{i t_{k'+1} L_{k'+1}} | \filtration_{n,0}])| \leq |E[e^{it_{k'+1} X_{n, k'+1}} | \filtration_{n,k'}] - E[e^{i t_{k'+1} L_{k'+1}} | \filtration_{n,0}]| \convp 0$.
Then the first term in the sum above is $\op(1)$ since the integrand converges in probability and is bounded, hence UI.
The final equality also uses our inductive hypothesis.
This finishes the proof.
\end{proof}

\begin{lem}[LLN] \label{lemma:lln}
Consider $A_n = n \inv \sum_{\group \in \groupset_n} \ugroup$, with $\groupset_n$ a collection of disjoint subsets of $[n]$. 
Let $(\filtrationgeneric)_{n \geq 1}$ be $\sigma$-algebras such that $\groupset_n$ is $\filtrationgeneric$-measurable, $E[u_{\group} | \filtrationgeneric] = 0$, and for all $\group \not = \group' \in \groupset_n$ $\ugroup \indep u_{\group'} | \filtrationgeneric$. 
If $n \inv \sum_{\group \in \groupset_n} E[|u_{\group}| \one(|\ugroup| > c_n) | \filtrationgeneric] \convp 0$ for $c_n = \omega(1)$, $c_n = o(n \half)$, then $A_n \convp 0$.
\end{lem}
\begin{proof}
By disjointness $|\groupset_n| \leq n$. 
Fix an indexing $\groupset_n = \{\group_s: 1 \leq s \leq n\}$, possibly with $\group_s = \emptyset$ for some $s$.
Define $\bar u_{sn} = u_s \one(|u_s| \leq c_n)$ and $\bar \mu_{sn} = E[u_s \one(|u_s| \leq c_n) | \filtrationgeneric]$. 
Expand
\begin{align*}
\frac{1}{n} \sum_{\group \in \groupset_n} \ugroup = \frac{1}{n} \sum_{s=1}^n u_s = \frac{1}{n} \sum_{s=1}^n [(u_s - \bar u_{sn}) + (\bar u_{sn} - \bar \mu_{sn}) + \bar \mu_{sn}] = T_{n1} + T_{n2} + T_{n3}
\end{align*}
Observe that $E[|T_{n1}| | \filtrationgeneric] \leq (1/n) \sum_{s=1}^n E[|u_s| \one(|u_s| > c_n) | \filtrationgeneric] = \op(1)$ by assumption.
Then $T_{n1} = \op(1)$ by conditional Markov (Lemma \ref{lemma:conditional_markov}). 
Next consider $T_{n2}$.
Note that by definition $E[\bar u_{sn} - \bar \mu_{sn} | \filtrationgeneric] = 0$ for each $1 \leq s \leq n$. 
Note that for $s \not = s'$ $ \cov(\bar u_{sn}, \bar u_{s'n} | \filtrationgeneric) = 0$ by the conditional independence assumption. 
Then $\var(T_{n2} | \filtrationgeneric) = n^{-2} \sum_{s=1}^n \var(\bar u_{sn} | \filtrationgeneric) \leq n^{-2} \sum_{s=1}^n E[\bar u_{sn}^2 | \filtrationgeneric] \leq n \inv c_n^2 = o(1)$, so that $T_{n2} = \op(1)$ by conditional Chebyshev. 
Finally, since $E[u_s | \filtrationgeneric] = 0$, we have $\bar \mu_{sn} = -E[u_s \one(|u_s| > c_n) | \filtrationgeneric]$.
Then $E[|T_{n3}| | \filtrationgeneric] \leq (1/n) \sum_{s=1}^n E[|u_s| \one(|u_s| > c_n) | \filtrationgeneric] = \op(1)$, so that $T_{n3} = \op(1)$ by conditional Markov as before.   
This finishes the proof.
\end{proof}

\begin{lem} \label{lemma:maximal-inequality}
For each $n$, let $(X_{i,n})_{i=1}^n$ be iid within its row. If $\{|X_{1,n}|^\alpha:n\geq1\}$ is uniformly integrable for some $\alpha>0$, then $\max_{i=1}^n |X_{i,n}| = \op(n^{1/\alpha})$.
In particular, if $(X_i)_{i=1}^n$ are iid and $E[|X|^\alpha]<\infty$, then $\max_{i=1}^n|X_i|=\op(n^{1/\alpha})$.
\end{lem}
\begin{proof}
For $\epsilon>0$, the union bound and identical distribution within each row give $P(\max_{i=1}^n|X_{i,n}|>\epsilon n^{1/\alpha})\leq\sum_{i=1}^nP(|X_{i,n}|>\epsilon n^{1/\alpha})=nP(|X_{1,n}|^\alpha>\epsilon^\alpha n)$.
By Markov's inequality, the last expression is at most $\epsilon^{-\alpha}E[|X_{1,n}|^\alpha\one(|X_{1,n}|^\alpha>\epsilon^\alpha n)]$, which converges to zero by uniform integrability.
For a fixed distribution, integrability of $|X|^\alpha$ implies uniform integrability of the constant family $\{|X|^\alpha:n\geq1\}$, proving the final claim.
\end{proof}

\begin{lem} \label{lemma:op-rate}
If $X_n \convp 0$, there is a deterministic sequence $\lambda_n \to \infty$ with $\lambda_n X_n \convp 0$.
\end{lem}
\begin{proof}
For each $k \geq 1$, choose $N_k$ strictly increasing such that $P(|X_n| > 2^{-k}) < 2^{-k}$ for all $n \geq N_k$.
Set $\lambda_n = 2^{k/2}$ for $N_k \leq n < N_{k+1}$, and $\lambda_n = 1$ for $n < N_1$.
Then $\lambda_n \to \infty$ since $N_k \to \infty$.
For $\epsilon > 0$ and $n \in [N_k, N_{k+1})$ with $k \geq 2 \log_2(1/\epsilon)$, $P(|\lambda_n X_n| > \epsilon) = P(|X_n| > \epsilon \cdot 2^{-k/2}) \leq P(|X_n| > 2^{-k}) < 2^{-k} \to 0$.
\end{proof}

In the following lemmas, $\filtrationgeneric$ denotes a generic $\sigma$-algebra with $\filtrationcandpsi \sub \filtrationgeneric$ and $\filtrationgeneric \indep \eta$.
In our applications, we take $\filtrationgeneric$ to be either $\filtrationcandpsi$ or $\filtrationhd$ depending on the context.

\begin{lem}[Group Aggregate Independence] \label{lemma:group_aggregate_independence}
Let $\Dn \sim \localdesigncond(\psi, \propfn(\psi))$ with associated partition $\groupset_n$ and randomness $\eta$ from Definition \ref{defn:design-construction}.
Let $\filtrationgeneric$ be a $\sigma$-algebra such that $\filtrationcandpsi \sub \filtrationgeneric$ and $\filtrationgeneric \indep \eta$.
Suppose $u_\group = \phi((\Di)_{i \in \group}, X, \group)$ for some $\filtrationgeneric$-measurable variable $X$ and deterministic measurable function $\phi(\cdot)$.
Then $(u_\group)_{\group \in \groupset_n}$ are jointly independent conditional on $\filtrationgeneric$.
\end{lem}
\begin{proof}
By Equation \eqref{equation:design-construction}, $(D_i)_{i \in \group} = f(\group, \eta_{s(\group)})$ for the truncation function $f(\cdot)$ and the $\filtrationcandpsi$-measurable injection $s: \groupset_n \to [L] \times [n+1]$ of Definition~\ref{defn:design-construction}.
Substituting, $u_\group = \phi(f(\group, \eta_{s(\group)}), X, \group) \equiv h(\group, X, \eta_{s(\group)})$ for the composition $h(g, X, \eta) = \phi(f(g, \eta), X, g)$.
By construction, $\groupsetn$ and $I = \{i: \Ti=1\}$ are $\filtrationcandpsi$ measurable, thus $\filtrationgeneric$ measurable.
The claim now follows from Lemma~\ref{lemma:cross_group_independence}, applied with finite index set $J = [L] \times [n+1]$, family $(\eta_{l,j})_{(l, j) \in J}$, injection $s(\cdot)$, and fixed function $h(\cdot)$ as above.
\end{proof}

\begin{lem}[Independence] \label{lemma:cross_group_independence}
Let $\groupsetn$ be a partition of $I \sub [n]$ with $\groupsetn, I$ both $\filtrationgeneric$-measurable.
Let $(\eta_t)_{t \in J}$ be jointly independent random variables with $J$ a finite set and $(\eta_t)_{t \in J} \indep \filtrationgeneric$, and let $s: \groupsetn \to J$ an $\filtrationgeneric$-measurable injection assigning $\group \in \groupsetn$ to $\eta_{s(\group)}$.
For each $\group \in \groupset_n$, let $Z_\group = h(\group, X, \eta_{s(\group)})$ for an $\filtrationgeneric$-measurable random vector $X$ and deterministic $h(\cdot)$.
Then $(Z_\group)_{\group \in \groupset_n}$ are jointly independent conditional on $\filtrationgeneric$.
\end{lem}
\begin{proof}
Let $A$ be an $\filtrationgeneric$-measurable event.
For each fixed partition $\mc{G}$ of $I$ and each injection $r : \mc{G} \to J$, define $A(r, \mc{G}) \equiv A \cap \{\groupset_n = \mc{G}\} \cap \{s(\cdot) = r(\cdot)\}$ and observe that $A(r, \mc G) \in \filtrationgeneric$.
Observe that the finite family $\{A(r, \mc{G})\}_{(r, \mc{G})}$ partitions $A$, that is $A = \cup_{(r, \mc{G})} A(r, \mc{G})$.
Then for any collection of bounded test functions $(\phi_\group)_{\group \sub I}$,
\begin{align*}
&E\bigl[\one_A \Pi_{\group \in \groupset_n} \phi_\group(Z_\group)\bigr] = E\bigl[\one_A \Pi_{\group \in \groupset_n} \phi_\group(h(\group, X, \eta_{s(\group)}))\bigr] \\
& = \sum_{(r, \mc{G})} E\bigl[\one_{A(r, \mc{G})} \Pi_{\group \in \mc{G}} \phi_\group(h(\group, X, \eta_{r(\group)}))\bigr] = \sum_{(r, \mc{G})} E\bigl[\one_{A(r, \mc{G})} E[\Pi_{\group \in \mc{G}} \phi_\group(h(\group, X, \eta_{r(\group)})) | \filtrationgeneric]\bigr] 
\end{align*}
The first equality is by definition of $Z_\group$.
The second since $A = \cup_{(r, \mc{G})} A(r, \mc{G})$ and $\groupset_n = \mc{G}$, $s = r$ on $A(r, \mc{G})$.
The third is by tower law and $\filtrationgeneric$-measurability of $A(r, \mc{G})$.
Continuing, 
\begin{align*}
&= \sum_{(r, \mc{G})} E\bigl[\one_{A(r, \mc{G})} \Pi_{\group \in \mc{G}} E[\phi_\group(h(\group, X, \eta_{r(\group)})) | \filtrationgeneric]\bigr] \\
&= \sum_{(r, \mc{G})} E\bigl[\one_{A(r, \mc{G})} \Pi_{\group \in \groupset_n} E[\phi_\group(Z_\group) | \filtrationgeneric]\bigr]
= E\bigl[\one_A \Pi_{\group \in \groupset_n} E[\phi_\group(Z_\group) | \filtrationgeneric]\bigr].
\end{align*}

The first equality is by conditional independence of $(\eta_{r(\group)})_{\group \in \mc{G}}$ given $\filtrationgeneric$ and since $X$ is $\filtrationgeneric$-measurable and $h(\cdot)$ is fixed. 
Conditional independence holds since indices $\{r(\group)\}_{\group \in \mc{G}}$ are distinct by injectivity of $r(\cdot)$ and $(\eta_t)_{t \in J} \indep \filtrationgeneric$ and are jointly independent.
The second equality uses $Z_\group = h(\group, X, \eta_{r(\group)})$ and $\groupset_n = \mc{G}$ on $A(r, \mc{G})$.
The final equality recombines the partition $A = \cup_{(r, \mc{G})} A(r, \mc{G})$.
Then we have shown $E\bigl[\one_A \Pi_{\group \in \groupset_n} \phi_\group(Z_\group)\bigr] = E\bigl[\one_A \Pi_{\group \in \groupset_n} E[\phi_\group(Z_\group) | \filtrationgeneric]\bigr]$ for any $\filtrationgeneric$-measurable event $A$.
By definition of conditional expectation, this implies $E\bigl[\Pi_{\group \in \groupset_n} \phi_\group(Z_\group) \bigm| \filtrationgeneric\bigr] = \Pi_{\group \in \groupset_n} E[\phi_\group(Z_\group) | \filtrationgeneric]$.
Since the test function collection $(\phi_\group)_{\group \sub I}$ was arbitrary, this shows the claim.
\end{proof}

\begin{lem}[Design Properties] \label{lemma:design_properties}
Let $\Dn \sim \localdesigncond(\psi, \propfn(\psi))$ with partition $\groupset_n$, group propensity $\propfng$ for $\group \in \groupset_n$, and randomness $\eta$ from Definition~\ref{defn:design-construction}.
Let $\filtrationgeneric$ be a $\sigma$-algebra such that $\filtrationcandpsi \sub \filtrationgeneric$ and $\filtrationgeneric \indep \eta$.
For $i \in \group$, $E[\Di | \filtrationgeneric] = \propfng$ and $\var(\Di | \filtrationgeneric) = \propfng(1-\propfng)$.
For $i \ne j$,
\[
\cov(\Di, \Dj | \filtrationgeneric) = \begin{cases} -\frac{\propfng(1-\propfng)}{|\group|-1} & i, j \in \group \text{ interior} \\ 0 & \text{otherwise.} \end{cases}
\]
\end{lem}
\begin{proof}
By Definition~\ref{defn:design-construction} applied to $\Dn$, conditional on $\filtrationgeneric$ the within-group vector $(\Di)_{i \in \group}$ is distributed as $\crdist(\propfng)$ for interior groups $\group$ and as iid Bernoulli$(\propfng)$ of length $|\group|$ for the remainder.
For $i \in \group$, $\Di$ is a single coordinate of $(\Di)_{i \in \group}$, marginally Bernoulli$(\propfng)$ in both cases.
Hence $E[\Di | \filtrationgeneric] = \propfng$ and $\var(\Di | \filtrationgeneric) = \propfng(1-\propfng)$.

For distinct $i, j$ with associated groups $\group(i) \neq \group(j)$, Lemma~\ref{lemma:group_aggregate_independence} applied with $u_\group = (\Di)_{i \in \group}$ implies the within-group vectors are jointly independent conditional on $\filtrationgeneric$, so $\cov(\Di, \Dj | \filtrationgeneric) = 0$.
For $i \neq j$ in the same remainder group, the iid Bernoulli structure of $(\Di)_{i \in \group}$ conditional on $\filtrationgeneric$ gives $\cov(\Di, \Dj | \filtrationgeneric) = 0$.
For $i \neq j$ in the same interior group $\group$ with $|\group| = k$ and $\propfng = a/k$, $(\Di, \Dj)$ has the joint marginal of $\crdist(a/k)$ on two distinct coordinates conditional on $\filtrationgeneric$:
\begin{align*}
\cov(\Di, \Dj | \filtrationgeneric) &= P(\Di = \Dj = 1 | \filtrationgeneric) - \propfng^2 = \binom{k}{a}\inv \binom{k-2}{a-2} - \propfng^2 \\
&= \frac{a(a-1)}{k(k-1)} - \frac{a^2}{k^2} = -\frac{(a/k)(1 - a/k)}{k-1} = -\frac{\propfng(1-\propfng)}{k-1}.
\end{align*}
\end{proof}

\begin{lem}[Stratified WLLN] \label{lemma:stratified_wlln}
Let $\Dn \sim \localdesigncond(\psi, \propfn(\psi))$.
Let $\filtrationgeneric$ be a $\sigma$-algebra such that $\filtrationcandpsi \sub \filtrationgeneric$ and $\filtrationgeneric \indep \eta$ in Definition \ref{defn:design-construction}.
Let $(h_n(W_i))_{i=1}^n$ be $\filtrationgeneric$-measurable.
\begin{enumerate}[noitemsep,topsep=0pt]
\item If $\sup_{n \geq 1} E[h_n(W)^2] < \infty$, then $\en[\Ti (\Di - \propfn(\psii)) h_n(W_i)] = \Op(\negrootn)$.
\item If $(h_n(W))_{n \geq 1}$ is uniformly integrable, then $\en[\Ti (\Di - \propfn(\psii)) h_n(W_i)] = \op(1)$.
\item Let $h$ be a fixed measurable function with $E[|h(W)|] < \infty$. Then $\en[\Ti h(W_i)] = E[\propselect(\psii) h(W_i)] + \op(1)$ and $\en[\Ti \Di h(W_i)] = E[\propselect(\psii) \propfn(\psii) h(W_i)] + \op(1)$.
\end{enumerate}
\end{lem}
\begin{proof}
Denote $a_i = \Ti (\Di - \propfn(\psii)) h_n(W_i)$.
Note $\Ti, \propfn(\psii), h_n(W_i) \in \filtrationgeneric$ by $\filtrationcandpsi \sub \filtrationgeneric$ and the hypothesis on $h_n$.
We first show $E[\en[a_i] | \filtrationgeneric] = 0$.
For $i$ with $\Ti = 1$, by the partition property $i$ lies in some group $\group(i) \in \groupset_n$ with $\propfng = \propfn(\psii)$, and Lemma~\ref{lemma:design_properties} gives $E[\Di | \filtrationgeneric] = \propfn(\psii)$.
For $i$ with $\Ti = 0$ we have $a_i = 0$.
In either case, we have $E[a_i | \filtrationgeneric] = \Ti h_n(W_i) (E[\Di | \filtrationgeneric] - \propfn(\psii)) = 0$.
and summing gives $E[\en[a_i] | \filtrationgeneric] = 0$.

Next, we bound $\var(\en[a_i] | \filtrationgeneric) = n^{-2} \sum_{i,j} \cov(a_i, a_j | \filtrationgeneric)$.
By $\filtrationgeneric$-measurability of the prefactors, $\cov(a_i, a_j | \filtrationgeneric) = \Ti \Tj h_n(W_i) h_n(W_j) \cov(\Di, \Dj | \filtrationgeneric)$.
Lemma~\ref{lemma:design_properties} gives $\var(\Di | \filtrationgeneric) = \propfng(1-\propfng) \leq 1/4$ for $i \in \group$, $|\cov(\Di, \Dj | \filtrationgeneric)| = \propfng(1-\propfng)/(|\group|-1) \leq |\group|\inv$ for distinct $i, j$ in an interior $\group$. 
To see the inequality, note $\propfng(1-\propfng) \leq 1/4 \leq (k-1)/k$ for $k = |\group| \geq 2$.
Also $\cov(\Di, \Dj | \filtrationgeneric) = 0$ for all other pairs $i, j$ by Lemma \ref{lemma:design_properties}.  
Then
\begin{align*}
\var(\en[a_i] | \filtrationgeneric) &\leq n^{-2} \sum_i \Ti h_n(W_i)^2 + n^{-2} \sum_{\group \text{ interior}} |\group|\inv \sum_{i \neq j \in \group} \Ti \Tj |h_n(W_i)| |h_n(W_j)| \\
&\leq n\inv \en[\Ti h_n(W_i)^2] + n^{-2} \sum_{\group \text{ interior}} |\group|\inv \biggl(\sum_{i \in \group} |h_n(W_i)|\biggr)^2 \\
&\leq n\inv \en[\Ti h_n(W_i)^2] + n^{-2} \sum_{\group \text{ interior}} \sum_{i \in \group} h_n(W_i)^2 \leq 2 n\inv \en[\Ti h_n(W_i)^2].
\end{align*}
The first inequality uses $\var(\Di | \filtrationgeneric) \leq 1$ and $|\cov(\Di, \Dj | \filtrationgeneric)| \leq |\group|\inv$, together with the triangle inequality.
The second uses $0 \leq \Ti \leq 1$ to drop the indicators and extends $\sum_{i \neq j \in \group}$ to $\sum_{i, j \in \group}$.
The third uses Cauchy-Schwarz: $(\sum_{i \in \group} |h_n(W_i)|)^2 \leq |\group| \sum_{i \in \group} h_n(W_i)^2$.

Suppose $\sup_{n \geq 1} E[h_n(W)^2] < \infty$.
By the law of total variance and work above we have $E[(\en[a_i])^2] = E[\var(\en[a_i] | \filtrationgeneric)] \leq 2 n\inv E[h_n(W)^2] \leq 2 n\inv \sup_n E[h_n(W)^2]$.  
The final quantity is $O(n\inv)$, so Markov inequality gives $\en[a_i] = \Op(\negrootn)$.

Suppose instead $(h_n(W))_{n \geq 1}$ is UI.
Let $\hinbar = h_n(W_i) \one(|h_n(W_i)| \leq n^{1/4})$.
Decompose $\en[a_i] = A_n + B_n$ where $A_n = \en[\Ti (\Di - \propfn(\psii))(h_n(W_i) - \hinbar)]$ and $B_n = \en[\Ti (\Di - \propfn(\psii)) \hinbar]$.
For $A_n$, we have the bound $|A_n| \leq \en[|h_n(W_i) - \hinbar|] = \en[|h_n(W_i)| \one(|h_n(W_i)| > n^{1/4})]$ using $|\Ti(\Di - \propfn(\psii))| \leq 1$.
By uniform integrability, $E[|h_n(W)| \one(|h_n(W)| > n^{1/4})] \to 0$, so $E[|A_n|] \to 0$ and $A_n = \op(1)$ by Markov.
For $B_n$, we have $E[B_n | \filtrationgeneric] = 0$ by the work above and the variance bound applied with $\hinbar$ in place of $h_n(W_i)$ gives $\var(B_n | \filtrationgeneric) \leq 2 n\inv \en[\hinbar^2] \leq 2 n^{-1/2} = o(1)$ since $|\hinbar| \leq n^{1/4}$.
Then $B_n = \op(1)$ by Chebyshev (Lemma~\ref{lemma:conditional_markov}).
Then $\en[a_i] = \op(1)$.

For (3), decompose $\en[\Ti h(W_i)] = \en[(\Ti - \propselect(\psii)) h(W_i)] + \en[\propselect(\psii) h(W_i)]$.
The first term is $\op(1)$ by (2) applied to the sampling design, via the substitution $\Dn \to \Tn$ and $\Tn \to 1$. 
Note the constant sequence $h_n = h$ is uniformly integrable since $E[|h(W)|] < \infty$.
The second term equals $E[\propselect(\psii) h(W_i)] + \op(1)$ by vanilla WLLN for iid data, using $|\propselect h| \leq |h|$ integrable.
This gives the first identity.
For the second identity, decompose $\en[\Ti \Di h(W_i)] = \en[\Ti (\Di - \propfn(\psii)) h(W_i)] + \en[\Ti \propfn(\psii) h(W_i)]$.
The first term is $\op(1)$ by (2) applied to the assignment design with constant sequence $h_n = h$ uniformly integrable since $E[|h(W)|] < \infty$.
The second term equals $E[\propselect(\psii) \propfn(\psii) h(W_i)] + \op(1)$ by the first identity applied with $\propfn h$ in place of $h$, integrable since $|\propfn h| \leq |h|$.
\end{proof}

\begin{lem}[Variance Identity] \label{lemma:variance_identity}
Let $\Dn \sim \localdesigncond(\psi, \propfn(\psi))$ with partition $\groupsetn$.
Let $\filtrationgeneric$ be a $\sigma$-algebra such that $\filtrationcandpsi \sub \filtrationgeneric$ and $\filtrationgeneric \indep \eta$ in Definition \ref{defn:design-construction}, and let $(h_n(W_i))_{i=1}^n$ be $\filtrationgeneric$-measurable.
Then $\var\bigl(\rootn \en[\Ti (\Di - \propfn(\psii)) h_n(W_i)] \bigm| \filtrationgeneric\bigr)$ is bounded above by
\begin{equation} \label{equation:variance-identity-bound}
n\inv \sum_{\group \in \groupset_n} |\group|\inv \sum_{i,j \in \group} (h_n(W_i) - h_n(W_j))^2 + n\inv \kboundn \nlevels \cdot \max_{i=1}^n h_n(W_i)^2.
\end{equation}
\end{lem}
\begin{proof}
Define $a_i = \Ti (\Di - \propfn(\psii)) h_n(W_i)$.
As in the proof of Lemma~\ref{lemma:stratified_wlln}, we have $\cov(a_i, a_j | \filtrationgeneric) = \Ti \Tj h_n(W_i) h_n(W_j) \cov(\Di, \Dj | \filtrationgeneric)$.
By Lemma~\ref{lemma:design_properties}, the covariance $\cov(\Di, \Dj | \filtrationgeneric) = 0$ unless $i, j$ lie in the same group, so
\begin{align*}
\var(\rootn \en[a_i] | \filtrationgeneric) &= n\inv \sum_{i, j} \Ti \Tj h_n(W_i) h_n(W_j) \cov(\Di, \Dj | \filtrationgeneric) \\
&= n\inv \sum_{\group \in \groupset_n} \sum_{i,j \in \group} h_n(W_i) h_n(W_j) \cov(\Di, \Dj | \filtrationgeneric).
\end{align*}
The second equality uses $\Ti = \Tj = 1$ for all $i, j \in \group$ since $\groupsetn$ partitions $\{i: \Ti=1\}$.
For interior $\group$ with $|\group| = k$ and $\propfng = a/k$, Lemma~\ref{lemma:design_properties} gives $\cov(\Di, \Di | \filtrationgeneric) = a(k-a)/k^2$ and $\cov(\Di, \Dj | \filtrationgeneric) = -a(k-a)/[k^2(k-1)]$ for $i \neq j$.
Then
\begin{align*}
\sum_{i,j \in \group} h_n(W_i) h_n(W_j) \cov(\Di, \Dj | \filtrationgeneric) &= \frac{a(k-a)}{k^2} \sum_{i \in \group} h_n(W_i)^2 - \frac{a(k-a)}{k^2(k-1)} \sum_{i \neq j \in \group} h_n(W_i) h_n(W_j) \\
&= \frac{a(k-a)}{k^2(k-1)} \biggl[(k-1) \sum_i h_n(W_i)^2 - \sum_{i \neq j} h_n(W_i) h_n(W_j)\biggr].
\end{align*}
Using $\sum_{i \neq j \in \group} h_n(W_i) h_n(W_j) = (\sum_i h_n(W_i))^2 - \sum_i h_n(W_i)^2$, the bracket equals $k \sum_i h_n(W_i)^2 - (\sum_i h_n(W_i))^2$, which equals $\tfrac{1}{2} \sum_{i,j \in \group}(h_n(W_i) - h_n(W_j))^2$ by the identity $\sum_{i,j}(h_n(W_i) - h_n(W_j))^2 = 2 [|\group| \sum h_n(W_i)^2 - (\sum h_n(W_i))^2]$.
Hence
\begin{align*}
\sum_{i,j \in \group} h_n(W_i) h_n(W_j) \cov(\Di, \Dj | \filtrationgeneric) &= \frac{a(k-a)}{2 k^2(k-1)} \sum_{i,j \in \group}(h_n(W_i) - h_n(W_j))^2 \\
&\leq \frac{1}{|\group|} \sum_{i,j \in \group}(h_n(W_i) - h_n(W_j))^2.
\end{align*}
Note we used $a(k-a)/[2 k^2(k-1)] \leq 1/k$ via $a(k-a) \leq k^2/4$ and $k(k-1) \geq k^2/2$ for $k \geq 2$.
For remainder groups $\group$, $(\Di)_{i \in \group}$ has independent components conditional on $\filtrationgeneric$, so $\cov(\Di, \Dj | \filtrationgeneric) = 0$ for $i \neq j$ and $\var(\Di | \filtrationgeneric) = \propfng(1 - \propfng) \leq 1$.
Hence $\sum_{i,j \in \group} h_n(W_i) h_n(W_j) \cov(\Di, \Dj | \filtrationgeneric) = \sum_{i \in \group} h_n(W_i)^2 \var(\Di | \filtrationgeneric) \leq \sum_{i \in \group} h_n(W_i)^2$.
Aggregating over groups, $\var(\rootn \en[a_i] | \filtrationgeneric)$ is bounded above by
\begin{align*}
n\inv \sum_{\group \text{ interior}} |\group|\inv \sum_{i,j \in \group}(h_n(W_i) - h_n(W_j))^2 + n\inv \sum_{\group \text{ remainder}} \sum_{i \in \group} h_n(W_i)^2.
\end{align*}
The second term is bounded by $\sum_{\group \text{ remainder}} \sum_{i \in \group} h_n(W_i)^2 \leq \max_{i=1}^n h_n(W_i)^2 \cdot \sum_l |\group_{l, \mathrm{rem}}| \leq \kboundn \nlevels \cdot \max_{i=1}^n h_n(W_i)^2$, since there are at most $\nlevels$ remainder groups each of size at most $\kboundn$.
Extending the interior sum to include remainder groups gives the claimed bound.
\end{proof}

\begin{lem}[Random Partitions] \label{lemma:partitions}
Let $\Wn$ be iid random variables, $h_i = h(W_i)$ for a measurable function $h$, and $\kappa$ a random element with $\kappa \indep \Wn | h_{1:n}$.
Define $\filtrationgeneric = \sigma(h_{1:n}, \kappa)$, and let $\groupset_n$ be an $\filtrationgeneric$-measurable partition of a subset of $[n]$.
For each $\group \in \groupset_n$, let $Z_\group = F(\group, (W_i)_{i \in \group})$ for a deterministic measurable function $F$.
Then $(Z_\group)_{\group \in \groupset_n}$ are jointly conditionally independent given $\filtrationgeneric$.
\end{lem}
\begin{proof}
Since $\kappa \indep \Wn | h_{1:n}$, we have $\Wn | h_{1:n}, \kappa \eqdist \Wn | h_{1:n}$.
Moreover, by iid sampling of $\Wn$ and $h_i = h(W_i)$, this factorizes as $\Wn | h_{1:n} \eqdist \prod_{i=1}^n (W_i | h_i)$.
Hence $(W_i)_{i=1}^n$ are jointly conditionally independent given $\filtrationgeneric$, with conditional marginal $W_i | h_i$.

Let $A \in \filtrationgeneric$ and $(\phi_\group)_{\group \sub [n]}$ a finite collection of bounded measurable test functions.
For each partition $\groupset$ of a subset of $[n]$, define $A(\groupset) \equiv A \cap \{\groupset_n = \groupset\} \in \filtrationgeneric$ by assumption, so the finite family $\{A(\groupset)\}_\groupset$ partitions $A = \cup_\groupset A(\groupset)$.
Then
\begin{align*}
E\bigl[\one_A \Pi_\group \phi_\group(Z_\group)\bigr] &= \sum_\groupset E\bigl[\one_{A(\groupset)} \Pi_{\group \in \groupset} \phi_\group(F(\group, (W_i)_{i \in \group}))\bigr] \\
&= \sum_\groupset E\bigl[\one_{A(\groupset)} E[\Pi_{\group \in \groupset} \phi_\group(F(\group, (W_i)_{i \in \group})) | \filtrationgeneric]\bigr] \\
&= \sum_\groupset E\bigl[\one_{A(\groupset)} \Pi_{\group \in \groupset} E[\phi_\group(F(\group, (W_i)_{i \in \group})) | \filtrationgeneric]\bigr] \\
&= \sum_\groupset E\bigl[\one_{A(\groupset)} \Pi_{\group \in \groupset_n} E[\phi_\group(Z_\group) | \filtrationgeneric]\bigr] = E\bigl[\one_A \Pi_\group E[\phi_\group(Z_\group) | \filtrationgeneric]\bigr].
\end{align*}
The first equality decomposes $A$ along $\groupset_n = \groupset$ and uses $Z_\group = F(\group, (W_i)_{i \in \group})$ with $\group$ ranging in $\groupset$ on $A(\groupset)$.
The second equality is by tower law and $\filtrationgeneric$-measurability of $A(\groupset)$.
The key third equality uses the joint conditional independence of $(W_i)_{i=1}^n$ given $\filtrationgeneric$ established above and the fact that $\groupset$ is a disjoint partition.
The fourth equality uses $\groupset_n = \groupset$ and $Z_\group = F(\group, (W_i)_{i \in \group})$ on $A(\groupset)$.
The fifth re-aggregates the partition $A = \cup_\groupset A(\groupset)$.
Since $A \in \filtrationgeneric$ and $(\phi_\group)$ were arbitrary, this gives $E[\Pi_\group \phi_\group(Z_\group) | \filtrationgeneric] = \Pi_\group E[\phi_\group(Z_\group) | \filtrationgeneric]$, hence joint conditional independence.
\end{proof}

\begin{lem}[Propensity Convergence] \label{lemma:propensity-convergence}
Suppose $|\hkest_d - \hk_d|_{2, \psi}^2 = \op(1)$ for $d = 0, 1$, with $\inf_\psi \hk_d(\psi) > 0$ and $\propfn(\psi) \in (\propbound, 1-\propbound)$, and discretization $\en[(\propselectestn(\psii) - \propselectest(\psii))^2] = \op(1)$.
Then $\en[(\propselectestn(\psii) - \propselectopt(\psii))^2] = \op(1)$.
\end{lem}

\begin{proof}
By the triangle inequality, $\en[(\propselectestn - \propselectopt)^2] \leq 2\en[(\propselectestn - \propselectest)^2] + 2\en[(\propselectest - \propselectopt)^2]$, with the first term $\op(1)$ by hypothesis.
It remains to bound $\en[(\propselectest - \propselectopt)^2]$.
Recall that 
\[
\propselectest(\psi) = \budget \cdot \frac{\sdavgest(\psi)\, \cost(\psi)^{-1/2}}{\en[\sdavgest(\psii)\, \cost(\psii)^{1/2}]}, \qquad \propselectopt(\psi) = \budget \cdot \frac{\sdavg(\psi)\, \cost(\psi)^{-1/2}}{E[\sdavg(\psi)\, \cost(\psi)^{1/2}]}.
\]
Our approach is to first establish $|\sdavgest - \sdavg|_{2,\psi}^2 = \op(1)$, then use it to control both the pointwise numerator difference $\sdavgest - \sdavg$ and the gap between the sample and population denominator normalizations.

Write $\hkavg(\psi) = \hk_1(\psi)/\propfn(\psi) + \hk_0(\psi)/(1-\propfn(\psi))$ and $\sdavg = \sqrt{\hkavg}$, with $\hkavgest, \sdavgest$ defined analogously from $\hkest_d$.
By hypothesis $\hk_d(\psi) \geq c_l > 0$ pointwise and $\propfn(\psi), 1-\propfn(\psi) \leq 1-\propbound$, so $\hkavg(\psi) \geq c_l/(1-\propbound) =: \tilde c > 0$, hence $\sdavgest + \sdavg \geq \sdavg \geq \sqrt{\tilde c}$.
Using the identity $(\sdavgest - \sdavg)^2 = (\hkavgest - \hkavg)^2 / (\sdavgest + \sdavg)^2$ and the pointwise bound $(\hkavgest - \hkavg)^2 \leq 2\propbound^{-2}[(\hkest_1 - \hk_1)^2 + (\hkest_0 - \hk_0)^2]$ (using $\propfn(\psi), 1-\propfn(\psi) \geq \propbound$),
\begin{equation} \label{equation:sdavg-rate}
|\sdavgest - \sdavg|_{2,\psi}^2 \lesssim \sum_{d=0,1} |\hkest_d - \hk_d|_{2,\psi}^2 = \op(1).
\end{equation}
Conditional on $\proprand$, $(\psii)_{i=1}^n$ are iid and $\sdavgest$ is fixed, so $E[\en[(\sdavgest(\psii) - \sdavg(\psii))^2] | \proprand] = |\sdavgest - \sdavg|_{2,\psi}^2 = \op(1)$ by \eqref{equation:sdavg-rate}.
By conditional Markov (Lemma~\ref{lemma:conditional_markov}), $\en[(\sdavgest(\psii) - \sdavg(\psii))^2] = \op(1)$.
Also $E[\hkavg(\psi)] \leq \propbound\inv E[\hk_1 + \hk_0] \leq 2\propbound\inv E[Y(1)^2 + Y(0)^2] < \infty$ by Young's inequality, $\propfn(\psi), 1-\propfn(\psi) \geq \propbound$, and Assumption~\ref{assumption:pilot-clt}(ii).

Define denominators $\mu = E[\sdavg(\psi) \cost(\psi)\half]$ and $\mu_n = \en[\sdavgest(\psii) \cost(\psii)\half]$.
Cauchy-Schwarz gives $\mu \leq E[\sdavg^2]\half E[\cost]\half < \infty$, and pointwise $\sdavg \cost\half \geq \tilde c\half C_l\half$ gives $\mu \geq \tilde c\half C_l\half > 0$.
Conditional on $\proprand$ the units $(\psii)_{i=1}^n$ are iid and $\sdavgest$ is a fixed function, so conditional Chebyshev gives $\en[\sdavgest \cost\half] = E[\sdavgest \cost\half |\proprand] + \Op(\negrootn)$.
By Cauchy-Schwarz, $(E[\sdavgest \cost\half |\proprand] - E[\sdavg \cost\half])^2 \leq C_u \cdot E_\psi[(\sdavgest - \sdavg)^2 |\proprand] = \op(1)$ using \eqref{equation:sdavg-rate}.
Combining and applying Young's inequality,
\begin{equation} \label{equation:mu-rate}
(\mu_n - \mu)^2 = \Op(n\inv) + \op(1) = \op(1).
\end{equation}
Then $\mu_n \geq \mu/2$ on an event of probability tending to one.
Finally, decompose
\[
\propselectest(\psi) - \propselectopt(\psi) = \frac{\budget \, \cost(\psi)^{-1/2}}{\mu_n \mu}\bigl[\mu \,(\sdavgest(\psi) - \sdavg(\psi)) - \sdavg(\psi)\, (\mu_n - \mu)\bigr].
\]
On the event $\{\mu_n \geq \mu/2\}$, using $\cost\inv \leq C_l\inv$ and $\mu \geq \tilde c\half C_l\half$,
\begin{align*}
\en[(\propselectest - \propselectopt)^2] &\lesssim \en[(\sdavgest(\psii) - \sdavg(\psii))^2] + (\mu_n - \mu)^2 \cdot \en[\sdavg(\psii)^2].
\end{align*}
The first term is $\op(1)$ by the work above. The second is $\op(1) \cdot \Op(1) = \op(1)$ using \eqref{equation:mu-rate} and $\en[\sdavg^2] \convp E[\sdavg^2] < \infty$ by the iid weak law.
\end{proof}

\medskip

We turn to consistency of the $\sate$ variance estimator $\varest_{\sate}$ of \eqref{equation:variance-estimator-varying}.

\begin{lem}[Conservative $\sate$ Inference] \label{lemma:sate-inference-consistency}
In the setting of Lemma~\ref{lemma:new-inference-consistency-conditioning},
\begin{equation} \label{equation:sate-varest-limit}
\varest_{\sate} \convp V_{\sate} + E[\hkte(\psi)].
\end{equation}
\end{lem}

\begin{proof}[Proof of Lemma \ref{lemma:sate-inference-consistency}]
Let $\filtrationhd = \sigma(\Wn, \Tn, \permn)$ as in the proof of Lemma~\ref{lemma:new-inference-consistency-conditioning}.
The within-stratum bracket in \eqref{equation:variance-estimator-varying} is $(k/\propselectl)\wh P_l^2 = (k/(\propselectl |\groupset_{nl}|)) \sum_{u \in \mathcal{G}_{nl}^{\nu}} (\estgone{u} - \estgtwo{u})^2$.
By sampling-subordinate assignment every union $u \in \groupsetnu$ lies in a single propensity stratum with $\propselect_u = \propselectl$, and $|\groupset_{nl}| = \nl \propselectl / k$, so $k / (\propselectl |\groupset_{nl}|) = k^2 / (\nl \propselectl^2)$.
Summing the within-stratum bracket over $l$ gives the aggregated form $\varest_{\sate} = \wh Q$ with
\[
\wh Q = \frac{1}{n} \sum_{u \in \groupsetnu} \frac{k^2}{\propselect_u^2} \bigl(\estgone{u} - \estgtwo{u}\bigr)^2.
\]
Expand the square.
Each $\group \in \groupset_n$ appears exactly once across the unions with $\propselect_u = \propselect_\group$, so the group-square terms collect into $\wh v_1$ of \eqref{equation:v1-v2}.
Comparing the resulting cross-product coefficient $2 k^2 / \propselect_u^2$ with that of $\wh v_2$ in \eqref{equation:v1-v2} gives $2 k^2 / \propselect_u^2 - 2 k (k - \propselect_u) / \propselect_u^2 = 2 k / \propselect_u$, hence
\[
\wh Q = \wh v_1 - \wh v_2 - T_n, \qquad T_n = \frac{2k}{n} \sum_{u \in \groupsetnu} \frac{1}{\propselect_u} \estgone{u} \estgtwo{u}.
\]
By \eqref{equation:v1-v2-est2-limit} and $\est \convp \ate$, we have $\wh v_1 - \wh v_2 \convp \varlocal + \ate^2$.

Next consider $T_n$.
Each summand of $T_n$ is a union-level cross product with bounded $\filtrationhd$-measurable weight $2 k / \propselect_u \leq 2 k / \propbound$, so the argument of Lemma~\ref{lemma:coupling-v1-v2} gives $T_n - E[T_n | \filtrationhd] = \op(1)$.
By Lemma~\ref{lemma:group_aggregate_independence}, $\estgone{u} \indep \estgtwo{u} | \filtrationhd$, and $E[\estg | \filtrationhd] = \theta_\group$, so $E[T_n | \filtrationhd] = \frac{2k}{n} \sum_u \propselect_u\inv \theta_{\group_1(u)} \theta_{\group_2(u)}$.
Define $A_i = \tau_i / \propselect(\psii)\half$, which is $\sigma(W_i)$-measurable with $E[A_i^2] \leq E[\tau_i^2] / \propbound < \infty$.
By sampling-subordinate assignment $\propselect_u = \propselect(\psii) = \propselect(\psij)$ for $i, j \in u$, so $A_i A_j = \tau_i \tau_j / \propselect_u$, and $\theta_\group = k\inv \sum_{i \in \group} \tau_i$ gives $\propselect_u\inv \theta_{\group_1(u)} \theta_{\group_2(u)} = k^{-2} \sum_{i \in \group_1(u), j \in \group_2(u)} A_i A_j$.
Using $u = \group_1(u) \cup \group_2(u)$ and that each $\group \in \groupset_n$ appears once across the unions,
\[
E[T_n | \filtrationhd] = \frac{2}{nk} \sum_{u \in \groupsetnu} \sum_{\substack{i \in \group_1(u) \\ j \in \group_2(u)}} A_i A_j = \frac{1}{nk} \biggl( \sum_{u \in \groupsetnu} \sum_{i \neq j \in u} A_i A_j - \sum_{\group \in \groupset_n} \sum_{i \neq j \in \group} A_i A_j \biggr).
\]
By the argument in the proof of Lemma~\ref{lemma:matched-union-tight-matching}, tight matching of $\groupset_n$ and the centroid-matching condition imposed in Section~\ref{section:inference} imply $F(\groupsetnu)=\op(1)$.
Thus Lemma~\ref{lemma:bilinear-form} with $B_i = A_i$ at $K = 2k$ on $\groupsetnu$ and at $K = k$ on $\groupset_n$ gives the two sums the limits $(2k - 1) E[\catefn(\psi)^2]$ and $(k - 1) E[\catefn(\psi)^2]$, using $E[A_i | \psii] = \catefn(\psii) / \propselect(\psii)\half$ and hence $E[\propselect(\psii) E[A_i | \psii]^2] = E[\catefn(\psii)^2]$.
Subtracting and dividing by $k$ gives $E[T_n | \filtrationhd] \convp E[\catefn(\psi)^2]$, hence $T_n \convp E[\catefn(\psi)^2]$.

Combining, $\wh Q \convp \varlocal + \ate^2 - E[\catefn(\psi)^2] = \varlocal - \var(\catefn(\psi))$, using $E[\catefn(\psi)^2] - \ate^2 = \var(\catefn(\psi))$.
By the variance formula \eqref{equation:asymptotic-theorem}, $\varlocal - \var(\catefn(\psi)) = E\bigl[\propselect(\psi)\inv \bigl(\hk_1(\psi) / \propfn + \hk_0(\psi) / (1 - \propfn)\bigr)\bigr]$.
It remains to identify this limit.
With $\psisamp = \psiassign = \psi$ and $\propfn$ constant, $\var(\te | \psi) = \hkte(\psi)$ and the identity \eqref{equation:ylevel-variance-identity} give the pointwise equality $\frac{1 - \propselect(\psi)}{\propselect(\psi)} \hkte(\psi) + \frac{\var(\ylevel | \psi)}{\propselect(\psi) \propfn (1 - \propfn)} = \propselect(\psi)\inv \bigl(\hk_1(\psi) / \propfn + \hk_0(\psi) / (1 - \propfn)\bigr) - \hkte(\psi)$.
Substituting into the variance of Corollary~\ref{cor:clt-sate-extended} gives $V_{\sate} = E[\propselect(\psi)\inv (\hk_1(\psi) / \propfn + \hk_0(\psi) / (1 - \propfn))] - E[\hkte(\psi)]$.
Hence $\varest_{\sate} \convp V_{\sate} + E[\hkte(\psi)]$.
\end{proof}

The within-group estimator \eqref{equation:P2-within} admits an exact, design-based characterization of its conditional bias. We state it for $\propselect = 1$ without stratified sampling, where $\filtrationhd = \sigma(\Wn, \permn)$ carries the data and the matching randomness, $E[\est | \filtrationhd] = \sate$, and only the assignment $\Dn$ is random.

\begin{lem}[Within-group Variance Estimator] \label{lemma:within-group-unbiased}
Let $\propselect = 1$, $\propfn = a/k$, and $\min(a, k-a) \geq 2$. 
Then for $\bar S^2_\tau = \frac{1}{|\groupset_n|} \sum_{\group \in \groupset_n} \frac{1}{k-1}\sum_{i \in \group}(\te_i - \theta_\group)^2$, the within-group estimator has 
\[
E[\varest_{\sate} | \filtrationhd] = n \var(\est | \filtrationhd) + \bar S^2_\tau.
\]
The bias $\bar S^2_\tau$ obeys the bound $0 \leq \bar S^2_\tau \leq \frac{k}{k-1}\, S^2_\tau \le 2 S^2_{\tau}$ for $S^2_\tau = \frac{1}{n-1}\sum_{i=1}^n (\te_i - \sate)^2$.
\end{lem}
\begin{proof}
Write $\hti = (\Di - \propfn)/(\propfn(1-\propfn))$ and $\yleveli = (1-\propfn)Y_i(1) + \propfn Y_i(0)$, so that $\hti Y_i = \te_i + \hti \yleveli$ and, since $\propselect = 1$, $\estg = k\inv \sum_{i \in \group}\hti Y_i = \theta_\group + k\inv\sum_{i \in \group}\hti\yleveli$ with $\est = |\groupset_n|\inv\sum_\group\estg$.

\emph{Unbiasedness of $\wh P^2_N$.} For an interior group $\group$, treated units reveal $\Di Y_i = \Di Y_i(1)$, so $(a-1)\, s^2_{1,\group} = \frac{a-1}{a}\sum_{i \in \group}\Di Y_i(1)^2 - \frac{1}{a}\sum_{i \neq j \in \group}\Di\Dj\, Y_i(1)Y_j(1)$. By Lemma~\ref{lemma:design_properties} with $\filtrationgeneric = \filtrationhd$, $E[\Di | \filtrationhd] = a/k$ and $E[\Di\Dj | \filtrationhd] = a(a-1)/(k(k-1))$ for $i \neq j \in \group$, and substituting $\sum_{i \neq j}Y_i(1)Y_j(1) = k^2\mu_\group(1)^2 - \sum_i Y_i(1)^2$ with $\mu_\group(d) = k\inv\sum_{i \in \group}Y_i(d)$,
\[
E[s^2_{1,\group} | \filtrationhd] = \frac{1}{k-1}\sum_{i \in \group}\bigl(Y_i(1) - \mu_\group(1)\bigr)^2 =: S^2_{1,\group},
\]
Likewise, $E[s^2_{0,\group} | \filtrationhd] = S^2_{0,\group}$. Hence $E[\varest_{\sate} | \filtrationhd] = k\, E[\wh P^2_N | \filtrationhd] = k\bar P^2_N$ for $\bar P^2_N = |\groupset_n|\inv\sum_\group(S^2_{1,\group}/a + S^2_{0,\group}/(k-a))$.

\emph{Design variance.} Since $E[\hti | \filtrationhd] = 0$, $E[\est | \filtrationhd] = |\groupset_n|\inv\sum_\group\theta_\group = \sate$. The groups are conditionally independent given $\filtrationhd$, and Lemma~\ref{lemma:design_properties} gives $\var(\hti | \filtrationhd) = [\propfn(1-\propfn)]\inv$ and $\cov(\hti, H_j | \filtrationhd) = -[(k-1)\propfn(1-\propfn)]\inv$ for $i \neq j \in \group$, so for variance $S^2_{\ylevel,\group} = \frac{1}{k-1}\sum_{i \in \group}\bigl(\yleveli - \tfrac1k\textstyle\sum_{j \in \group}\ylevelj\bigr)^2$, $\var(\estg | \filtrationhd) = k^{-2}\sum_{i,j \in \group}\yleveli\ylevelj\cov(\hti, H_j | \filtrationhd) = \frac{S^2_{\ylevel,\group}}{k\,\propfn(1-\propfn)}$.
The within-group identity $\yleveli - k\inv\sum_{j}\ylevelj = (1-\propfn)(Y_i(1) - \mu_\group(1)) + \propfn(Y_i(0) - \mu_\group(0))$, combined with $S^2_{\tau,\group} = S^2_{1,\group} + S^2_{0,\group} - 2S^2_{10,\group}$ for within-group covariance $S^2_{10,\group}$, expands to
\[
\frac{S^2_{\ylevel,\group}}{\propfn(1-\propfn)} = \frac{S^2_{1,\group}}{\propfn} + \frac{S^2_{0,\group}}{1-\propfn} - S^2_{\tau,\group} = k\bigl(\frac{S^2_{1,\group}}{a} + \frac{S^2_{0,\group}}{k-a}\bigr) - S^2_{\tau,\group},
\]
where $S^2_{\tau,\group} = (k-1)\inv\sum_{i \in \group}(\te_i - \theta_\group)^2$. Using $\var(\est | \filtrationhd) = |\groupset_n|^{-2}\sum_\group\var(\estg | \filtrationhd)$ and $n = k|\groupset_n|$,
\[
n\var(\est | \filtrationhd) = \frac{1}{|\groupset_n|}\sum_\group\frac{S^2_{\ylevel,\group}}{\propfn(1-\propfn)} = k\bar P^2_N - \bar S^2_\tau.
\]
Combined with the previous display, $E[\varest_{\sate} | \filtrationhd] = k\bar P^2_N = n\var(\est | \filtrationhd) + \bar S^2_\tau$.

\emph{Bound.} The finite-population decomposition $\sum_{i=1}^n(\te_i - \sate)^2 = \sum_\group\sum_{i \in \group}(\te_i - \theta_\group)^2 + k\sum_\group(\theta_\group - \sate)^2$ and nonnegativity of the second term give $(k-1)|\groupset_n|\,\bar S^2_\tau \leq (n-1)S^2_\tau$, hence $\bar S^2_\tau \leq \frac{n-1}{(k-1)|\groupset_n|}S^2_\tau < \frac{k}{k-1}S^2_\tau \leq 2 S^2_\tau$ for $k \geq 2$.
\end{proof}

\begin{lem}[Matching for Unions] \label{lemma:matched-union-tight-matching}
Let $\groupset_n$ satisfy $n\inv \sum_{\group \in \groupset_n}\sum_{i \in \group}|\psii - \bar\psi_\group|_2^2 = \op(1)$ and $E[|\psi|_2^\alpha] < \infty$ for some $\alpha > \dimpsi + 1$.
Suppose $\groupsetnu$ matches group centroids $(\bar \psi_\group)_{\group \in \groupsetn}$ into pairs using either the algorithm of Section~\ref{section:algorithms} or \cite{derigs1988}.
Then
\[
n\inv \sum_{u \in \groupsetnu}\sum_{i \in u}|\psii - \bar\psi_u|_2^2 = \op(1).
\]
\end{lem}

\begin{proof}[Proof of Lemma \ref{lemma:matched-union-tight-matching}]
In what follows, let $k = \max_{g \in \groupsetn} |g|$. 
Note we have the identity $\sum_{i, j \in \group}|\psii - \psij|_2^2 = 2|\group|\sum_{i \in \group}|\psii - \bar\psi_\group|_2^2$.
Each $u \in \groupsetnu$ has the form $\group \cup \groupmatching(\group)$ for $\group \in \groupset_n$. 
Indexing by $\group$ double-counts each union, so  
\[
n\inv \sum_{u \in \groupsetnu}\sum_{i \in u}|\psii - \bar\psi_u|_2^2 \le \frac{1}{2n}\sum_{u \in \groupsetnu}\sum_{i, j \in u}|\psii - \psij|_2^2 = \frac{1}{4n}\sum_{\group \in \groupset_n}\sum_{i, j \in \group \cup \groupmatching(\group)}|\psii - \psij|_2^2.
\]
Splitting the inner sum into within-half and cross-half pairs and using the bijection $\group \mapsto \groupmatching(\group)$ to combine the two within-half contributions,
\[
\frac{1}{n}\sum_{u \in \groupsetnu}\sum_{i, j \in u}|\psii - \psij|_2^2 = \frac{1}{n}\sum_{\group \in \groupset_n}\sum_{i, j \in \group}|\psii - \psij|_2^2 + \frac{1}{n}\sum_{\group \in \groupset_n}\sum_{\substack{i \in \group \\ j \in \groupmatching(\group)}}|\psii - \psij|_2^2.
\]
The first term is $\op(1)$ by the tight matching assumption for $\groupset_n$.
For the cross term, by Jensen's inequality $|\psii - \psij|_2^2 \leq 3|\psii - \bar\psi_\group|_2^2 + 3|\bar\psi_\group - \bar\psi_{\groupmatching(\group)}|_2^2 + 3|\bar\psi_{\groupmatching(\group)} - \psij|_2^2$. 
Then 
\begin{align*}
\frac{1}{n}\sum_{\group \in \groupset_n}\sum_{\substack{i \in \group \\ j \in \groupmatching(\group)}}|\psii - \psij|_2^2 \leq{}& \frac{3k}{n}\sum_{\group \in \groupset_n}\sum_{i \in \group}|\psii - \bar\psi_\group|_2^2 + \frac{3k^2}{n}\sum_{\group \in \groupset_n}|\bar\psi_\group - \bar\psi_{\groupmatching(\group)}|_2^2 \\
&+ \frac{3k}{n}\sum_{\group \in \groupset_n}\sum_{j \in \groupmatching(\group)}|\psij - \bar\psi_{\groupmatching(\group)}|_2^2.
\end{align*}
The first and third terms are $\op(1)$ by the tight matching assumption for $\groupset_n$.
It remains to handle the second term.
Define the squared maximal coordinate range of the centroids by $R_n^2 \equiv \max_{1 \leq j \leq \dimpsi}(\max_{\group \in \groupset_n}\bar\psi_\group^{(j)}-\min_{\group \in \groupset_n}\bar\psi_\group^{(j)})^2$.
Because the groups have fixed size $k$, $|\groupset_n| \asymp n$, so Theorem \ref{thm:algorithm-guarantee} applied to the centroids gives $n\inv \sum_{\group \in \groupset_n}|\bar\psi_\group - \bar\psi_{\groupmatching(\group)}|_2^2 \leq R_n^2 \cdot O(n^{-2/(\dimpsi + 1)})$.
The same bound holds when $\groupmatching$ is obtained using \cite{derigs1988}, since optimal pairs minimize the left-hand side.
Moreover, by Jensen's inequality and Lemma \ref{lemma:maximal-inequality}, $R_n^2 \leq 4\max_{\group \in \groupset_n}|\bar\psi_\group|_2^2 \leq 4\max_{i \in [n]}|\psii|_2^2 = \op(n^{2/\alpha})$.
Hence the second term is $\op(n^{2/\alpha - 2/(\dimpsi + 1)})=\op(1)$ since $\alpha > \dimpsi+1$.
\end{proof}

\begin{lem}[Coupling for Inference] \label{lemma:coupling-v1-v2}
Under the conditions of Lemma \ref{lemma:new-inference-consistency-conditioning}, we have $\wh v_1 - E[\wh v_1 | \filtrationhd] = \op(1)$ and $\wh v_2 - E[\wh v_2 | \filtrationhd] = \op(1)$ for $\filtrationhd = \sigma(\Wn, \Tn, \permn)$.
\end{lem}

\begin{proof}[Proof of Lemma \ref{lemma:coupling-v1-v2}]
First consider $\wh v_1$. 
Set $u_\group = (k^2/\propselect_\group^2)(\estg^2 - E[\estg^2 | \filtrationhd])$, so that $\wh v_1 - E[\wh v_1 | \filtrationhd] = n\inv \sum_{\group \in \groupset_n} u_\group$.
We verify the conditions of Lemma \ref{lemma:lln} with the filtration $\filtrationhd$.
The partition $\groupset_n$ is $\filtrationhd$-measurable, and $E[u_\group | \filtrationhd] = 0$ by construction.
As noted in the proof of Lemma~\ref{lemma:new-inference-consistency-conditioning}, $\filtrationcandpsi \sub \filtrationhd$ and $\filtrationhd \indep \eta$ for the assignment randomness $\eta$ of Definition~\ref{defn:design-construction}.
Since $\estg$ is a function of $(\Di)_{i \in \group}$ and the $\filtrationhd$-measurable potential outcomes $(Y_i(0), Y_i(1))_{i \in \group}$, the term $u_\group = (k^2/\propselect_\group^2)(\estg^2 - E[\estg^2 | \filtrationhd])$ has the form $\phi((\Di)_{i \in \group}, X, \group)$ for a deterministic function $\phi$ and the $\filtrationhd$-measurable variable $X = (W_{1:n}, (E[\wh\theta_{\group'}^2 | \filtrationhd])_{\group' \in \groupset_n})$.
By Lemma~\ref{lemma:group_aggregate_independence}, the $(u_\group)_{\group \in \groupset_n}$ are jointly independent conditional on $\filtrationhd$; in particular $u_\group \indep u_{\group'} | \filtrationhd$ for $\group \neq \group'$.

For $i \in \group$, $\Ti = 1$ and $Y_i = \Di Y_i(1) + (1 - \Di) Y_i(0)$, so $|\estg| = |a^{-1}\sum_{i \in \group}\Di Y_i(1) - (k-a)^{-1}\sum_{i \in \group}(1-\Di) Y_i(0)| \leq \max_{i \in \group}|Y_i(1)| + \max_{i \in \group}|Y_i(0)|$.
Squaring and using $(p + q)^2 \leq 2(p^2 + q^2)$ together with $\max_{i \in \group} Y_i(d)^2 \leq \sum_{i \in \group} Y_i(d)^2$ gives $\estg^2 \leq 2\sum_{i \in \group}(Y_i(1)^2 + Y_i(0)^2)$.
The right side is $\filtrationhd$-measurable, so the bound is preserved under conditional expectation: $E[\estg^2 | \filtrationhd] \leq 2\sum_{i \in \group}(Y_i(1)^2 + Y_i(0)^2)$.
Hence $|u_\group| \leq \tilde u_\group := (4k^2/\propselect_\group^2)\sum_{i \in \group}(Y_i(1)^2 + Y_i(0)^2)$, which is $\filtrationhd$-measurable, and since $\one(|u_\group| > c_n) \leq \one(\tilde u_\group > c_n)$, we have $E[|u_\group|\one(|u_\group| > c_n) | \filtrationhd] \leq \tilde u_\group \one(\tilde u_\group > c_n)$.

Note that for positive constants $(a_l)_{l=1}^m$, $\sum_l a_l \one(\sum_l a_l > c) \leq m \sum_l a_l \one(a_l > c/m)$.
Viewing $\tilde u_\group$ as a sum of $2k$ terms indexed by $(i, d) \in \group \times \{0, 1\}$, each equal to $4k^2 Y_i(d)^2/\propselect_\group^2$, and applying the inequality with $m = 2k$, we get 
\[
\tilde u_\group \one(\tilde u_\group > c_n) \leq (8k^3/\propselect_\group^2)\sum_{i \in \group, d}Y_i(d)^2 \one \big (Y_i(d)^2/\propselect_\group^2 > c_n/(8k^3) \big).
\]
Since $\groupset_n$ partitions $\{i : \Ti = 1\}$ and $\propselect_\group = \propselect(\psii)$ for $i \in \group$, setting $Z_i(d) = Y_i(d)^2/\propselect(\psii)^2$,
\begin{equation} \label{equation:tilde-u-sum-bound}
\frac{1}{n}\sum_{\group \in \groupset_n}\tilde u_\group \one(\tilde u_\group > c_n) \leq \frac{8 k^3}{n}\sum_{i = 1}^n \Ti \sum_{d \in \{0, 1\}} Z_i(d) \one\bigl(Z_i(d) > c_n/(8k^3)\bigr).
\end{equation}
By assumption $\propselect(\psii) \geq \propbound > 0$, so that $E[Z_i(d)] \leq E[Y_i(d)^2]/\propbound^2 < \infty$. 
Then by dominated convergence $E[Z_i(d)\one(Z_i(d) > c_n/(8k^3))] \to 0$ as $c_n \to \infty$.
Using $\Ti \leq 1$ and iid sampling of $W_i$, the right side of \eqref{equation:tilde-u-sum-bound} has expectation at most $8k^3 \sum_d E[Z_i(d)\one(Z_i(d) > c_n/(8k^3))] \to 0$, hence is $\op(1)$ by Markov for any $c_n \to \infty$.
Chaining through the conditional truncation bound $E[|u_\group|\one(|u_\group| > c_n) | \filtrationhd] \leq \tilde u_\group \one(\tilde u_\group > c_n)$ established above,
\[
\frac{1}{n}\sum_{\group \in \groupset_n} E\bigl[|u_\group|\one(|u_\group| > c_n) \bigm| \filtrationhd\bigr] \leq \frac{1}{n}\sum_{\group \in \groupset_n}\tilde u_\group \one(\tilde u_\group > c_n) = \op(1).
\]
Choosing $c_n = n^{1/4}$, Lemma \ref{lemma:lln} gives $n\inv \sum_{\group \in \groupset_n} u_\group = \op(1)$, i.e., $\wh v_1 - E[\wh v_1 | \filtrationhd] = \op(1)$.
The argument for $\wh v_2$ is similar, with $u_u = (2k(k-\propselect_u)/\propselect_u^2)(\estgone{u}\estgtwo{u} - E[\estgone{u}\estgtwo{u} | \filtrationhd])$ for $u \in \groupsetnu$ in place of $u_\group$.
Each $u_u$ is a function of $((\Di)_{i \in \group})_{\group \in \{\group_1(u), \group_2(u)\}}$ and $\filtrationhd$-measurable quantities, and since distinct unions are built from disjoint groups of $\groupset_n$, the joint conditional independence of $((\Di)_{i \in \group})_{\group \in \groupset_n}$ given $\filtrationhd$ from Lemma~\ref{lemma:group_aggregate_independence} yields $(u_u)_{u \in \groupsetnu}$ jointly conditionally independent given $\filtrationhd$.
The conditional truncation bound is verified as for $\wh v_1$, using $|\estgone{u}\estgtwo{u}| \leq \tfrac12(\estgone{u}^2 + \estgtwo{u}^2)$.
\end{proof}

\textbf{Rounding Procedure.}
Next, we discuss a specific rounding procedure that can be used to satisfy the assumptions on $\propselectestn(\psi)$ for the sampling design $\Tn \sim \localdesigncond(\psi, \propselectestn(\psi))$.
For an estimator $\propselectest(\psi)$ with $\propbound \leq \propselectest(\psi) \leq 1$, define the discretization $\propselectestn(\psi) = \lfloor k_n \propselectest(\psi) \rceil / k_n$, where $\lfloor x \rceil$ denotes the nearest integer to $x$.
Then $\propselectestn(\psi)$ takes values in the rational set $\{a/k_n : a = 0, 1, \dots, k_n\}$, with at most $\nlevels \leq k_n + 1$ distinct propensity levels and common group size $\kboundn = \max_{l \in [\nlevels]} \kl = k_n$.
The following lemma shows that this discretization scheme satisfies part (iv) of Assumption \ref{assumption:pilot-clt}.

\begin{lem}[Discretization by Rounding] \label{lemma:discretization-rounding}
Under the discretization scheme above, $\en[(\propselectestn(\psii) - \propselectest(\psii))^2] \leq 1/(4 k_n^2)$ and $\kboundn \nlevels = O(k_n^2)$.
In particular, if $k_n \to \infty$ and $k_n = o(n^{(1 - (\dimpsi+1)/\alpha)/2} \wedge n^{1/4})$, then part (iv) of Assumption \ref{assumption:pilot-clt} holds.
\end{lem}

\begin{proof}
By construction $\propselectestn(\psi) = a(\psi)/k_n$ with $a(\psi) = \lfloor k_n \propselectest(\psi) \rceil \in \{0, 1, \dots, k_n\}$, since $\propselectest(\psi) \in [\propbound, 1] \subseteq [0,1]$.
Hence $\nlevels \leq k_n + 1$ and $\kboundn = k_n$.
Rounding to the nearest integer satisfies $|\lfloor x \rceil - x| \leq 1/2$, so $\sup_\psi |\propselectestn(\psi) - \propselectest(\psi)| \leq 1/(2 k_n)$, and squaring gives $\en[(\propselectestn(\psii) - \propselectest(\psii))^2] \leq 1/(4 k_n^2)$.
Thus $\kboundn \nlevels \leq k_n(k_n + 1) = O(k_n^2)$.
For the final claim, $k_n = o(n^{(1 - (\dimpsi+1)/\alpha)/2})$ gives $\kboundn \nlevels = o(n^{1 - (\dimpsi+1)/\alpha})$, while $\kboundn = k_n = o(n^{1/4})$, so part (iv) of Assumption \ref{assumption:pilot-clt} holds.
\end{proof}

\subsection{Targeting the SATE} \label{appendix:sate-targeting}

Recall the motivating example in Section~\ref{section:method}, where \cite{chioda2026making} sample $\nsampled \approx 4{,}402$ Ugandan applicants into an experiment from a pool of $n = 7{,}431$ applicants to an entrepreneurship training program.
The superpopulation interpretation embeds these $n$ applicants into a larger distribution $P$ of potential applicants to such programs from which they are imagined to be drawn iid, with $\ate = E_P[Y(1)-Y(0)]$.

However, the sampling process into the eligible pool is not controlled by the researcher, and it's unclear exactly who the superpopulation $P$ represents in this case.
Because of this, the researcher may instead wish to target $\sate = \en[Y_i(1) - Y_i(0)]$ in the realized eligible population that we physically sample from.

What is the interpretation of the superpopulation model in this case?
To explain this, we note that our work in Corollary \ref{cor:clt-sate-extended} in fact implies a stronger conditional weak convergence statement $\rootn(\est - \sate) | \Wn \convwprocess \normal(0, V_{\sate})$ in the sense that
\[
\sup_{x \in \mr} \bigl| P\bigl(\rootn(\est - \sate) \leq x \bigm| \Wn\bigr) - \Phi\bigl(x / \sqrt{V_{\sate}}\bigr) \bigr| \convp 0.
\]

This is a quasi design-based result, since asymptotic normality is delivered by design randomness $\Tn \sim \localdesigncond(\psisamp, \propselect)$ and $\Dn \sim \localdesigncond(\psiassign, \propfn)$ alone.
The iid model for the data further constrains the form of the limiting variance, showing that it converges in probability to the simple analytical expression presented in the text.

Extension to a fully design-based model is beyond the scope of this paper.
See \cite{bai2026variance} for a more robust purely design-based model of finely stratified assignment and a related theoretical discussion.

\subsection{Empirical Application Details} \label{appendix:paper-details}

\begin{enumerate}[label={(\arabic*)}, itemindent=.5pt, itemsep=.4pt]
\item \cite{abebe2021} estimates the effect of an application incentive on the ability of applicants for clerical employment in Ethiopia.
We exclude the high-wage offer arm and let $D$ indicate the application incentive rather than control, giving $\propfn = 1/2$ within each randomization stratum.
We let $Y$ be the authors' index of cognitive ability and $\psi$ include gender, age, work experience in months, self-reported gpa, previous wage, and indicators for being born in Addis Ababa, speaking Amharic, and studying engineering.
\item \cite{baysan2022} estimates the effect of political information campaigns about concentration of executive power in Turkey on voter polarization.
Using ballot-box data aggregated to the neighborhood level, we let $Y_i$ be the share of ballots cast for ``No'' in the 2017 referendum and $D$ indicate either campaign, assigned to 100 of 550 neighborhoods ($\propfn = 2/11$).
We let $\psi$ include neighborhood-level average vote share for the CHP in the 2015 election, turnout, and the number of ballots collected.
\item \cite{bergquist2020} uses randomized demand subsidies to study competition among agricultural traders in Kenya.
We let $D=1$ indicate one of five discounts of 200--400 Ksh and $D=0$ be one of five discounts of 0--150 Ksh, giving $\propfn=1/2$ and matched pairs, and let $Y$ be the log change in purchased quantity following the discount.
We let $\psi$ include the negotiated baseline price and quantity, market road access and coordinates, and four baseline trader-ethnicity shares.
\item \cite{casey2021} estimates the effect of increased information given to political parties about voter preferences over candidates on whether the most voter-preferred candidate was selected to run by the party.
We let $D$ indicate the information package, assigned to 23 of 46 constituencies within each party ($\propfn = 1/2$), $Y$ whether the most preferred candidate was selected, and $\psi$ include race competitiveness and indices of candidate professional qualifications and public service motivation.
\item \cite{dellavigna2022} estimates the effect of employer gifts and other ``social preference'' related interventions on worker productivity.
In the first experiment, we code positive and in-kind gifts as $D=1$ and no and negative gifts as $D=0$.
Because the available treatments changed from three arms to four across sessions, we use $\propfn = 1/2$ for this pooled contrast.
We let $Y$ be total worker output in the final two working periods and $\psi$ include total productivity during the first 8 periods, gender, and age.
\item \cite{domurat2021} estimates the effect of informational interventions about Covered California insurance policies on takeup of insurance.
We code the no-direct-letter control arm 1 as $D=0$ and letter arms 2--5 as $D=1$.
Because the five arms were approximately equal in size, we set $\propfn = 4/5$.
We let $Y$ indicate insurance takeup and $\psi$ include household income, mean age, household subsidy size, and an indicator for being Latino.
\item \cite{fink2020} studies the effect of seasonal loans on agricultural production in Zambia.
We pool the cash- and maize-loan arms as $D=1$ against pure control ($D=0$), giving $\propfn=2/3$ and matched triples under the three-arm village randomization, and let $Y$ be log harvest value at reported prices in the first harvest survey.
We aggregate to the randomized village level and let $\psi$ include baseline cash and grain reserves, harvest value, maize and cash-crop acreage, farm workers, prior agricultural labor participation, asset quintile, livestock value, and input value.
\item \cite{hussam2022} estimates the value of employment on measures of psychosocial wellbeing.
Treatment is assigned at the refugee-camp block level, with five individuals sampled per block.
We exclude the 33 large-cash blocks and compare 83 employment blocks ($D=1$) with 33 small-cash blocks ($D=0$), approximating $83/116$ by $\propfn = 3/4$.
After averaging within blocks, $Y$ is the endline mental health index and $\psi$ includes baseline mental health, gender, whether a family member was killed, and sociability.
\item \cite{lowe2021} studies how collaborative and adversarial intergroup contact affects cross-caste friendships in cricket leagues.
We let $D$ indicate assignment to one of the 13 mixed-caste teams in each 20-team league, approximating $13/20$ by $\propfn = 2/3$.
We let $Y$ be the number of other-caste friends at endline and $\psi$ include the number at baseline, age, and cricket ability.
\end{enumerate}

\subsection{Empirical Results for the SATE} \label{appendix:empirical-sate}

This section presents our empirical results for $\sate$ estimation and inference, using the variance estimator $\varest_{\sate}$ from Section~\ref{section:inference} on the same nine DGPs and seven designs as in Section~\ref{section:empirical}.
Table~\ref{table:empirical-sate} reports percent changes in the standard deviation of $\est - \sate$ and CI length relative to complete randomization, and empirical coverage of $\sate$.

% Generated deterministically by the replication package.
% Design: jittered DGP; pool=4n; budget=0.75; constant q=1/4.
% Target estimand: $\sate$ (col_sd='est_sd_around_sate', col_ci='ci_radius_sate', col_cov='coverage_sate').
% Top panel reports %\Delta SD (vs CR baseline) so the SD/CI rows are in the same units;
% the gap between |%\Delta SD| and |%\Delta CI| is a visual diagnostic for over-conservative inference.
\begin{table}[htbp]
\begin{adjustbox}{width=0.9\columnwidth,center}
  \centering
    \begin{tabular}{rrrrrrrrrrr}
\cmidrule{2-11}          & \multicolumn{1}{c}{Design, Paper} & \multicolumn{1}{c}{Abe.} & \multicolumn{1}{c}{Bay.} & \multicolumn{1}{c}{Ber.} & \multicolumn{1}{c}{Cas.} & \multicolumn{1}{c}{Del.} & \multicolumn{1}{c}{Dom.} & \multicolumn{1}{c}{Fin.} & \multicolumn{1}{c}{Hus.} & \multicolumn{1}{c}{Low.} \\
\cmidrule{2-11}    \multicolumn{1}{c}{\multirow{7}[2]{*}{$\%\Delta$SD}} & \multicolumn{1}{c|}{CR} & \multicolumn{1}{c}{0} & \multicolumn{1}{c}{0} & \multicolumn{1}{c}{0} & \multicolumn{1}{c}{0} & \multicolumn{1}{c}{0} & \multicolumn{1}{c}{0} & \multicolumn{1}{c}{0} & \multicolumn{1}{c}{0} & \multicolumn{1}{c}{0} \\
          & \multicolumn{1}{c|}{CR, Loc} & \multicolumn{1}{c}{-8} & \multicolumn{1}{c}{-46} & \multicolumn{1}{c}{-4} & \multicolumn{1}{c}{0} & \multicolumn{1}{c}{-41} & \multicolumn{1}{c}{0} & \multicolumn{1}{c}{-19} & \multicolumn{1}{c}{-6} & \multicolumn{1}{c}{-13} \\
          & \multicolumn{1}{c|}{Loc} & \multicolumn{1}{c}{-9} & \multicolumn{1}{c}{-46} & \multicolumn{1}{c}{-1} & \multicolumn{1}{c}{0} & \multicolumn{1}{c}{-41} & \multicolumn{1}{c}{1} & \multicolumn{1}{c}{-20} & \multicolumn{1}{c}{-7} & \multicolumn{1}{c}{-13} \\
          & \multicolumn{1}{c|}{Hom.} & \multicolumn{1}{c}{-15} & \multicolumn{1}{c}{-45} & \multicolumn{1}{c}{-15} & \multicolumn{1}{c}{-9} & \multicolumn{1}{c}{-43} & \multicolumn{1}{c}{-5} & \multicolumn{1}{c}{-25} & \multicolumn{1}{c}{-14} & \multicolumn{1}{c}{-11} \\
          & \multicolumn{1}{c|}{Pilot S} & \multicolumn{1}{c}{-15} & \multicolumn{1}{c}{-43} & \multicolumn{1}{c}{-21} & \multicolumn{1}{c}{-9} & \multicolumn{1}{c}{-41} & \multicolumn{1}{c}{2} & \multicolumn{1}{c}{-25} & \multicolumn{1}{c}{-13} & \multicolumn{1}{c}{-19} \\
          & \multicolumn{1}{c|}{Pilot L} & \multicolumn{1}{c}{-14} & \multicolumn{1}{c}{-45} & \multicolumn{1}{c}{-23} & \multicolumn{1}{c}{-9} & \multicolumn{1}{c}{-43} & \multicolumn{1}{c}{-4} & \multicolumn{1}{c}{-24} & \multicolumn{1}{c}{-11} & \multicolumn{1}{c}{-20} \\
          & \multicolumn{1}{c|}{Obs} & \multicolumn{1}{c}{-15} & \multicolumn{1}{c}{-45} & \multicolumn{1}{c}{-21} & \multicolumn{1}{c}{-10} & \multicolumn{1}{c}{-42} & \multicolumn{1}{c}{-4} & \multicolumn{1}{c}{-24} & \multicolumn{1}{c}{-14} & \multicolumn{1}{c}{-20} \\
\cmidrule{2-11}    \multicolumn{1}{c}{\multirow{7}[2]{*}{$\%\Delta$CI}} & \multicolumn{1}{c|}{CR} & \multicolumn{1}{c}{0} & \multicolumn{1}{c}{0} & \multicolumn{1}{c}{0} & \multicolumn{1}{c}{0} & \multicolumn{1}{c}{0} & \multicolumn{1}{c}{0} & \multicolumn{1}{c}{0} & \multicolumn{1}{c}{0} & \multicolumn{1}{c}{0} \\
          & \multicolumn{1}{c|}{CR, Loc} & \multicolumn{1}{c}{-6} & \multicolumn{1}{c}{-45} & \multicolumn{1}{c}{0} & \multicolumn{1}{c}{-1} & \multicolumn{1}{c}{-38} & \multicolumn{1}{c}{-1} & \multicolumn{1}{c}{-15} & \multicolumn{1}{c}{0} & \multicolumn{1}{c}{-12} \\
          & \multicolumn{1}{c|}{Loc} & \multicolumn{1}{c}{-6} & \multicolumn{1}{c}{-45} & \multicolumn{1}{c}{0} & \multicolumn{1}{c}{-1} & \multicolumn{1}{c}{-38} & \multicolumn{1}{c}{-1} & \multicolumn{1}{c}{-14} & \multicolumn{1}{c}{0} & \multicolumn{1}{c}{-11} \\
          & \multicolumn{1}{c|}{Hom.} & \multicolumn{1}{c}{-13} & \multicolumn{1}{c}{-45} & \multicolumn{1}{c}{-12} & \multicolumn{1}{c}{-7} & \multicolumn{1}{c}{-39} & \multicolumn{1}{c}{-5} & \multicolumn{1}{c}{-19} & \multicolumn{1}{c}{-2} & \multicolumn{1}{c}{-10} \\
          & \multicolumn{1}{c|}{Pilot S} & \multicolumn{1}{c}{-12} & \multicolumn{1}{c}{-41} & \multicolumn{1}{c}{-16} & \multicolumn{1}{c}{-7} & \multicolumn{1}{c}{-38} & \multicolumn{1}{c}{1} & \multicolumn{1}{c}{-19} & \multicolumn{1}{c}{-2} & \multicolumn{1}{c}{-16} \\
          & \multicolumn{1}{c|}{Pilot L} & \multicolumn{1}{c}{-12} & \multicolumn{1}{c}{-43} & \multicolumn{1}{c}{-18} & \multicolumn{1}{c}{-6} & \multicolumn{1}{c}{-39} & \multicolumn{1}{c}{-4} & \multicolumn{1}{c}{-18} & \multicolumn{1}{c}{-2} & \multicolumn{1}{c}{-17} \\
          & \multicolumn{1}{c|}{Obs} & \multicolumn{1}{c}{-13} & \multicolumn{1}{c}{-42} & \multicolumn{1}{c}{-16} & \multicolumn{1}{c}{-7} & \multicolumn{1}{c}{-39} & \multicolumn{1}{c}{-5} & \multicolumn{1}{c}{-21} & \multicolumn{1}{c}{-2} & \multicolumn{1}{c}{-17} \\
\cmidrule{2-11}    \multicolumn{1}{c}{\multirow{7}[2]{*}{Cover}} & \multicolumn{1}{c|}{CR} & \multicolumn{1}{c}{0.96} & \multicolumn{1}{c}{0.95} & \multicolumn{1}{c}{0.96} & \multicolumn{1}{c}{0.95} & \multicolumn{1}{c}{0.95} & \multicolumn{1}{c}{0.96} & \multicolumn{1}{c}{0.95} & \multicolumn{1}{c}{0.94} & \multicolumn{1}{c}{0.95} \\
          & \multicolumn{1}{c|}{CR, Loc} & \multicolumn{1}{c}{0.97} & \multicolumn{1}{c}{0.95} & \multicolumn{1}{c}{0.97} & \multicolumn{1}{c}{0.95} & \multicolumn{1}{c}{0.96} & \multicolumn{1}{c}{0.95} & \multicolumn{1}{c}{0.96} & \multicolumn{1}{c}{0.95} & \multicolumn{1}{c}{0.96} \\
          & \multicolumn{1}{c|}{Loc} & \multicolumn{1}{c}{0.96} & \multicolumn{1}{c}{0.95} & \multicolumn{1}{c}{0.96} & \multicolumn{1}{c}{0.96} & \multicolumn{1}{c}{0.96} & \multicolumn{1}{c}{0.96} & \multicolumn{1}{c}{0.96} & \multicolumn{1}{c}{0.96} & \multicolumn{1}{c}{0.96} \\
          & \multicolumn{1}{c|}{Hom.} & \multicolumn{1}{c}{0.96} & \multicolumn{1}{c}{0.94} & \multicolumn{1}{c}{0.97} & \multicolumn{1}{c}{0.95} & \multicolumn{1}{c}{0.96} & \multicolumn{1}{c}{0.95} & \multicolumn{1}{c}{0.96} & \multicolumn{1}{c}{0.96} & \multicolumn{1}{c}{0.95} \\
          & \multicolumn{1}{c|}{Pilot S} & \multicolumn{1}{c}{0.97} & \multicolumn{1}{c}{0.95} & \multicolumn{1}{c}{0.97} & \multicolumn{1}{c}{0.96} & \multicolumn{1}{c}{0.96} & \multicolumn{1}{c}{0.95} & \multicolumn{1}{c}{0.96} & \multicolumn{1}{c}{0.96} & \multicolumn{1}{c}{0.96} \\
          & \multicolumn{1}{c|}{Pilot L} & \multicolumn{1}{c}{0.96} & \multicolumn{1}{c}{0.95} & \multicolumn{1}{c}{0.97} & \multicolumn{1}{c}{0.96} & \multicolumn{1}{c}{0.97} & \multicolumn{1}{c}{0.95} & \multicolumn{1}{c}{0.96} & \multicolumn{1}{c}{0.96} & \multicolumn{1}{c}{0.96} \\
          & \multicolumn{1}{c|}{Obs} & \multicolumn{1}{c}{0.97} & \multicolumn{1}{c}{0.95} & \multicolumn{1}{c}{0.97} & \multicolumn{1}{c}{0.95} & \multicolumn{1}{c}{0.96} & \multicolumn{1}{c}{0.95} & \multicolumn{1}{c}{0.96} & \multicolumn{1}{c}{0.96} & \multicolumn{1}{c}{0.96} \\
\cmidrule{2-11}          & \multicolumn{1}{c|}{$\nsampled$} & \multicolumn{1}{c}{1451} & \multicolumn{1}{c}{550} & \multicolumn{1}{c}{1359} & \multicolumn{1}{c}{91} & \multicolumn{1}{c}{446} & \multicolumn{1}{c}{1000} & \multicolumn{1}{c}{175} & \multicolumn{1}{c}{116} & \multicolumn{1}{c}{770} \\
          & \multicolumn{1}{c|}{$\dim(\psi)$} & \multicolumn{1}{c}{8} & \multicolumn{1}{c}{3} & \multicolumn{1}{c}{9} & \multicolumn{1}{c}{3} & \multicolumn{1}{c}{3} & \multicolumn{1}{c}{4} & \multicolumn{1}{c}{9} & \multicolumn{1}{c}{4} & \multicolumn{1}{c}{3} \\
\cmidrule{2-11}          &       &       &       &       &       &       &       &       &       &  \\
    \end{tabular}%
    \end{adjustbox}
    \caption{Empirical Results (SATE).}
    \label{table:empirical-sate}%
\end{table}%

Coverage of $\sate$ is close to nominal across all 63 (paper, design) cells, in line with the asymptotically valid inference guarantee from Theorem~\ref{thm:inference}.
By Corollary~\ref{cor:clt4sate} and Theorem~\ref{thm:inference}, the probability limits of $\varest$ and $\varest_{\sate}$ differ by $\var(\catefn(\psi))$.
This difference is small in our DGPs, where the nonparametric $R^2$ of $\psi$ for treatment effects averages only $4\%$.
Accordingly, the $\sate$ confidence intervals are only $0.05\%$ shorter than their $\ate$ counterparts on average.

\end{document}